\newif\ifpublic\publicfalse
\publictrue

\documentclass[11pt,a4paper]{article}

\usepackage[a4paper,text={160mm,247mm},centering]{geometry}
\usepackage{layout}

\usepackage{setspace}
\usepackage[utf8]{inputenc}

\usepackage[english]{babel}
\usepackage{hyphenat}
\usepackage[nosort]{cite}

\newcommand{\rcite}[2][]{%
  ref.~%
  \ifx\relax#1\relax
    \cite{#2}%
  \else
    \cite[#1]{#2}%
  \fi
}
\newcommand{\rcites}[1]{refs.~\cite{#1}}

\usepackage{amsmath, amssymb, amsfonts, amsthm}
\usepackage{mathtools}
\usepackage[normalem]{ulem}

\theoremstyle{plain}
\newtheorem{theorem}{Theorem}
\newtheorem*{theorem*}{Theorem}
\newtheorem{proposition}[theorem]{Proposition}
\newtheorem{propdef}[theorem]{Definition-Proposition}

\newtheorem{lemma}[theorem]{Lemma}
\newtheorem{definition}[theorem]{Definition}

\newtheorem{conjecture}[theorem]{Conjecture}

\usepackage{mathfixs}
\ProvideMathFix{autobold}
\ProvideMathFix{multskip}
\ProvideMathFix{frac,root,rootclose}

\allowdisplaybreaks[1]

\makeatletter
\def\mr@ignsp#1 {\ifx\:#1\@empty\else #1\expandafter\mr@ignsp\fi}%
\newcommand{\multiref}[1]{\begingroup
\xdef\mr@no@sparg{\expandafter\mr@ignsp#1 \: }%
\def\mr@comma{}%
\@for\mr@refs:=\mr@no@sparg\do{\mr@comma\def\mr@comma{,}\ref{\mr@refs}}%
\endgroup}
\renewcommand{\eqref}[1]{(\multiref{#1})}
\makeatother

\makeatletter
\newcommand{\namedref}[2]{\hyperref[#2]{#1~\ref*{#2}}}%
\newcommand{\namedreff}[2]{\hyperref[#2]{#1\,\ref*{#2}}}%
\newcommand{\secref}{\namedreff{\S}}
\newcommand{\subsecref}{\namedreff{\S}}
\newcommand{\appref}{\namedref{Appendix}}
\newcommand{\tabref}{\namedref{Table}}
\newcommand{\figref}{\namedref{Figure}}
\newcommand{\thmref}{\namedref{Theorem}}
\newcommand{\conref}{\namedref{Conjecture}}
\newcommand{\defref}{\namedref{Definition}}

\newcommand{\propref}{\namedref{Proposition}}

\newcommand{\lemref}{\namedref{Lemma}}
\newcommand{\sitref}{\namedref{Situation}}
\makeatother

\renewcommand{\theequation}{\thesection.\arabic{equation}}
\numberwithin{equation}{section}

\newcommand{\eqn}[1]{eq.~\eqref{#1}}

\newcommand{\eqns}[2]{eqs.~\eqref{#1} and~\eqref{#2}}

\providecommand{\href}[2]{#2}

\usepackage[font=small,labelfont=bf]{caption}

\usepackage{graphbox}

\usepackage[compile=false,fonts]{mpostinl}
\begin{mposttex}[]
\usepackage{amsmath,amssymb,amsfonts}
\end{mposttex}
\begin{mpostdef}
	input metaobj;
	
	def drawzigzag(expr pzz)=
	pair zz[];
	zz[1]:= point 0 of pzz ; zz[2]:= point 1 of pzz;
	nczigzag(zz[1])(zz[2]) "coilwidth(0.15xu)", "linewidth(1pt)", "linearc(.01xu)", "arrows(-)", "coilarmA(0)", "coilarmB(0)";
	enddef;
\end{mpostdef}

\begin{mpostdef}
	pair vpos[];
	pair epos[];
	pair ext[];
	pair exta[];
	path paths[];
	picture pic;
	picture pic[];
	picture savepic;
	path cpath;
	path hpath;
	xu:=1cm;
	yu:=1cm;
\end{mpostdef}

\begin{mpostdef}
	def pensize(expr s)=withpen pencircle scaled s enddef;
	def fillshape(expr p,ci,tb,cb)=
	fill p withcolor ci;
	draw p pensize(tb) withcolor cb;
	enddef;
	def filldot(expr z,s,ci)=
	fillshape(fullcircle scaled s shifted z, ci, 0.5pt, 0.0white);
	enddef;
	def filltrig(expr z,s,r,ci)=
	fillshape(((dir 60)--(dir 180)--(dir 300)--cycle) scaled s rotated r shifted z, ci, 0.5pt, 0.0white);
	enddef;
	def fillsqr(expr z,s,r,ci)=
	fillshape(((dir 45)--(dir 135)--(dir 225)--(dir 315)--cycle) scaled s rotated r shifted z, ci, 0.5pt, 0.0white);
	enddef;
	def drawcross(expr z,s,r,t,c)=
	draw ((-0.5,-0.5)--(+0.5,+0.5)) scaled s rotated r shifted z pensize(t) withcolor c;
	draw ((+0.5,-0.5)--(-0.5,+0.5)) scaled s rotated r shifted z pensize(t) withcolor c;
	enddef;
	
	def midarrow (expr p, t) =
	fill (arrowhead subpath(0,arctime(arclength(subpath (0,t) of p)+0.5ahlength) of p) of p) withcolor ecolor;
	enddef;
	
	def midarrowred (expr p, t) =
	fill (arrowhead subpath(0,arctime(arclength(subpath (0,t) of p)+0.5ahlength) of p) of p) withcolor ecolor;
	enddef;
	
	def midarrowblue (expr p, t) =
	fill (arrowhead subpath(0,arctime(arclength(subpath (0,t) of p)+0.5ahlength) of p) of p) withcolor bcolor;
	enddef;
	
	def dprop expr p=draw p pensize(3pt) withcolor feyncolor enddef;
	def dpropf expr p=draw p pensize(3pt) withcolor feyncolorfaint enddef;
	
	def drawprop expr p=draw p pensize(1.5pt) enddef;
	def drawline expr p=draw p pensize(0.5pt) enddef;
	def drawpos expr p=drawcross(p, 5pt, 0, 1.0pt, 0.0white) enddef;
	def drawvertex expr p=filldot(p, 5pt, 0.5white+0.5red) enddef;
	def drawblackdotscale expr p=filldot(p, 0.4*xu, black) enddef;
	def drawreddotscale expr p=filldot(p, 0.4*xu, red) enddef;
	def drawwhitedotscale expr p=filldot(p, 0.4*xu, white) enddef;
	def drawblackdot expr p=filldot(p, 8pt, black) enddef;
	def drawwhitedot expr p=filldot(p, 8pt, white) enddef;
	def drawlinearrow expr p=draw p pensize(0.5pt); midarrow (p,0.5); enddef;
	def drawproparrow expr p=draw p pensize(1.5pt); midarrow (p,0.5); enddef;
	def drawproparrowdouble expr p=draw p pensize(1.5pt); midarrow (p,0.65); midarrow (reverse p,0.65); enddef;
	def drawwhitedotmr (expr p, sz)=filldot(p, sz*xu, white) enddef;
	def drawblackdotmr (expr p, sz)=filldot(p, sz*xu, black) enddef;
	
	def framed(expr p)=fill bbox p withcolor white; draw bbox p pensize(0.8pt); draw p; enddef;
	def framedd(expr p)=fill bbox p withcolor 0.94white; draw bbox p pensize(1pt); draw p; enddef;
\end{mpostdef}
\begin{mpostdef}
	color ecolor; ecolor:=(.72,.03,.06);
	color bcolor; bcolor:=(.07,.05,.66);
	color gcolor; gcolor:=(.0,.44,.0);
	color feyncolor; feyncolor:=(.06,.4,.01);
	color feyncolorfaint; feyncolorfaint:=(.12,.8,.02);
	color arrowcolor; arrowcolor:=(0,0,0);
	color faint; faint:=(.7,0.7,0.7);
	color fgray; fgray:=(.7,0.7,0.7);
	vsize := 0.16xu;
	rad:= 2xu;
	cpath:= for i=0 upto 35: 
	dir (i*10) scaled rad ..
	endfor cycle;
	def vert(expr pos)=
	fill fullcircle scaled vsize shifted point pos of cpath withcolor white;
	draw halfcircle scaled vsize rotated (90+10*pos) shifted point pos of cpath withcolor ecolor pensize(1.3pt);
	enddef;
	def vertthick(expr pos)=
	fill fullcircle scaled vsize shifted point pos of cpath withcolor white;
	draw halfcircle scaled vsize rotated (90+10*pos) shifted point pos of cpath withcolor ecolor pensize(2.0pt);
	enddef;
	def vertcol(expr pos,col)=
	fill fullcircle scaled vsize shifted point pos of cpath withcolor white;
	draw halfcircle scaled vsize rotated (90+10*pos) shifted point pos of cpath withcolor col pensize(1.3pt);
	enddef;
	def drawbdry(expr startpos,endpos)=
	draw subpath(startpos,endpos) of cpath withcolor ecolor pensize(1.3pt);
	enddef;
	def drawbdryblack(expr startpos,endpos)=
	draw subpath(startpos,endpos) of cpath withcolor black pensize(1.3pt);
	enddef;
	def drawbdrythick(expr startpos,endpos)=
	draw subpath(startpos,endpos) of cpath withcolor ecolor pensize(2.0pt);
	enddef;
	def drawbdrythickblack(expr startpos,endpos)=
	draw subpath(startpos,endpos) of cpath withcolor black pensize(2.0pt);
	enddef;
	def drawbdrydashed(expr startpos,endpos)=
	draw subpath(startpos,endpos) of cpath withcolor ecolor dashed withdots scaled 0.5 pensize(1.3pt);
	enddef;
	def drawbdrydashedthick(expr startpos,endpos)=
	draw subpath(startpos,endpos) of cpath withcolor ecolor dashed withdots scaled 0.5 pensize(2.0pt);
	enddef;
	def vlabel(expr pos, vscal, tt)=
	vert(pos);
	label(tt,point pos of cpath scaled vscal);
	enddef;
	
	def vlabelcol(expr pos, vscal, tcol, vcol, tt)=
	vertcol(pos, vcol);
	label(tt,point pos of cpath scaled vscal) withcolor tcol;
	enddef;
	
	def ppath(expr p,spos,epos)=
	subpath(arctime(spos*arclength(p))of p,arctime(epos*arclength(p)) of p) of p
	enddef;
\end{mpostdef}

\begin{mpostdef}
	picture nicearrow; currentpicture:= nullpicture;
	ahlengthsave:=ahlength; ahlength:=12pt;
	drawarrow (0xu,0xu){dir 20}..(2xu,0xu) pensize(3pt) withcolor arrowcolor;
	ahlength:=ahlengthsave;
	nicearrow:= currentpicture; currentpicture:= nullpicture;
\end{mpostdef}

\begin{mpostdef}
	eNum:= 2.718281828459045235;
	numeric startangle, lsdist, scalefactor, stepangle;
	path ls;
	
	vardef logspiralder(expr lspos, endpoint, maxangle, stepnumber, enddir) =
	startangle := ((angle(endpoint-lspos) - (maxangle mod 360)) mod 360); 
	stepangle := (1.0maxangle)/(1.0stepnumber);
	lsdist := length(endpoint-lspos); 
	scalefactor := mlog(lsdist+1.0)/(1.0*maxangle)/256;
	ls:=lspos for t = startangle step stepangle until (maxangle + startangle + 0.02): .. (((mexp(256((t-startangle)*scalefactor))-1.0) * cosd(t),(mexp(256((t-startangle)*scalefactor))-1.0)  * sind(t))+lspos) endfor;
	ls:= subpath(1,(length ls) - 1) of ls;
	ls:= ls .. {dir enddir}endpoint;
	ls
	enddef;
\end{mpostdef}

\begin{mpostdef}
def roundedrectangle(expr xa, ya, xb, yb, r) =
  (xa + r, ya) 
  -- (xb - r, ya) {dir 0} .. {dir 90} (xb, ya + r) 
  -- (xb, yb - r) {dir 90} .. {dir 180} (xb - r, yb) 
  -- (xa + r, yb) {dir 180} .. {dir 270} (xa, yb - r) 
  -- (xa, ya + r) {dir 270} .. {dir 0} cycle
enddef;

def SchottkyBackground(expr xoffset, yoffset, w, h) =
fregion := 0.2xu;
fsteps := 80;

for ii = 0 upto fsteps:
  numeric margin, g;
  margin := ii / fsteps * fregion;
  gcol := 1.0 - 0.3 * (ii / fsteps); 

  draw roundedrectangle(margin+xoffset, margin+yoffset,w - margin+xoffset, h-margin+yoffset, fregion * (1-ii / fsteps) ) withcolor (gcol, gcol, gcol) pensize( fregion / (fsteps - 1));
endfor
  fill roundedrectangle(margin+xoffset, margin+yoffset,w - margin+xoffset, h-margin+yoffset, 0 ) withcolor (0.7,0.7,0.7);
enddef;

def SchottkyAxes(expr xstartpos, xendpos, ystartpos, yendpos) = 
  drawarrow xstartpos-- xendpos pensize (1.0pt);  
  drawarrow ystartpos -- yendpos pensize (1.0pt);  
enddef;

def SchottkyCircle(expr centerPos, rad, col, bcol) = 
  fill fullcircle scaled (2*rad) shifted centerPos withcolor bcol; 
  draw fullcircle scaled (2*rad) shifted centerPos withcolor col pensize(0.7pt);
enddef;

def SchottkyCircleNoFilling(expr centerPos, rad, col) = 
  draw fullcircle scaled (2*rad) shifted centerPos withcolor col pensize(0.7pt);
enddef;

def SchottkyACircle(expr centerPos, rad, col, bcol) = 
  fill fullcircle scaled (2*rad) shifted centerPos withcolor bcol;
  drawarrow reverse(fullcircle) scaled (2*rad) shifted centerPos withcolor col pensize(0.7pt);
enddef;

def SchottkyCircleDashed(expr centerPos, rad, col, bcol) = 
	fill fullcircle scaled (2*rad) shifted centerPos withcolor bcol;
	draw fullcircle scaled (2*rad) shifted centerPos withcolor col pensize(0.7pt) dashed evenly;
enddef;

def SchottkyCircleDashedNoFilling(expr centerPos, rad, col) = 
	draw fullcircle scaled (2*rad) shifted centerPos withcolor col pensize(0.7pt) dashed evenly;
enddef;

def SchottkyCross(expr PosA, Sangle, col) = 
  transform ST;
  ST := identity rotated (Sangle + 45 + 15) shifted PosA; 
  draw ((-0.05xu,0xu)--(0.05xu,0xu)) transformed ST withcolor col;
  draw ((0,-0.05xu)--(0,0.05xu)) transformed ST withcolor col;
enddef;

def SchottkyCrossS(expr PosA, Sangle, size, col) = 
  transform ST;
  ST := identity rotated (Sangle + 45 + 15) shifted PosA; 
  draw ((-size,0xu)--(size,0xu)) transformed ST withcolor col;
  draw ((0,-size)--(0,size)) transformed ST withcolor col;
enddef;

def SchottkySquare(expr PosA, Sangle) = 
	transform ST;
	ST := identity rotated (Sangle + 45 + 15) shifted PosA; 
	draw ((-0.03xu,-0.03xu)--(0.03xu,-0.03xu)--(0.03xu,0.03xu)--(-0.03xu,0.03xu)--(-0.03xu,-0.03xu)) transformed ST;
enddef;

def SchottkyPathArrow(expr PosA,PosB) =
  path SPA;
  numeric SPAdir;
  SPAdir := angle(PosB-PosA);
  SPA := PosA {dir (SPAdir-15)} .. {dir (SPAdir+15)} PosB;
  drawarrow subpath (0.04 * length SPA, 0.96 * length SPA) of SPA pensize(0.7pt);
  SchottkyCross (PosB, SPAdir, (0,0,0));
enddef;

def SchottkyPathArrowAlt(expr PosA,PosB) =
  path SPA;
  numeric SPAdir;
  SPAdir := angle(PosB-PosA);
  SPA := PosA {dir (SPAdir+15)} .. {dir (SPAdir-15)} PosB;
  drawarrow subpath (0.04 * length SPA, 0.96 * length SPA) of SPA pensize(0.7pt);
  SchottkyCross (PosB, SPAdir, (0,0,0));
enddef;
\end{mpostdef}

\begin{mpostdef}
	path quartercircle;
	quartercircle := subpath (0, 2) of fullcircle;
\end{mpostdef}

\begin{mpostfig}[label=SchottkySetup]
SchottkyCircleDashed((2xu,1.6xu),1.1xu,ecolor, fgray);
SchottkyCircle((-0.8xu,6xu),0.5xu,bcolor, fgray);
SchottkyCircleDashed((2xu,5.1xu),1.4xu,gcolor, fgray);
SchottkyCircleDashed((2.3xu,5.7xu),0.5xu,gcolor, fgray);
SchottkyCircle((-2.5xu,1.7xu),0.7xu,ecolor, fgray);
SchottkyCircleDashed((-0.4xu,3.4xu),0.9xu,bcolor, fgray);
SchottkyCircle((-3xu,4.7xu),1.3xu,gcolor, fgray);
SchottkyCircle((-2.5xu,4.1xu),0.3xu,gcolor, fgray);
SchottkyPathArrow((-0.7xu,1.5xu),(1.8xu,1.3xu));
SchottkyPathArrow((1.8xu,1.3xu),(-0.2xu,3.9xu));
SchottkyPathArrow((-0.2xu,3.9xu),(2xu,4.3xu));
SchottkyPathArrow((2xu,4.3xu),(2.5xu,5.9xu));
SchottkyPathArrow((2.5xu,5.9xu),(-1.0xu,5.9xu));
SchottkyPathArrow((-1.0xu,5.9xu),(-3.5xu,4.4xu));
SchottkyPathArrow((-3.5xu,4.4xu),(-2.5xu,4.1xu));
SchottkyPathArrow((-2.5xu,4.1xu),(-0.6xu,3.3xu));
numeric p;
p = 0.2xu;
SchottkyCross((-0.7xu,1.5xu), 0, black);
label(btex $z_0$ etex, (-0.7xu,1.5xu)+(-p,p));
numeric s;
s = 0.9;
numeric r;
pair c;
r := 1.1xu;
c := (2xu,1.6xu);
label(btex $C_1'$ etex, c+s*(r,-r));
r := 0.5xu;
c := (-0.8xu,6xu);
label(btex $C_3$ etex, c+1.1*(r,-r));
r := 1.4xu;
c := (2xu,5.1xu);
label(btex $C_2'$ etex, c+s*(r,-r));
r := 0.5xu;
c := (2.3xu,5.7xu);
label(btex $\sigma_2C_2'$ etex, c+1.2*(r,-r));
r := 0.7xu;
c := (-2.5xu,1.7xu);
label(btex $C_1$ etex, c+1.1*(r,-r));
r := 0.9xu;
c := (-0.4xu,3.4xu);
label(btex $C_3'$ etex, c+s*(r,-r));
r := 1.3xu;
c := (-3xu,4.7xu);
label(btex $C_2$ etex, c+s*(-r,-r));
r := 0.3xu;
c := (-2.5xu,4.1xu);
label(btex $\sigma_2^{-1}C_2$ etex, c+2.0*(r-0.2xu,r));
\end{mpostfig}

\begin{mpostfig}[label=hyperelliptic]
pair c[], cp[], p[], pp[], labelShift;
c1 = (1.3xu,3xu);
cp1 = (4.7xu,3xu);
p1 = (1.35xu,3xu);
pp1 = (4.65xu,3xu);
c2 = (3xu+0.5*1.7xu,3xu+0.866*1.7xu);
cp2 = (3xu-0.5*1.7xu,3xu-0.866*1.7xu);
p2 = (3xu+0.5*1.9xu,3xu+0.866*1.9xu);
pp2 = (3xu-0.5*1.9xu,3xu-0.866*1.9xu);
c3 = (3xu-0.5*1.7xu,3xu+0.866*1.7xu);
cp3 = (3xu+0.5*1.7xu,3xu-0.866*1.7xu);
p3 = (3xu-0.5*1.5xu,3xu+0.866*1.5xu);
pp3 = (3xu+0.5*1.5xu,3xu-0.866*1.5xu);
labelShift = (0.2xu, 0.2xu);
SchottkyAxes((0xu,3xu),(6xu,3xu),(3xu,0.5xu),(3xu,5.5xu));
SchottkyCircle(c1,0.6xu,ecolor, fgray);
label(btex $C_1$ etex, c1 + (0.65xu, -0.65xu));
SchottkyCircleDashed(cp1,0.6xu,ecolor, fgray);
label(btex $C_1^\prime$ etex, cp1 + (0.7xu, -0.7xu));
SchottkyCircle(c2,0.5xu,bcolor, fgray);
label(btex $C_2$ etex, c2 + (0.6xu, -0.6xu));
SchottkyCircleDashed(cp2,0.5xu,bcolor, fgray);
label(btex $C_2^\prime$ etex, cp2 + (0.6xu, -0.6xu));
SchottkyCircle(c3,0.7xu,gcolor, fgray);
label(btex $C_3$ etex, c3 + (0.6xu, -0.75xu));
SchottkyCircleDashed(cp3,0.7xu,gcolor, fgray);
label(btex $C_3^\prime$ etex, cp3 + (0.7xu, -0.7xu));
SchottkyCross(p1, 0, black);
label(btex $P_1$ etex, p1 + labelShift);
SchottkyCross(pp1, 0, black);
label(btex $P_1^\prime$ etex, pp1 + labelShift);
SchottkyCross(p2, 0, black);
label(btex $P_2$ etex, p2 - labelShift);
SchottkyCross(pp2, 0, black);
label(btex $P_2^\prime$ etex, pp2 + labelShift);
SchottkyCross(p3, 0, black);
label(btex $P_3$ etex, p3 + labelShift);
SchottkyCross(pp3, 0, black);
label(btex $P_3^\prime$ etex, pp3 + labelShift);
label(btex $0$ etex, (3.20xu, 3.20xu));

drawarrow (6.5xu,3xu){dir 20}..(8.5xu,3xu) pensize(1pt);
label(btex $\theta$ etex, (7.5xu, 3.5xu));

pair offset;
offset = (9xu, 0xu);
SchottkyAxes((0xu,3xu)+offset,(6xu,3xu)+offset,(3xu,0.5xu)+offset,(3xu,5.5xu)+offset);
SchottkyCircleDashed(c1+offset,0.6xu,ecolor, fgray);
label(btex $C_1^\prime$ etex, c1+offset + (0.65xu, -0.65xu));
SchottkyCircle(cp1+offset,0.6xu,ecolor, fgray);
label(btex $C_1$ etex, cp1+offset + (0.7xu, -0.7xu));
SchottkyCircleDashed(c2+offset,0.5xu,bcolor, fgray);
label(btex $C_2^\prime$ etex, c2+offset + (0.6xu, -0.6xu));
SchottkyCircle(cp2+offset,0.5xu,bcolor, fgray);
label(btex $C_2$ etex, cp2+offset + (0.6xu, -0.6xu));
SchottkyCircleDashed(c3+offset,0.7xu,gcolor, fgray);
label(btex $C_3^\prime$ etex, c3+offset + (0.6xu, -0.75xu));
SchottkyCircle(cp3+offset,0.7xu,gcolor, fgray);
label(btex $C_3$ etex, cp3+offset + (0.7xu, -0.7xu));
SchottkyCross(p1+offset, 0, black);
label(btex $P_1^\prime$ etex, p1+offset + labelShift);
SchottkyCross(pp1+offset, 0, black);
label(btex $P_1$ etex, pp1+offset + labelShift);
SchottkyCross(p2+offset, 0, black);
label(btex $P_2^\prime$ etex, p2+offset - labelShift);
SchottkyCross(pp2+offset, 0, black);
label(btex $P_2$ etex, pp2+offset + labelShift);
SchottkyCross(p3+offset, 0, black);
label(btex $P_3^\prime$ etex, p3+offset + labelShift);
SchottkyCross(pp3+offset, 0, black);
label(btex $P_3$ etex, pp3+offset + labelShift);
label(btex $0$ etex, (3.20xu, 3.20xu)+offset);
\end{mpostfig}

\begin{mpostfig}[label=Situation1]
pair ci, cj, cpj, pj, ppj, x, ipj, ippj, ix;
ci = (2xu,1.8xu);
cj = (0xu, 1.2xu);
cpj = (1.2xu, -0.3xu);
ri := 1xu;
rj := 0.6xu;
rpj := 0.9xu;
SchottkyACircle(ci, ri, ecolor, fgray);
label(btex $\mathfrak{A}_i$ etex, ci + (ri,-ri));
SchottkyCircleDashed(cj, rj, bcolor, fgray);
SchottkyCircle(cpj, rpj, bcolor, fgray);

x = (0xu, 0xu);
pj = cj + (0.1xu, 0.2xu);
ppj = cpj + (-0.2xu,0.2xu);
SchottkyCross(x, 0, black);
label(btex $x$ etex, x+(-0.2xu,-0.2xu));
SchottkyCross(pj, 0, black);
label(btex $P_j^\prime$ etex, pj+(-0.25xu,-0.1xu));
SchottkyCross(ppj, 0, black);
label(btex $P_j$ etex, ppj+(-0.25xu,+0.1xu));

ix = ci + (-0.3xu, -0.5xu);
ipj = ci + (0xu, 0.3xu);
ippj = ci + (0.45xu,-0.35xu);
label(btex $\gamma_1 x$ etex, ix+(-0.1xu,+0.2xu));
SchottkyPathArrowAlt(x, ix);
label(btex $\gamma_2 P_j^\prime$ etex, ipj+(-0.1xu,0.2xu));
SchottkyPathArrow(pj, ipj);
label(btex $\gamma_2 P_j$ etex, ippj+(-0.1xu,0.2xu));
SchottkyPathArrow(ppj, ippj);
\end{mpostfig}

\begin{mpostfig}[label=Situation2]
pair ci, cj, cpj, pj, ppj, x, ipj, ippj, ix;
ci = (2xu,1.8xu);
cj = (0xu, 1.2xu);
cpj = (1.2xu, -0.3xu);
ri := 1xu;
rj := 0.6xu;
rpj := 0.9xu;
SchottkyACircle(ci, ri, ecolor, fgray);
label(btex $\mathfrak{A}_i$ etex, ci + (ri,-ri));
SchottkyCircleDashed(cj, rj, bcolor, fgray);
SchottkyCircle(cpj, rpj, bcolor, fgray);

x = (0xu, 0xu);
pj = cj + (0.1xu, 0.2xu);
ppj = cpj + (-0.15xu,0.3xu);
SchottkyCross(x, 0, black);
label(btex $x$ etex, x+(-0.2xu,-0.2xu));
SchottkyCross(pj, 0, black);
label(btex $P_j^\prime$ etex, pj+(-0.25xu,-0.1xu));
SchottkyCross(ppj, 0, black);
label(btex $P_j$ etex, ppj+(-0.25xu,+0.1xu));

ix = cpj + (0xu,-0.25xu);
ipj = ci + (0xu, 0.3xu);
ippj = ci + (0.45xu,-0.35xu);
label(btex $\gamma_1 x$ etex, ix+(0.1xu,-0.2xu));
SchottkyPathArrow(x, ix);
label(btex $\gamma_2 P_j^\prime$ etex, ipj+(-0.1xu,0.2xu));
SchottkyPathArrow(pj, ipj);
label(btex $\gamma_2 P_j$ etex, ippj+(-0.1xu,0.2xu));
SchottkyPathArrow(ppj, ippj);
\end{mpostfig}

\begin{mpostfig}[label=Situation3]
pair ci, cj, cpj, pj, ppj, x, ipj, ippj, ix;
ci = (2xu,1.8xu);
cj = (0xu, 1.2xu);
cpj = (1.2xu, -0.3xu);
ri := 1xu;
rj := 0.6xu;
rpj := 0.9xu;
SchottkyACircle(ci, ri, ecolor, fgray);
label(btex $\mathfrak{A}_i$ etex, ci + (ri,-ri));
SchottkyCircleDashed(cj, rj, bcolor, fgray);
SchottkyCircle(cpj, rpj, bcolor, fgray);

x = (0xu, 0xu);
pj = cj + (0.1xu, 0.2xu);
ppj = cpj + (-0.2xu,0.2xu);
SchottkyCross(x, 0, black);
label(btex $x$ etex, x+(-0.2xu,-0.2xu));
SchottkyCross(pj, 0, black);
label(btex $P_j^\prime$ etex, pj+(-0.25xu,-0.1xu));
SchottkyCross(ppj, 0, black);
label(btex $P_j$ etex, ppj+(-0.25xu,+0.1xu));

ix = ci + (-0.3xu, -0.5xu);
ipj = cpj + (-0.3xu, -0.3xu);
ippj = cpj + (0.1xu,-0.25xu);
label(btex $\gamma_1 x$ etex, ix+(-0.1xu,+0.2xu));
SchottkyPathArrowAlt(x, ix);
label(btex $\gamma_2 P_j^\prime$ etex, ipj+(0.15xu,-0.25xu));
SchottkyPathArrow(pj, ipj);
label(btex $\gamma_2 P_j$ etex, ippj+(0.25xu,0.3xu));
SchottkyPathArrow(ppj, ippj);
\end{mpostfig}

\begin{mpostfig}[label=Situation4]
pair ci, cpi, fpi, fppi, x, ix[];
ci = (2xu,1.8xu);
cpi = (1.2xu, -0.3xu);
ri :=  1xu;
rpi := 0.9xu;
SchottkyACircle(ci, ri, ecolor, fgray);
label(btex $\mathfrak{A}_i$ etex, ci + (ri,-ri));
SchottkyCircleDashed(cpi, rpi, ecolor, fgray);

x = (0xu, 0xu);
SchottkyCross(x, 0, black);
label(btex $x$ etex, x+(-0.2xu,-0.2xu));

ix1 = ci + (-0.3xu, -0.4xu);
ix2 = cpi + (+0xu, -0.25xu);
fpi = ci + (0.3xu, -0.2xu);
fppi = cpi + (0.2xu, 0.15xu);
label(btex $\gamma_1 x$ etex, ix1+(-0.1xu,+0.2xu));
SchottkyPathArrowAlt(x, ix1);
SchottkyPathArrow(x, ix2);
label(btex $\gamma_1 x$ etex, ix2+(0.1xu,-0.25xu));
SchottkyCross(fpi, 0, black);
label(btex $P_i$ etex, fpi+(0.1xu,0.25xu));
SchottkyCross(fppi, 0, black);
label(btex $P_i^\prime$ etex, fppi+(-0.25xu,0.1xu));
\end{mpostfig}

\begin{mpostfig}[label=genustwoSchottky]
		pens:=0.7pt;
		drawarrow (-2.2xu, 0) -- (2.5xu, 0) pensize (1.0pt);  
		drawarrow (0, -2.0xu) -- (0, 2.3xu) pensize (1.0pt);  
		pair rpos; pair rrpos; 
		pair midpos, midposright, midposleft;
		path con, rA, rrA;
		path lspr, lsprr, lsprupper, lsprlower, lsprrupper, lsprrlower;
		
		rcirc:=0.9xu;
		rpos:=(1.4xu,1.5xu);
		rrcirc:=1.3xu;
		rrpos:=(-1.0xu,1.0xu);
		fill fullcircle scaled (rcirc+0.11xu) shifted rpos pensize(pens) withcolor fgray;
		fill fullcircle scaled (rrcirc+0.11xu) shifted rrpos pensize(pens) withcolor fgray;
		rA:=fullcircle scaled (rcirc+0.11xu) shifted rpos;
		rrA:=fullcircle scaled (rrcirc+0.11xu) shifted rrpos;
		con = rpos--rrpos;
		midpos := ((rcirc/(rcirc + rrcirc))[rpos,rrpos]);
		midposright := ((1.06*rcirc/(rcirc + rrcirc))[rpos,rrpos]);
		midposleft := ((0.94*rcirc/(rcirc + rrcirc))[rpos,rrpos]);
		
		numeric appanglea, appangleb;
		appanglea:=76;
		appangleb:=appanglea + 180;
		rpos:=(1.4xu,1.3xu);
		rrpos:=(-0.8xu,1.1xu);
		lsprr:=logspiralder(rrpos, midpos, 560, 12, appanglea);
		draw lsprr pensize(pens);
		lspr:=logspiralder(rpos, midpos, 560, 12, appangleb);
		draw lspr pensize(pens);
		
		lsprupper:=logspiralder(rpos, midposright, 595, 12, appangleb);
		lsprlower:=logspiralder(rpos, midposleft, 528, 12, appangleb);
		lsprrupper:=logspiralder(rrpos, midposright, 528, 12, appanglea);
		lsprrlower:=logspiralder(rrpos, midposleft, 588, 12, appanglea);
		
		draw lsprupper cutbefore rA pensize (pens) withcolor bcolor;
		midarrowblue ((lsprrupper cutbefore rrA),2.4);
		draw lsprlower cutbefore rA pensize(pens) withcolor bcolor;
		draw lsprrupper cutbefore rrA pensize (pens) withcolor bcolor;
		draw lsprrlower cutbefore rrA pensize (pens) withcolor bcolor;
		midarrowblue (lsprlower cutbefore rA,2.4);

		draw (rA cutafter lsprupper) pensize (pens) withcolor ecolor;
		draw (rA cutbefore lsprlower) pensize (pens) withcolor ecolor;
		midarrow (reverse(rA cutbefore lsprlower),0.3);
		draw (rrA cutbefore lsprrupper) pensize (pens) withcolor ecolor;
		draw (rrA cutafter lsprrlower) pensize (pens) withcolor ecolor;
		midarrow (reverse(rrA cutafter lsprrlower),3.4);

		label(btex $\mathfrak{A}_1^{-1}$ etex, (-1.75xu, 1.7xu));
		label(btex $\mathfrak{A}_1$ etex, (1.9xu, 0.9xu));
		label(btex $\mathfrak{B}_1$ etex, (0.66xu, 0.75xu));
		label(btex $\mathfrak{B}_1^{-1}$ etex, (0.35xu, 1.8xu));
		label.bot(btex $P_1^\prime$ etex, rrpos);
		label.top(btex $P_1$ etex, rpos+(-0.1xu,0xu));

		draw fullcircle scaled (0.03xu) shifted rpos pensize(pens) withcolor black; 
		draw fullcircle scaled (0.03xu) shifted rrpos pensize(pens) withcolor black; 
		
		rcirc:=1.2xu;
		rpos:=(1.4xu,-1.1xu);
		rrcirc:=1.0xu;
		rrpos:=(-1.2xu,-1.0xu);
		fill fullcircle scaled (rcirc+0.11xu) shifted rpos pensize(pens) withcolor fgray;
		fill fullcircle scaled (rrcirc+0.11xu) shifted rrpos pensize(pens) withcolor fgray;

		rA:=fullcircle scaled (rcirc+0.11xu) shifted rpos;
		rrA:=fullcircle scaled (rrcirc+0.11xu) shifted rrpos;
		midpos := ((rcirc/(rcirc + rrcirc))[rpos,rrpos]);
		midposright := ((1.06*rcirc/(rcirc + rrcirc))[rpos,rrpos]);
		midposleft := ((0.94*rcirc/(rcirc + rrcirc))[rpos,rrpos]);
		
		numeric appanglea, appangleb;
		appanglea:=56;
		appangleb:=appanglea + 180;
		rpos:=(1.4xu,-1.3xu);
		rrpos:=(-1.05xu,-0.95xu);
		lsprr:=logspiralder(rrpos, midpos, 560, 12, appanglea);
		draw lsprr pensize(pens);
		lspr:=logspiralder(rpos, midpos, 560, 12, appangleb);
		draw lspr pensize(pens);
		
		lsprupper:=logspiralder(rpos, midposright, 595, 12, appangleb);
		lsprlower:=logspiralder(rpos, midposleft, 528, 12, appangleb);
		lsprrupper:=logspiralder(rrpos, midposright, 528, 12, appanglea);
		lsprrlower:=logspiralder(rrpos, midposleft, 588, 12, appanglea);
		
		draw lsprupper cutbefore rA pensize (pens) withcolor bcolor;
		midarrowblue (reverse(lsprupper cutbefore rA),0.35);
		draw lsprlower cutbefore rA pensize(pens) withcolor bcolor;
		midarrowblue ((lsprlower cutbefore rA),1.6);
		draw lsprrupper cutbefore rrA pensize (pens) withcolor bcolor;
		draw lsprrlower cutbefore rrA pensize (pens) withcolor bcolor;

		draw (rA cutafter lsprupper) pensize (pens) withcolor ecolor;
		draw (rA cutbefore lsprlower) pensize (pens) withcolor ecolor;
		midarrow (reverse(rA cutbefore lsprlower),0.3);
		draw (rrA cutbefore lsprrupper) pensize (pens) withcolor ecolor;
		draw (rrA cutafter lsprrlower) pensize (pens) withcolor ecolor;
		midarrow (reverse(rrA cutafter lsprrlower),3.4);

		label(btex $\mathfrak{A}_2^{-1}$ etex, (-1.85xu, -0.43xu));
		label(btex $\mathfrak{A}_2$ etex, (2.1xu, -1.8xu));
		label(btex $\mathfrak{B}_2$ etex, (0.42xu, -1.15xu));
		label(btex $\mathfrak{B}_2^{-1}$ etex, (0.35xu, -0.24xu));
		label.bot(btex $P_2^\prime$ etex, rrpos);
		label.top(btex $P_2$ etex, rpos);

		draw fullcircle scaled (0.03xu) shifted rpos pensize(pens) withcolor black; 
		draw fullcircle scaled (0.03xu) shifted rrpos pensize(pens) withcolor black; 

\end{mpostfig}

\begin{mpostfig}[label=genustwo]
	
	draw (-3xu,-0.15xu) .. controls (-2.9xu,-0.18xu) and (-2.9xu,-0.92xu) .. ( -3xu,-0.95xu) dashed evenly pensize (1.0pt) withcolor ecolor;
	draw (-4.2xu,0xu) ..controls (-4.2xu,0.6xu) and (-1.8xu,0.6xu) .. (-1.8xu,0xu) pensize (1.0pt) withcolor bcolor;
	draw (-4.2xu,0xu) ..controls (-4.2xu,-0.6xu) and (-1.8xu,-0.6xu) .. (-1.8xu,0xu) pensize (1.0pt) withcolor bcolor;
	midarrowblue ((-4.2xu,0xu) ..controls (-4.2xu,-0.6xu) and (-1.8xu,-0.6xu) .. (-1.8xu,0xu),0.6);
	draw (-4xu,0.03xu){dir -20}..(-2xu,0.03xu) pensize (1.0pt);
	draw (-3.8xu,-0.02xu){dir 25}..(-2.2xu,-0.02xu) pensize (1.0pt);
	draw (-3xu,-0.15xu) .. controls (-3.1xu,-0.18xu) and (-3.1xu,-0.92xu) .. ( -3xu,-0.95xu) pensize (1.0pt) withcolor ecolor;
	midarrowred ((-3xu,-0.15xu) .. controls (-3.1xu,-0.18xu) and (-3.1xu,-0.92xu) .. ( -3xu,-0.95xu),0.65);
	
	draw (0xu,-0.15xu) .. controls (0.1xu,-0.18xu) and (0.1xu,-0.92xu) .. ( 0xu,-0.95xu) dashed evenly pensize (1.0pt) withcolor ecolor;
	draw (-1.2xu,0xu) ..controls (-1.2xu,0.6xu) and (1.2xu,0.6xu) .. (1.2xu,0xu) pensize (1.0pt) withcolor bcolor;
	draw (-1.2xu,0xu) ..controls (-1.2xu,-0.6xu) and (1.2xu,-0.6xu) .. (1.2xu,0xu) pensize (1.0pt) withcolor bcolor;
	midarrowblue ((-1.2xu,0xu) ..controls (-1.2xu,-0.6xu) and (1.2xu,-0.6xu) .. (1.2xu,0xu),0.6);
	draw (-1xu,0.03xu){dir -20}..(1xu,0.03xu) pensize (1.0pt);
	draw (-0.8xu,-0.02xu){dir 25}..(0.8xu,-0.02xu) pensize (1.0pt);
	draw (0xu,-0.15xu) .. controls (-0.1xu,-0.18xu) and (-0.1xu,-0.92xu) .. ( 0xu,-0.95xu) pensize (1.0pt) withcolor ecolor;
	midarrowred ((0xu,-0.15xu) .. controls (-0.1xu,-0.18xu) and (-0.1xu,-0.92xu) .. ( 0xu,-0.95xu),0.65);
	
	draw halfcircle xscaled 1.9xu yscaled 3.5xu rotated -90 ..(-1.5xu,0.7xu){left}..(-3xu,0.95xu){left} pensize (1pt);
	draw halfcircle xscaled 1.9xu yscaled 3.5xu rotated 90 shifted (-3xu,0xu) ..(-1.5xu,-0.7xu){right}..(0xu,-0.95xu){right} pensize (1pt);

	label(btex $\mathfrak{A}_1$ etex, (-2.7xu, -1.2xu));
	label(btex $\mathfrak{A}_2$ etex, (0.3xu, -1.2xu));
	label(btex $\mathfrak{B}_1$ etex, (-1.63xu, 0.3xu));
	label(btex $\mathfrak{B}_2$ etex, (1.37xu, 0.3xu));

\end{mpostfig}

\begin{mpostfig}[label=shottky-degeneration]
SchottkyAxes((-3.1238xu, 0.0000xu), (3.1238xu, 0.0000xu), (0.0000xu, -3.0870xu), (0.0000xu, 3.0870xu));
SchottkyCircleDashed((-1.4000xu,2.1000xu), 0.7000xu, ecolor, fgray);
label(btex $C_1^\prime$ etex, (-0.5550xu, 1.2550xu));
SchottkyCircle((1.7500xu,1.1900xu), 0.3500xu, ecolor, fgray);
label(btex $C_1$ etex, (2.3475xu, 0.5925xu));
SchottkyCircleDashed((1.4000xu,-0.7000xu), 1.0500xu, bcolor, fgray);
label(btex $C_2^\prime$ etex, (2.4925xu, -1.7925xu));
SchottkyCircle((-2.4500xu,0.5600xu), 0.5250xu, bcolor, fgray);
label(btex $C_2$ etex, (-1.7288xu, -0.1612xu));
SchottkyCircleDashedNoFilling((0.7000xu,2.4500xu), 0.4900xu, (0.0000,0.3000,0.0000));
SchottkyCircleNoFilling((-1.4000xu,-1.4000xu), 1.0500xu, (0.0000,0.3000,0.0000));
SchottkyCircleDashedNoFilling((0.7204xu,2.4146xu), 0.3217xu, (0.2748,0.4717,0.2748));
SchottkyCircleNoFilling((-1.4204xu,-1.3646xu), 0.6894xu, (0.2748,0.4717,0.2748));
SchottkyCircleDashedNoFilling((0.7396xu,2.3813xu), 0.1633xu, (0.5334,0.6334,0.5334));
SchottkyCircleNoFilling((-1.4396xu,-1.3313xu), 0.3499xu, (0.5334,0.6334,0.5334));
SchottkyCircleDashedNoFilling((0.7588xu,2.3480xu), 0.0049xu, (0.7920,0.7950,0.7920));
SchottkyCircleNoFilling((-1.4588xu,-1.2980xu), 0.0105xu, (0.7920,0.7950,0.7920));
label(btex $C_3^\prime$ etex, (1.3965xu, 1.7535xu));
label(btex $C_3$ etex, (-0.3075xu, -2.4925xu));
\end{mpostfig}

\begin{mpostfig}[label=schottky-degeneration-sep]
SchottkyAxes((-3.1238xu, 0.0000xu), (3.1238xu, 0.0000xu), (0.0000xu, -3.0870xu), (0.0000xu, 3.0870xu));
SchottkyCircleDashed((-1.4000xu,2.1000xu), 0.7000xu, ecolor, fgray);
label(btex $C_1^\prime$ etex, (-0.5550xu, 1.2550xu));
SchottkyCircle((1.7500xu,1.1900xu), 0.3500xu, ecolor, fgray);
label(btex $C_1$ etex, (2.3475xu, 0.5925xu));
SchottkyCircleDashed((1.4000xu,-0.7000xu), 1.0500xu, bcolor, fgray);
label(btex $C_2^\prime$ etex, (2.4925xu, -1.7925xu));
SchottkyCircle((-2.4500xu,0.5600xu), 0.5250xu, bcolor, fgray);
label(btex $C_2$ etex, (-1.7288xu, -0.1612xu));
drawcross((0.7594xu,2.3469xu),2,0,0.7pt,(0.0000,0.3000,0.0000));
drawcross((-1.4594xu,-1.2969xu),2,0,0.7pt,(0.0000,0.3000,0.0000));
SchottkyCircleDashedNoFilling((0.7000xu,2.4500xu), 0.4900xu, (0.0000,0.3000,0.0000));
SchottkyCircleNoFilling((-1.4000xu,-1.4000xu), 1.0500xu, (0.0000,0.3000,0.0000));
drawcross((0.4784xu,1.4786xu),2,0,0.7pt,(0.2960,0.4850,0.2960));
drawcross((-0.9194xu,-0.8171xu),2,0,0.7pt,(0.2960,0.4850,0.2960));
SchottkyCircleDashedNoFilling((0.4410xu,1.5435xu), 0.3087xu, (0.2960,0.4850,0.2960));
SchottkyCircleNoFilling((-0.8820xu,-0.8820xu), 0.6615xu, (0.2960,0.4850,0.2960));
drawcross((0.2202xu,0.6806xu),2,0,0.7pt,(0.5680,0.6550,0.5680));
drawcross((-0.4232xu,-0.3761xu),2,0,0.7pt,(0.5680,0.6550,0.5680));
SchottkyCircleDashedNoFilling((0.2030xu,0.7105xu), 0.1421xu, (0.5680,0.6550,0.5680));
SchottkyCircleNoFilling((-0.4060xu,-0.4060xu), 0.3045xu, (0.5680,0.6550,0.5680));
drawcross((0.0076xu,0.0235xu),2,0,0.7pt,(0.7920,0.7950,0.7920));
drawcross((-0.0146xu,-0.0130xu),2,0,0.7pt,(0.7920,0.7950,0.7920));
SchottkyCircleDashedNoFilling((0.0070xu,0.0245xu), 0.0049xu, (0.7920,0.7950,0.7920));
SchottkyCircleNoFilling((-0.0140xu,-0.0140xu), 0.0105xu, (0.7920,0.7950,0.7920));
label(btex $C_3^\prime$ etex, (1.3965xu, 1.7535xu));
label(btex $C_3$ etex, (-0.3075xu, -2.4925xu));
label(btex $P_3^\prime$ etex, (0.4794xu, 2.5569xu));
label(btex $P_3$ etex, (-1.6694xu, -1.5069xu));
\end{mpostfig}

\begin{mpostfig}[label=degen_nonsep]
	draw (-3xu,-0.15xu) .. controls (-2.9xu,-0.18xu) and (-2.9xu,-0.92xu) .. ( -3xu,-0.95xu) dashed evenly pensize (1.0pt) withcolor ecolor;
	draw (-4.2xu,0xu) ..controls (-4.2xu,0.6xu) and (-1.8xu,0.6xu) .. (-1.8xu,0xu) pensize (1.0pt) withcolor bcolor;
	draw (-4.2xu,0xu) ..controls (-4.2xu,-0.6xu) and (-1.8xu,-0.6xu) .. (-1.8xu,0xu) pensize (1.0pt) withcolor bcolor;
	draw (-4xu,0.03xu){dir -20}..(-2xu,0.03xu) pensize (1.0pt);
	draw (-3.8xu,-0.02xu){dir 25}..(-2.2xu,-0.02xu) pensize (1.0pt);
	draw (-3xu,-0.15xu) .. controls (-3.1xu,-0.18xu) and (-3.1xu,-0.92xu) .. ( -3xu,-0.95xu) pensize (1.0pt) withcolor ecolor;
	
	draw (0xu,-0.15xu) .. controls (0.1xu,-0.18xu) and (0.1xu,-0.92xu) .. ( 0xu,-0.95xu) dashed evenly pensize (1.0pt) withcolor ecolor;
	draw (-1.2xu,0xu) ..controls (-1.2xu,0.6xu) and (1.2xu,0.6xu) .. (1.2xu,0xu) pensize (1.0pt) withcolor bcolor;
	draw (-1.2xu,0xu) ..controls (-1.2xu,-0.6xu) and (1.2xu,-0.6xu) .. (1.2xu,0xu) pensize (1.0pt) withcolor bcolor;
	draw (-1xu,0.03xu){dir -20}..(1xu,0.03xu) pensize (1.0pt);
	draw (-0.8xu,-0.02xu){dir 25}..(0.8xu,-0.02xu) pensize (1.0pt);
	draw (0xu,-0.15xu) .. controls (-0.1xu,-0.18xu) and (-0.1xu,-0.92xu) .. ( 0xu,-0.95xu) pensize (1.0pt) withcolor ecolor;
	
	draw halfcircle xscaled 1.9xu yscaled 3.5xu rotated -90 ..(-1.5xu,0.7xu){left}..(-3xu,0.95xu){left} pensize (1pt);
	draw halfcircle xscaled 1.9xu yscaled 3.5xu rotated 90 shifted (-3xu,0xu) ..(-1.5xu,-0.7xu){right}..(0xu,-0.95xu){right} pensize (1pt);
	pic[1]:=currentpicture;
	currentpicture:=nullpicture;

	draw (-3xu,-0.15xu) .. controls (-2.9xu,-0.18xu) and (-2.9xu,-0.92xu) .. ( -3xu,-0.95xu) dashed evenly pensize (1.0pt) withcolor ecolor;
	draw (-4.2xu,0xu) ..controls (-4.2xu,0.6xu) and (-1.8xu,0.6xu) .. (-1.8xu,0xu) pensize (1.0pt) withcolor bcolor;
	draw (-4.2xu,0xu) ..controls (-4.2xu,-0.6xu) and (-1.8xu,-0.6xu) .. (-1.8xu,0xu) pensize (1.0pt) withcolor bcolor;
	draw (-4xu,0.03xu){dir -20}..(-2xu,0.03xu) pensize (1.0pt);
	draw (-3.8xu,-0.02xu){dir 25}..(-2.2xu,-0.02xu) pensize (1.0pt);
	draw (-3xu,-0.15xu) .. controls (-3.1xu,-0.18xu) and (-3.1xu,-0.92xu) .. ( -3xu,-0.95xu) pensize (1.0pt) withcolor ecolor;

	draw halfcircle rotated -90 xscaled 0.048xu yscaled 0.12xu shifted (0xu,-0.56xu) pensize (1.0pt) withcolor ecolor;
	draw (-1.2xu,0xu) ..controls (-1.2xu,0.6xu) and (1.2xu,0.6xu) .. (1.2xu,0xu) pensize (1.0pt) withcolor bcolor;
	draw (-1.2xu,0xu) ..controls (-1.2xu,-0.75xu) and (1.2xu,-0.75xu) .. (1.2xu,0xu) pensize (1.0pt) withcolor bcolor;
	draw (-1xu,0.03xu){dir -20}..{dir 0}(0xu, -0.5xu){dir 0}..{dir 20}(1xu,0.03xu) pensize (1.0pt);
	draw (-0.85xu,-0.02xu){dir 18}..(0.85xu,-0.02xu) pensize (1.0pt);
	draw (0xu,-0.62xu){right} .. (0.75xu, -0.65xu){right} .. quartercircle xscaled 3.5xu yscaled 1.9xu ..(-1.5xu,0.7xu){left}..(-3xu,0.95xu){left} pensize (1pt);
	draw halfcircle xscaled 1.9xu yscaled 3.5xu rotated 90 shifted (-3xu,0xu) ..(-1.5xu,-0.55xu){right}..(-0.75xu, -0.65xu){right}..(0xu,-0.62xu){right} pensize (1pt);
	draw halfcircle rotated 90 xscaled 0.048xu yscaled 0.12xu shifted (0xu,-0.56xu) pensize (1.0pt) withcolor ecolor;
	pic[2]:=currentpicture;
	currentpicture:=nullpicture;

	draw (-3xu,-0.15xu) .. controls (-2.9xu,-0.18xu) and (-2.9xu,-0.92xu) .. ( -3xu,-0.95xu) dashed evenly pensize (1.0pt) withcolor ecolor;
	draw (-4.2xu,0xu) ..controls (-4.2xu,0.6xu) and (-1.8xu,0.6xu) .. (-1.8xu,0xu) pensize (1.0pt) withcolor bcolor;
	draw (-4.2xu,0xu) ..controls (-4.2xu,-0.6xu) and (-1.8xu,-0.6xu) .. (-1.8xu,0xu) pensize (1.0pt) withcolor bcolor;
	draw (-4xu,0.03xu){dir -20}..(-2xu,0.03xu) pensize (1.0pt);
	draw (-3.8xu,-0.02xu){dir 25}..(-2.2xu,-0.02xu) pensize (1.0pt);
	draw (-3xu,-0.15xu) .. controls (-3.1xu,-0.18xu) and (-3.1xu,-0.92xu) .. ( -3xu,-0.95xu) pensize (1.0pt) withcolor ecolor;
	draw fullcircle xscaled 1.9xu yscaled 3.5xu rotated -90 shifted (-3xu,0xu)  pensize (1pt);
	draw (-1.5xu,-0.2xu){dir -130} .. (-2.5xu,-0.7xu) pensize (0.7pt) withcolor bcolor;
	drawcross((-1.5xu,-0.2xu),2,35,0.7pt,ecolor);
	drawcross((-2.5xu,-0.7xu),2,10,0.7pt,ecolor);

	pic[3]:=currentpicture;
	currentpicture:=nullpicture;
	
	draw pic[1] shifted (0xu,0xu);
	draw nicearrow scaled 0.4 rotated -90 shifted (1.2xu,-0.85xu);
	draw pic[2] shifted (0xu,-2.5xu);
	draw nicearrow scaled 0.4 rotated -100 shifted (0.5xu,-3.6xu);
	draw pic[3] shifted (1.5xu,-5xu);
\end{mpostfig}

\begin{mpostfig}[label=degen_sep]
	draw (-3xu,-0.15xu) .. controls (-2.9xu,-0.18xu) and (-2.9xu,-0.92xu) .. ( -3xu,-0.95xu) dashed evenly pensize (1.0pt) withcolor ecolor;
	draw (-4.2xu,0xu) ..controls (-4.2xu,0.6xu) and (-1.8xu,0.6xu) .. (-1.8xu,0xu) pensize (1.0pt) withcolor bcolor;
	draw (-4.2xu,0xu) ..controls (-4.2xu,-0.6xu) and (-1.8xu,-0.6xu) .. (-1.8xu,0xu) pensize (1.0pt) withcolor bcolor;
	draw (-4xu,0.03xu){dir -20}..(-2xu,0.03xu) pensize (1.0pt);
	draw (-3.8xu,-0.02xu){dir 25}..(-2.2xu,-0.02xu) pensize (1.0pt);
	draw (-3xu,-0.15xu) .. controls (-3.1xu,-0.18xu) and (-3.1xu,-0.92xu) .. ( -3xu,-0.95xu) pensize (1.0pt) withcolor ecolor;
	
	draw (0xu,-0.15xu) .. controls (0.1xu,-0.18xu) and (0.1xu,-0.92xu) .. ( 0xu,-0.95xu) dashed evenly pensize (1.0pt) withcolor ecolor;
	draw (-1.2xu,0xu) ..controls (-1.2xu,0.6xu) and (1.2xu,0.6xu) .. (1.2xu,0xu) pensize (1.0pt) withcolor bcolor;
	draw (-1.2xu,0xu) ..controls (-1.2xu,-0.6xu) and (1.2xu,-0.6xu) .. (1.2xu,0xu) pensize (1.0pt) withcolor bcolor;
	draw (-1xu,0.03xu){dir -20}..(1xu,0.03xu) pensize (1.0pt);
	draw (-0.8xu,-0.02xu){dir 25}..(0.8xu,-0.02xu) pensize (1.0pt);
	draw (0xu,-0.15xu) .. controls (-0.1xu,-0.18xu) and (-0.1xu,-0.92xu) .. ( 0xu,-0.95xu) pensize (1.0pt) withcolor ecolor;
	
	draw (-1.5xu,0.7xu) .. controls (-1.3xu,0.65xu) and (-1.3xu,-0.65xu) .. (-1.5xu,-0.7xu) dashed evenly pensize (1.0pt) withcolor gcolor;
	draw halfcircle xscaled 1.9xu yscaled 3.5xu rotated -90 ..(-1.5xu,0.7xu){left}..(-3xu,0.95xu){left} pensize (1pt);
	draw (-1.5xu,0.7xu) .. controls (-1.7xu,0.65xu) and (-1.7xu,-0.65xu) .. (-1.5xu,-0.7xu) pensize (1.0pt) withcolor gcolor;
	draw halfcircle xscaled 1.9xu yscaled 3.5xu rotated 90 shifted (-3xu,0xu) ..(-1.5xu,-0.7xu){right}..(0xu,-0.95xu){right} pensize (1pt);
	pic[1]:=currentpicture;
	currentpicture:=nullpicture;

	draw (-3.4xu,-0.04xu) .. controls (-3.3xu,-0.04xu) and (-3.3xu,-0.8xu) .. ( -3.4xu,-0.8xu) dashed evenly pensize (1.0pt) withcolor ecolor;
	draw (-4.2xu,0xu) ..controls (-4.2xu,0.5xu) and (-2.6xu,0.5xu) .. (-2.6xu,0xu) pensize (1.0pt) withcolor bcolor;
	draw (-4.2xu,0xu) ..controls (-4.2xu,-0.5xu) and (-2.6xu,-0.5xu) .. (-2.6xu,0xu) pensize (1.0pt) withcolor bcolor;
	draw (-3.8xu,0.03xu){dir -20}..(-3xu,0.03xu) pensize (1.0pt);
	draw (-3.7xu,0xu){dir 25}..(-3.1xu,0xu) pensize (1.0pt);
	draw (-3.4xu,-0.04xu) .. controls (-3.5xu,-0.04xu) and (-3.5xu,-0.8xu) .. ( -3.4xu,-0.8xu) pensize (1.0pt) withcolor ecolor;
	
	draw (0.4xu,-0.04xu) .. controls (0.5xu,-0.04xu) and (0.5xu,-0.8xu) .. ( 0.4xu,-0.8xu) dashed evenly pensize (1.0pt) withcolor ecolor;
	draw (-0.4xu,0xu) ..controls (-0.4xu,0.5xu) and (1.2xu,0.5xu) .. (1.2xu,0xu) pensize (1.0pt) withcolor bcolor;
	draw (-0.4xu,0xu) ..controls (-0.4xu,-0.5xu) and (1.2xu,-0.5xu) .. (1.2xu,0xu) pensize (1.0pt) withcolor bcolor;
	draw (0xu,0.03xu){dir -20}..(0.8xu,0.03xu) pensize (1.0pt);
	draw (0.1xu,0xu){dir 25}..(0.7xu,0xu) pensize (1.0pt);
	draw (0.4xu,-0.04xu) .. controls (0.3xu,-0.04xu) and (0.3xu,-0.8xu) .. ( 0.4xu,-0.8xu) pensize (1.0pt) withcolor ecolor;
	
	draw halfcircle rotated -90 xscaled 0.048xu yscaled 0.12xu shifted (-1.5xu,0xu) pensize (1.0pt) withcolor gcolor;
	draw halfcircle xscaled 1.6xu yscaled 2.7xu rotated -90 shifted (0.4xu, 0xu) ..(-0.5xu,0.6xu)..(-1.5xu,0.06xu){left}..(-2.5xu, 0.6xu)..(-3.4xu,0.8xu){left} pensize (1pt);
	draw halfcircle xscaled 1.6xu yscaled 2.7xu rotated 90 shifted (-3.4xu,0xu) .. (-2.5xu,-0.6xu)..(-1.5xu,-0.06xu){right}..(-0.5xu, -0.6xu)..(0.4xu,-0.8xu){right} pensize (1pt);
	draw halfcircle rotated 90 xscaled 0.048xu yscaled 0.12xu shifted (-1.5xu,0xu) pensize (1.0pt) withcolor gcolor;
	pic[2]:=currentpicture;
	currentpicture:=nullpicture;

	draw (-3.4xu,-0.04xu) .. controls (-3.3xu,-0.04xu) and (-3.3xu,-0.8xu) .. ( -3.4xu,-0.8xu) dashed evenly pensize (1.0pt) withcolor ecolor;
	draw (-4.2xu,0xu) ..controls (-4.2xu,0.5xu) and (-2.6xu,0.5xu) .. (-2.6xu,0xu) pensize (1.0pt) withcolor bcolor;
	draw (-4.2xu,0xu) ..controls (-4.2xu,-0.5xu) and (-2.6xu,-0.5xu) .. (-2.6xu,0xu) pensize (1.0pt) withcolor bcolor;
	draw (-3.8xu,0.03xu){dir -20}..(-3xu,0.03xu) pensize (1.0pt);
	draw (-3.7xu,0xu){dir 25}..(-3.1xu,0xu) pensize (1.0pt);
	draw (-3.4xu,-0.04xu) .. controls (-3.5xu,-0.04xu) and (-3.5xu,-0.8xu) .. ( -3.4xu,-0.8xu) pensize (1.0pt) withcolor ecolor;
	
	draw (0.4xu,-0.04xu) .. controls (0.5xu,-0.04xu) and (0.5xu,-0.8xu) .. ( 0.4xu,-0.8xu) dashed evenly pensize (1.0pt) withcolor ecolor;
	draw (-0.4xu,0xu) ..controls (-0.4xu,0.5xu) and (1.2xu,0.5xu) .. (1.2xu,0xu) pensize (1.0pt) withcolor bcolor;
	draw (-0.4xu,0xu) ..controls (-0.4xu,-0.5xu) and (1.2xu,-0.5xu) .. (1.2xu,0xu) pensize (1.0pt) withcolor bcolor;
	draw (0xu,0.03xu){dir -20}..(0.8xu,0.03xu) pensize (1.0pt);
	draw (0.1xu,0xu){dir 25}..(0.7xu,0xu) pensize (1.0pt);
	draw (0.4xu,-0.04xu) .. controls (0.3xu,-0.04xu) and (0.3xu,-0.8xu) .. ( 0.4xu,-0.8xu) pensize (1.0pt) withcolor ecolor;

	draw fullcircle xscaled 1.6xu yscaled 2.7xu rotated -90 shifted (0.4xu, 0xu) pensize (1pt);
	draw fullcircle xscaled 1.6xu yscaled 2.7xu rotated 90 shifted (-3.4xu,0xu)  pensize (1pt);

	drawcross((-0.8xu,0xu),2,0,0.7pt,gcolor);
	drawcross((-2.2xu,0xu),2,0,0.7pt,gcolor);
	pic[3]:=currentpicture;
	currentpicture:=nullpicture;
	
	draw pic[1] shifted (0xu,0xu);
	draw nicearrow scaled 0.4 rotated -90 shifted (1.2xu,-0.85xu);
	draw pic[2] shifted (0xu,-2.5xu);
	draw nicearrow scaled 0.4 rotated -90 shifted (1.2xu,-3.35xu);
	draw pic[3] shifted (0xu,-5xu);
\end{mpostfig}

\usepackage{graphicx}
\usepackage{subcaption}
\usepackage{caption}

\usepackage[dvipsnames,table]{xcolor}
\definecolor{mygreen}{rgb}{0,0.4,0}
\definecolor{myblue}{rgb}{0,0.0,0.4}
\definecolor{refrcolor}{rgb}{0,0.4,0}
\definecolor{cgreen}{rgb}{0,0.7,0}
\definecolor{ecolor}{rgb}{.52,.03,.06}
\definecolor{bgcolor}{rgb}{.96,.95,.80}
\definecolor{bgcolordark}{rgb}{.80,.80,.67}
\definecolor{faint}{rgb}{.80,.80,.80}

\usepackage{tcolorbox}
\tcbuselibrary{breakable, skins}
\usetikzlibrary{patterns}

\newtcolorbox{myproofbox}[1][]{
	enhanced,
	breakable,
	borderline west={3pt}{0pt}{faint},
	notitle,
	before skip=10pt,
	after skip=10pt,
	colback=white, 
	colframe=white,
	frame hidden,
	boxrule=0pt, 
	boxsep=0pt,
	sharp corners,
	left=9pt, right=0pt, top=1pt, bottom=0pt,
	fontupper=\small,
	#1
}

\makeatletter
\renewenvironment{proof}[1][\proofname]{%
	\begin{myproofbox}\par\pushQED{\qed}\normalfont%
		\topsep6\p@\@plus6\p@\relax%
		\trivlist%
		\item[\hskip\labelsep\itshape#1\@addpunct{.}]%
	}%
	{\popQED\endtrivlist\end{myproofbox}\@endpefalse%
	\if@noskipsec\leavevmode\fi\noindent\ignorespacesafterend}
\makeatother

\newtheorem{example}[theorem]{Example}

\newtcolorbox{myexamplebox}[1][]{%
	enhanced,
	breakable,
	borderline west={3pt}{0pt}{faint},
	notitle,
	before skip=10pt,
	after skip=10pt,
	colback=white, 
	colframe=white,
	frame hidden,
	boxrule=0pt, 
	boxsep=0pt,
	sharp corners,
	left=9pt, right=0pt, top=1pt, bottom=0pt,
	#1
}

\makeatletter
\renewenvironment{example}[1][]{%
	\begin{myexamplebox}\refstepcounter{theorem}%
		\topsep6\p@\@plus6\p@\relax%
		\trivlist%
		\item[\hskip\labelsep\textbf{Example \theexample\@addpunct{.}}]%
	}%
	{\endtrivlist\end{myexamplebox}\@endpefalse%
	\if@noskipsec\leavevmode\fi\noindent\ignorespacesafterend}
\makeatother

\usepackage{enumitem}

\usepackage{tabularx}
\usepackage{booktabs}

\newcounter{sitcounter}

\newcommand{\itemwithlabel}[3][n]{%
  \def\tempposarg{#1} 
  \def\templeftarg{l}
  \def\temprightarg{r}
  \vspace{1em} 
  \noindent
  \refstepcounter{sitcounter}
  \ifx\tempposarg\templeftarg 
  \begin{tabularx}{\linewidth}{@{}m{10em}@{\hspace{2em}}X@{}}
    \raisebox{\dimexpr\ht\strutbox-\height}{#2} & \textbf{Situation \arabic{sitcounter}:} #3
  \end{tabularx}%
  \else
  \ifx\tempposarg\temprightarg 
  \begin{tabularx}{\linewidth}{@{}X@{\hspace{2em}}m{10em}@{}}
    \textbf{Situation \arabic{sitcounter}:} #3 & \raisebox{\dimexpr\ht\strutbox-\height}{#2}
  \end{tabularx}%
  \else
  \begin{tabularx}{\linewidth}{@{\hspace{2em}}X@{}}
    \textbf{Situation \arabic{sitcounter}:} #3
  \end{tabularx}%
  \fi
  \fi
  \par
}

\newlist{situations}{enumerate}{1}
\setlist[situations]{
    labelsep=0.8em,
    leftmargin=0pt, 
    align=left,
    labelindent=0in, 
    labelwidth=0in,
    itemindent=0.8em,
    label=\textbf{Situation \arabic*:},
    ref=\arabic*
}

\makeatletter
\providecommand*{\shuffle}{%
  \mathbin{\mathpalette\shuffle@{}}%
}
\newcommand*{\shuffle@}[2]{%
  \sbox0{$#1\vcenter{}$}%
  \kern .15\ht0 
  \rlap{\vrule height .25\ht0 depth 0pt width 2.5\ht0}%
  \raise.1\ht0\hbox to 2.5\ht0{%
    \vrule height 1.75\ht0 depth -.1\ht0 width .17\ht0 %
    \hfill
    \vrule height 1.75\ht0 depth -.1\ht0 width .17\ht0 %
    \hfill
    \vrule height 1.75\ht0 depth -.1\ht0 width .17\ht0 %
  }%
  \kern .15\ht0 
}
\makeatother

\usepackage{xparse}

\ExplSyntaxOn
\NewDocumentCommand{\Gtargz}{m m}
{
 \Gt\left(\begin{smallmatrix}
 \Gtargz_print:n {#1} \\
 \Gtargz_print:n {#2}
 \end{smallmatrix};z\right)
}
\seq_new:N \l_Gtargz_list_seq
\cs_new_protected:Npn \Gtargz_print:n #1
{
  \seq_set_split:Nnn \l_Gtargz_list_seq { , } { #1 }
  \seq_use:Nn \l_Gtargz_list_seq { , & }
}
\ExplSyntaxOff

\ExplSyntaxOn
\NewDocumentCommand{\Gtargt}{m m}
{
 \Gt\left(\begin{smallmatrix}
 \Gtargt_print:n {#1} \\
 \Gtargt_print:n {#2}
 \end{smallmatrix};t\right)
}
\seq_new:N \l_Gtargt_list_seq
\cs_new_protected:Npn \Gtargt_print:n #1
{
  \seq_set_split:Nnn \l_Gtargt_list_seq { , } { #1 }
  \seq_use:Nn \l_Gtargt_list_seq { , & }
}
\ExplSyntaxOff

\ExplSyntaxOn
\NewDocumentCommand{\Gtargtheta}{m m}
{
 \Gt\left(\begin{smallmatrix}
 \Gtargtheta_print:n {#1} \\
 \Gtargtheta_print:n {#2}
 \end{smallmatrix};\vt\right)
}
\seq_new:N \l_Gtargtheta_list_seq
\cs_new_protected:Npn \Gtargtheta_print:n #1
{
  \seq_set_split:Nnn \l_Gtargtheta_list_seq { , } { #1 }
  \seq_use:Nn \l_Gtargtheta_list_seq { , & }
}
\ExplSyntaxOff

\ExplSyntaxOn
\NewDocumentCommand{\Gtargzt}{m m}
{
 \Gt\left(\begin{smallmatrix}
 \Gtargzt_print:n {#1} \\
 \Gtargzt_print:n {#2}
 \end{smallmatrix};z|\tau\right)
}
\seq_new:N \l_Gtargzt_list_seq
\cs_new_protected:Npn \Gtargzt_print:n #1
{
  \seq_set_split:Nnn \l_Gtargzt_list_seq { , } { #1 }
  \seq_use:Nn \l_Gtargzt_list_seq { , & }
}
\ExplSyntaxOff

\ExplSyntaxOn
\NewDocumentCommand{\Gtargxit}{m m}
{
 \Gt\left(\begin{smallmatrix}
 \Gtargxit_print:n {#1} \\
 \Gtargxit_print:n {#2}
 \end{smallmatrix};\xi|\tau\right)
}
\seq_new:N \l_Gtargxit_list_seq
\cs_new_protected:Npn \Gtargxit_print:n #1
{
  \seq_set_split:Nnn \l_Gtargxit_list_seq { , } { #1 }
  \seq_use:Nn \l_Gtargxit_list_seq { , & }
}
\ExplSyntaxOff

\ExplSyntaxOn
\NewDocumentCommand{\Gtargtt}{m m}
{
 \Gt\left(\begin{smallmatrix}
 \Gtargtt_print:n {#1} \\
 \Gtargtt_print:n {#2}
 \end{smallmatrix};t|\tau\right)
}
\seq_new:N \l_Gtargtt_list_seq
\cs_new_protected:Npn \Gtargtt_print:n #1
{
  \seq_set_split:Nnn \l_Gtargtt_list_seq { , } { #1 }
  \seq_use:Nn \l_Gtargtt_list_seq { , & }
}
\ExplSyntaxOff

\ExplSyntaxOn
\NewDocumentCommand{\Gtargzg}{m m}
{
 \Gt\left(\begin{smallmatrix}
 \Gtargzg_print:n {#1} \\
 \Gtargzg_print:n {#2}
 \end{smallmatrix};z|\SGroup\right)
}
\seq_new:N \l_Gtargzg_list_seq
\cs_new_protected:Npn \Gtargzg_print:n #1
{
  \seq_set_split:Nnn \l_Gtargzg_list_seq { , } { #1 }
  \seq_use:Nn \l_Gtargzg_list_seq { , & }
}
\ExplSyntaxOff

\ExplSyntaxOn
\NewDocumentCommand{\Gtargxi}{m m}
{
 \Gt\left(\begin{smallmatrix}
 \Gtargxi_print:n {#1} \\
 \Gtargxi_print:n {#2}
 \end{smallmatrix};\xi\right)
}
\seq_new:N \l_Gtargxi_list_seq
\cs_new_protected:Npn \Gtargxi_print:n #1
{
  \seq_set_split:Nnn \l_Gtargxi_list_seq { , } { #1 }
  \seq_use:Nn \l_Gtargxi_list_seq { , & }
}
\ExplSyntaxOff

\ExplSyntaxOn
\NewDocumentCommand{\Gtargxig}{m m}
{
 \Gt\left(\begin{smallmatrix}
 \Gtargxig_print:n {#1} \\
 \Gtargxig_print:n {#2}
 \end{smallmatrix};\xi|\SGroup\right)
}
\seq_new:N \l_Gtargxig_list_seq
\cs_new_protected:Npn \Gtargxig_print:n #1
{
  \seq_set_split:Nnn \l_Gtargxig_list_seq { , } { #1 }
  \seq_use:Nn \l_Gtargxig_list_seq { , & }
}
\ExplSyntaxOff

\ExplSyntaxOn
\NewDocumentCommand{\Gtargtg}{m m}
{
 \Gt\left(\begin{smallmatrix}
 \Gtargtg_print:n {#1} \\
 \Gtargtg_print:n {#2}
 \end{smallmatrix};t|\SGroup\right)
}
\seq_new:N \l_Gtargtg_list_seq
\cs_new_protected:Npn \Gtargtg_print:n #1
{
  \seq_set_split:Nnn \l_Gtargtg_list_seq { , } { #1 }
  \seq_use:Nn \l_Gtargtg_list_seq { , & }
}
\ExplSyntaxOff

\newcommand{\SI}[1]{\Sel[#1]}

\ExplSyntaxOn
\NewDocumentCommand{\SIE}{m m}
{
\SelE\!\Big[\begin{smallmatrix}
 \SI_print:n {#1} \\
 \SI_print:n {#2}
 \end{smallmatrix}\Big]
}
\seq_new:N \l_SI_list_seq
\cs_new_protected:Npn \SI_print:n #1
{
  \seq_set_split:Nnn \l_SI_list_seq { , } { #1 }
  \seq_use:Nn \l_SI_list_seq { , & }
}
\ExplSyntaxOff

\ExplSyntaxOn
\NewDocumentCommand{\Gargbare}{m m m m}
{
	\Gt\left(\begin{smallmatrix}
		\Gargbare_print:n {#1} \\
		\Gargbare_print:n {#2}
	\end{smallmatrix}\,\Big|\,\Gargbare_print:n {#3}\,;\Gargbare_print:n {#4}\right)
}
\seq_new:N \l_Gargbare_list_seq
\cs_new_protected:Npn \Gargbare_print:n #1
{
	\seq_set_split:Nnn \l_Gargbare_list_seq {,} { #1 }
	\seq_use:Nn \l_Gargbare_list_seq {\hspace{-0.1em},\hspace{-0.18em} & }
}
\ExplSyntaxOff

\newcommand{\Greg}{\tilde{\Gamma}_{\text{reg}}}

\ExplSyntaxOn
\NewDocumentCommand{\Gtreg}{m m m m}
{
	\Greg\!\left(\!\begin{smallmatrix}
		\Gtreg_print:n {#1} \\
		\Gtreg_print:n {#2}
	\end{smallmatrix}\!\,\Big|\,\Gtreg_print:n {#3}\,;\Gtreg_print:n {#4}\right)
}
\seq_new:N \l_Gtreg_list_seq
\cs_new_protected:Npn \Gtreg_print:n #1
{
	\seq_set_split:Nnn \l_Gtreg_list_seq {,} { #1 }
	\seq_use:Nn \l_Gtreg_list_seq {\hspace{-0.1em},\hspace{-0.18em} & }
}
\ExplSyntaxOff

\ExplSyntaxOn
\NewDocumentCommand{\GtregG}{m m m m}
{
	\Greg\!\left(\!\begin{smallmatrix}
		\Gtreg_print:n {#1} \\
		\Gtreg_print:n {#2}
	\end{smallmatrix}\,;\,\GtregG_print:n {#3}\,,\GtregG_print:n {#4}\,\Big|\,\SGroup\right)
}
\seq_new:N \l_GtregG_list_seq
\cs_new_protected:Npn \GtregG_print:n #1
{
	\seq_set_split:Nnn \l_GtregG_list_seq {,} { #1 }
	\seq_use:Nn \l_GtregG_list_seq {\hspace{-0.1em},\hspace{-0.18em} & }
}
\ExplSyntaxOff

\newcommand{\Gtargone}[2]{\Gt\left(\begin{smallmatrix} #1 \\ #2\end{smallmatrix};1\right)}

\usepackage[pdfencoding=auto,bookmarks=true,hyperfigures=true]{hyperref}%
\PassOptionsToPackage{unicode}{hyperref}
\providecommand{\hypersetup}[1]{}
\providecommand{\texorpdfstring}[2]{#1}
\hypersetup{plainpages=false}
\hypersetup{pdfpagemode=UseNone}
\hypersetup{bookmarksnumbered=true}
\hypersetup{pdfstartview=FitH}
\hypersetup{colorlinks=false}
\hypersetup{citebordercolor={0.7 0.7 1}}
\hypersetup{urlbordercolor={.4 .8 1}}
\hypersetup{linkbordercolor={1 .8 .6}}
\hypersetup{colorlinks=true, urlcolor=[rgb]{0.13,0.30,0.45}, linkcolor=[rgb]{0.13,0.30,0.45}, citecolor=[rgb]{0.55,0.0,0.05}}

\usepackage{tocloft}

\makeatletter
\newcommand{\appendixsection}[1]{%
  \refstepcounter{section}
  \section*{\thesection\quad#1}
  \addcontentsline{toc}{appendixsec}{\protect\numberline{\thesection}#1}
}
\newcommand{\appendixsubsection}[1]{
  \refstepcounter{subsection}
  \subsection*{\thesubsection\quad#1}
  \addcontentsline{toc}{appendixsubsec}{\protect\numberline{\thesubsection}#1}
}
\newcommand{\l@appendixsec}[2]{%
  \vspace{0.1ex}%
  \@dottedtocline{1}{\cftsubsecindent}{\cftsubsecnumwidth}{#1}{#2}%
}
\newcommand{\l@appendixsubsec}[2]{%
  \vspace{0.1ex}%
  \@dottedtocline{2}{\cftsubsubsecindent}{\cftsubsubsecnumwidth}{#1}{#2}%
}
\providecommand*{\toclevel@appendixsec}{1}
\providecommand*{\toclevel@appendixsubsec}{2}
\def\Hy@toc@appendixsec{section}
\def\Hy@toc@appendixsec{subsection}
\makeatother

\ifpublic
\else
\usepackage[notcite,notref]{showkeys}

\def\showkeysrefformat#1{{\normalfont\tiny\ttfamily#1}}
\makeatletter
\def\SK@@ref#1>#2\SK@{%
{\@inlabelfalse\leavevmode\vbox to\z@{%
\vss\SK@refcolor\rlap{\vrule\raise .75em%
 \hbox{\showkeysrefformat{#2}}}}}}
\makeatother
\fi

\makeatletter
\let\@keywords\@empty
\let\@subject\@empty
\providecommand{\keywords}[1]{\gdef\@keywords{#1}}
\providecommand{\subject}[1]{\gdef\@subject{#1}}
\def\thetitle{\@title}
\def\theauthor{\@author}
\def\thesubject{\@subject}
\def\thedate{\@date}
\def\thekeywords{\@keywords}
\makeatother
\AtBeginDocument{
\hypersetup{pdftitle={\thetitle}}%
\hypersetup{pdfauthor={\theauthor}}%
\hypersetup{pdfsubject={\thesubject}}%
\hypersetup{pdfkeywords={\thekeywords}}%
}

\newif\ifnote\notetrue
\ifpublic
\notefalse
\fi

\let\Im\relax\DeclareMathOperator{\Im}{Im}

\let\Im\undefined\DeclareMathOperator{\Im}{Im}

\DeclareMathOperator{\sgn}{sgn}

\DeclareMathOperator{\id}{id}

\DeclareMathOperator{\Res}{Res}

\DeclareMathOperator{\Gt}{\tilde{\Gamma}}

\DeclareMathOperator{\Sel}{S}
\DeclareMathOperator{\SelE}{S^E}

\newcommand{\SGroup}{\mathrm{G}}
\newcommand{\SCosetL}[1]{\SGroup_{#1}\backslash\SGroup}
\newcommand{\SCosetR}[1]{\SGroup\slash\SGroup_{#1}}
\newcommand{\SCosetLR}[2]{\SGroup_{#1}\backslash\SGroup\slash\SGroup_{#2}}
\newcommand{\CustomSR}[2]{{#1}\slash{#2}} 
\newcommand{\CustomSL}[2]{{#1}\backslash{#2}} 

\newcommand{\sgen}{\sigma}
\newcommand{\ssplit}{\mathrm{Sp}}
\newcommand{\ssubsetl}[3]{#1^{[#2]}_{#3\rightarrow}}
\newcommand{\ssubsetr}[3]{#1^{[#2]}_{\leftarrow #3}}
\newcommand{\csubsetl}[3]{\mathrm{C}#1^{[#2]}_{#3\rightarrow}}
\newcommand{\csubsetr}[3]{\mathrm{C}#1^{[#2]}_{\leftarrow #3}}

\ProvideMathFix{rfrac,vfrac}
\newcommand{\iunit}{{\mathring{\imath}}}

\newcommand{\half}{\rfrac{1}{2}}

\newcommand{\der}{\mathrm{d}}
\newcommand{\diff}[2][.]{\mathinner{\der#2\if #1.\else^{#1}\fi}}

\newcommand{\field}[1]{\mathbb{#1}}

\newcommand{\Integers}{\field{Z}}

\newcommand{\Complex}{\field{C}}
\newcommand{\ComplexComplete}{\bar{\Complex}}

\newcommand{\grp}[1]{\mathrm{#1}}
\newcommand{\alg}[1]{\mathfrak{#1}}

\newcommand{\Order}{\mathcal{O}}

\let\qed\relax\newcommand{\qed}
{\hfill\ensuremath{\Box}}

\newcommand{\re}{\ve}

\newcommand{\vt}{\vartheta}
\newcommand{\ve}{\varepsilon}

\newcommand{\dd}{\mathrm{d}}

\newcommand{\zC}{\mathbb C}

\newcommand{\zN}{\mathbb N}

\newcommand{\zQ}{\mathbb Q}

\newcommand{\zZ}{\mathbb Z}

\newcommand{\cF}{\mathcal{F}}       
 
\newcommand{\cH}{\mathcal{H}}

\newcommand{\cL}{\mathcal{L}}

\newcommand{\cO}{\mathcal{O}}

\def\reg{\text{reg}}

\newcommand{\nn}{\nonumber}

\newcommand{\with}{\,|\,}

\newcommand{\abel}[0]{\mathfrak{u}}

\newcommand{\acyc}[0]{\mathfrak A}
\newcommand{\bcyc}[0]{\mathfrak B}
\newcommand{\Acycle}[0]{\mathfrak A}

\newcommand{\genus}{\ensuremath{h}}

\newcommand{\periodmatrix}{\tau}
\newcommand{\Atxt}{$\mathfrak{A}$}
\newcommand{\Btxt}{$\mathfrak{B}$}
\newcommand{\omell}{\omega^{\mathrm{ell}}}

\newcommand{\gkern}[1]{g^{(#1)}}

\newcommand{\skern}[1]{s^{(#1)}}

\newcommand{\funddom}{\cF}

\newcommand{\mindx}[1]{\mathbf{#1}}
\newcommand{\RSurf}{\Sigma_\genus}
\newcommand{\hgMZV}{hgMZV}
\newcommand{\hgMPL}{hgMPL}
\newcommand{\wrt}[1]{w.r.t.~#1}
\newcommand{\zetaA}[1]{\zeta_{\raisebox{-2.8pt}{\scriptsize$\acyc_#1$}}\!}
\newcommand{\zetaAreg}[1]{\zeta_{\raisebox{-2.8pt}{\scriptsize$\acyc_#1$}}^{\mathrm{reg}}}

\newcommand{\zetaAg}[2]{{}^{[#2]}\hspace{-0.4pt}\zeta_{\raisebox{-2.8pt}{\scriptsize$\acyc_#1$}}\!}
\newcommand{\zetaAgt}[2]{{}^{[#2]}\hspace{-0.4pt}\tilde{\zeta}_{\raisebox{-2.8pt}{\scriptsize$\acyc_#1$}}\!}

\newcommand{\hgzeta}{\zetaA}

\newcommand{\hgcorr}{\sigma}
\newcommand{\len}[1]{\left| #1 \right|}

\renewcommand{\mid}{\,|\,}

\DeclareMathOperator{\Sp}{\text{Sp}}

\newcommand{\alt}{\text{Alt}}
\newcommand{\Coeff}{\mathcal{C}}
\newcommand{\COeff}{\mathbf{C}}

\newcommand{\psplt}[2]{\mathrm{Sp}^+_{#1}(#2)}
\newcommand{\nsplt}[2]{\mathrm{Sp}^-_{#1}(#2)}

\newcommand{\bn}[1]{\mathrm{B}_{#1}}
\newcommand{\rn}[1]{\mathrm{R}_{#1}}

\newcommand{\ndiff}{\Omega^{(y-t)}}

\newcommand{\texteq}{\,{=}\,}

\newcommand{\dct}{Cycle~exchange}

\title{\textbf{\texorpdfstring{Higher-genus multiple zeta values}{Higher-genus multiple zeta values}}}
\author{
Konstantin Baune, 
Johannes Broedel, 
Egor Im, 
\texorpdfstring{\\}{}
Zhexian Ji, 
Yannis Moeckli 
}
\date{\today}

\begin{document}

\pdfbookmark[1]{Title Page}{title} 
\thispagestyle{empty}
\vspace*{1.0cm}
\begin{center}%
  \begingroup\LARGE\bfseries\thetitle\par\endgroup
\vspace{1.0cm}

\begingroup\large\theauthor\par\endgroup
\vspace{9mm}
\begingroup\itshape
Institute for Theoretical Physics, ETH Zurich\\Wolfgang-Pauli-Str.~27, 8093 Zurich, Switzerland\\[4pt]
\par\endgroup
\vspace*{7mm}

\begingroup\ttfamily
baunek@ethz.ch, jbroedel@ethz.ch, egorim@ethz.ch,\\
zhexji@ethz.ch, moeckliy@ethz.ch 
\par\endgroup

\vspace*{2.0cm}

\textbf{Abstract}\vspace{5mm}

\begin{minipage}{13.4cm}
Multiple zeta values arise as special values of polylogarithms defined on Riemann surfaces of various genera. Building on the vast knowledge for classical and elliptic multiple zeta values, we explore a canonical extension of the formalism to Riemann surfaces of higher genera, which yields higher-genus multiple zeta values. \\
We provide a regularization prescription for higher-genus polylogarithms, which we extend to higher-genus multiple zeta values. Our regularization uses the Schottky uniformization to trace back higher-genus endpoint regularization to known regularization at genus one.\\ 
Additionally, we are commenting on relations among higher-genus multiple zeta values implied by degeneration of the underlying geometry, where we distinguish between the two types of separating and non-separating degeneration.\\
Finally, employing functional relations for higher-genus polylogarithms in the Schottky uniformization, we explore relations among higher-genus multiple zeta values and check them against our numerical testing setup. We identify relations for higher-genus multiple zeta values beyond those implied by polylogarithm identities, thereby matching the situation for genus zero and genus one. While we find several known structures for elliptic multiple zeta values to generalize to relations for higher-genus multiple zeta values, there are further classes of relations arising from the interplay and combinatorics of different cycles.
\end{minipage} 
\end{center}
\vfill

\newpage
\setcounter{tocdepth}{2}
\tableofcontents

\newpage

\section{Introduction}\label{sec:introduction}

The persisting and long-lasting interest in multiple zeta values (MZVs) originates from two directions of research: on the one hand --- being special values of polylogarithms --- MZVs make a prominent appearance in calculations in quantum field theory and string theory. On the other hand, MZVs are of interest to number theorists and algebraic geometers: exploiting the rich algebraic structure of MZVs allows to investigate their properties ranging from transcendental behavior to functional relations. While MZVs found their entry into mathematics from the sum representation~\cite{Basel1,Basel2} initially, it was their representation as iterated integrals over (mostly single-pole) differentials, which allowed better control over analytic and algebraic properties and paved the way towards various generalizations of MZVs. 

Numerous algebraic, analytic and numeric tools towards treating and calculating iterated integrals on Riemann surfaces of higher genus have been developed recently. In this article, we will employ iterated integrals of Enriquez' kernels and their Schottky representation to define higher-genus multiple zeta values as iterated integrals evaluated on non-contractible cycles of Riemann surfaces. Before describing the content of this article in more detail below, we are going to take a bird's-eye perspective on the field of MZVs on Riemann surfaces of genus zero and genus one. The exploration of single-valued versions of higher-genus multiple zeta values is not part of the current article: accordingly this field of research will not be taken into account in the two following paragraphs. 

\paragraph{Genus zero.} We will refer to numbers arising as special values of polylogarithms~\cite{GONCHAROV1995197,Goncharov:1998kja,Remiddi:2000,Goncharov:2001iea,GiNaC:2005}, and which are defined as iterated integrals over Abelian differentials of a genus-zero Riemann surface, as \textit{classical MZVs} or \textit{genus-zero MZVs}.
Alternatively (and equivalently), one can translate the sum representation of classical polylogarithms to a sum representation of genus-zero MZVs. Since every Riemann surface of genus zero is isomorphic to the Riemann sphere, there is no parameter describing the genus-zero geometry: classical MZVs are real numbers. 

The theory of classical MZVs has been developed and understood over many years, see for example \rcites{Chen,Borwein1997,Hoffman:MHS,Zagier23,BrownThesis} and \rcite{FresanGil} for a review. The algebraic structure of their integral representation, in conjunction with a suitable regularization procedure as well as an understanding as periods of a category of mixed Tate motives (see e.g. \rcites{DG,Brown1102.1312}), led to a translation of functional relations for polylogarithms to MZVs and finally to the identification of a basis for MZVs of given transcendentality~\cite{browndepthgraded,Glanois:Thesis}. Furthermore, these structures paved the way for calculational tools like the MZV data mine~\cite{Blumlein:2009cf}. In quantum field theory and string theory, most recent calculations employ MZVs and generalizations thereof, see \rcites{SnowmassString,LocalizationReview,Alday:2023mvu} for examples in different research directions.   

\paragraph{Genus one.} Zeta values which can be represented as special values of elliptic polylogarithms, that is, as iterated integrals over Abelian differentials along either the \Atxt- or \Btxt-cycle of a Riemann surface of genus one, are called \textit{elliptic MZVs} (eMZVs). Different from their genus-zero cousins, eMZVs are functions of the geometry of the surface. By a slight abuse of nomenclature, they will be referred to as \textit{values} anyway. Depending on the integration cycle, eMZVs come as \Atxt-cycle or \Btxt-cycle eMZVs, respectively. They represent the same class of functions and can be converted into each other (see e.g.~\rcites{Enriquez:EllAss,Enriquez:Emzv}).

Based on the theory of elliptic polylogarithms~\cite{Levin, BrownLevin}, eMZVs appeared first in \rcite{Enriquez:Emzv}, have been used and developed for one-loop open-string scattering in \rcite{Broedel:2014vla}, have been exposed to underlying (derivation) algebra structures in \rcites{Pollack, CEE} and different languages in \rcites{Broedel:2015hia, Broedel:2017kkb}. Moreover, they found numerous applications in particle physics phenomenology, see \rcite{Bourjaily:2022bwx} for an overview. An eMZV data mine has been established in \rcite{eMZVWebsite}. Relations among eMZVs have been explored in different formulations mentioned above, see \rcite{Matthes:Thesis} and the references therein for a review. Several functional relations have been investigated in \rcites{ZagierGangl,Broedel_2020}. Statements about transcendentality of elliptic polylogarithms and eMZVs can be found in e.g.~\rcites{Broedel:2018qkq,DHoker:2019blr}.

\medskip

In \secref{sec:reviewRS} and \secref{sec:rev01} of this article we set the stage by reviewing higher-genus Riemann surfaces, with a special focus on the Schottky uniformization, which is a key component of our constructions later on, as well as polylogarithms and MZVs for genus zero and one. Afterward, we review the construction of Enriquez' higher-genus integration kernels in \secref{sec:revKernels} as well as of the higher-genus multiple polylogarithms based on them. We also review the construction of the kernels in terms of integrals over the prime form and as Poincar\'e series on a Schottky cover.

\paragraph{Higher genus.} In this article, we are considering periods defined as iterated integrals over different \Atxt-cycles of a Riemann surface of genus higher than one. The resulting objects are called \textit{higher-genus MZVs} (\hgMZV{}s) and we will use for example the symbol\footnote{Please find the definition in \secref{sec:hgMZV}.} 
\begin{equation}\label{eqn:zetanotation}
	\zetaAg{1}{3}(121,1^232\mid\SGroup)
\end{equation}
to denote a genus-three \hgMZV{} obtained as a depth-two iterated integral around the cycle $\acyc_1$ over the differential forms (Enriquez' kernels) $\omega_{121}$ and $\omega_{1132}$. 
The resulting hgMZVs are going to be functions of the underlying complex geometry. The dependence will be mathematically expressed either by the period matrix $\tau$ or --- equivalently, as in the above example --- by the Schottky group $\SGroup$. If clear from the context, we are going to omit the genus information as well as the dependence on the geometry. Our notation for \hgMZV{}s in \eqn{eqn:zetanotation} is related to previous notations for (e)MZVs as summarized in \tabref{tab:zetanotation}.  

\begin{table}[h!]
\centering
\def\arraystretch{1.5}
\begin{tabular}{@{}ll@{}}
\toprule
\textbf{Genus zero:}   & $\displaystyle\zeta(2)=\zeta_{0,1}={}^{[0]}\zeta_{0,1}=\int_0^1 \omega^{\lbrace 0\rbrace}\circ\omega^{\lbrace 1\rbrace}=\int_0^1 \frac{\dd t_1}{t_1-0}\int_0^{t_1}\frac{\dd t_2}{t_2-1}$ \\
	\textbf{Genus one:}    & $\displaystyle\omell(1,2)=\int_0^1\!\gkern{1}(t_1\mid\tau)\!\int_0^{t_1}\!\gkern{2}(t_2\mid\tau)\,\dd t_2\, \dd t_1\propto\zetaAg{1}{1}(1^2,1^3)=\int_0^1\omega_{11}\circ\omega_{111}$ \\
	\textbf{Higher genus:} & $\displaystyle\zetaAg{1}{3}(121,1^232\mid\SGroup)=\int_{\acyc_1}\omega_{121}\circ\omega_{1132}$ \\
\bottomrule
\end{tabular}
\caption{Connecting common notations for multiple zeta values at different genera to our notation.}
\label{tab:zetanotation}
\end{table}
In order to properly define \hgMZV{}s, one will have to deal with endpoint divergences. As a preparation, we suggest a regularization scheme for higher-genus polylogarithms in \secref{sec:hgMPLreg}, which employs the Schottky uniformization in order to trace back the higher-genus regularization to the known regularization of genus-one polylogarithms. In order to extend this regularization from depth-one polylogarithms to those of higher depth, we need additional linear relations available only for Enriquez' kernels at genera higher than one. As a further preparation to the investigation of \hgMZV{}s, we investigate integration of Enriquez' kernels along closed \Atxt-cycles in the Schottky language in \secref{sec:TechniquesSchottky}, providing an explicit example in \secref{sec:thm2}, by considering statements of \rcite{DHoker:2025dhv} in the Schottky language.  

The techniques explored and developed in \secref{sec:TechniquesSchottky} highlight the benefit of the Schottky uniformization: the convoluted cycle-integrals on Riemann surfaces can be tackled using methods of complex analysis together with combinatorics applied to the Poincar\'e series over Schottky groups. This facilitates the analysis and proof of several sophisticated relations among \hgMZV{}s investigated in \secref{sec:hgmzvrel}.

Definitions and the theory of \hgMZV{}s are unfolded in \secref{sec:hgMZV}. We devote \secref{sec:hgmzvreg} to extending the regularization of higher-genus polylogarithms discussed in \secref{sec:hgMPLreg} to \hgMZV{}s: other than for polylogarithms, there can be divergences to take care of at both ends of the closed integration cycle. We conclude the section by giving an explicit formula for depth-one \hgMZV{}s in \secref{sec:relations_depth_one}, obtained using the Schottky uniformization, which was already proven in \rcites{EnriquezHigher,EZ1}.

Identities for \hgMZV{}s related to the degeneration of the underlying geometry are discussed in \secref{sec:degen} using mainly the Schottky language. We consider two types of degenerations of Riemann surfaces: non-separating (in \secref{sec:nonseparatingA}) and separating (in \secref{sec:separatingB}). Both types of degeneration are realized as particular limits applied to the Schottky generators and circles, and thus have a clear geometric interpretation. We discuss the effect of the degeneration on Enriquez' kernels and the associated \hgMZV{}s, connecting \hgMZV{}s at different genera in the degeneration limit.

In \secref{sec:hgmzvrel}, we discuss and investigate relations among \hgMZV{}s of equal genus. We start with the Fay(-like) identities for Enriquez' integration kernels, which imply relations among the corresponding \hgMZV{}s, in \secref{sec:Fayid}. These identities are higher-genus echoes of mechanisms already known from genus one. 
In \secref{sec:relations_cyc_exchange}, we discuss and prove the validity of an identity relating \hgMZV{}s obtained as integrals over different \Atxt-cycles. We furthermore comment on possible extensions of this identity.
We will consider a set of identities valid only for hyperelliptic Riemann surfaces in \secref{sec:altid}, where the enhanced symmetry provided by the hyperelliptic involution triggers the existence of new identities. 

\secref{sec:weightex} and \secref{sec:emzvhemzv} are devoted to the identities for which we have conjectures backed by extensive numerical checks, but which currently lack an analytic proof. In  \secref{sec:weightex}, we discuss the weight-exchange identities at depth two: these can be viewed as a generalization of the reflection identities known for eMZVs.
In \secref{sec:emzvhemzv}, we discuss in which way identities available for eMZVs will cause a higher-genus echo. Naturally, one would like to identify a genus-lifting mechanism as counterpart to the degeneration identities, but things do not seem to be straightforward.

Several of the starting points for identities mentioned above for \hgMZV{}s have been influenced by a numerical survey using our package \texttt{SchottkyTools} \cite{SchottkyTools}, which in principle allows numerical evaluation of multiple polylogarithms at any genus. Practically, evaluations are of course constrained by time and the accuracy goal. All the relations mentioned in \secref{sec:hgMZV}, \secref{sec:degen} and \secref{sec:hgmzvrel} (regardless if proven analytically or not) were verified numerically for generic genus-two and genus-three \hgMZV{}s of depth two and three of total weight up to six.

Within the current article, we are neither going to touch upon \Btxt-cycle \hgMZV{}s and their relation to the corresponding \Atxt-cycle \hgMZV{}s nor on the question of existence and construction of single-valued \hgMZV{}s. Those topics and questions as well as further considerations are collected in \secref{sec:openqs}. Several appendices supplement this article, showing details for some of our computations.

\subsection*{Acknowledgments}
We are grateful to Manuel Berger for various discussions and work on related projects.
We also thank Manuel Berger, Eric D'Hoker, Artyom Lisitsyn and Oliver Schlotterer for carefully proof-reading this manuscript and providing valuable feedback.
The work of K.B., J.B., E.I., and Y.M. is partially supported by the Swiss National Science Foundation through the NCCR SwissMAP. 

\section{Higher-genus Riemann surfaces}\label{sec:reviewRS}

In this section we will first review the basics of higher-genus Riemann surfaces (without specifying a uniformization) in \subsecref{sec:basicConcepts} before focusing on the main ingredients of the Schottky uniformization for Riemann surfaces in \subsecref{sec:Schottky}, which will be heavily used in later sections.

\subsection{Basic concepts}\label{sec:basicConcepts}

\subsubsection{Cycles and period matrices}\label{sec:cycles}
We consider a Riemann surface $\RSurf$ of genus $\genus\,{\geq}\,1$ with non-trivial cycles $\acyc_1,\ldots,\acyc_\genus$, $\bcyc_1,\ldots,\bcyc_\genus$ generating the fundamental group of $\RSurf$ such that their intersections are
\begin{equation}
	\acyc_i\#\acyc_j=0=\bcyc_i\#\bcyc_j,\quad \acyc_i\#\bcyc_j=\delta_{ij},\quad\text{for }i,j\in\{1,\ldots,\genus\}.
\end{equation}
We refer to the $\acyc_i$ and $\bcyc_i$ as \Atxt- and \Btxt-cycles, respectively, where we understand the \Atxt-cycles to be going around the \genus{} handles and the \Btxt-cycles around the \genus{} holes of the Riemann surface~$\RSurf$ (cf.~\figref{fig:genustwo}). 

For the \genus{} independent holomorphic differentials $\omega_i$, $i\,{\in}\,\{ 1,\ldots,\genus\}$, on $\RSurf$ we choose the normalized basis satisfying
\begin{equation}\label{eqn:periodmatrix}
	\int_{\acyc_i}\omega_j=\delta_{ij},\quad\int_{\bcyc_i}\omega_j=\tau_{ij}.
\end{equation}
Consequently, all geometric properties of $\RSurf$ are captured by the \Btxt-cycle periods of the holomorphic differentials through the period matrix $\tau\texteq(\tau_{ij})$. The period matrix $\tau$ is symmetric ($\tau_{ij}\texteq\tau_{ji}$) with positive definite imaginary part ($\Im\tau\,{>}\,0$), such that $\tau$ is an element of the genus-\genus{} Siegel upper half-plane $\cH_\genus$. Given a Riemann surface $\RSurf$, its period matrix is not unique, but two symmetric matrices $\tau,\tau'\,{\in}\,\cH_\genus$ capture the same Riemann surface if they can be related via a modular transformation, i.e.
\begin{equation}\label{eqn:modtranPM}
	\tau'=(a\tau+b)(c\tau+d)^{-1}, \quad\text{with }M=\bigg(\!\!\begin{array}{cc}a&b\\c&d\end{array}\!\!\bigg)\in\Sp(2\genus,\zZ).
\end{equation}
Such a modular transformation accounts for switching to a different basis of cycles $\acyc_i'$, $\bcyc_i'$ and holomorphic differentials $\omega_i'$ related through
\begin{equation}
	\bigg(\!\!\begin{array}{c}\vec{\bcyc}'\\\vec{\acyc}'\end{array}\!\!\bigg)=M\bigg(\!\!\begin{array}{c}\vec{\bcyc}\\\vec{\acyc}\end{array}\!\!\bigg)
	\quad\text{and}\quad
	\vec{\omega}'=\vec{\omega}\,(c\tau+d)^{-1}\,.
\end{equation}

Points $z$ and $z_0$ on (the universal cover of) a Riemann surface $\RSurf$ are mapped to the \textit{Jacobian variety} $J(\RSurf)\texteq\zC^\genus/(\zZ^\genus\,{+}\,\tau\zZ^\genus)$ using \emph{Abel's map} $\abel$, which is defined as
\begin{equation}
	\abel_i(z,z_0\mid\tau)=\int_{z_0}^z\omega_i.
\end{equation}
\subsubsection{Hyperelliptic Riemann surfaces}\label{sec:hyperelliptic}
A compact Riemann surface $\RSurf$ for $\genus\,{\geq}\, 2$ is called \textit{hyperelliptic} if there exists a two-to-one branched covering $\pi\colon\RSurf\,{\to}\,\ComplexComplete\texteq\zC\,{\cup}\,\lbrace\infty\rbrace$ of the Riemann sphere. In the language of algebraic curves, this two-sheeted covering can be represented as
\begin{equation}
y^2=P(x)\,,
\end{equation}
where $P(x)$ is a polynomial of degree $2\genus\,{+}\,1$ or $2\genus\,{+}\,2$ with distinct roots. The corresponding hyperelliptic involution 
\begin{equation}\label{eqn:hyperellipticinvolution}
	\iota\colon(x,y)\to(x,-y)
\end{equation}
fixes the branch points. Accordingly, the quotient of the Riemann surface and the involution yields the Riemann sphere: $\RSurf/\iota\,{\simeq}\,\ComplexComplete$.

\subsection{Schottky uniformization}\label{sec:Schottky}

\subsubsection{Construction}\label{sec:Schottly_defs}

There are different ways to parameterize compact Riemann surfaces: one way that has proven useful for considering meromorphic kernels for higher-genus polylogarithms is the \emph{Schottky uniformization}~\cite{Schottky1887}. We are going to provide a lightning review in this subsection, whereas a complete introduction can be found in \rcites{Bobenko:2011,herrlich1schottky}. More details on Schottky uniformization in the context of higher-genus polylogarithms is provided in \rcite{Baune:2024biq}.

Let $C_1, \ldots, C_\genus, C_1^{\prime}, \ldots, C_\genus^{\prime}$ be $2\genus$ mutually disjoint Jordan curves on~$\ComplexComplete$. The \emph{Schottky group}~$\SGroup$ is then a discrete subgroup of $\grp{PSL}(2, \Complex)$, which is generated by $\genus$ loxodromic transformations $\sigma_1,\ldots ,\sigma_\genus\,{\in}\,\grp{PSL}(2,\Complex)$. Each Schottky transformation $\sigma_i$ is defined by mapping the exterior of $C_i$ to the interior of $C_i^{\prime}$. We call the transformations $\sigma_i$, $i\,{\in}\,\{1,\ldots,\genus\}$, \emph{Schottky generators}\footnote{In contrast to \rcite{Baune:2024biq}, we use the symbols $\sigma_j$ for the Schottky generators instead of $\gamma_j$, as the latter will be used to denote elements of certain cosets of the Schottky group below.}.

Each Schottky generator $\sigma_i$ can be characterized by its two fixed points $P_i,P_i'$, defined as
\begin{equation}\label{eqn:fixedpoints}
	P_i\,{\coloneqq} \lim_{n\to\infty} \sigma_i^{-n}z\qquad \text{and}\qquad P_i'\,{\coloneqq} \lim_{n\to\infty} \sigma_i^{n}z,
\end{equation}
and by a multiplier~$\lambda_i$, which can be chosen with $0 \,{<}\,|\lambda_i|\,{ <}\, 1 $. From the definition, it is clear that $P_i'\,{\in}\,\mathring{C_i'}$ and $P_i\,{\in}\,\mathring{C_i}$, where $\mathring{S}$ denotes the interior of a subset $S\,{\in}\,\ComplexComplete$. We call the fixed point $P_i'$ the \textit{attractive} fixed point whereas $P_i$ is called \textit{repulsive}. This terminology arises from the fact that repeated application of $\sigma_i$ to a generic point $z_0\,{\in}\,\ComplexComplete\setminus\{P_i',P_i\}$ results in $|z_0\,{-}\,P_i'|\,{\rightarrow}\,0$ in the limit of infinitely many applications of $\sigma_i$. On the other hand, $|z_0\,{-}\,P_i|$ is bounded from below independently of the number of applications of $\sigma_i$. The generator $\sigma_i$, the multiplier $\lambda_i$ and the two fixed points $P_i',P_i$ are related by the \textit{fixed point equation}
\begin{equation}\label{eqn:fixedpteq}
	\frac{\sigma_iz-P_i'}{\sigma_iz-P_i}=\lambda_i\frac{z-P_i'}{z-P_i}.
\end{equation}
The \emph{Schottky uniformization theorem} (cf.~\cite[Thm.~29]{Bobenko:2011}) ensures that for every Riemann surface $\RSurf$, there exists a Schottky group $\SGroup$ uniformizing $\RSurf$. The Riemann surface is then given by
\begin{equation}
	\RSurf=\Omega(\SGroup)\slash\SGroup\, ,
\end{equation} 
where $\Omega(\SGroup)\,{\subset}\,\ComplexComplete$ is the \emph{region of discontinuity} of $\SGroup$. The connected region $\funddom\,{\subset}\,\ComplexComplete$ bounded by all $2\genus$ curves $\lbrace C_i,C_i' \rbrace_{i=1}^h$ is called the \textit{fundamental domain} of the Riemann surface. When the Jordan curves $C_i, C_i^{\prime}$ are restricted to circles, a \textit{classical Schottky group} is obtained. In the following, we are going to focus on classical Schottky groups. 

\begin{figure}[t]
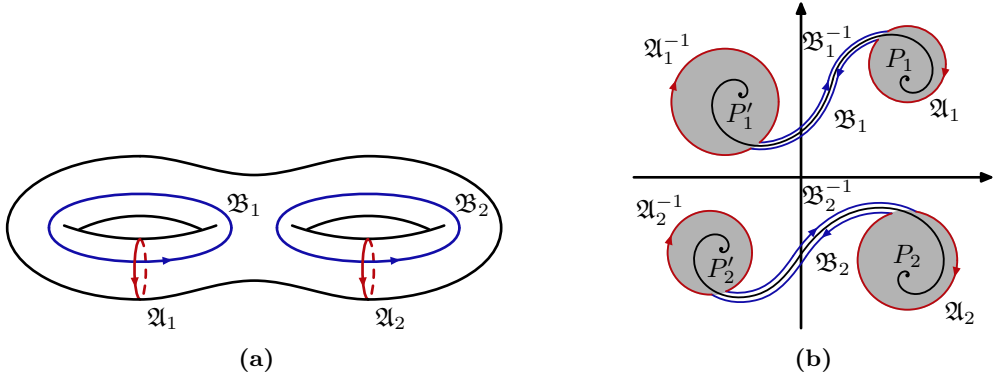

	\centering
    \begin{subfigure}[b]{0.45\textwidth}
        \centering 
        \mpostuse{genustwo}
        \caption{}
    \end{subfigure}
    \begin{subfigure}[b]{0.45\textwidth}
        \centering 
        \mpostuse{genustwoSchottky}
        \caption{}
    \end{subfigure}
	\caption{Schottky uniformization (b) of the genus-two Riemann surface in (a). \Atxt- and \Btxt-cycles are drawn in red and blue, respectively. The fundamental domain $\funddom$ of the Schottky cover is the region outside all circles.}
	\label{fig:genustwo}
\end{figure}
In \figref{fig:genustwo} we sketch a Riemann surface of genus two and its corresponding Schottky uniformization. The (canonical) homology basis $\acyc_1,\ldots,\acyc_\genus$ and $\bcyc_1,\ldots,\bcyc_\genus$ on the Schottky cover is chosen in the following way: upon projecting from the Schottky cover to the Riemann surface, the circles $C_i$ (and all their copies under the Schottky group) descend to a representative of the homology class\footnote{By abuse of notation, we use the same notation $\acyc_i$ and $\bcyc_i$, $i\,{\in}\,\{1,\ldots,\genus\}$ for the homology basis and its representatives, i.e.~the actual (non-contractible) paths on the Riemann surface. Moreover, we will not distinguish between $\acyc_i$ on the Riemann surface and the path representing it on the covering space. Accordingly, we will freely switch between these notions.}
of $\acyc_i$ for all $i\,{\in}\,\{1,\ldots,\genus\}$. Accordingly, we define the circle $C_i$ with negative (i.e.~clockwise) orientation to be the path on the covering space representing the cycle $\acyc_i$ on the Riemann surface. Furthermore, we represent $\bcyc_i$ by the path on the covering space connecting $C_i$ and $C_i'$, starting at $C_i$ and ending at $C_i'$. More precisely, the path on the Schottky cover representing $\bcyc_i$ connects a point $z_0$ on $C_i$ with the point $\sigma_iz_0$ on $C_i'$, hence $\bcyc_i$ can be identified with the action of the generator $\sigma_i$ on the Schottky cover. The conventions for both $\acyc_i$ and $\bcyc_i$ are illustrated in \figref{fig:genustwo}. Finally, let us introduce some further terminology associated to $\acyc_i$. By definition of the Schottky cover, the $\acyc_i$ are Jordan curves on the Riemann sphere and therefore divide $\ComplexComplete$ into two disjoint connected subsets by the Jordan curve theorem. We call the two resulting subsets the \textit{sides} of the curve. Moreover, for a given orientation of the curve, we refer to the side to the right of the curve as being \textit{enclosed by the curve}. Notice that this ensures that $\acyc_i$, i.e.~the circle $C_i$ with negative orientation, encloses the bounded side (see \figref{fig:genustwo}). 

Let us introduce the subgroups $\SGroup_i\,{\coloneqq}\,\langle\sigma_i\rangle\,{\leq}\, \SGroup$, for $i\,{\in}\,\{1,\ldots ,\genus\}$ and consider the right and left cosets
\begin{align}\label{eqn:cosetdefinition}
	\SCosetR{i} = \{ \sigma_{j_1}^{n_1} \cdots \sigma_{j_k}^{n_k} \mid \sigma_{j_k} \neq \sigma_i \}\,, \nonumber\\
	\SCosetL{i} = \{ \sigma_{j_1}^{n_1} \cdots \sigma_{j_k}^{n_k} \mid \sigma_{j_1} \neq \sigma_i \}\,.
\end{align}
Then, the normalized holomorphic differentials of a Riemann surface (cf.~\eqn{eqn:periodmatrix}) can be expressed on the Schottky cover with Schottky group $\SGroup$ as the Poincar\'e series 
\begin{equation}\label{eqn:schottky-holomorphic-basis}
\begin{aligned}
	\omega_i(z\mid\SGroup) 
	&= \frac{1}{2\pi \iunit}\sum_{\gamma \in \SCosetR{i}} \left(\frac{1}{z-\gamma P_i'} - \frac{1}{z-\gamma P_i}\right) \der z \\
    &= \frac{1}{2\pi \iunit}\sum_{\gamma \in \SCosetL{i}} \left(\frac{1}{\gamma z-P_i'} - \frac{1}{\gamma z-P_i}\right) \der(\gamma  z),
\end{aligned}
\end{equation}
where $P_i$ and $P_i'$ are the fixed points defined in \eqn{eqn:fixedpoints}. Employing the above representation of the holomorphic differential, Abel's map can be calculated directly on the Schottky cover: for a pair of points $z$, $z_0\,{\in}\,\funddom$ one finds for $i\,{\in}\,\{1,\ldots,\genus\}$
\begin{equation}\label{eq:abel-map-schottky}
	\abel_i(z,z_0\mid\SGroup) = \int_{z_0}^z \omega_i = \frac{1}{2\pi \iunit} \sum_{\gamma \in \SCosetR{i}} \log \{z,\gamma P_i',z_0,\gamma P_i\}\,,
\end{equation}
where the cross-ratio is defined as
\begin{equation}\label{eqn:crossratio}
	\{z_1,z_2,z_3,z_4\} = \frac{(z_1-z_2)(z_3-z_4)}{(z_1-z_4)(z_3-z_2)}.
\end{equation}
Moreover, the corresponding period matrix $\tau(\SGroup)$, has entries (see e.g.~\rcite{Baker})
\begin{equation}\label{eq:period-matrix-schottky}
\begin{aligned} 
	\periodmatrix_{ij}(\SGroup) &= \int_{\bcyc_j}\omega_i= \frac{1}{2\pi\iunit}\sum_{\gamma\in\SCosetLR{j}{i}}\log\{P_j', \gamma P_i', P_j, \gamma P_i\}\,,\quad i\neq j,\\
	\periodmatrix_{ii}(\SGroup) &= \int_{\bcyc_i}\omega_i= \frac{1}{2\pi\iunit}\log\lambda_i+\frac{1}{2\pi\iunit}\sum_{\substack{\gamma\in\SCosetLR{i}{i}\\\gamma\neq\id}}\log\{P_i', \gamma P_i', P_i, \gamma P_i\}\,.
\end{aligned}
\end{equation}
The Schottky uniformization of a Riemann surface is not unique: conjugation by a M\"obius transformation $\mu\,{\in}\,\grp{PSL}(2,\mathbb C)$ corresponds to a different choice of coordinates on $\ComplexComplete$ and leaves the Riemann surface unchanged. Accordingly, the cross-ratio~\eqref{eqn:crossratio} --- and thus also Abel's map and the period matrix --- are invariant under $\mu$. As each point on the Schottky cover transforms according to $z\,{\mapsto}\, \mu z$, the Schottky group elements $\gamma\,{\in}\, \SGroup$ transform as $\gamma\,{\mapsto}\, \mu \gamma \mu^{-1}$. Furthermore, one can show the holomorphic differentials to be invariant under the M\"obius transformation.

In later sections we will often be interested in objects defined on genus-one covers defined through the Schottky group $\SGroup_j$, which we will refer to as \emph{subcovers}. For brevity of notation, we introduce the following shorthands for these subcovers\footnote{Note the abuse of notation for $\abel_j$ below, which here refers to a one-dimensional Abel map on the $j$-th subcover, not to be confused with the $j$-th component of the \genus{}-dimensional Abel map defined above in \eqn{eq:abel-map-schottky}.}:
\begin{subequations}\label{eqn:shorthands}
\begin{align}
	&\tau_j:=\tau(\SGroup_j),\\
	&q_j:=\exp(2\pi\iunit\tau_j) = \lambda_j,\\
	&\abel_j(z,z_0):=\abel(z,z_0\mid\!\SGroup_j).
\end{align}
\end{subequations}
%

\subsubsection{Schottky uniformization for hyperelliptic surfaces}\label{sec:hyperellipticSchottky}

Schottky uniformizations for hyperelliptic surfaces inherit a substantial amount of symmetry from the hyperelliptic involution~\eqref{eqn:hyperellipticinvolution}. Those symmetries imply certain additional identities valid only on hyperelliptic Riemann surfaces. Our prime example is the alternating identity described and explored in \secref{sec:altid}.

As a preparation for dealing with these identities below, consider the following theorem, in which the hyperelliptic property is related to particularly symmetric configurations of the repulsive and attractive fixed points of the individual Schottky generators:
\begin{theorem}[Keen~\cite{Keen_1980}]\label{thrm:hyper}
    Let $\RSurf$ be a hyperelliptic Riemann surface of genus $\genus\,{\geq}\,1$. Then $\RSurf$ can be uniformized by a Schottky group\footnote{Notice that this group will in general not be classical, i.e.~its generators are not necessarily defined via circles.} $\SGroup$, which admits a set of generators $\varsigma\texteq\{\sigma_1,\ldots,\sigma_\genus\}$ satisfying
    \begin{equation}\label{eqn:hyperellFP}
        P_j' = -P_j \, , \quad \forall\,j\in\{1,\ldots,\genus\},
    \end{equation}
    for $P_j'$, $P_j$ the fixed points of $\sigma_j$.
\end{theorem}
\begin{proof}
    The proof is essentially carried out in \rcite{Keen_1980}. In particular, ref.~\cite[Lemma 2]{Keen_1980} implies that $\SGroup$ admits a set of generators $\varsigma\texteq\{\sigma_1,\ldots,\sigma_\genus\}$, whose fixed points $P_j', P_j$ satisfy\footnote{Our definition of the cross-ratio differs from the convention used in \rcite{Keen_1980} by the permutation $(23)\,{\in}\, S_4$, where a permutation $\rho\,{\in}\, S_4$ acts on the cross-ratio as $\rho(\{z_1,z_2,z_3,z_4\})\texteq\{z_{\rho(1)},z_{\rho(2)},z_{\rho(3)},z_{\rho(4)}\}$.}
    \begin{equation}\label{eqn:proofhyperellFP}
        \{P_j',Q',P_j,Q\}=-1
    \end{equation}
    for all $j\,{\in}\,\{1,\ldots,\genus\}$. Here, $(Q',Q)$ is a fixed pair of complex numbers (or $\infty$) called the \textit{harmonic conjugates} of $\{(P_j',P_j)\}_{i=1}^\genus$. Now we can use the freedom of the Schottky uniformization under conjugation with an arbitrary M\"obius transformation: consider $\mu\,{\in}\,\grp{PSL}(2,\zC)$ to be a transformation that maps $Q'\,{\mapsto}\,0$ and $Q\,{\mapsto}\,\infty$. Since the cross-ratio is invariant under M\"obius transformations, \eqn{eqn:proofhyperellFP} immediately implies
    \begin{equation}
        -1 = \{P_j',Q',P_j,Q\} = \{\mu P_j',0,\mu P_j,\infty\} = \frac{\mu P_j'}{\mu P_j} \, .
    \end{equation}
    From this, we can see that the desired Schottky group is given by $\tilde\SGroup\texteq \mu\SGroup \mu^{-1}$ with generators $\tilde\varsigma\texteq\mu\varsigma \mu^{-1}$, which concludes the proof.
\end{proof}
An example of a Schottky cover representing a hyperelliptic surface is shown in \figref{fig:hyper}. This symmetry translates into a rather technical relation for the cross-ratios involving the fixed points of the generators, as collected in the following lemma:
\begin{figure}[t]
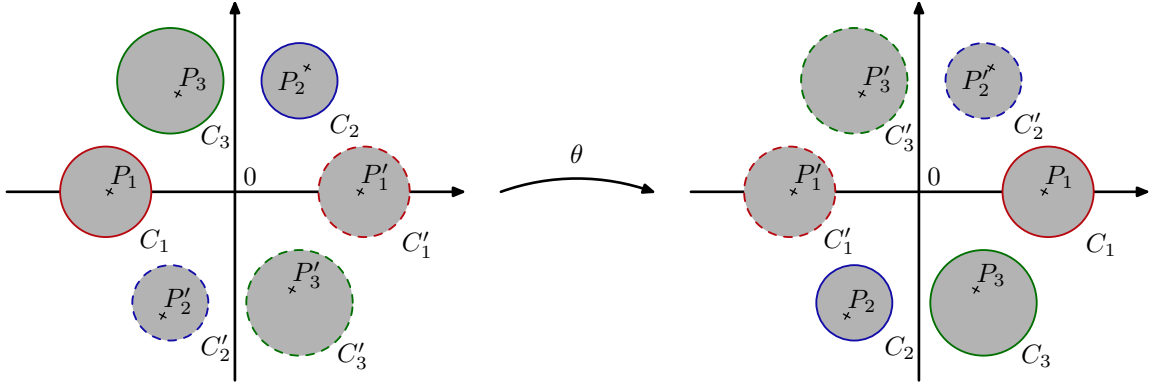

	\centering
	\mpostuse{hyperelliptic}
	\caption{A genus-three Schottky cover with a circle setup such that the fixed points fulfill $P_j'\texteq{-}\,P_j$, making this surface manifestly hyperelliptic. Applying the transformation $\theta$ from the proof of \lemref{lemma:altid} swaps the fixed points, leading to the picture on the right side.}
	\label{fig:hyper}
\end{figure}
\begin{lemma}\label{lemma:altid}
    Let $\RSurf$ be a hyperelliptic Riemann surface of genus $\genus\,{\geq}\, 2$. Then, there exists a Schottky group $\SGroup$ uniformizing $\RSurf$, which admits a set of generators $\sgen_1,\ldots,\sgen_\genus$ such that 
    \begin{equation}
        \frac{(\tilde\gamma P_{e_1}'- P_{e_2}')(\tilde\gamma P_{e_1}- P_{e_2})}{(\tilde\gamma P_{e_1}-P_{e_2}')(\tilde\gamma P_{e_1}'- P_{e_2})}=\frac{(\gamma P_{e_1}'- P_{e_2}')(\gamma P_{e_1}- P_{e_2})}{(\gamma P_{e_1}- P_{e_2}')(\gamma P_{e_1}'- P_{e_2})} \, ,
    \end{equation}
    where $e_1,e_2\,{\in}\,\{1,\ldots,\genus\}$, $\gamma\,{\in}\,\SGroup$ and $P_{e_1}', P_{e_1}, P_{e_2}', P_{e_2}\,{\in}\,\ComplexComplete$ denote the fixed points of the generators $\sigma_{e_1}$ and $\sigma_{e_2}$, respectively. The element $\tilde\gamma\,{\in}\,\SGroup$ is defined by replacing $m_l\,{\rightarrow}\, {-}\,m_l$ in $\gamma\texteq\sigma_{j_1}^{m_1}\cdots\sigma_{j_s}^{m_s}$ for all $l\,{\in}\,\{1,\ldots,s\}$.
\end{lemma}
\begin{proof}
    Schematically, the idea of the proof is to find an appropriate M\"obius transformation $\theta$ that performs the exchange $\gamma\,{\rightarrow}\,\tilde\gamma$ in the cross-ratio. Conjugating with this function then directly yields the statement of the lemma. Let us make this more precise. Since the surface $\RSurf$ is hyperelliptic, \thmref{thrm:hyper} guarantees that we can find generators $\sgen_i$, $i\,{\in}\,\{1,\ldots,\genus\}$ such that their fixed points satisfy $P_i\texteq{-}\,P_i'$ for all $i\,{\in}\,\{1,\ldots,\genus\}$. Let now $\theta(z)\texteq{-}\,z$ be the M\"obius transformation\footnote{For $\genus\texteq2$, the map $\theta$ can be constructed for general configurations of the Schottky uniformization, which reflects the fact that every Riemann surface of genus two is hyperelliptic.} permuting the fixed points $P_i'$ and $P_i$ for all $i\,{\in}\,\{1,\ldots,\genus\}$. As a next step, we want to show that $\theta$ indeed implements the desired transformation $\gamma\,{\rightarrow}\,\tilde\gamma$. To do this, note that
    \begin{align}
        \sgen_i\theta(P_i')&=\sgen_i(P_i)=P_i,
    \end{align}
    as well as $\sgen_i\theta(P_i)\texteq P_i'$, $i\,{\in}\,\{1,2\}$, accordingly. Upon conjugation with an appropriate transformation sending $P_i'\,{\rightarrow}\,0$ and $P_i\,{\rightarrow}\,\infty$, this shows that $\sgen_i\theta$ is an inversion, i.e.~$\sigma_i\theta(z)\texteq\frac{c}{z}$ for some $c\,{\in}\,\zC^\times$, in this specific coordinate system where $P_i'\texteq0$ and $P_i\texteq\infty$. Since inversions are involutions\footnote{A M\"obius transformation $\mu$ is an involution if it satisfies $\mu^2\texteq\id$.}, we can conclude that $(\sgen_i\theta)^2\texteq\id$ in this coordinate system. However, if a M\"obius transformation is an involution in one coordinate system, it is so in every coordinate system. Thus, we have shown that $(\sgen_i\theta)^2\texteq\id$ generically. This now implies that $\theta\sigma_i\theta^{-1}\texteq\sgen_i^{-1}$, $i\,{\in}\,\{1,\ldots,\genus\}$, after restructuring and applying $\theta^{-1}\texteq \theta$. We can now see that $\theta$ indeed satisfies $\theta\gamma\theta^{-1}\texteq\tilde\gamma$ as it replaces each generator by its inverse. This directly yields the statement as
    \begin{equation}
        \begin{aligned}
            \frac{(\gamma P_{e_1}'-P_{e_2}')(\gamma P_{e_1}-P_{e_2})}{(\gamma P_{e_1}'-P_{e_2})(\gamma P_{e_1}-P_{e_2}')} &= \frac{((\theta\gamma \theta^{-1})\theta(P_{e_1}')-\theta(P_{e_2}'))((\theta\gamma \theta^{-1})\theta(P_{e_1})-\theta(P_{e_2}))}{((\theta\gamma \theta^{-1})\theta(P_{e_1}')-\theta(P_{e_2}))((\theta\gamma \theta^{-1})\theta(P_{e_1})-\theta(P_{e_2}'))} \\
            &= \frac{(\tilde\gamma P_{e_1}-P_{e_2})(\tilde\gamma P_{e_1}'-P_{e_2}')}{(\tilde\gamma P_{e_1}-P_{e_2}')(\tilde\gamma P_{e_1}'-P_{e_2})} \, ,
        \end{aligned}
    \end{equation}
    where we have used the invariance of the cross-ratio under M\"obius transformations.
\end{proof}
To obtain some intuition for the above proof, notice that the map $\theta(z)\texteq{-}\,z$ swaps the roles of the attractive and repulsive fixed points on a hyperelliptic Schottky cover. This induces the transformation $\sgen_i\,{\rightarrow}\,\sgen_i^{-1}$, $i\,{\in}\,\{1,\ldots,\genus\}$, of the generators of the Schottky group, as the attractive (repulsive) fixed point of the generator $\sgen_i^{-1}$ corresponds to the repulsive (attractive) fixed point of $\sgen_i$ for all $i\,{\in}\,\{1,\ldots,\genus\}$.

\subsubsection{Splitting of Schottky words}\label{sec:splitting_of_Schottky_words}

To prepare for later sections, let us define various operations on elements of the Schottky group~$\SGroup$. We are going to refer to particular Schottky group elements as \textit{Schottky words} 
\begin{equation}\label{eqn:reduced}
\gamma = \sgen_{i_1}^{n_1} \sgen_{i_2}^{n_2} \cdots \sgen_{i_k}^{n_k} \in \SGroup
\end{equation}
built from the alphabet of Schottky generators $\sgen_i$ with $i_l\,{\ne}\, i_{l+1}$, $i_l \,{\in}\, \{1,\ldots,\genus\}$ and $n_l \,{\in}\, \Integers \setminus \{0\}$. 
\begin{definition}\label{def:splitting}
For a Schottky word $\gamma\,{\in}\,\SGroup$, we define the \emph{splitting}
\begin{align}
    \psplt{i}{\gamma}\,,\,\nsplt{i}{\gamma}\,\,{\subset}\, \SGroup\,{\times}\, \SGroup
\end{align}
as the set of all possible ways to split $\gamma$ either after $\sgen_i$ or before $\sgen_i^{-1}$, respectively. The \emph{splitting number} is denoted by 
\begin{align}\label{eqn:nsplitting}
    N_i^\pm(\gamma) = \len{\ssplit_i^\pm(\gamma)}\,,
\end{align}
where $\len{S}$ denotes the number of elements in the set $S$.
\end{definition}
\begin{example}
For example, for $\gamma \texteq \sigma_j\sigma_i^2\sigma_j \sigma_i^{-1}\sigma_j$ with $i\,{\neq}\, j$, we find
\begin{align}
    &\ssplit_i^+(\gamma)= \left\{(\sigma_j\sigma_i^2,\sigma_j\sigma_i^{-1}\sigma_j),(\sigma_j\sigma_i,\sigma_i\sigma_j\sigma_i^{-1}\sigma_j)\right\},&\quad N_i^+(\gamma)=2\nonumber\, ,\\
    &\ssplit_i^-(\gamma)=\left\{(\sigma_j \sigma_i^2\sigma_j,\sigma_i^{-1}\sigma_j)\right\},&\quad N_i^-(\gamma)=1 \, .
\end{align}
\end{example}
When working with the above splittings, we often use special subsets or cosets of the Schottky group $\SGroup$:
\begin{definition}
Let $\mathrm{S} \,{\subset}\, \SGroup$ be any subset of the Schottky group, and let $i\,{\in}\,\{1,\ldots,\genus\}$. Then we define
\begin{subequations}\label{eqn:set}
\begin{align}
    \label{eqn:set1}
    \ssubsetl{\mathrm{S}}{\pm}{i}&=\left\{\sigma_i^n\gamma\in \mathrm{S}\with n\in\zZ_\pm,\, \gamma\in\CustomSL{\SGroup_i}{\mathrm{S}}\right\}, \\
    \label{eqn:set2}
    \ssubsetr{\mathrm{S}}{\pm}{i}&=\left\{\gamma\sigma_i^n\in \mathrm{S}\with n\in\zZ_\pm,\, \gamma\in\CustomSR{\mathrm{S}}{\SGroup_i}\right\},
\end{align}
\end{subequations}
where $\zZ_{\pm}\texteq\{n\,{\in}\,\zZ \with {\pm} n\,{>}\,0\}$. 
We also use a shorthand notation for their counterparts
\begin{subequations}\label{eqn:cset}
\begin{align}
    \label{eqn:cset1}
    \csubsetl{\mathrm{S}}{\pm}{i}&=\left\{\sigma_i^n\gamma\in \mathrm{S}\with n\notin\zZ_\pm,\, \gamma\in\CustomSL{\SGroup_i}{\mathrm{S}}\right\},\\
    \label{eqn:cset2}
    \csubsetr{\mathrm{S}}{\pm}{i}&=\left\{\gamma\sigma_i^n\in \mathrm{S}\with n\notin\zZ_\pm,\, \gamma\in\CustomSR{\mathrm{S}}{\SGroup_i}\right\}.
\end{align}
\end{subequations}
\end{definition}
\noindent For typical applications, the subset $\mathrm{S}$ will take the form of the whole group $\SGroup$ or one of the cosets $\SCosetL{j}$, $\SCosetR{j}$ for $j\,{\in}\,\{1,\ldots,\genus\}$.

\section{Review of polylogarithms and MZVs at low genera}\label{sec:rev01}

Before focusing on higher-genus multiple polylogarithms and their associated multiple zeta values, we will take a step back and review the simpler structure of polylogarithms and multiple zeta values on Riemann surfaces of genus zero and one. We will take the most straightforward approach and review the situation at genus zero in \secref{sec:rev0}, then advance to genus one in \secref{sec:rev1}.

Polylogarithms are defined in terms of iterated integrals over a suitable set of differential forms on a Riemann surface. As such, they satisfy a set of general identities attached to their definition in terms of iterated integrals~\cite{Chen,Brown:2013qva}. In particular, any iterated integral of the form $\int_\mu\eta_1\circ\cdots\circ\eta_k$ for smooth one-forms $\eta_1,\ldots,\eta_k$ on a manifold $M$ and an integration path $\mu$ obeys the following identities:
\begin{itemize}
	\item \textbf{Path inversion.} \hspace*{2.5em}$\displaystyle\int_\mu\eta_1\circ\cdots\circ\eta_k=(-1)^k\int_{\mu^{-1}}\eta_k\circ\cdots\circ\eta_1$.\hfill\refstepcounter{equation}\textup{(\theequation)}\label{eqn:pathinversion}
	\item \textbf{Path concatenation.} $\displaystyle\int_\mu\eta_1\circ\cdots\circ\eta_k=\sum_{\ell=0}^k\int_{\mu_1}\eta_1\circ\cdots\circ\eta_\ell\int_{\mu_2}\eta_{\ell+1}\circ\cdots\circ\eta_k$\hfill\refstepcounter{equation}\textup{(\theequation)}\label{eqn:pathconcatenation}\\
	for two paths $\mu_1,\mu_2$ such that $\mu\texteq\mu_1 \mu_2$ is the concatenation of the two paths.
	\item \textbf{Shuffle product.} \hspace*{2em}$\displaystyle\int_{\mu}\eta_1\circ\cdots\circ\eta_k\int_{\mu}\eta_{k+1}\circ\cdots\circ\eta_{k+\ell}=\sum_{\nu\in{\eta_1\cdots\eta_k}\shuffle{\eta_{k+1}\cdots\eta_{k+\ell}}}\int_{\mu}\nu$,\hfill\refstepcounter{equation}\textup{(\theequation)}\label{eqn:shuffle}\\ with $\nu$ a $(k{+}\ell)$-form, which stems from the shuffle product $\shuffle$.
\end{itemize}
%

\subsection{Genus zero}\label{sec:rev0}

\subsubsection{Polylogarithms}

\emph{Multiple polylogarithms} (MPLs), also called \emph{Goncharov polylogarithms}, are iterated integrals
\begin{align}\label{eqn:GPL}
	G(a_1, a_2, \ldots, a_n ; z):=\int_0^z\frac{\dd t}{t-a_1} G(a_2, \ldots, a_n ; t), \quad G(; z)=1.
\end{align}
The integrals in \eqn{eqn:GPL} can have endpoint divergences, which need to be regularized. The general idea behind regularization of MPLs is to slightly shift the endpoints of the integration path by a small amount $\re\,{>}\,0$ in case of a divergence~\cite{Goncharov:2001iea}. The resulting (iterated) integral can be expanded in $\re$ and the regularized value is taken to be the regular part of this expansion in the limit $\re\,{\rightarrow}\,0$. This is referred to as \emph{canonical regularization} in~\rcite{Goncharov:2001iea}. More precisely, the only divergent MPL at depth one\footnote{Other conventions might refer to the depth of an MPL as the number of $a_i\,{\neq}\,0$. In order to be consistent with later sections, we define the depth as the number of iterated integrations.} (for generic $z\,{\in}\,\ComplexComplete$) is given by
\begin{equation}\label{eqn:depth-oneMPL}
    G(0;z)=\int_0^z\frac{\dd t}{t} \, .
\end{equation}
Shifting the lower integration boundary by a small amount $\re\,{>}\,0$ (along the positive real axis) results in the expansion
\begin{equation}
    G_\re(0;z)=\int_\re^z\frac{\dd t}{t}=\log(z)-\log(\re)
\end{equation}
for \eqn{eqn:depth-oneMPL}. The regularized value is then
\begin{equation}\label{eqn:depth-oneMPLreg}
    G_{\text{reg}}(0;z)=\lim_{\re\rightarrow0}\left(G_\re(0;z)+\log(\re)\right)=\log(z).
\end{equation}
This regularization can be shown to be compatible with the shuffle relations for MPLs. Hence, regularization at higher depth can be reduced to \eqn{eqn:depth-oneMPLreg} by (repeatedly) using the shuffle relations. For example, we have
\begin{align}
    G(\underbrace{0,0,\ldots,0}_n;z):=\frac{1}{n!} \log ^n z,
\end{align}
and more complicated cases can be derived analogously.

\subsubsection{Multiple zeta values}

At genus zero, \emph{multiple zeta values} (MZVs) have been historically defined through the sums
\begin{equation}
	\zeta(n_1,\ldots,n_s)=\sum_{0<k_1<\cdots<k_s}\frac{1}{k_1^{n_1}\cdots k_s^{n_s}},
\end{equation}
which can be identified with the MPLs of \eqn{eqn:GPL} evaluated at $z\texteq1$,
\begin{equation}\label{eqn:MZVint}
	\zeta(n_1 \ldots, n_s)=(-1)^s\, G(\underbrace{0, \ldots, 0}_{n_s-1}, 1, \ldots, \underbrace{0, \ldots, 0}_{n_2-1}, 1, \underbrace{0, \ldots, 0}_{n_1-1}, 1 ; 1).
\end{equation}
These special values of MPLs and their structure are well-understood, and they obey shuffle and stuffle relations (see e.g.~\cite{Brown:2013qva} for a review on MZVs). Since the iterated integrals of MZVs~\eqref{eqn:MZVint} encounter endpoint divergences when the indices of the corresponding MPL start with a $1$ or end with a $0$, regularization of these divergences is necessary. To achieve this, one employs the regularization procedure for MPLs reviewed in the previous paragraph for both integration boundaries, e.g.
\begin{equation}
G(0;1)=\lim_{\re\rightarrow0}\left(\int_\re^1\frac{dt}{t}+\log(\re)\right)=0\, ,\quad G(1;1)=\lim_{\re\rightarrow0}\left(\int_0^{1-\re}\frac{dt}{t-1}+\log(\re)\right)=0 \, ,
\end{equation} 
leading to the regularized depth-one MZVs\footnote{One defines $\zeta_0\texteq G(0;1)$.} $\zeta_0\texteq\zeta_1\texteq0$. Similarly as for MPLs, regularized MZVs at higher depth then follow through shuffle relations from the depth-one regularization, e.g.~$\zeta_{1,0}\texteq{-}\,\zeta_{0,1}\,{+}\,\zeta_1\zeta_0\texteq{-}\,\zeta_{0,1}\texteq{-}\,\zeta(2)$.

\subsection{Genus one}\label{sec:rev1}

\subsubsection{Polylogarithms}\label{sec:g1reg}

At genus one, polylogarithms are constructed using an infinite set of integration kernels\footnote{More precisely, the integration kernels are defined on the universal cover of the complex torus $\zC\slash(\zZ+\tau\zZ)$, which happens to coincide with $\zC$.} defined on $\zC$. These integration kernels are obtained from the expansion of a generating function known as the \emph{Kronecker function}. This generating function $F(\xi,\alpha\mid\tau)$ can be represented using the odd theta functions as~\cite{BrownLevin}
\begin{align}\label{eqn:KronForm}
    F(\xi, \alpha \mid \tau)=\frac{\theta^{\prime}(0 \mid \tau)\, \theta(\xi+\alpha \mid \tau)}{\theta(\xi \mid \tau) \,\theta(\alpha \mid \tau)},
\end{align}
where $\xi\,{\in}\, \mathbb{C}$ and $\alpha$ is a (formal) power-counting variable.

By formally expanding $F(\xi,\alpha\mid\tau)$ in $\alpha$, we obtain a set of integration kernels $g^{(m)}(\xi \mid \tau)\,\dd\xi$, $m\,{\geq}\,0$, as expansion coefficients of
\begin{align}\label{eqn:gen1Kronecker}
   \alpha\, F(\xi,\alpha \mid \tau)\,\dd\xi = \sum^\infty_{m=0} g^{(m)}(\xi\mid\tau)\,\dd\xi \,\alpha^m.
\end{align}
The monodromies and parity of the odd theta function translate to the properties
\begin{subequations}
    \begin{align}
        F(\xi+1,\alpha\mid\tau)=F(\xi,\alpha\mid\tau),\qquad& F(\xi+\tau,\alpha\mid\tau)=e^{-2\pi\iunit \alpha}\,F(\xi,\alpha\mid\tau),\\
        F(-\xi,\alpha\mid\tau)=&\,-F(\xi,-\alpha\mid\tau)
    \end{align}
\end{subequations}
for the Kronecker function. These properties descend to the kernels, giving
\begin{subequations}
\begin{align}\label{eqn:gen1kernelprops}
    &g^{(n)}(\xi+1\mid\tau)=g^{(n)}(\xi\mid\tau),\qquad g^{(n)}(\xi+\tau\mid\tau)=\sum_{k=0}^n(-2 \pi \iunit)^k g^{(n-k)}(\xi \mid \tau),\\ 
    &g^{(n)}(-\xi\mid\tau)=(-1)^n g^{(n)}(\xi\mid\tau).\label{eqn:gen1kernelparity}
\end{align}
\end{subequations}
The Kronecker function also satisfies the Fay identity, which reads
\begin{equation}\label{eqn:g1Fay}
	F(\xi-\chi,\alpha)\,F(\tilde{\xi}-\chi,\tilde{\alpha})=
	F(\xi-\tilde{\xi},\alpha)\,F(\tilde{\xi}-\chi,\alpha+\tilde{\alpha})+
	F(\xi-\chi,\alpha+\tilde{\alpha})\,F(\tilde{\xi}-\xi,\tilde{\alpha}).
\end{equation}
Expanding the above Fay identity using \eqn{eqn:gen1Kronecker} generates relations among the integration kernels $g^{(n)}$, which in turn imply relations between polylogarithms. A convenient language for exploring those functional relations is given by the symbol, as explicitly explored in~\rcite{Broedel:2018iwv}. Further relations have been investigated in~\rcite{ZagierGangl} and exemplified in \rcite{Broedel_2020}.

Polylogarithms on a Riemann surface of genus one, or \emph{elliptic polylogarithms} (eMPLs), are defined as~\cite{BrownLevin,Broedel:2014vla}
\begin{equation}\label{eqn:eMPL}
	\Gtargxi{n_1 & n_2 & \ldots & n_r}{\chi_1 & \chi_2 & \ldots & \chi_r}=\int_0^\xi \dd \vt\, g^{(n_1)}(\vt-\chi_1)\Gtargtheta{n_2 & \ldots & n_r}{\chi_2 & \ldots & \chi_r},\qquad\Gtargxi{}{}=1,
\end{equation}
in terms of the integration kernels $g^{(m)}(\xi\,{-}\,\chi)$. The dependence on the elliptic modulus $\tau$ is usually omitted.

Once again, regularization of the polylogarithms $\tilde\Gamma$ is required (see e.g.~\rcite{EZ3} for a detailed account). To see this, notice that the integration kernel $g^{(1)}(\xi)$ has a pole at $\xi\texteq0$, as evident from its $q$-expansion~\cite{Broedel:2014vla}
\begin{align}\label{eqn:g1q}
    g^{(1)}(\xi)\,\dd\xi=\pi \cot (\pi \xi)\,\dd\xi +4 \pi \sum_{m,n=1}^{\infty} \sin (2 \pi m \xi)\, q^{m n}\, \dd\xi,\qquad q=\exp(2\pi\iunit \tau).
\end{align}
Therefore, regularization is needed for eMPLs involving the kernel $g^{(1)}(\xi)$. Based on the above $q$-expansion, we can apply the same regularization procedure as described in \secref{sec:rev0}. Shifting the endpoint by a small amount $\re\,{>}\,0$ along the positive real axis, we obtain~\cite{Broedel:2019gba}
\begin{align}\label{eqn:genus1reg}
    \tilde\Gamma_{\mathrm{reg}}\!\left(\,{ }_0^1 \,; \chi\right) & =\log(1-e^{2 \pi \iunit \chi})-\pi \iunit \chi+4 \pi \sum_{k, l>0} \frac{1}{2 \pi k}(1-\cos (2 \pi k \chi))\, q^{k l}
\end{align}
for the value of the regularized eMPL $\tilde\Gamma_{\mathrm{reg}}\!\left(\,{ }_0^1 \,; \chi\right)$. As for genus zero, it can be shown that applying this procedure to a general eMPL is compatible with the shuffle relations~\cite{EZ3}. Therefore, the regularization of eMPLs at higher depth can always be traced back to \eqn{eqn:genus1reg} by using the shuffle relations. For example, at depth $n$, one has
\begin{align}
    \tilde\Gamma_{\text {reg}}\!\,\Big(\underbrace{\,{ }_0^1\,\ldots\,{ }_0^1\,}_{n} ; \chi\Big)=\frac{1}{n!}\left(\tilde\Gamma_{\mathrm{reg}}\!\left({ }_0^1\, ; \chi\right)\right)^n.
\end{align}
In general, the shuffle relations in combination with \eqn{eqn:genus1reg} provide an algorithm for regularization of endpoint divergences of arbitrary eMPLs. As we will see later in \secref{sec:hgMPLreg}, a similar algorithm can be applied at higher genera. As a result, we obtain a generalization of \eqn{eqn:genus1reg} as well as a procedure to reduce the regularization of arbitrary higher-genus polylogarithms to depth one.

\subsubsection{Multiple zeta values}\label{sec:emzv}

On the genus-one Riemann surface, the \emph{elliptic multiple zeta values} (eMZVs) are defined as\footnote{Note that, in this definition, we reverse the order of the labels compared to the standard references~\cite{Broedel:2014vla,Broedel:2015hia}. However, this will not change later results since any eMZV identity considered below remains valid by means of the reflection identity~\cite[eq.~(2.24)]{Broedel:2014vla}.}
\begin{align}\label{eqn:eMZVs}
\omell(n_1, n_2, \ldots, n_r) & = \int_{0}^{1} g^{\left(n_1\right)}\left(\xi_1\right) \mathrm{d} \xi_1\int_{0}^{\xi_1} g^{\left(n_2\right)}\left(\xi_2\right) \mathrm{d} \xi_2\cdots \int_{0}^{\xi_{r-1}}g^{\left(n_r\right)}\left(\xi_r\right) \mathrm{d} \xi_r\notag \\
&=\Gtargone{n_1,n_2,\ldots,n_r}{0, \ 0,\ \ldots,\ 0},
\end{align}
so that eMZVs are obtained by choosing the $\acyc$-cycle as the integration path of an eMPL from \eqn{eqn:eMPL}.

As discussed in \subsecref{sec:g1reg}, eMPLs involving the elliptic kernel $g^{(1)}$ can exhibit endpoint divergences. For eMZVs, which are defined by integrating over the closed non-contractible cycle $\acyc\texteq[0,1]$ on the torus, this implies that endpoint divergences can appear at both $0$ and $1$ as $g^{(1)}(\xi)$ has trivial monodromy along $\acyc$ (c.f.~\eqn{eqn:gen1kernelprops}). Consequently, they require a regularization procedure (see e.g.~\cite{Broedel:2014vla,Broedel:2015hia,Matthes:Thesis,Broedel:2019gba} for a detailed account). In order to demonstrate the equivalence with the procedure for higher-genus MZVs carried out in \secref{sec:hgmzvreg}, let us briefly summarize the procedure here. By definition, we can write
\begin{equation}
    \omell(1)=\tilde\Gamma_{\text{reg}}\!\left(\,{ }_0^1\,; 1\right)=\left(\int_0^1g^{(1)}(\xi)\,\dd\xi\right)_{\text{reg}} \, ,
\end{equation}
where $g^{(1)}(\xi)$ diverges at both endpoints $0$ and $1$. By applying path concatenation, we arrive at
\begin{align}
     \omell(1)&=\tilde\Gamma_{\text{reg}}\!\left(\,{ }_0^1\,; 1\right) =\left(\int_0^1g^{(1)}(\xi)\,\dd\xi\right)_{\text{reg}}=\left(\int_0^{\xi_0}g^{(1)}(\xi)\,\dd\xi\right)_{\text{reg}}+\left(\int_{\xi_0}^1g^{(1)}(\xi)\,\dd\xi\right)_{\text{reg}}\nn \\
     &=\tilde\Gamma_{\text{reg}}\!\left(\,{ }_0^1\, ; \xi_0\right) -\tilde\Gamma_{\text{reg}}\!\left(\,{ }_0^1\, ; 1-\xi_0\right),
\end{align}
where $\xi_0\,{\in}\,(0,1)$ and we have used the substitution $\xi\,{\rightarrow}\,1\,{-}\,\xi$ in combination with the trivial monodromy of $g^{(1)}(\xi)$ along $[0,1]$ in the last equality. This allows us to employ the eMPL regularization reviewed in \secref{sec:g1reg}. In particular, \eqn{eqn:genus1reg} yields
\begin{equation}
    \omell(1)=\pi\iunit-\log(-1)=0,
\end{equation}
where we chose the standard branch of the logarithm $\log(-1)\texteq\pi\iunit$. As for eMPLs, this regularization is shuffle-compatible. Hence, regularization of eMZVs at higher depth can be reduced to the regularization of $\omell(1)$ by using the shuffle relations.

\section{Polylogarithms on higher-genus Riemann surfaces}\label{sec:reviewMPL}

In order to define multiple zeta values on a Riemann surface of higher genus, we need to define a space of polylogarithms on these surfaces. Then, higher-genus multiple zeta values can be defined analogously as for genera zero and one by evaluating the polylogarithms on the non-contractible loops of the homology basis. In \secref{sec:revKernels}, we review the construction of a set of integration kernels on higher-genus Riemann surfaces based on a connection defined by Enriquez in \rcite{EnriquezHigher} as well as the associated construction of polylogarithms. In a second step, we address the issue of regularization for these polylogarithms in \secref{sec:hgMPLreg}. 

\subsection{Review of Enriquez' higher-genus integration kernels and MPLs}\label{sec:revKernels}

\subsubsection{Definition and properties}

\begin{propdef}[Enriquez~\cite{EnriquezHigher}, Enriquez--Zerbini~\cite{EZ1}]\label{def:propdef} Let $\RSurf$ be a Riemann surface of genus $\genus\,{\geq}\,1$ and let $\alg{t}$ be the algebra freely generated by the formal variables\footnote{We will refer to these variables as \textit{letters}.} $a_1,\ldots,a_\genus$ and $b_1,\ldots,b_\genus$. Moreover, let $\alg{b}\,{\subset}\,\alg{t}$ be the subalgebra freely generated by the letters $b_1,\ldots,b_\genus$. Then, there exists a unique (meromorphic) flat connection form $K(z,x)$ on the universal cover of $\RSurf$, which is valued in the algebra $\alg{t}$. Furthermore, it is uniquely determined by the properties
\begin{subequations}\label{eqn:defprop}
    \begin{align}
        &K(\sigma_iz,x)=\,e^{b_i}K(z,x),\qquad \qquad \quad &\text{(quasi-periodicity)}\\
        \label{eqn:residueK}&(-2\pi\iunit)\Res_{z=x}\,K(z,x)=\sum_{j=1}^\genus b_j\,a_j,\quad\qquad\qquad\, &\text{(residue)} \\
        &\int_{\acyc_i}K(z,x)=\frac{b_i\,a_i}{e^{b_i}-1},\qquad\qquad\qquad\ \ \ &\text{(\,$\acyc$-cycles)}\label{eqn:Acycconnection}
    \end{align}
\end{subequations}
where $i\,{\in}\,\{1,\ldots,\genus\}$ and the expression $\sigma_iz$ refers to the point $z$ after moving along the cycle $\bcyc_i$, which can be identified with the application of the generator $\sigma_i$ on the Schottky cover.
\end{propdef}
The connection form $K(z,x)$ can be expanded in the algebra generators as
\begin{align}\label{eqn:connectionexpansion}
	K(z,x)&=\sum_{j=1}^{\genus}K_j(z,x) \,a_j=\sum_{r=0}^\infty\sum_{i_1,\ldots,i_r,j=1}^\genus\omega_{i_1\cdots i_rj}(z,x)\,b_{i_1}\cdots b_{i_r}\,a_j,
\end{align}
where the expansion coefficients $\omega_{i_1\cdots i_r j}(z,x)$ define the meromorphic so-called \emph{Enriquez' kernels} and the $K_j$ are called \emph{component forms}. The kernels are one-forms in $z$ and scalar functions in $x$ and can be used for the construction of polylogarithms on higher-genus Riemann surfaces as shown in \subsecref{sec:hgmpls}.

For the construction of higher-genus polylogarithms one can likewise choose a connection that is non-meromorphic and single-valued (in contrast to the meromorphic and quasi-periodic connection of Enriquez), whose expansion yields non-meromorphic and single-valued integration kernels alike. This method was used in \rcite{DHoker:2023vax} and the relation to Enriquez' connection was established in~\rcite{DHoker:2025szl}. Furthermore, one can construct a meromorphic and single-valued connection by allowing the connection to admit poles of higher order at the punctures. This construction was considered by Levin--Racinet at genus one~\cite{LevinRacinet} and generalized to higher genera by Enriquez--Zerbini in \rcite{EZ1}.

The defining relations~\eqref{eqn:defprop} translate into properties for the kernels as follows~\cite{EnriquezHigher,EZ1}:
\begin{subequations}\label{eqn:Enriquezkernels}
\begin{align}
	\omega_{i_1\cdots i_r j}(\sigma_kz,x)&=\sum_{\ell=0}^r\frac{1}{\ell!}\delta_{ki_1\cdots i_\ell}\,\omega_{i_{\ell+1}\cdots i_rj}(z,x),\\
	\omega_{i_1\cdots i_rij}(z,\sigma_k x)&=\omega_{i_1\cdots i_rij}(z,x)+\delta_{ij}\sum_{\ell=0}^r\frac{(-1)^{\ell+1}}{(\ell+1)!}\delta_{ki_r\cdots i_{r-\ell+1}}\,\omega_{i_1\cdots i_{r-\ell}k}(z,x), \\
    \Res_{z=x}\omega_{i_1\cdots i_rj}(z,x)&=\frac{1}{(-2\pi\iunit)}\delta_{r1}\,\delta_{i_1j}, \label{eqn:kernelresidue}\\
    \int_{\acyc_k}\omega_{i_1\cdots i_rj}(z,x)&=\frac{\bn{r}}{r!}\delta_{i_1\cdots i_rjk},\label{eqn:d1}
\end{align}
\end{subequations}
where $\delta_{ki_1\cdots i_\ell}\texteq\prod_{n=1}^{\ell}\delta_{ki_n}$ and $z,x\,{\in}\,\funddom$. By $\bn{n}$, we denote the $n$-th Bernoulli number, arising as coefficients in the generating series
\begin{equation}\label{eqn:Bernoulli}
	\frac{x}{e^x-1}=\sum\limits_{k=0}^{\infty}\bn{k}\frac{x^k}{k!}\,,
\end{equation}
where, in particular, $\bn{1}\texteq{-}\,1/2$. From \eqn{eqn:kernelresidue}, one can infer that (only) the kernels $\omega_{jj}(z,x)$ have a pole at $z\texteq x$, which is why we will sometimes refer to the second variable of those kernels as the \textit{pole variable}.

In \rcites{Baune:2024ber,DHoker:2024ozn} quadratic identities between higher-genus integration kernels have been derived, generalizing the Fay identities~\cite{BrownLevin} at genus one. Even though it has not been shown that these are indeed the higher-genus Fay identities formulated in the language of Riemann theta functions~\cite{fay}, we will refer to them as \emph{Fay-like identities}.
They can be written as
\begin{align}
	\label{eqn:FaylikeId}
		\omega_{i_1\cdots i_rk}(z,x)\,\omega_{p_1\cdots p_spp'}(y,z)\big|_{p' = p}\!&=\omega_{i_1\cdots i_rk}(z,x)\,\omega_{p_1\cdots p_spp'}(y,x)\big|_{p'=p} \\
		&\quad-\sum_{l=0}^{r-1}\sum_{m=0}^s(-1)^{m-s}\omega_{(i_1\cdots i_l\shuffle p_s\cdots p_{m+1})ji_{l+1}\cdots i_rk}(z,x)\,\omega_{p_1\cdots p_mj}(y,x)\notag \\
		&\quad-\sum_{m=0}^s(-1)^{m-s}\omega_{(i_1\cdots i_r\shuffle p_s\cdots p_{m+1})jk}(z,y)\,\omega_{p_1\cdots p_mj}(y,x) \notag\\
		&\quad-\sum_{l=0}^r\sum_{m=0}^s(-1)^{m-s}\omega_{(i_1\cdots i_l\shuffle p_s\cdots p_{m+1})j}(z,y)\,\omega_{p_1\cdots p_mji_{l+1}\cdots i_rk}(y,x),\notag
\end{align}
with implicit summation over repeated indices $j$. It can be used to derive identities between higher-genus polylogarithms and their zeta values as discussed later in \secref{sec:Fayid}. Additionally, Enriquez' kernels satisfy a certain kind of linear $3$-point identity due to Enriquez~\cite[Prop.~10]{EnriquezHigher} (see also \rcite{Baune:2024ber}), which states that differences of higher-genus kernels are independent of the choice of basepoint: for any points $u$ and $v$ in the fundamental domain, one has for $r\geq0$
\begin{align}\label{eqn:linearid}
    \omega_{i_1\cdots i_rjj}(z,u)-\omega_{i_1\cdots i_rkk}(z,u)=\omega_{i_1\cdots i_rjj}(z,v)-\omega_{i_1\cdots i_rkk}(z,v).
\end{align}

While these properties are useful for proofs concerning polylogarithms, one needs explicit formulas to evaluate the kernels and the resulting higher-genus polylogarithms analytically or numerically. Ref.~\cite{Baune:2024biq} gave a representation of Enriquez' kernels in terms of Schottky sums (shown in \secref{sec:EnriquezSchottky}), which was also suitable for numerical evaluation due to quickly converging Poincar\'e series. A different formulation of Enriquez' kernels (not relying on any particular uniformization of the surface) was recently put forward in \rcite{DHoker:2025dhv}, giving a representation of the first non-trivial kernels $\omega_{ij}$ as an \Atxt-cycle integral over prime forms and expressing all other kernels as convolutions of the former (discussed in \secref{sec:EnriquezDHS}). In \secref{sec:TechniquesSchottky}, we are going to compare these two formulations of Enriquez' kernels.

\subsubsection{Enriquez' kernels in Schottky uniformization}\label{sec:EnriquezSchottky}

The Schottky representation of Enriquez' kernels in \rcite{Baune:2024biq} paved the way to writing higher-genus kernels as Schottky sums over elliptic integration kernels. The derivation starts from formulating the component forms $K_j(z,x)$ \eqn{eqn:connectionexpansion} in the Schottky uniformization as
\begin{equation}\label{eqn:schottkykronecker}
	K_j(z,x)=\frac{1}{(-2\pi\iunit)}\sum_{\gamma\in\SGroup}\left(\frac{\dd z}{z-\gamma x}-\frac{\dd z}{z-\gamma P_j}\right)W(\gamma)\,b_j,
\end{equation}
for $W\colon \sgen_{i_1}^{n_1}\cdots\sgen_{i_k}^{n_k}\,{\mapsto}\,e^{n_1b_{i_1}}\cdots e^{n_kb_{i_k}}$. 
Importantly, the above explicit expression of the component forms requires and relies on the choice and orientation of cycles $\acyc_j$ described in \subsecref{sec:Schottky} and depicted in \figref{fig:genustwo} in order to be consistent with the \Atxt-cycle integrals in \eqn{eqn:Acycconnection}. If one would instead choose the contour $\acyc_j$ to be around the circle $C_j'$ in anti-clockwise direction, the \Atxt-cycle periods would be altered to feature the other convention\footnote{The other convention (in contrast to \eqn{eqn:Bernoulli}) for Bernoulli numbers is through the generating series
\begin{equation*}
	\frac{x\,e^x}{e^x-1}=\sum_{k=0}^{\infty}\tilde{\mathrm{B}}_k\frac{x^k}{k!},
\end{equation*}
which only switches the sign of $\bn{1}\texteq{-}\frac{1}{2}$ to $\tilde{\mathrm{B}}_1\texteq\frac12$.
}
$\tilde{\mathrm{B}}_k$ for the Bernoulli numbers compared to \eqn{eqn:Bernoulli}. However, it is possible to formulate a consistent integration framework based on each choice of cycles and orientations.

Expanding expression~\eqref{eqn:schottkykronecker} in the algebra generators $b_i$, Enriquez' higher-genus integration kernels are expressed as weighted sums over genus-one integration kernels $g^{(n)}$ as follows:
\begin{align}\label{eqn:highergenuskernels}
	\omega_{i_1 \cdots i_s j}(z, x \mid\! \SGroup)=&\frac{1}{(-2 \pi \iunit)} \sum_{\gamma \in \SCosetR{j}} \sum_{k=0}^{\delta_{j i_s}{n_s}} C\big(b_{i_1}^{n_1} \cdots b_{i_s}^{n_s-k}, \gamma\big)\, s^{(k)}_j(\gamma^{-1} z, x),
\end{align}
where the form $s^{(n)}_j$ is defined as
\begin{align}\label{eqn:sn}
	s^{(n)}_j(z, x)=(-2 \pi \iunit)^{1-n} \, g^{(n)}_j(z, x)\, \omega(z \mid \SGroup_j),
\end{align}
where $g^{(n)}$ are the integration kernels for eMPLs discussed and reviewed in \secref{sec:g1reg}. Moreover, we used the shorthands from \eqn{eqn:shorthands} and defined
\begin{equation}
	g_j^{(n)}(z,x)=g^{(n)}(\abel_j(z,x)\mid\tau_j).\label{eqn:gj}
\end{equation}
The coefficients $C$ from \eqn{eqn:highergenuskernels} are given recursively through\footnote{A similar expression for the expansion coefficients of a family of bi-differentials can be found in \rcite{EnriquezHigher}.}~\cite{Baune:2024biq}
\begin{align} \label{eqn:Ccoeff}
	&C(b_{i_1}^{n_1} \cdots b_{i_s}^{n_s} , \sigma_{j_1}^{m_1} \cdots \sigma_{j_l}^{m_l}) =
	\begin{cases} s = 0 : & 1, \\
		s \neq 0 = l : & 0,\\
		i_1 \neq j_1 : & C(b_{i_1}^{n_1} \cdots b_{i_s}^{n_s} , \sigma_{j_2}^{m_2} \cdots \sigma_{j_l}^{m_l}),\\
		i_1 = j_1 : & \sum_{k=0}^{n_1} \frac{(m_1)^k}{k!} C(b_{i_1}^{n_1-k} \cdots b_{i_s}^{n_s} , \sigma_{j_2}^{m_2} \cdots \sigma_{j_l}^{m_l}).
	\end{cases}
\end{align}
From this definition, we can derive a useful property of these coefficients, which turns out to be equivalent to the definition in the end.
\begin{lemma}\label{lem:Ccoeff}
    Let $r\geq0$. Moreover, let $n,n_1,\cdots,n_r\geq0$ and $j,i_1,\ldots,i_r\in\{1,\ldots,\genus\}$ such that $j\neq i_1$ as well as $m,l\in\zZ$ and $\gamma\in\SCosetL{j}$. Then we have 
    \begin{equation}\label{eqn:recC}
        C(b_j^nb_{i_1}^{n_1}\cdots b_{i_r}^{n_r},\sigma_j^{m\pm l}\gamma)=\sum_{p=0}^n\frac{(\pm l)^p}{p!}C(b_j^{n-p}b_{i_1}^{n_1}\cdots b_{i_r}^{n_r},\sigma^m_j\gamma)
    \end{equation}
\end{lemma}
\begin{proof}
    We start by explicitly expanding the recursive definition of the $C$-coefficients~\eqref{eqn:Ccoeff}. This gives
    \begin{equation}
        C(b_j^nb_{i_1}^{n_1}\cdots b_{i_r}^{n_r},\sigma_j^{m\pm l}\gamma)=\sum_{k=0}^n\frac{(m\pm l)^k}{k!}C(b_j^{n-k}b_{i_1}^{n_1}\cdots b_{i_r}^{n_r},\gamma) \, .
    \end{equation}
    We can now use the binomial theorem to simplify
    \begin{equation}
        \begin{aligned}
            \sum_{k=0}^n\frac{(m\pm l)^k}{k!}C(b_j^{n-k}b_{i_1}^{n_1}\cdots b_{i_r}^{n_r},\gamma)=&\sum_{k=0}^n\sum_{p=0}^k\binom{k}{p}\frac{m^{k-p}(\pm l)^p}{k!}C(b_j^{n-k}b_{i_1}^{n_1}\cdots b_{i_r}^{n_r},\gamma) \\
            =&\sum_{p=0}^n\sum_{k=p}^n\frac{m^{k-p}(\pm l)^p}{(k-p)!p!}C(b_j^{n-k}b_{i_1}^{n_1}\cdots b_{i_r}^{n_r},\gamma) \\
            =&\sum_{p=0}^n\frac{(\pm l)^p}{p!}\sum_{k=0}^{n-p}\frac{m^k}{k!}C(b_j^{n-k-p}b_{i_1}^{n_1}\cdots b_{i_r}^{n_r},\gamma) \, ,
        \end{aligned}
    \end{equation}
    where we have swapped the summations and simplified the binomial coefficient in the second line as well as relabelled $k\rightarrow k-p$ in the last line. Now, in the last line, we can again recognize the recursive definition~\eqref{eqn:Ccoeff}, which immediately yields \eqn{eqn:recC} and thus finishes the proof.
\end{proof}

\subsubsection{Enriquez' kernels as integrals over prime forms}\label{sec:EnriquezDHS}

In \rcite{DHoker:2025dhv}, a formula for the first kernel of Enriquez $\omega_{ij}$ was derived. Using the prime form\footnote{The prime form is defined as~\cite{fay}
\begin{equation*}
	E(z,x)=\frac{\theta_\nu(z-x)}{h_\nu(z)\,h_\nu(x)}
\end{equation*} 
for a theta function of (arbitrary) odd characteristic $\nu$ and the $(\frac{1}{2},0)$-form $h(z)\texteq\sqrt{\sum_{j=1}^\genus\partial_j\theta(0)\omega_j(z)}$.} $E(z,x)$, this can be written as
\begin{equation}\label{eqn:omegaij}
	\omega_{ij}(z,x)=\frac{1}{(-2\pi\iunit)}\int_{t\in\acyc_i}\omega_j(t)\,\dd_z\log\frac{E(z,x)}{E(z,t)}-\frac12\delta_{ij}\,\omega_j(z)\, ,
\end{equation}
where we can identify the ratio of prime forms as a representation of the \textit{normalized differential of the third kind} (see \appref{app:Ndiff} for an outline of its properties)
\begin{equation}
	\label{eqn:fundamentalDiff3}
	\Omega^{(x_1-x_0)}(z)
	\coloneqq\dd_z\log\frac{E(z,x_1)}{E(z,x_0)}
\end{equation}
for $x_1,x_0\,{\in}\,\funddom$. In order to match this representation to the Schottky representation~\eqref{eqn:highergenuskernels} later in \secref{sec:IntegrationSchottky}, it is convenient to adapt a slightly different convention compared to the one chosen in \eqn{eqn:omegaij}. More precisely, notice that \eqn{eqn:omegaij} has a pole at $z\texteq t$, which manifests when integrating \eqn{eqn:omegaij} around the cycle $\acyc_i$ \wrt{$z$}. To circumvent this complication, we can consider a modified integration contour to calculate the $\acyc_i$-period. In other words, we consider a closed contour $A_i$ homotopic to $\acyc_i$ obtained by slightly shifting $\acyc_i$ into the fundamental domain. Then we integrate \eqn{eqn:omegaij} \wrt{$z$} over this modified contour $A_i$. The modification of the contour results in $t$ being enclosed by the contour $A_i$. Therefore, the integral around $A_i$ \wrt{$z$} admits a residue contribution from $\Omega^{(x-t)}(z)$. In the end, the integral \wrt{$z$} evaluates to 
\begin{equation}
	\frac{1}{(-2\pi\iunit)}\int_{z\in A_i}\int_{t\in\acyc_i}\omega_j(t)\,\Omega^{(x-t)}(z)= (-2\pi\iunit)\Res_{z=t}\left(\frac{1}{(-2\pi\iunit)}\int_{t\in\acyc_i}\omega_j(t)\,\Omega^{(x-t)}(z)\right)=-\delta_{ij}
\end{equation}
by the residue theorem. Hence we have to modify \eqn{eqn:omegaij} to ensure the correct $\acyc_i$-period~\eqref{eqn:d1} for Enriquez' kernel $\omega_{ij}(z,x)$. It turns out that
\begin{equation}\label{eqn:omegaij2}
	\omega_{ij}(z,x)=\frac{1}{(-2\pi\iunit)}\int_{t\in\acyc_i}\omega_j(t)\,\dd_z\log\frac{E(z,x)}{E(z,t)}+\frac12\delta_{ij}\,\omega_j(z)
\end{equation}
implements the desired behavior. It is important to highlight that $t$ is regarded as being enclosed by $A_i$ and therefore equivalently may be considered to be shifted outside of the fundamental domain by a $\bcyc_i$-cycle. This shift is precisely accommodated by the change in sign when going from \eqn{eqn:omegaij} to \eqn{eqn:omegaij2}.

For kernels of higher weight, \rcite{DHoker:2025dhv} provides a recursive formula through convolution integrals over \Atxt-cycles, reading
\begin{align}\label{eqn:DSthm2}
	\omega_{\ell i_1\cdots i_rk}(z,x)&=-\sum_{j=1}^\genus\int_{t\in\acyc_\ell}\omega_{jk}(z,t)\,\omega_{i_1\cdots i_rj}(t,x)-\sum_{n=1}^{r-1}\frac{\bn{n}}{n!}\delta_{i_1\cdots i_n\ell}\,\omega_{\ell i_{n+1}\cdots i_rk}(z,x)\notag\\
	&\quad-\omega_k(z)\frac{\bn{r+1}}{r!}\delta_{i_1\cdots i_rk\ell},
\end{align}
for $z,x\,{\in}\,\funddom$ and not on the cycle $\acyc_l$. Implicitly, the $\acyc$-cycle integration contains similar modifications to the integration contour as described above in order to work around the poles of $\omega_{jj}(z,t)$ at $z\texteq t$. Furthermore, the statement of \eqn{eqn:DSthm2} was generalized in \rcite{DHoker:2025dhv} to show that the space of Enriquez' kernels is closed under taking convolution integrals of the type appearing in \eqn{eqn:DSthm2}.

In \secref{sec:TechniquesSchottky}, we will compare the above formulas for Enriquez' kernels with the Schottky representation shown in \secref{sec:EnriquezSchottky}.

\subsubsection{Enriquez' polylogarithms on higher-genus Riemann surfaces}\label{sec:hgmpls}

Polylogarithms on higher-genus Riemann surfaces have been defined for different kinds of integration kernels (see \rcite{DHoker:2023vax} for single-valued integration kernels, \rcite{EnriquezHigher,Baune:2024biq} for meromorphic and quasi-periodic kernels and \rcite{EZ1,EZ2} for meromorphic and single-valued kernels with poles of higher orders).
In this article, we focus on higher-genus polylogarithms arising when integrating Enriquez' kernels. The translation between the single-valued kernels and Enriquez' kernels has been established in \rcite{DHoker:2025szl}.
\begin{definition}[\!\!\cite{Baune:2024biq}]
	Using Enriquez' higher-genus kernels to form iterated integrals, \emph{higher-genus multiple polylogarithms} (\hgMPL{}s) are defined as 
	\begin{equation}\label{eqn:hgMPL}
		\Gargbare{\mindx{i}_1,\ldots,\mindx{i}_k}{x_1,\ldots,x_k}{z}{z_0}\coloneqq\int_{t_1=z_0}^{z}\omega_{\mindx{i}_1}(t_1,x_1)\int_{t_2=z_0}^{t_1}\omega_{\mindx{i}_2}(t_2,x_2)\int\cdots\int_{t_k=z_0}^{t_{k-1}}\omega_{\mindx{i}_k}(t_k,x_k).
	\end{equation}
\end{definition}
In this definition we used multi-indices $\mindx{i}_\ell$ to label the $\ell$-th kernel, the $x_\ell$ denote the pole variables of the kernels and $z_0$ is the basepoint of the integration. We call the number of integrations $k$ the \emph{depth} of a \hgMPL{}, while the number $w\texteq|\mindx{i}_1|\,{+}\,\ldots\,{+}\,|\mindx{i}_k|$ is called the \emph{weight}, where we defined the weight of a single multi-index $\mindx{i}\texteq i_1\cdots i_rj$ as $|\mindx{i}|\texteq r$. 
Note that the weight of a single multi-index is one less than the number of indices it contains. This makes the connection to the elliptic case and ensures that the holomorphic forms $\omega_j$ have weight 0. Furthermore, \hgMPL{}s depend on the geometry (i.e.~the period matrix $\tau$ or the Schottky group $\SGroup$) of the Riemann surface under consideration, where we often suppress this dependence.

Relations between different \hgMPL{}s can be derived using general properties of iterated integrals and characteristics of the integration kernels, including quadratic kernel identities, so-called Fay-like identities~\cite{DHoker:2024ozn,Baune:2024ber}. Several \hgMPL{} identities are spelled out in these references.

\subsection{Regularization of higher-genus multiple polylogarithms}\label{sec:hgMPLreg}

\subsubsection{General formalism}\label{sec:GeneralformalismPolylog}

Generating series for polylogarithms on an arbitrary higher-genus Riemann surface $\RSurf$ can be identified as the solution to a differential equation of the form
\begin{equation}\label{eqn:KZhg}
    \dd L(z) +K(z,z_0)\, L(z) = 0 \,,
\end{equation}
where $K(z,z_0)$ is the Enriquez connection defined in \defref{def:propdef}. It can be checked that the expression 
\begin{equation}\label{eqn:gensol}
    L_\re(z) = \exp{\left(-\int_{\eta_\re}K(z',z_0)\right)}(\lambda\re)^{-\Res_{z_0}(K)}
\end{equation}
for $\eta_\re(s)\texteq\eta((1\,{-}\,s)\,\re\,{+}\,s)$, $\eta\colon[0,1]\,{\rightarrow}\,\RSurf$ being a path from $z_0$ to $z$, is a solution to \eqn{eqn:KZhg} for each $0\,{<}\,\re\,{\ll}\, 1$. Furthermore, it can be shown that the solution is finite for $\re\,{\rightarrow}\,0$ and so we define
\begin{equation}\label{eqn:Leps}
    L(z) = \lim_{\re\rightarrow0}L_\re(z)
\end{equation}
to be the generating series of (regularized) polylogarithms. The parameter $\lambda\,{\in}\,\zC^\times$ fixes the asymptotic behavior of $L(z)$ as $z\,{\rightarrow}\, z_0$ and therefore uniquely determines the solution.

This general approach described above will be the guiding principle for regularization of hgMPLs below: we are going to use the expansion of Enriquez' connection~\eqref{eqn:connectionexpansion} and expand the kernels therein into genus-one kernels as reviewed in \subsecref{sec:EnriquezSchottky}. As the regularization for genus-one kernels and polylogarithms is well known and reviewed in \subsecref{sec:g1reg}, the algorithm for hgMPL{}s can be phrased in terms of the regularization at genus one and Abel's map. In the following, we will explicitly compute $L(z)$ to lowest non-trivial order, i.e.~for \hgMPL{}s of depth one. As a next step, we then extend the regularization to higher depth by describing an algorithm to reduce the regularization of a higher-depth \hgMPL{} to the regularization at depth one by means of the shuffle relations~\eqref{eqn:shuffle} and the linear identity~\eqref{eqn:linearid}.

Importantly, while we assume that the general approach to regularization described above is compatible with the shuffle relations and the linear identity, it is in principle possible that they yield different regularization prescriptions for \hgMPL{}s\footnote{For eMPLs, compatibility of the two approaches to regularization has been shown in \rcite{EZ3}, and we expect the same to hold true at higher genus.}. In this article, we restrict to applying the general formalism described by \eqns{eqn:gensol}{eqn:Leps} exclusively to \hgMPL{}s of depth one and extend the regularization to higher depth by means of the shuffle relations and the linear identity. Thereby, we do not make any statement on the compatibility of the general approach described by \eqns{eqn:gensol}{eqn:Leps} on the one hand and its application at depth one, extended to higher depth by means of \eqn{eqn:shuffle} and \eqn{eqn:linearid}, on the other hand.

\subsubsection{Regularization at depth one} 

Let us consider the expansion of the solution~\eqref{eqn:Leps} to lowest non-trivial order (coefficient of $b_ia_j$) for the expanded Enriquez connection $K(z,z_0)$ from \eqn{eqn:connectionexpansion}. Referring to \eqn{eqn:residueK}, the only residues for $K$ occur for words of the form $b_ja_j$. So the only \hgMPL{}s of depth one, which need regularization are the ones containing $\omega_{jj}$, $j\,{\in}\,\{1,\ldots,\genus\}$: these kernels exhibit a simple pole in the fundamental domain $\funddom$ at $z \texteq z_0$ as underlined by \eqn{eqn:kernelresidue}. Using \eqn{eqn:highergenuskernels}, they can be expanded in the Schottky language as 
\begin{align}\label{eqn:omegajjSchottky}
    \omega_{jj}(z, z_0 \mid\! \SGroup)  =\frac{1}{(-2\pi\iunit)}\sum_{\gamma \in \SCosetR{j}} C(b_j, \gamma)\,s_j^{(0)}(\gamma^{-1}z)
+\frac{1}{(-2 \pi \iunit)} \sum_{\gamma \in \SCosetR{j}} \skern{1}_j(\gamma^{-1} z, z_0)\,.
\end{align}
Let us consider the integral of $\skern{1}_j(t,z_0)$ along a path $\eta(s),\,s\,{\in}\,[0,1]$ with $\eta(0)\texteq z_0,\,\eta(1)\,{=}\,z$ on the Schottky cover,
\begin{equation}
	\int_{\eta}\skern{1}_j(t,z_0),
\end{equation}
which exhibits an endpoint divergence at the lower integration boundary $s\texteq0$. Following the argument in \subsecref{sec:GeneralformalismPolylog}, which is a formalization of the algorithm for regularization of eMPLs reviewed in \subsecref{sec:g1reg}, we introduce a small parameter $\re\,{>}\,0$ and denote by $\eta_\re$ a path as in \eqn{eqn:gensol}. Thus, the parameter $\re$ measures the displacement from the basepoint $z_0$ at the lower integration boundary, along the direction of the path $\eta$. Regularization is then achieved by evaluating the $\re$-dependent integral and taking $\re$ to zero after isolating the divergent behavior. Unsurprisingly, the divergent behavior is again logarithmic in $\re$.

Let us now take the integration along the path $\eta_\re$ from the Schottky cover to the Jacobi variety of the $j$-th subcover by using Abel's map~\eqref{eq:abel-map-schottky} as follows:
\begin{equation}\label{eqn:regint}
	\int_{\eta_\re}\skern{1}_j(t,z_0)=\int^z_{z_\re}\gkern{1}(\abel_j(t, z_0)\mid\tau_j)\underbrace{\omega(t\mid\SGroup_j)}_{\dd \abel_j(t,z_0)}\stackrel{\vt=\abel_j(t,z_0)}{=}\int_{\abel_j(z_\re,z_0)}^{\abel_j(z,z_0)} \gkern{1}(\vt\mid\tau_j)\, \dd \vt\,.
\end{equation}
For generic $z$ only the lower integration boundary is problematic. Evaluating $\abel_j(z_\re,z_0)$, thereby employing \eqn{eq:abel-map-schottky} and the shorthands from~\eqref{eqn:shorthands}, we find
\begin{align}
	\abel_j(z_\re,z_0)&=\frac{1}{2\pi\iunit}\log(\{z_\re,P_j',z_0,P_j\})\notag\\
	&=\frac{1}{2\pi\iunit}\Bigg[\log\underbrace{\{z_0,P_j',z_0,P_j\}}_{=1}+\Bigg(\underbrace{\frac{1}{z_0-P_j}-\frac{1}{z_0-P_j'}}_{=\abel'_j(z_0)}\Bigg)\eta'(0)\,\re\Bigg]+\cO(\re^2)\notag\\
	&=\frac{1}{2\pi\iunit}\abel'_j(z_0)\,\eta'(0)\,\re+\cO(\re^2)\,,
\end{align}
where we have used that $z_\re$ expands as $z_\re\texteq\eta(\re)\texteq z_0\,{+}\,\eta'(0)\,\re\,{+}\,\cO(\re^2)$ for small $0\,{<}\,\re\,{<}\,1$.

Integrating the $q$-expansion of $\gkern{1}$ in \eqn{eqn:g1q} and plugging in the limits of the last integral in \eqn{eqn:regint} leads to
\begin{align}
	\log(\sin(\pi\abel_j(z_\re,z_0)))&=\log\left(\frac{1}{2i}\left(e^{\frac12\abel'_j(z_0)\eta'(0)\re+\cO(\re^2)}-e^{-\frac12\abel'_j(z_0)\eta'(0)\re+\cO(\re^2)}\right)\right)\notag\\
	&=\log\left(\frac{1}{2\iunit}\right)+\log\left(1-e^{\abel'_j(z_0)\eta'(0)\re+\cO(\re^2)}\right)+\log\left(-e^{-\frac12\abel'_j(z_0)\eta'(0)\re+\cO(\re^2)}\right)\notag\\
	&=\log\left(\frac{1}{2\iunit}\right)+\log\left(-\abel'_j(z_0)\,\eta'(0)\right)+\log\re+\log(-1)+\cO(\re),\label{eqn:logsinze}\\
	\log(\sin(\pi\abel_j(z,z_0)))&=\log\left(\frac{1}{2\iunit}\right)+\log\left(1-e^{2\pi\iunit\abel_j(z,z_0)}\right)-\pi\iunit\,\abel_j(z,z_0)+\log(-1).\label{eqn:logsinz}
\end{align}
Plugging in everything, we find
\begin{align}\label{eqn:hgMZVregularization}
   \int_{\abel_j(z_\re,z_0)}^{\abel_j(z,z_0)} \gkern{1}(\vt\mid\tau_j)\, \dd \vt &=  \log(\sin(\pi \abel_j(z,z_0))) - \log(\sin(\pi \abel_j(z_\re,z_0)) )\nonumber\\
    &\quad-4\pi \sum^\infty_{m=1}\frac{\cos(2\pi m\, \abel_j(z,z_0))-\cos(2\pi m\, \abel_j(z_\re,z_0))}{2\pi m} \sum^\infty_{n=1}q_j^{nm}\nonumber\\
    &= \log\left(1-e^{2\pi\iunit\abel_j(z,z_0)}\right)-\pi\iunit\,\abel_j(z,z_0)-\log\!\left(-\abel'_j(z_0)\,\eta'(0)\right)-\log\re\nonumber\\
    &\quad+4\pi \sum^\infty_{m=1}\frac{1-\cos(2\pi m\, \abel_j(z,z_0))}{2\pi m} \sum^\infty_{n=1}q_j^{nm}+\cO(\re).
\end{align}
To obtain a finite expression, we need to get rid of the divergent $\log\re$ term. In terms of the general approach described in \eqn{eqn:gensol}, this is immediate as $(\lambda\re)^{-\Res_{z_0}(K)}$ yields a term proportional to $\log(\lambda\re)\Res_{z_0}(K)$ upon expansion in the generators of the algebra $\alg{t}$. This amounts to essentially appending a term $\log(\lambda)\,{+}\,\log(\re)$ to \eqn{eqn:hgMZVregularization}, which implements the desired cancellation of the divergent term. Taking the limit $\re \,{\to}\, 0$ then finally leads to the regularized integral 
\begin{align}\label{eqn:hgMZVregularization2}
\left( \int^z_{z_0}\skern{1}_j(t, z_0)\!\right)_{\!\text{reg}}&=\log\Big(1-e^{2\pi\iunit\abel_j(z,z_0)}\Big)-\pi\iunit\,\abel_j(z,z_0)-\log\!\big({-}\,\abel'_j(z_0)\,\eta'(0)\big)\nonumber\\
&\quad+4\pi \sum^\infty_{m=1}\frac{1-\cos(2\pi m\, \abel_j(z,z_0))}{2\pi m} \sum^\infty_{n=1}q_j^{nm} + \log(\lambda)\, ,
\end{align}
where $\log(\lambda)$ arises from the regularization described by \eqns{eqn:gensol}{eqn:Leps}. 

Note that the constant $\lambda\,{\in}\,\zC^\times$ provides a certain freedom in what constants can be cancelled in \eqn{eqn:hgMZVregularization2}. Choosing for example
\begin{equation}
    \lambda=-\abel'_j(z_0)\,\eta'(0)
\end{equation}
would remove the term $-\!\log\!\left(-\abel'_j(z_0)\,\eta'(0)\right)$ in \eqn{eqn:hgMZVregularization2}. However, this cannot be achieved for all $i\,{\in}\,\{1,\ldots,\genus\}$ simultaneously since $\lambda$ amounts to only one complex degree of freedom. The remaining regularized integrals for $i\,{\neq}\, j$ would then be left with a modified constant term of the form
\begin{equation}
    -\log\!\left(-\abel'_i(z_0)\,\eta'(0)\right)\rightarrow-\log\left(\frac{\abel'_i(z_0)}{\abel'_j(z_0)}\right) \, .
\end{equation}
This illustrates that there will always remain a dependence on the moduli of the underlying complex geometry in the regularized hgMPL. This behavior will be radically different when considering regularization of higher-genus multiple zeta values in \secref{sec:hgmzvreg}: the regularized higher-genus multiple zeta values at depth one will simply be (rational) numbers, having lost any information on the underlying complex geometry. In other words, the dependence on the complex structure moduli contained in both endpoint regularizations exactly cancels in the regularization of the depth-one higher-genus multiple zeta values.

\subsubsection{Regularization at higher depth} \label{sec:higher_dep_reg_hgmpl}

Now that we have settled the depth-one regularization of hgMPLs for arbitrary covers, we can use this to infer the combinatorics of regularization for higher-depth hgMPLs. In \secref{sec:rev01}, we have seen that regularization of higher-depth MPLs and eMPLs could be traced back to the regularization at depth one by the shuffle product, hence demonstrating that it is sufficient to regularize the depth-one polylogarithms in order to resolve all divergences at higher depth. However, starting from genus $\genus\texteq2$, this statement has to be adjusted by taking into account that we have $\genus$ kernels exhibiting a simple pole instead of just one. A simple illustration is provided by the example
\begin{align}
       \Gargbare{(ii),(jj)}{z_0,z_0}{z}{z_0} ,
\end{align}
which is divergent in general. In order to regularize this, we would like to employ the shuffle product. Naively applying the shuffle relations gives the formal\footnote{In the sense that both sides of the equation are generally divergent.} relation
\begin{equation}
    \Gargbare{(ii),(jj)}{z_0,z_0}{z}{z_0}= 
    \Gargbare{(ii)}{z_0}{z}{z_0} \Gargbare{(jj)}{z_0}{z}{z_0} - \Gargbare{(jj), (ii)}{z_0,z_0}{z}{z_0}.
\end{equation}
Unlike the situation for $\genus \,{\leq}\, 1$, the right-hand side here contains a different divergent depth-two \hgMPL{}, which implies that this cannot be resolved by the shuffle relation alone since we do not know the regularization of all constituents on the right-hand side.

To overcome this, we use the linear identity~\eqref{eqn:linearid} for Enriquez' kernels. Integrating this identity permits the replacement of higher-depth integrals by differences that depend on an auxiliary point $u\,{\in}\,\funddom$, but are finite. For instance, we can write
\begin{align}\label{eqn:differenceformula}
    \Gargbare{(ii),(jj)}{z_0,z_0}{z}{z_0}-
    \Gargbare{(ii),(kk)}{z_0,z_0}{z}{z_0}=
    \Gargbare{(ii),(jj)}{z_0,u}{z}{z_0}-
    \Gargbare{(ii),(kk)}{z_0,u}{z}{z_0}
\end{align}
for generic $i,j,k\,{\in}\,\{1,\ldots,\genus\}$. As a consequence, any divergent higher-depth integral can be expressed in terms of combinations that involve either depth-one regularized integrals or finite differences of divergent expressions. For the depth-two case, we obtain
\begin{align}
    \Gargbare{(ii),(jj)}{z_0,z_0}{z}{z_0}
    &=\Gargbare{(ii)}{z_0}{z}{z_0}  \Gargbare{(jj)}{z_0}{z}{z_0}- \Gargbare{(jj),(ii)}{z_0,z_0}{z}{z_0} \nonumber\\
    &\quad + \underbrace{\Gargbare{(jj),(jj)}{z_0,z_0}{z}{z_0}- \Gargbare{(jj),(jj)}{z_0,z_0}{z}{z_0}}_{=0}  \nonumber\\
    &=
    \Gargbare{(ii)}{z_0}{z}{z_0} \Gargbare{(jj)}{z_0}{z}{z_0}+\Gargbare{(jj),(jj)}{z_0,u}{z}{z_0}\nonumber\\
&\quad-\Gargbare{(jj),(ii)}{z_0,u}{z}{z_0}-\frac{1}{2}\left(\Gargbare{(ii)}{z_0}{z}{z_0}\right)^{2}.
\end{align}
where we used a shuffle identity and inserted a trivial difference to apply \eqn{eqn:differenceformula}. In conclusion, we can now define
\begin{equation} \label{eqn:regularized_dep_2_hgmpl}
    \begin{aligned}
        \Gtreg{(ii),(jj)}{z_0,z_0}{z}{z_0}=&\Gtreg{(ii)}{z_0}{z}{z_0} \Gtreg{(jj)}{z_0}{z}{z_0}+\Gtreg{(jj),(jj)}{z_0,u}{z}{z_0} \\
    &-\Gtreg{(jj),(ii)}{z_0,u}{z}{z_0}-\frac{1}{2}\left(\Gtreg{(ii)}{z_0}{z}{z_0}\right)^{2},
    \end{aligned}
\end{equation}
since we know how to regularize all the terms on the right-hand side.

Iterating this method, we can reduce any divergent \hgMPL{} of arbitrary depth to expressions involving only depth-one regularized \hgMPL{}s and finite products of regular integrals. Consequently, we will henceforth drop the explicit ``reg'' subscript and understand all \hgMPL{}s to be given in their regularized form.

The implications for the regularization of higher-genus multiple zeta values will be discussed in detail in \secref{sec:hgmzvreg}.

\section{Integration on the Schottky cover}\label{sec:TechniquesSchottky}

In this section, we collect several techniques and considerations for higher-genus polylogarithms and their kernels based on employing the Schottky uniformization from \subsecref{sec:EnriquezSchottky}.  We provide an account of how to understand integration along \Atxt-cycles in the Schottky language in \subsecref{sec:IntegrationSchottky}. As an example, we are going to consider a Schottky version of a theorem from \rcite{DHoker:2025dhv} in \subsecref{sec:thm2}. 

\subsection{\texorpdfstring{\Atxt-cycle integration on the Schottky cover}{A-cycle integration on the Schottky cover}}\label{sec:IntegrationSchottky}

The Schottky uniformization introduced in \subsecref{sec:Schottky} provides an explicit representation of functions, differential forms and homology cycles on higher-genus Riemann surfaces in terms of Poincar\'e series over the Schottky group~$\SGroup$ and simple (closed\footnote{Strictly speaking, the Schottky cover does not represent $\bcyc$-cycles in terms of closed curves, but lifts them to simple curves connecting the circles on the covering space.}) curves on the Riemann sphere $\ComplexComplete$. In this way, it provides a convenient tool to perform explicit calculations of the higher-genus polylogarithms defined in \secref{sec:hgmpls}. We demonstrate the necessary techniques to carry out such calculations on the Schottky cover by explicitly evaluating the analytic representation of Enriquez' kernel $\omega_{ij}(z,x)$ given in \eqn{eqn:omegaij2} on the Schottky cover, thereby showing compatibility of the two representations.

We start by expressing all the objects in \eqn{eqn:omegaij2} in terms of their explicit representation on the Schottky cover. The normalized basis of holomorphic differentials admits the representation given in \eqn{eqn:schottky-holomorphic-basis}. In order to evaluate the integral in \eqn{eqn:omegaij2}, it is necessary to rewrite the normalized differential of the third kind~\eqref{eqn:fundamentalDiff3} in terms of a Poincar\'e series defined on the Schottky uniformization. By comparing analytic properties, we can write
\begin{equation}
    \label{eqn:fundDiffschottky}
    \ndiff(x) = \sum_{\gamma\in\SGroup}\left(\frac{1}{\gamma x - y}-\frac{1}{\gamma x - t}\right)\dd(\gamma x)
\end{equation}
on the Schottky cover. The proof of this statement is deferred to \appref{app:Ndiff}. Once this is established, we can write \eqn{eqn:omegaij2} as
\begin{equation}
    \begin{aligned}
        \label{eqn:omegaij-schottky}
        \omega_{ij}(x,y)=&\frac{1}{(-2\pi\iunit)}\sum_{\substack{\gamma_1\in \SGroup\\\gamma_2\in \SCosetR{j}}}\dd(\gamma_1x)\int_{t\in\acyc_i}\dd t\frac{1}{2\pi\iunit}\left(\frac{1}{t-\gamma_2 P_j'}-\frac{1}{t-\gamma_2 P_j}\right)\left(-\frac{1}{\gamma_1x-t}\right) \\
        &+\frac{1}{(-2\pi\iunit)}\sum_{\gamma_1\in \SGroup}\frac{\dd(\gamma_1x)}{\gamma_1x-y}\int_{t\in\acyc_i}\omega_j(t) \\
        &+\frac{1}{2}\delta_{ij}\,\omega_j(x)\, ,
    \end{aligned}
\end{equation}
where we have already substituted the definition of the normalized differential of the third kind~\eqref{eqn:fundDiffschottky} as well as the holomorphic differentials on the Schottky cover given by \eqn{eqn:schottky-holomorphic-basis}. The first line in \eqn{eqn:omegaij-schottky} contains the integral over $\omega_j(t)$ convoluted with the second term of \eqn{eqn:fundDiffschottky}, which carries the $t$-dependence. The second line is the convolution integral of $\omega_j(t)$ with the first term in \eqn{eqn:fundDiffschottky}, which is independent of $t$ and can henceforth be taken out of the integral. The remaining integral over the holomorphic form $\omega_j(t)$ immediately yields $\delta_{ij}$ by \eqn{eqn:periodmatrix}. Evaluation of the first line, however, is more involved. Our strategy is to use the residue theorem to perform the integration around $\acyc_i$.

\paragraph{Residue analysis.}

In order to apply the residue theorem, we need to find all (simple) poles of the function
\begin{equation}
    \mathcal{I}_j(t)=\left(\frac{1}{t-\gamma_2 P_j'}-\frac{1}{t-\gamma_2 P_j}\right)\frac{1}{\gamma_1x-t} \, ,
\end{equation}
which are located inside the contour $\acyc_i$ for fixed $\gamma_1\,{\in}\,\SCosetR{j}$ and $\gamma_2\,{\in}\,\SGroup/\SGroup_j$. From the above equation, $\mathcal{I}_j(t)$ can have simple poles in $t$ at the locations $p_{P'}(\gamma_2)\texteq\gamma_2P_j'$, $p_{P}(\gamma_2)\texteq\gamma_2P_j$ and $p_x(\gamma_1)\texteq\gamma_1x$. Moreover the points $p_{P'}(\gamma_2)$ and $p_{P}(\gamma_2)$ are always located on the same side of the contour except for $\gamma_2\texteq\id$. Several situations need to be distinguished:

\itemwithlabel[l]{\mpostuse{Situation1}}{\label{item:1}
  The first situation is given by all the singularities being located on the same side of the contour. By reversing the orientation if necessary, we can assume that the integrand is holomorphic on the side enclosed by the contour. Accordingly, in this situation, the integral vanishes by the residue theorem. 
}

\itemwithlabel[r]{\mpostuse{Situation2}}{\label{item:2} 
  The second situation is given when the poles $p_{P'}(\gamma_2)$ and $p_{P}(\gamma_2)$ are enclosed by the contour, but $p_x(\gamma_1)$ is not. This happens when $\gamma_2 \,{\in}\, \ssubsetl{(\SCosetR{j})}{-}{i}$ and $\gamma_1 \,{\in}\, \csubsetl{\SGroup}{-}{i}$ (cf.~\eqns{eqn:set}{eqn:cset} for the notation). In that case, the two enclosed residues contribute 
    \begin{equation}
        \label{eqn:res1}
	(-2\pi\iunit)(\Res_{p_{P'}}+\Res_{p_{P}})\,\mathcal{I}_j=\frac{(-2\pi\iunit)}{\gamma_1x-\gamma_2 P_j'}-\frac{(-2\pi\iunit)}{\gamma_1x-\gamma_2 P_j} \, ,
    \end{equation}
    where we have also included a minus sign to attribute that the $\acyc_i$-cycle conventionally has negative orientation on the Schottky cover (cf.~\secref{sec:Schottly_defs}).   
}

\itemwithlabel[l]{\mpostuse{Situation3}}{\label{item:3} 
  The next situation involves the opposite structure. The pole at $p_x(\gamma_1)$ is enclosed by the contour, but the other poles are on the other side of the contour. This happens when $\gamma_1 \,{\in}\, \ssubsetl{\SGroup}{-}{i}$ and $\gamma_2 \,{\in}\, \csubsetl{(\SCosetR{j})}{-}{i}$.
  To calculate the contribution to the residue theorem, we can reverse the orientation of the contour, which results in once again enclosing the poles $p_{P'}(\gamma_2)$ and $p_P(\gamma_2)$. According to \sitref{item:2}, this therefore yields
    \begin{equation}
        \label{eqn:res2}
        (-2\pi\iunit)\Res_{p_x}\mathcal{I}_j=-(-2\pi\iunit)(\Res_{p_{P'}}+\Res_{p_{P}})\,\mathcal{I}_j \, ,
    \end{equation}
    where the additional minus sign on the right-hand side is caused by reversing the orientation of the contour. Notice that, when $i\texteq j$, we explicitly need to exclude the case $\gamma_2\texteq\id$ from the above considerations, because in that case the pole $p_P(\id)$ would be enclosed by the contour $\acyc_i$ as well. We will treat this case separately in \sitref{item:4} below. 
}

\itemwithlabel[r]{\mpostuse{Situation4}}{\label{item:4} 
Finally, consider the case $i\texteq j$: then the point $P_j$ is enclosed by the contour $\acyc_i$ and additional pole configurations make an appearance. 
\begin{enumerate}
\item The first one arises when $\gamma_2\texteq\id$ and $\gamma_1 \,{\in}\, \csubsetl{\SGroup}{-}{i}$. 
This leads to $p_P(\id)$ being enclosed by the contour, but the other poles being located outside. The associated contribution to the residue theorem reads
\begin{equation}
    \label{eqn:res3}
(-2\pi\iunit)\Res_{p_P}\,\mathcal{I}_j=-(-2\pi\iunit)\frac{1}{\gamma_1x-P_j} \, ,
\end{equation}
where we again accounted for the negative orientation of $\acyc_i$ with an additional minus sign.
\item A similar situation occurs for $\gamma_2\texteq\id$ and $\gamma_1 \,{\in}\, \ssubsetl{\SGroup}{-}{i}$.
Here $p_P(\id)$ and $p_x(\gamma_1)$ are enclosed, but $p_{P'}(\id)$ is located on the opposite side. The residue contribution is
\begin{equation}
    \label{eqn:res3_v2}
(-2\pi\iunit)(\Res_{p_P}+\Res_{p_x})\,\mathcal{I}_j=-(-2\pi\iunit)\frac{1}{\gamma_1x-P_j'},
\end{equation}
with the usual minus sign accounting for the negative orientation of $\acyc_i$.
\end{enumerate}
}
%

\paragraph{Manipulating Schottky sums.} 

In the next step, all contributions from \sitref{item:1} to \sitref{item:4} above are to be combined and the result needs to be simplified. While the detailed combinatorics are explicitly lined out in \appref{app:TechnicalitiesIntegrationSchottky}, we constrain ourselves to presenting the intuition behind this procedure as well as the final result here.

It is instructive to imagine the action of the generators $\sigma_i$, $i\,{\in}\,\{1,\ldots,\genus\}$, of the Schottky group $\SGroup$ with the operation of moving along a cycle $\bcyc_i$ on the covering space. Hence an element $\Upsilon\,{\in}\,\SGroup$ can be regarded as a path on the Schottky cover\footnote{The direction of the path is defined by applying $\Upsilon$ to some $z_0$ on the Schottky cover. The path is then set to start at $z_0$ and end at $\Upsilon z_0$, moving along $\bcyc$-cycles in between}.
In \appref{app:TechnicalitiesIntegrationSchottky}, it is shown that the final expression for $\omega_{ij}(x,y)$ in fact only depends on the product $\Upsilon\texteq\gamma_1^{-1}\gamma_2\,{\in}\,\SCosetR{j}$. Let us fix such an element and consider it as a path moving between circles on the covering space (not passing the circle pair corresponding to $\sigma_j$ in the beginning due to the coset). Hence the conditions on $\gamma_1$ and $\gamma_2$ in\footnote{\sitref{item:4} only contains terms with $\gamma_2\texteq\id$ and requires separate treatment.}~\sitref{item:1} -- \sitref{item:3}, which originate from the residue analysis, can be interpreted as a certain condition on this path. This is displayed in \figref{fig:pathschottky}. 
\begin{figure}[t]
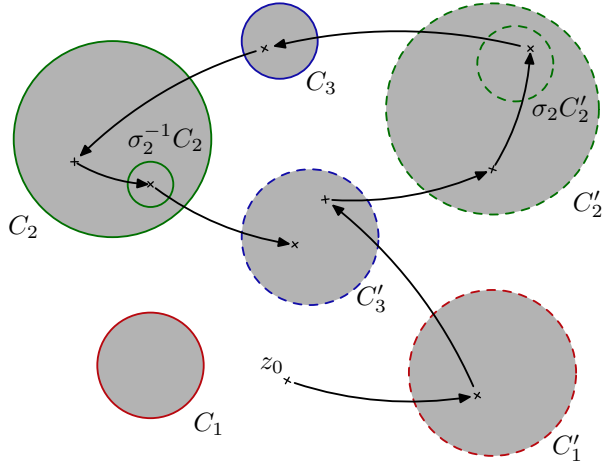

	\centering
	\mpostuse{SchottkySetup}
	\caption{Illustration of the residue conditions in \sitref{item:1} -- \sitref{item:3}. It shows a Schottky cover corresponding to a surface of genus three. The path intuitively corresponds to the element $\Upsilon\texteq\sigma_3\sigma_2^{-2}\sigma_3^{-1}\sigma_2^2\sigma_3\sigma_1\,{\in}\,\SGroup$ (applied to $z_0\,{\in}\,\funddom$). The condition in \sitref{item:2} amounts to counting the positive intersections of the path and the circles $C_2$ and $\sigma_2^{-1}C_2$. Correspondingly, the condition in \sitref{item:3} corresponds to counting the positive intersections of the path and the circles $C_2'$ and $\sigma_2C_2'$.}
    \label{fig:pathschottky}
\end{figure}

For example, the conditions in \sitref{item:2} translate to the requirement that the path passes at least once through the area enclosed by the integration contour\footnote{The condition on $\gamma_1$ in fact prohibits that the path is trivially walked backwards after entering the area enclosed by $\acyc_i$.} $\acyc_i$ since $\gamma_2\texteq\sigma_i^n\cdots$ for $n\,{<}\,0$ maps all points to the area enclosed by $C_i$. More precisely, let us define a \textit{positive intersection} of a path (corresponding to $\Upsilon\,{\in}\,\SGroup$) with a closed contour $C$ to be such that the path is intersecting $C$ from the left \wrt{the orientation of the contour $C$}. Then the conditions in~\sitref{item:2} can be regarded as counting the number of positive intersections of the path corresponding to an element $\Upsilon\,{\in}\,\SCosetR{j}$ with the (negatively oriented) circles $\sigma_i^nC_i$, $n\,{\leq}\,0$ as this number corresponds to the possible ways to decompose $\Upsilon$ into $\gamma_1^{-1}\gamma_2$ for $\gamma_2\,{\in}\, \ssubsetl{(\SCosetR{j})}{-}{i}$ and $\gamma_1 \,{\in}\, \csubsetl{\SGroup}{-}{i}$ as in \sitref{item:2}. Notice that this also corresponds to the number of generators $\sigma_i^{-1}$ contained in $\Upsilon$. We can interpret the conditions in~\sitref{item:3} analogously by replacing the circle $C_i$ corresponding to $\acyc_i$ with its associated circle $C_i'$. The residue theorem therefore counts the number of generators $\sigma_i^{\pm1}$ in $\Upsilon$ weighted by the residue contributions~\eqref{eqn:res1} and~\eqref{eqn:res2} (modulo additional contributions for $i\texteq j$ occurring when $\gamma_2\texteq\id$). Accordingly, we can rewrite \eqn{eqn:omegaij-schottky} as
\begin{equation}
    \label{eqn:omegaij-schottkyv2}
    \omega_{ij}(x,y) = \frac{1}{(-2\pi\iunit)}\sum_{\Upsilon\in \SCosetR{j}}(N_i^+(\Upsilon)-N_i^-(\Upsilon))\,s_j^{(0)}(\Upsilon^{-1}x) + \delta_{ij}\,[\,\cdots] \, ,
\end{equation}
where the ellipsis collects the contributions only relevant when $i\texteq j$ arising when $\gamma_2\texteq\id$ as well as the remaining terms from the last two lines in \eqn{eqn:omegaij-schottky}. Moreover, $N_i^{\pm}(\Upsilon)$ is defined in \eqn{eqn:nsplitting} and encodes the number of generators $\sigma_i^{\pm1}$ in $\Upsilon$ (or equivalently the number of positive intersections of the path corresponding to $\Upsilon$ and the circles $\sigma_i^nC_i'$, $n\,{\geq}\,0$ or $\sigma_i^nC_i$, $n\,{\leq}\,0$, respectively). The holomorphic kernel $s_j^{(0)}(\Upsilon^{-1}x)$ arises from the residue contributions~\eqref{eqn:res1} and~\eqref{eqn:res2}. The precise details of the calculation are contained in \appref{app:TechnicalitiesIntegrationSchottky}. By definition of the $C$-coefficients~\eqref{eqn:Ccoeff}, we can now conclude that
\begin{equation}
    N_i^+(\Upsilon)-N_i^-(\Upsilon)=C(b_i,\Upsilon) \, .
\end{equation}
This finalizes the appearance of the first term and it remains to simplify the contributions from~\sitref{item:4} for $i\texteq j$. It can be shown that these exactly recombine with the remaining terms in \eqn{eqn:omegaij-schottky} to give
\begin{equation}
    \frac{1}{(-2\pi\iunit)}\sum_{\Upsilon\in \SCosetR{j}}s_j^{(1)}(\Upsilon^{-1}x,y) \, .
\end{equation}
For details, we again refer to \appref{app:TechnicalitiesIntegrationSchottky}. Combining all of this, \eqn{eqn:omegaij-schottkyv2} finally becomes
\begin{equation}
    \label{eqn:final}
    \omega_{ij}(x,y)=\frac{1}{(-2\pi\iunit)}\sum_{\Upsilon\in \SCosetR{j}}\left(C(b_i,\Upsilon)\,s_j^{(0)}(\Upsilon^{-1}x)+\delta_{ij}\,s_j^{(1)}(\Upsilon^{-1}x,y)\right) \, ,
\end{equation}
which exactly coincides with \eqn{eqn:highergenuskernels}.

\subsection{Example: kernel recursions in the Schottky language}\label{sec:thm2}

In this subsection, we show that the Schottky representation~\eqref{eqn:highergenuskernels} and the convolution definition~\eqref{eqn:DSthm2} of Enriquez' kernels are compatible, by illustrating that the result of~\rcite[Thm.~3]{DHoker:2025dhv} can also be derived using the Schottky language. 
We are considering the simplest non-trivial case of~\cite[Thm.~3]{DHoker:2025dhv} (the general case is stated in \eqn{eqn:DSthm3}), which reads\footnote{In \rcite{DHoker:2025dhv}, a different normalization is employed, which translates via $g^{i_1 \cdots i_r}{ }_j(x, y)=(-2\pi\iunit)^r \omega_{i_1 \cdots i_rj}(x, y)$.}
\begin{align}\label{eqn:thm3simple}
    \int_{t \in \acyc_j} \omega_{aa}(y, t)\, \omega_k(t)= -  \omega_{jk}(y)\, ,
\end{align}
where $a,j,k\in\{1,\ldots,\genus\}$ and we also assume $j\,{\neq}\, k$. More complicated situations are discussed in \appref{app:33}. 

First, without loss of generality, we can restrict to the case $a\,{\neq}\, j$: by \eqn{eqn:linearid}, the difference $\omega_{aa}(y,t) - \omega_{cc}(y,t)$ is independent of $t$ for any choice of $a$ and $c$. Note that the choice of $a$ in this example, as well as in further examples, is primarily for convenience in the calculation. Carrying out the calculation for an arbitrary point $x \,{\in}\, \funddom$, one obtains
\begin{align}\label{eqn:wlog}
	&\int_{t \in \acyc_j} \omega_{aa}(y, t)\, \omega_{k}(t) 
	- \int_{t \in \acyc_j} \omega_{cc}(y, t)\, \omega_{k}(t)
	= \left[\, \omega_{aa}(y, x) - \omega_{cc}(y, x) \,\right] 
	\underbrace{\int_{t\in\acyc_j} \omega_{k}(t)}_{\delta_{jk}}.
\end{align}
Since we assumed $j \,{\neq}\, k$ in \eqn{eqn:thm3simple}, the result of \eqn{eqn:wlog} is zero, allowing us to limit our consideration to $a\,{\neq}\,j$ in the following.

Next, by explicitly expanding the kernels in the Schottky uniformization using \eqn{eqn:highergenuskernels}, the left-hand side of \eqn{eqn:thm3simple} can be written as
\begin{align}\label{eqn:thm3simpleSchottky}
    \int_{t \in \acyc_j}\! \omega_{aa}(y, t)\, \omega_k(t)=\frac{1}{(-2\pi \iunit)^2}\sum_{\substack{\gamma_a \in \SCosetR{a}\\\gamma_k\in\SCosetR{k}}}\int_{t \in \acyc_j} \!\!\left(C(b_a, \gamma_a)\, \skern{0}_a(\gamma_a^{-1}y )+ s^{(1)}_a(\gamma_a^{-1} y, t )\!
    \right)\!\skern{0}_k(\gamma_k^{-1}t).
\end{align}
We are going to use the same arguments as in \subsecref{sec:IntegrationSchottky} to single out contributing terms from the above equation based on their pole configurations with respect to the integration path.

\paragraph{Residue analysis.}

Since $\skern{0}(\gamma_a^{-1}y)$ in \eqn{eqn:thm3simpleSchottky} does not depend on the integration variable~$t$, it factors out of the integral, i.e.
\begin{align}
	\frac{1}{(-2\pi \iunit)^2}\sum_{\substack{\gamma_a \in \SCosetR{a}\\\gamma_k\in\SCosetR{k}}}&\int_{t \in \acyc_j}C(b_a, \gamma_a)\, \skern{0}_a(\gamma_a^{-1}y )\,\skern{0}_k(\gamma_k^{-1}t)\\
	&=\frac{1}{(-2\pi \iunit)}\sum_{\gamma_a \in \SCosetR{a}}C(b_a, \gamma_a)\, \skern{0}_a(\gamma_a^{-1}y )\underbrace{\frac{1}{(-2\pi \iunit)}\sum_{\gamma_k\in\SCosetR{k}}\int_{t \in \acyc_j}\skern{0}_k(\gamma_k^{-1}t)}_{=\int_{\acyc_j}\omega_k=\delta_{jk}},\notag
\end{align}
which vanishes due to the assumption $j\,{\neq}\,k$.
Thus, the non-vanishing contribution to the left-hand side of \eqn{eqn:thm3simple} reduces to
\begin{align}\label{eqn:exThm3}
\int_{t \in \acyc_j} \omega_{aa}(y, t)\, \omega_k(t, z)= 
\frac{1}{(-2 \pi \iunit)^2} \sum_{\substack{\gamma_a \in \SCosetR{a}\\\gamma_k \in \SCosetR{k}}}
\int_{t \in  \acyc_j}\skern{1}_a(\gamma_a^{-1} y,  t)\, \skern{0}_k(\gamma_k^{-1}t).
\end{align}
We use the explicit representation of the integration kernels on the Schottky cover to identify all simple poles of the integrand
\begin{align}
    \mathcal{I}_{ak}(t)=\skern{1}_a(\gamma_a^{-1} y,  t)\, \skern{0}_k(\gamma_k^{-1}t).
\end{align}
From \eqns{eqn:sn}{eqn:schottky-holomorphic-basis} as well as~\appref{app:TechnicalitiesIntegrationSchottky} (see eqs.~\eqref{eqn:qExpDiv}--\eqref{eqn:qExpNondiv2} for details), we get the following expressions:
\begin{align}
    \label{eqn:holDiffschottkyapp}
    s^{(0)}_k(y)&=-\dd y\left(\frac{1}{y-P_k'}-\frac{1}{y-P_k}\right) \, , \\
    \label{eqn:g1schottkyapp}
    s^{(1)}_k(y,t)&=-\Bigg(\sum_{n>0}\frac{\dd(\sigma_k^ny)}{\sigma_k^ny-P_k}+\sum_{n<0}\frac{\dd(\sigma_k^ny)}{\sigma_k^n y-P_k'} -\sum_{n\neq0}\frac{\dd(\sigma_k^n y)}{\sigma_k^n y-t} \\
    &\hspace{25ex}+\frac{\dd y}{2}\left(\frac{1}{y-P_k'}+\frac{1}{y-P_k}\right)-\frac{\dd y}{y-t}\Bigg) \nonumber.
\end{align}
Since \eqn{eqn:exThm3} is an integration over the variable $t$, it can be seen from \eqns{eqn:holDiffschottkyapp}{eqn:g1schottkyapp} that the integrand has simple poles located at
\begin{align}
&p_{P}(\gamma_k) = \gamma_k P_k \quad \text{and} \quad
p_{P'}(\gamma_k) = \gamma_k P'_k,
&&\text{from } \skern{0}_k(\gamma_k^{-1} t),  \nonumber\\
&p_y^n(\gamma_a) = \sigma_a^n \gamma_a^{-1} y,
&&\text{from } \skern{1}_a(\gamma_a^{-1}y,t) ,
\end{align}
with $n \,{\in}\, \mathbb{Z}$.
Let us adapt the residue analysis, performed in \secref{sec:IntegrationSchottky}, to the configuration of poles at hand:
\begin{situations}
    \item \label{3.3item:1} All simple poles lie on the same side of the contour. In this case, the residue contribution is zero.
    \item \label{3.3item:2} The point $p_y^n(\gamma_a)$ lies outside the integration contour, while $p_{P}(\gamma_k)$ and $p_{P'}(\gamma_k)$ are located inside. This occurs either when $n \,{\neq}\, 0$ and $\gamma_a \,{\in}\, \SCosetR{a}$, or when $n \texteq 0$, $\gamma_a \,{\in}\, \mathrm{C}(\SCosetR{a})_{\leftarrow j}^{[+]}$, and $\gamma_k \,{\in}\, (\SCosetR{k})^{[-]}_{j \rightarrow}$. We choose the contour to enclose the poles located outside. Inverting the contour orientation introduces an additional minus sign, resulting in the residue contribution
    \begin{equation}
        -2\pi \iunit \Res_{p_y^n}\mathcal{I}_{ak}(t) 
        = {2\pi \iunit}\, \skern{0}_k(\gamma_k^{-1} \sigma_a^n \gamma_a^{-1}y).
    \end{equation}
    \item \label{3.3item:3} The points $p_{P}(\gamma_k)$, $p_{P'}(\gamma_k)$ are outside, while $p_y^0(\gamma_a)$ is inside of the contour. This occurs when $\gamma_a \,{\in}\, \SGroup^{[+]}_{\leftarrow j}$ and $\gamma_k \,{\in}\, \csubsetl{(\SCosetR{k})}{-}{j}$. The residue contribution is
    \begin{equation}
        2\pi \iunit\Res_{p_y^0}\mathcal{I}_{ak}(t) 
        = {2\pi \iunit}\, \skern{0}_k(\gamma_k^{-1} \gamma_a^{-1}y).
    \end{equation}
\end{situations}
%

\paragraph{Manipulating Schottky sums.}

This is where splitting Schottky words becomes necessary. For integration along the $\acyc_j$ direction, we analyze the splitting defined in \secref{sec:splitting_of_Schottky_words}. 
For example, in \sitref{3.3item:2}, we have the two splittings before $\sigma_j^{-1}$ that give the same residue contribution:
\begin{align}
    &\gamma_a = \sigma_k,\quad \gamma_k = \sigma_j^{-2} \sigma_a,\\
    &\gamma_a = \sigma_k\sigma_j^{-1} ,\quad \gamma_k = \sigma_j^{-1}\sigma_a,
\end{align}
which both result in
$
(-{2\pi \iunit})\, \skern{0}_k(\Upsilon^{-1}y)
$
with $\Upsilon\texteq\sigma_k^{-1}\sigma_j^{-2} \sigma_a$.
Together, this yields
\begin{align}
   (-2\pi\iunit) \,N_j^-(\Upsilon) \,\skern{0}_k(\Upsilon^{-1}y).
\end{align}
Here, the negative sign comes from the orientation of the contour $\acyc_j$, which was conventionally chosen to be negative (see \secref{sec:Schottly_defs}).
To obtain the total contribution, we must also include the terms from $\ssplit_j^+(\Upsilon)$, which happen to be empty in this specific example. Proceeding analogously for all $\Upsilon\,{\in}\, \SCosetR{k}$, we finally arrive at
\begin{align} \label{eqn:thm2_r0_case_calc}
\int_{t \in \acyc_j} \omega_{aa}(y, t)\, \omega_k(t, z)
&= \sum_{\Upsilon\in \SCosetR{k}} \frac{1}{(-2 \pi \iunit)}\left( N_j^-(\Upsilon) -N_j^+(\Upsilon)
\right)\skern{0}_k(\Upsilon^{-1}y) \nonumber\\
&= \sum_{\Upsilon\in \SCosetR{k}}
 \frac{1}{2 \pi \iunit} C(b_j,\Upsilon)\, \skern{0}_k(\Upsilon^{-1}y) \nonumber\\
 &=-\omega_{jk}(y,z).
\end{align}
The general relation between splitting number $N$ and the $C$-coefficient is a purely combinatorial problem, which is discussed in detail in \appref{app:combinatoricProblem}.
The case when $j\texteq k$ requires the use of the recursive integral representation of the genus-one integration kernels. The corresponding proof is provided in \appref{app:L=Kproof}.


\section{Higher-genus multiple zeta values}\label{sec:hgMZV}

Having reviewed languages and techniques for treating functions on Riemann surfaces of higher genus in the previous sections, we are now going to define higher-genus multiple zeta values (\hgMZV{}s). Following the pattern at genus zero and genus one, we are extending the regularization prescription of \hgMPL{}s to the second endpoint in \subsecref{sec:hgmzvreg}. 
\secref{sec:degen} and \secref{sec:hgmzvrel} are then devoted to studying degenerations and relations of \hgMZV{}s, respectively. 

\subsection{Definition}

Using Enriquez' higher-genus integration kernels for the definition of hgMPLs in \eqn{eqn:hgMPL}, we can now define higher-genus multiple zeta values:
\begin{definition}
	Let $\RSurf$ be a compact Riemann surface of genus $\genus\,{\geq}\,1$. For $k\,{>}\,0$, $j\,{\in}\,\{1,\ldots,\genus\}$, and $\mindx{i}_\ell$, $\ell\,{\in}\,\{1,\ldots,k\}$ multi-indices labeling Enriquez' kernels on $\RSurf$, we define \Atxt-cycle \emph{higher-genus multiple zeta values} (\hgMZV{}s) as higher-genus polylogarithms evaluated along the $j$-th \Atxt-cycle: 
	\begin{align}\label{eqn:hgMZV}
		\zetaAg{j}{\genus}(\mindx{i}_1,\mindx{i}_2,\ldots,\mindx{i}_k)&\coloneqq\Gargbare{\mindx{i}_1,\ldots,\mindx{i}_k}{z_0,\ldots,z_0}{z_0\,{+}\,\Acycle_j}{z_0}\notag\\
		&=\int_{t_1=z_0}^{z_0+\Acycle_j}\omega_{\mindx{i}_1}(t_1,z_0)\int_{t_2=z_0}^{t_1}\omega_{\mindx{i}_2}(t_2,z_0)\,\cdots\int_{t_k=z_0}^{t_{k-1}}\omega_{\mindx{i}_k}(t_k,z_0)\\ 
		&=\int_{\Acycle_j}\omega_{\mindx{i}_1}\circ\omega_{\mindx{i}_2}\circ\cdots\circ\omega_{\mindx{i}_k},\notag
	\end{align}
	where the point $z_0$ is an arbitrary point on the cycle $\Acycle_j$ and the pole variables for each integration kernel are chosen to be $z_0$. We refer to the arguments $\mindx{i}_1,\ldots,\mindx{i}_k$ of the \hgMZV{} as entries or slots.
\end{definition}
The notation $z_0\,{+}\,\Acycle_j$ denotes that we end our integration at $z_0$ after moving once around the cycle $\Acycle_j$. The last line of \eqn{eqn:hgMZV} is a shorthand for the iterated integration over the cycle $\Acycle_j$ (note that in this shorthand all pole variables of the kernels are set to $z_0$). If the genus is clear from the context, we are going to avoid the label $[\genus{}]$ for the genus frequently.

Just as for \hgMPL{}s above (cf.~\secref{sec:hgmpls}), we define the depth $k$ as the number of integrations of a \hgMZV{} and the weight as $w\texteq|\mindx{i}_1|\,{+}\,\ldots\,{+}\,|\mindx{i}_k|$. For brevity of notation, we will write powers for indices repeated within one multi-index, e.g.~$211133\texteq21^33^2$.

The \hgMZV{}s defined here are functions of the geometric information of the underlying Riemann surface, i.e.~functions of the period matrix $\tau$ or equivalently of the Schottky group~$\SGroup$. We will usually omit this dependence.

While \hgMZV{}s are just (special values of) iterated integrals and, thus, adhere to the general properties~\eqref{eqn:pathinversion},~\eqref{eqn:pathconcatenation} and~\eqref{eqn:shuffle} of iterated integrals, they also respect the Fay-like identities~\eqref{eqn:FaylikeId} that are special for the integration kernels used for \hgMZV{}s.

\subsection{Regularization of higher-genus MZVs}\label{sec:hgmzvreg}

Regularization of MZVs for any genus relies on the mechanism of regulating endpoint divergences for polylogarithms, as reviewed in \secref{sec:g1reg} and \secref{sec:hgMPLreg}. Other than (generic) polylogarithms, which are calculated as iterated path integrals starting from a basepoint $z_0$ and ending at the argument $z$ of the polylogarithm, \hgMZV{}s are integrated around a full \Atxt-cycle, which for convenience can be assumed as path connecting $z_0$ to $z_0\,{+}\,\acyc$. Accordingly, they need regularization at the basepoint $z_0$ twice: once when the integration path leaves from $z_0$ and once when approaching $z_0$ at the end of the integration path.

\paragraph{General formalism.} 

The generating series for (regularized) MZVs is typically identified with the associator of the solution to \eqn{eqn:KZhg} around a non-contractible cycle in the homology basis\footnote{This is the canonical definition for genera $\genus\,{\geq}\,1$. For $\genus\texteq0$, the associator is usually defined by connecting solutions at different punctures on the Riemann sphere.} (cf.~\rcites{gonzalez2020surfacedrinfeldtorsorsi,DHoker:2023vax} for higher-genus associators). Concretely, we consider the monodromy of the solution $L(z)$ around an $\acyc$-cycle and define the associator corresponding to $\acyc_i$, $i\,{\in}\,\{1,\ldots,\genus\}$ based at the puncture $z_0$ as
\begin{equation}\label{eqn:associator}
    Z_i=L(z)^{-1}L(z+\acyc_i)=\lim_{\re\rightarrow0}(\lambda\re)^{\Res_{z_0}(K)}\exp{\left(-\int_{(\acyc_i)_\re}K(z',z_0)\right)}(\lambda\re)^{-\Res_{z_0}(K)} \, ,
\end{equation} 
where $(\acyc_i)_\re\colon[0,1]\,{\rightarrow}\,\RSurf$ is defined by $(\acyc_i)_\re(s)\texteq\acyc_i((1\,{-}\,s)\,\re\,{+}\,s\,(1\,{-}\,\re))$. The (regularized) MZVs are now defined as the expansion coefficients of $Z_i$ for $i\,{\in}\,\{1,\ldots,\genus\}$, e.g. for depth one
\begin{equation}
    \zeta_{\acyc_i}^\reg(w)=[w]\,Z_i \, ,
\end{equation}
where $[w]$ extracts the coefficient of the word\footnote{By slight abuse of notation, we denote the \hgMZV{} as $\zeta_{\acyc_i}^{\text{reg}}(w)$ for $w\,{\in}\,\alg{t}$ here.} $w\,{\in}\,\alg{t}$ in $Z_i$. In the next paragraph, we will carry out this procedure to first order explicitly. Through the prescription in \eqn{eqn:associator} endpoint divergences at both ends of the integration path are taken care of automatically in this way.

\paragraph{Regularization of higher-genus MZVs.}\label{sec:hgmzv}

In \subsecref{sec:hgMPLreg}, we addressed the regularization of hgMPLs, where we only needed to handle the endpoint divergence at one integration boundary. However, according to the definition in \eqn{eqn:hgMZV}, because hgMZV involve a closed contour integration, divergences occur at both endpoints.

We start by considering the depth-one hgMZVs, which require regularization: $\hgzeta{i}(jj)$.
As discussed for \hgMPL{}s in \secref{sec:hgMPLreg}, the only divergent term in the expansion into genus-one kernels is $\skern{1}_j(\gamma^{-1}z,z_0)$ when $\gamma \texteq \id$. Thus, for $z_0\texteq\acyc_i(0)\texteq\acyc(1)$, we focus on the integral
\begin{align}\label{eqn:hgmzvreg1}
	\int_{(\acyc_i)_\re} \skern{1}_j(t, z_0 )
	&=\int_{\abel_j(z_\re,z_0)}^{\abel_j(z_{1-\re},z_0)+\delta_{ij}} \gkern{1}(\vt\mid\tau_j)\, \dd \vt,
\end{align}
with $(\acyc_i)_\re$ as defined in \eqn{eqn:associator} and where $z_\re\texteq(\acyc_i)_\re(0)$ and $z_{1-\re}\texteq(\acyc_i)_\re(1)$, respectively. Furthermore, notice that the upper integration boundary acquired a $\delta_{ij}$ due to the cycle $(\acyc_i)_\re$ crossing the branch cut of (the logarithm in) Abel's map $\abel_j$ for $i\texteq j$. More precisely, Abel's map $\abel_j$ has a branch cut connecting the fixed points $P_j'$ and $P_j$. Hence, when going around the cycle $(\acyc_i)_\re$ for $i\texteq j$, we cross the branch cut and enter a different sheet of the logarithm. This results in the additive factor of $\delta_{ij}$ in the upper integration boundary. Next, let us split the contour $(\acyc_i)_\re$ at some point\footnote{W.l.o.g.~we assume that the branch cut crosses the contour ``after'' $z_1$, i.e.~if the crossing happens at the point $z_c\texteq(\acyc_j)_\re(s_c)$, we assume $s_c\,{>}\,s_1$.} $z_1\texteq(\acyc_i)_\re(s_1)$ for $s_1\,{\in}\,(0,1)$. Then, we can rewrite \eqn{eqn:hgmzvreg1} as 
\begin{equation}\label{eqn:hgmzvreg2}
    \begin{aligned}
        \int_{(\acyc_i)_\re} \skern{1}_j(t, z_0 )=&\int_{\abel_j(z_\re,z_0)}^{\abel_j(z_{1-\re},z_0)+\delta_{ij}} \gkern{1}(\vt\mid\tau_j)\, \dd \vt \\
        =& \int_{\abel_j(z_\re,z_0)}^{\abel_j(z_1,z_0)} \gkern{1}(\vt\mid\tau_j)\, \dd \vt+\int_{\abel_j(z_1,z_0)}^{\abel_j(z_{1-\re},z_0)+\delta_{ij}} \gkern{1}(\vt\mid\tau_j)\, \dd \vt .
    \end{aligned}
\end{equation}
We can now use the above result to trace the \hgMZV{} regularization back to the endpoint regularization of hgMPLs described in \secref{sec:hgMPLreg}. Using in particular \eqn{eqn:hgMZVregularization2}, we can write \eqn{eqn:hgmzvreg2} as
\begin{equation}
    \begin{aligned}
        \int_{(\acyc_i)_\re} \skern{1}_j(t, z_0 )&=\delta_{ij}\pi\iunit-\log\!\big(\!-\abel'_j(z_0)\,{(\acyc_i)_{\re}'}(0)\,\re\big)+\log\!\big(\abel'_j(z_0)\,{(\acyc_i)_{\re}'}(1)\,\re\big)+\cO(\re) \\
        &=\delta_{ij}\pi\iunit-\log(-1)+\cO(\re)\, ,
    \end{aligned}
\end{equation}
where the term $\delta_{ij}\pi\iunit$ arises due to the branch cut contribution in the upper boundary and $\abel'_j(z)$ and $(\acyc_i)_{\re}'(z)$ are the derivative of Abel's map and the tangent vector to the integration path evaluated at $z$, respectively. Moreover, the term $\log(-1)$ appears due to the fact that $\acyc_i(1\,{-}\,\re)\texteq\acyc_i^{-1}(\re)$, which implies that the upper boundary can also be reached by walking the path by an amount $\re$ in the opposite direction. Finally, we can perform the limit $\re\,{\rightarrow}\,0$ to obtain
\begin{equation}\label{eqn:hgmzvreg3}
	\left(\int_{\acyc_i} \skern{1}_j(t, z_0 )\right)_{\text{reg}}=\lim_{\re\to0}\int_{(\acyc_i)_\re} \skern{1}_j(t, z_0 )=(\delta_{ij}-1)\,\pi\iunit.
\end{equation}
The branch cut contribution exactly cancels the term $\log(-1)\texteq\pi\iunit$ in the case $i\texteq j$, so that \eqn{eqn:hgmzvreg3} implies
\begin{align}\label{eqn:hgmzvreg}
	\zetaAreg{i}(jj)&=
	\lim_{\re\to0}\int_{(\acyc_i)_\re}\omega_{jj}(t, z_0 \mid\! \SGroup)\notag\\
	&=(\delta_{ij}-1)\frac{\pi\iunit}{(-2\pi\iunit)}+\sum_{\gamma \in \SCosetR{j}} C(b_i, \gamma)\int_{\acyc_i}\omega(\gamma^{-1}t \mid\! \SGroup_j) 
	+\frac{1}{(-2 \pi \iunit)} \sum_{\substack{\gamma \in \SCosetR{j}\\\gamma\neq\id}} \int_{\acyc_i}\skern{1}_j(\gamma^{-1} z, z_0)\notag\\
	&=\begin{cases}
		0, & i=j,\\
		\frac12, & i\neq j,
	\end{cases}
\end{align}
for all $i,j\,{\in}\,\{1,\ldots,\genus\}$ as all terms with $\gamma\,{\neq}\,\id$ are regular at the boundaries and therefore cancel for a closed contour. Moreover, the second term vanishes also for $\gamma\texteq\id$ as $C(b_i,\id)\texteq0$ by \eqn{eqn:Ccoeff}. Thus, the \hgMZV{}s $\zetaAreg{i}(ii)$ vanish through the regularization, which is consistent with the elliptic case studied e.g.~in \rcites{Matthes:Thesis,Broedel:2015hia,Broedel:2019gba}, while the other cases for $i\,{\neq}\,j$ yield the non-vanishing result 1/2. This is in particular consistent with \eqn{eqn:linearid}: choosing a point $v$ in the fundamental domain, which is not on the cycle $\acyc_k$ and $z_0$ on the cycle $\acyc_k$, one finds for $k\,{\neq}\,j$
\begin{align}
	\zetaAreg{k}(jj)&=\lim_{\re\rightarrow0}\int_{(\acyc_k)_\re}\omega_{jj}(t,z_0)=\underbrace{\lim_{\re\rightarrow0}\int_{(\acyc_k)_\re}\omega_{kk}(t,z_0)}_{=\zetaAreg{k}(kk)=0}+\underbrace{\int_{\acyc_k}\omega_{jj}(t,v)}_{=0}-\underbrace{\int_{\acyc_k}\omega_{kk}(t,v)}_{=-\frac12}=\frac12,
\end{align}
where we first used \eqn{eqn:linearid} and then the regularized \hgMZV{} $\zetaAreg{k}(kk)\texteq0$ as well as the result~\eqref{eqn:d1}.

Let us stress the difference between the regularization result~\eqref{eqn:hgmzvreg} and relation~\eqref{eqn:d1}: for the definition of the \hgMZV{}s, the pole variable of the integration was assumed to lie \textit{on the integration path}, requiring the above regularization calculation. On the contrary, for \eqn{eqn:d1} the pole variable is required to lie \textit{inside the fundamental domain} (and not on the integration path), such that in the latter case no regularization is needed.

Finally, for regularization of the higher-depth \hgMZV{}s, we adopt the same strategy as for \hgMPL{}s in \secref{sec:higher_dep_reg_hgmpl}: by means of the shuffle relations~\eqref{eqn:shuffle} and the linear identity~\eqref{eqn:differenceformula}, one can express higher-depth \hgMZV{}s in terms of lower-depth counterparts in the same way as it is done in \eqn{eqn:regularized_dep_2_hgmpl}. In this way, one can systematically rewrite all the divergent \hgMZV{}s in terms of divergent depth-one quantities, which are regularized by \eqn{eqn:hgmzvreg}.

\subsection{Higher-genus MZVs of depth one}\hypertarget{para:dep1}{}\label{sec:relations_depth_one}

Now that we have defined the \hgMZV{}s, we would like to apply the Schottky formalism in order to reproduce the formula~\eqref{eqn:d1} for depth-one \hgMZV{}s, including the regularization from \secref{sec:hgmzvreg}. We thereby show that depth-one \hgMZV{}s directly reduce to the elliptic values, which in their turn are simply given by the ordinary zeta values at genus zero. More precisely, we have the following result:
\begin{proposition}[Enriquez~\cite{EnriquezHigher}, Enriquez--Zerbini~\cite{EZ1}]\label{prop:depth1}
Depth-one \hgMZV{}s are given by
\begin{equation}\label{eqn:hgmzvdepth1}
	\hgzeta{k}(i_1\cdots i_r j)=\begin{cases}
		-\delta_{jki_1\cdots i_r}\frac{2\zeta_r}{(-2\pi \iunit)^r}, & \text{for }r \text{ even} ,\\
        \delta_{r1}\delta_{j i_1}(1- \delta_{jk})\frac{1}{2}, & \text{for } r \text{ odd}.
	\end{cases}
\end{equation}
for $\zeta_0\,{\equiv}\,\frac12$ and the regularization shown in \secref{sec:hgmzvreg}.
\end{proposition}
In other words, we find eMZVs at depth-one, i.e.~$\zetaA{j}(j^{n+1})\texteq\omell(n)$, which are actually genus-zero zeta values of even arguments. The special case for odd $r$ is due to the regularization, see  \eqn{eqn:hgmzvreg}.
This proposition (in a slightly different form) was already proven in \rcites{EnriquezHigher,EZ1} (and also provided in \rcite{DHoker:2025szl}) using only the properties of Enriquez' kernels. Here, we provide an alternative proof using Schottky uniformization.
\begin{proof}
To simplify the notation, let us consider the depth-one \hgMZV{}s of the form $\zetaA{k}(i_1^{n_1}\cdots i_r^{n_r} j)$, where $i_k\,{\ne}\, i_{k+1}$ and $n_k \,{\ge}\, 1$ for all $k\,{\in}\, \{1,\ldots,r\}$.
The case $\zetaA{k}(jj)\texteq (1-\delta_{jk})\half$ follows from the regularization procedure, cf.~\eqn{eqn:hgmzvreg}. For all other cases, where regularization is not needed, we use \eqn{eqn:highergenuskernels} to expand $\zetaA{k}(i_1^{n_1}\cdots i_r^{n_r} j)$ in the Schottky formalism, thereby obtaining
\begin{equation}
    \label{eqn:depth-one}
    \zetaA{k}(i_1^{n_1}\cdots i_r^{n_r} j)=\frac{1}{(-2\pi\iunit)}\sum_{\gamma\in\SCosetR{j}}\sum_{l=0}^{\delta_{i_rj}n_r}C(b_{i_1}^{n_1}\cdots b_{i_r}^{n_r-l},\gamma)\int_{\acyc_k}s^{(l)}_j(\gamma^{-1}t,t_0) \,.
\end{equation}
To calculate the integral in \eqn{eqn:depth-one}, we employ the residue theorem and the techniques developed in \secref{sec:TechniquesSchottky}. First of all, notice that as $s^{(l)}_j(t,t_0)$ only depends on the $j$-th subcover, its poles can at most be located at images of $P_j'$, $P_j$ and $\sigma_j^nt_0$, $n\,{\in}\,\zZ$ under $\SCosetR{j}$. Denote these poles by $p_{P_j'}(\gamma)\texteq\gamma P_j'$, $p_{P_j}(\gamma)\texteq\gamma P_j$ and $p_{t_0}(\gamma)\texteq\gamma\,\sigma_j^nt_0$ for $\gamma\,{\in}\,\SCosetR{j}$. Next, notice that for $\gamma\,{\neq}\id$, these poles are always located on the same side of the contour $\acyc_k$. Therefore, by inverting the orientation of the contour if necessary, we can disregard these contributions in the residue theorem and the only case to be considered is $\gamma\texteq\!\id$. We can therefore write
\begin{equation}
    \begin{aligned}
    \label{eqn:depth-oneCs}
    \sum_{\gamma\in\SCosetR{j}}&C(b_{i_1}^{n_1}\cdots b_{i_r}^{n_r-l},\gamma)\int_{\acyc_k}s^{(l)}_j(\gamma^{-1}t,t_0)=C(b_{i_1}^{n_1}\cdots b_{i_r}^{n_r-l},\id)\,\delta_{jk}\int_{\acyc_j}s^{(l)}_j(t,t_0)\, . 
    \end{aligned}
\end{equation}
for $0\,{\leq}\, l\,{\leq}\, n_r$. Notice that $\delta_{jk}$ encodes the fact that $j\texteq k$ is the only situation, where we can possibly have poles on both sides of the integration contour. The integral on the right-hand side is now purely dependent on the $j$-th subcover and hence can be related to an eMZV as 
\begin{equation}
    \int_{\acyc_j}s^{(l)}_j(t,t_0)=(-2\pi\iunit)^{1-l}\int_{\acyc_j}g^{(l)}_j(t,t_0)\,\omega(z\mid\SGroup_j)=(-2\pi\iunit)^{1-l}\,\omell(l) \, ,
\end{equation}    
where the last equality follows from the substitution $\xi\texteq\abel_j(t,t_0)$. Notice that we consider all values of $l$ other than $1$, as the divergent integral with $l\texteq1$ only appears within $\zetaA{k}(jj)$ with a non-zero coefficient, which we have already addressed using regularization. With this information, we find
\begin{equation}
    \label{eqn:depth-oneDeltas}
    \begin{aligned}
        \sum_{\gamma\in\SCosetR{j}}C(b_{i_1}^{n_1}\cdots b_{i_r}^{n_r-l},\gamma)\int_{\acyc_k}s^{(l)}_j(\gamma^{-1}t,t_0)=(-2\pi\iunit)^{1-l}\,\delta_{r1}\,\delta_{n_1l}\,\delta_{jk}\,\omell(l) \, ,
    \end{aligned}
\end{equation} 
where we have used the definition~\eqref{eqn:Ccoeff} of the $C$-coefficients as well. In conclusion, we can use \eqn{eqn:depth-oneDeltas} to rewrite \eqn{eqn:depth-one} as
\begin{equation}\label{eqn:zetadepth-one}
    \zetaA{k}(i_1^{n_1}\cdots i_r^{n_r} j)=\delta_{jk}\left(\delta_{r0}\,\omell(0)+(-2\pi\iunit)^{-n_1}\,\delta_{r1}\,\delta_{i_1j}\,\omell(n_1)\right) \, ,
\end{equation}
where the first term corresponds to $l\texteq0$ and the second term corresponds to $l\,{>}\,0$, which only occurs if $\delta_{i_1j}$. This is now equivalent to \eqn{eqn:hgmzvdepth1}. To see this, note that \eqn{eqn:zetadepth-one} vanishes unless all indices are the same. In the latter case, the formula evaluates to 
\begin{equation}
    \zetaA{j}(j^{n+1})=\frac{\omell(n)}{(-2\pi\iunit)^n} \, .
\end{equation}
This exactly matches \eqn{eqn:hgmzvdepth1} as eMZVs of depth one $\omell(n)$ vanish if $n$ is odd and equal $-2\zeta_n$ for $n$ even and positive (and $\omell(0)\texteq1$)~\cite{Broedel:2015hia}.
\end{proof}
%

\section{Degenerations of higher-genus MZVs}\label{sec:degen}

In this section, we study the behavior of \hgMZV{}s under degenerations of the Riemann surface to which they are associated. As a preparation, let us briefly review the types of degenerations available for Riemann surfaces.

Topologically, the degeneration of a Riemann surface can be done in two different ways~\cite{fay}: one can pinch the surface along a separating curve, that is a simple closed curve homologous to zero\footnote{A curve is said to belong to the zeroth homology class if and only if it bounds a 2-chain, i.e.~a 2-dimensional subregion.}, or along a non-separating curve, i.e.~a curve belonging to a non-trivial homology class. The respective curves are shrunk to minimal length and then cut open. The former degeneration is commonly referred to as \emph{separating degeneration}, while the latter is called a \emph{non-separating degeneration}. A common choice for the non-separating curve is an \Atxt{}-cycle, which we use in what follows. \figref{fig:degen} illustrates both types of degenerations.

A non-separating degeneration results in a surface of one genus lower with two additional marked points (or punctures) corresponding to the pinched cycle, cf.~\figref{fig:degen_nonsep}. In terms of the period matrix, degenerating the cycle $\acyc_j$ leads to $\tau_{jj} \,{\to}\, \iunit \infty$. 
On the other hand, a separating degeneration produces two disconnected surfaces of lower genera, with a marked point on each of them as in \figref{fig:degen_sep}. Therefore, the period matrix takes block-diagonal form in the degeneration, where each block corresponds to one of the disconnected components.

In \secref{sec:nonseparatingA}, we are considering the behavior of \hgMZV{}s under the non-separating degenerations, that is, under pinching \Atxt-cycles, which was discussed in \rcite{Baune:2024biq} for the Schottky uniformization. Given that in our conventions Schottky circles correspond to \Atxt-cycles, pinching the cycle $\acyc_j$ translates into shrinking the corresponding circles $C_j$ and $C_j'$ to their respective fixed points $P_j$ and $P_j'$, see \figref{fig:degeneration-schottky}. Accordingly, the Schottky generator $\sigma_j$ degenerates to the map\footnote{For $z\,{\neq}\, P_j,P_j'$.} $z\,{\mapsto}\, P_j'$ (and $\sigma^{-1}$ to $z\,{\mapsto}\, P_j$), and is thus no longer a M\"obius transformation\footnote{Note that it is not clear, whether the degeneration of a Schottky cover results in a well-defined Schottky cover of the lower-genus Riemann surface (see discussion in~\cite[Sec.~5.3]{Baune:2024biq}), but the functions behave well under the degenerations.}. 
As pinching of an \Atxt-cycle leads to a Riemann surface of one genus lower, it is expected that \hgMZV{}s of a genus-\genus{} surface degenerate to \hgMZV{}s of a surface of genus $\genus{-}1$, as we will discuss in the following. We will derive these relations using the degeneration properties of the integration kernels. Thus, the degeneration properties derived for \hgMZV{}s here hold analogously for \hgMPL{}s.

In \secref{sec:separatingB}, we are going to introduce the separating degeneration in the Schottky picture. This degeneration corresponds to shrinking two fixed points $P_j$ and $P_j'$ of a generator $\sigma_j$ to a single point while keeping the multiplier parameter $\lambda_j$ fixed, see \figref{fig:degeneration-sep-schottky}. Hereby, effectively two scales are created on the Schottky cover: from the point of view of the untouched generators, the shrinking produces a single marked point on the cover; from the point of view of the shrunken generators, all the Schottky circles corresponding to the other generators are moved away to $\infty$, creating a marked point there. These marked points at different scales exactly correspond to the touching point of the lower-genus components of the original surface in the separating degeneration limit. We will demonstrate that starting from a genus-\genus{} surface, we can recover two disconnected surfaces of genera 1 and $\genus{-}1$ at the level of the period matrix. Moreover, by considering the two different scales, we show that the Schottky--Kronecker forms (cf.~\eqn{eqn:schottkykronecker}) degenerate into those of the lower-genus surfaces in the Schottky language. In this way, we obtain the degeneration of Enriquez' kernels and, thus, of the \hgMZV{}s.

\begin{figure}[t]
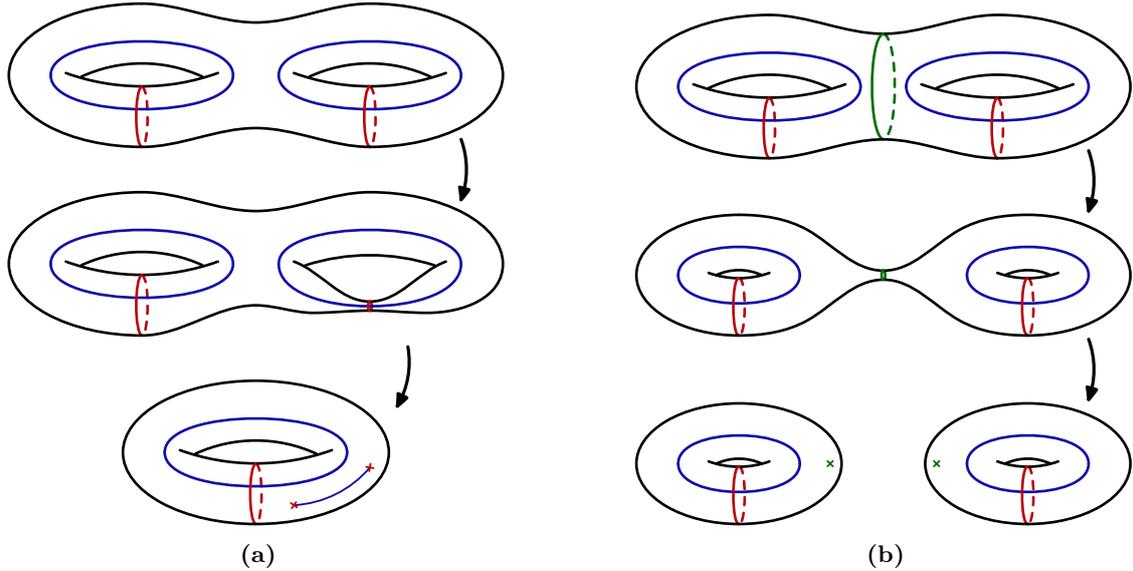

    \centering
    \begin{subfigure}[b]{0.45\textwidth}
        \centering
        \mpostuse{degen_nonsep}
        \caption{}
        \label{fig:degen_nonsep}
    \end{subfigure}
    \hspace{0.05\textwidth}
    \begin{subfigure}[b]{0.45\textwidth}
        \centering
        \mpostuse{degen_sep}
        \caption{}
        \label{fig:degen_sep}
    \end{subfigure}
    \caption{Degenerations of a genus-two Riemann surface. (a) Non-separating degeneration, where an $\mathfrak{A}$-cycle is pinched, resulting in a genus-one surface with two punctures. (b) Separating degeneration, where the surface splits into two lower-genus components.}
    \label{fig:degen}
\end{figure}
%

\subsection{Non-separating degeneration}\label{sec:nonseparatingA}

In order to derive the degeneration of \hgMZV{}s, we first need to know the degeneration of Enriquez' integration kernels. This can be found through the degeneration of the component forms $K_k$, as was derived and discussed in ref.~\cite[Sec.~5.3]{Baune:2024biq} and is as follows:
we consider a genus-\genus{} Riemann surface and degenerate in the Schottky picture by pinching one cycle $\acyc_j$, for $j\,{\in}\,\{1,\ldots,\genus\}$.

\paragraph{Degeneration of an unaffected direction.} 

For $k\,{\neq}\, j$, in \rcite{Baune:2024biq}, it was found that the genus-$\genus$ component forms $K_k$ degenerate to the component forms of genus $\genus{-}1$,
\begin{equation}
	K_k(z,x\mid\SGroup^{(\genus)})\overset{\acyc_j\to0}{\longrightarrow }K_k(z,x\mid\SGroup^{(\genus-1)}),
\end{equation}
where $\SGroup^{(\genus-1)}$ refers to the Schottky group of the degenerate surface, that is the group generated by the generators $\{\sigma_i\}_{i\neq j}$. For the expansion into words as in \eqn{eqn:connectionexpansion}, this implies that the letters $b_j$ do not appear anymore in the degenerate Schottky--Kronecker form and thus all kernels $\omega_{\cdots j\cdots k}$ vanish in this degeneration, as we can also observe numerically. Thus, any \hgMZV{} containing such a kernel will vanish when degenerated, e.g.~at arbitrary genus $\genus$, for $j\,{\neq}\, i$,
\begin{equation}
	\zetaAg{i}{\genus}(ji,i^3)\overset{\acyc_j\to0}{\longrightarrow}0.
\end{equation}
Other \hgMZV{}s only involving kernels from $K_k$ for $k\,{\neq}\, j$ will simply reduce to \hgMZV{}s of one genus lower, e.g.~at arbitrary genus \genus{}, for $i,k\,{\neq}\, j$,
\begin{equation}
	\zetaAg{i}{\genus}(ik^2,k^2i^2)\overset{\acyc_j\to0}{\longrightarrow}\zetaAg{i}{\genus-1}(ik^2,k^2i^2).
\end{equation}
\begin{figure}[t]
	\centering
	\begin{subfigure}[b]{0.5\textwidth}
		\centering
		\mpostuse{shottky-degeneration}
		\caption{}
        \label{fig:degeneration-schottky}
	\end{subfigure}%
	\begin{subfigure}[b]{0.5\textwidth}
		\centering
		\includegraphics[width=\textwidth]{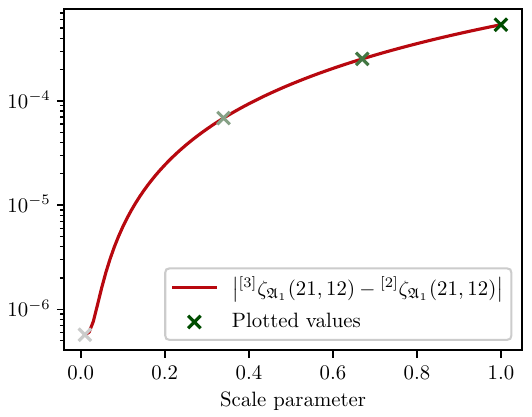}
		\caption{}
        \label{fig:degeneration-numerics}
	\end{subfigure}
	\caption{Figure (a) shows a graphical representation of degenerating a Schottky cover of genus three by shrinking one pair of circles, which corresponds to the non-separating degeneration in \figref{fig:degen_nonsep}. Figure (b) shows the absolute value of the difference between the value of $\zetaAg{1}{3}(21,12)$ (which depends on the 
    scaling parameter) and the value of $\zetaAg{1}{2}(21,12)$. The scaling parameter controls the radii of the green circles in (a) (depicted as fading green circles). Fading green crosses in (b) correspond to the shrunken pairs of circles in (a).}
    \label{fig:degeneration}
\end{figure}
\vspace{-3ex}

\noindent This shows that we can get any \hgMZV{} by degenerating a \hgMZV{} of higher genus, such that the desired \hgMZV{} is a ``descendant'' of another \hgMZV{} of one genus higher after degeneration. \figref{fig:degeneration-numerics} depicts a numerical example of a genus-three \hgMZV{} degenerating to a genus-two \hgMZV{}.

In particular, if we degenerate all directions except one, the result is zero unless all indices and the integration cycle are equal to the remaining direction and we find the corresponding eMZV of the non-degenerate genus-one subcover (with modulus $\tau_i$)
\begin{equation}
	\zetaA{i}(\mindx{i}_1,\ldots,\mindx{i}_k)\overset{\{\acyc_j\to0\}_{j\neq i}}{\longrightarrow}(-2\pi\iunit)^{-\sum_{m=1}^k\len{\mindx{i}_m}}\,\delta_{\mindx{i}_1,\ldots,\mindx{i}_k,i}\ \omell(\len{\mindx{i}_1},\ldots,\len{\mindx{i}_k}\mid \tau_i).
\end{equation}
%

\paragraph{Degeneration of the affected direction.} 

The case $k\texteq j$ is more involved: here the degeneration of the component form $K_j$ in the Schottky language (cf.~\eqn{eqn:schottkykronecker}) is
\begin{align}
	K_j(z,x\mid\SGroup^{(\genus)})\overset{\acyc_j\to0}{\longrightarrow}&\sum_{\gamma\in \SGroup^{(\genus-1)}}\frac{\dd z}{z-\gamma P_j'}\frac{\gamma P_j-\gamma P_j'}{\gamma P_j-z}W(\gamma)\sum_{n<0}w_j^n\notag\\
	&+\sum_{\gamma\in \SGroup^{(\genus-1)}}\frac{\dd z}{z-\gamma x}\frac{\gamma P_j-\gamma x}{\gamma P_j-z}W(\gamma),
\end{align}
where again $\SGroup^{(\genus-1)}$ refers to the Schottky group of the degenerate surface. Thus, there are kernels coming from the expansion of this component form, which remain finite after the degeneration. Curiously, all dependence on the generator $\sigma_j$ vanishes. However, some of the kernels still depend on the locations of the fixed points $P_j$ and $P_j'$ of $\sigma_j$. In this sense, if one views the degenerate surface as a Riemann surface of genus $\genus{-}1$ with two marked points $P_j,\,P_j'$, the form $K_j$ will carry the information of the marked points in the limit, while the other forms $K_k$, $k\,{\neq}\, j$, do not see any contributions from the degenerate cycle $\acyc_j\,{\to}\,0$.

Thus, \hgMZV{}s with kernels of this form (e.g.~$\zetaA{i}(ij,i)$ for $\acyc_j\,{\to}\,0$) are not expected to degenerate purely to lower-genus \hgMZV{}s, but that they carry also geometric information of the degenerate direction through the fixed points.

\subsection{Separating degeneration}\label{sec:separatingB}

In order to define the separating degeneration in the Schottky language, let us assume without loss of generality that $z \texteq 0$ is in the fundamental domain of a given Schottky cover with the generators $\{\sigma_i\}_{i=1}^{\genus}$. Each generator $\sigma_i$ is uniquely determined by the two fixed points $P_i$ and $P_i^\prime$ and the multiplier parameter $\lambda_i$ (cf.~\eqn{eqn:fixedpteq}).  
Now let us consider a different Schottky cover, for which we keep the parameters $P_i$, $P_i^\prime$ and 
$\lambda_i$ of the $i$-th Schottky generator intact for $i\, {\ne}\, j$, while for the $j$-th generator we rescale\footnote{There is a subtlety that for each value of $\epsilon$ the corresponding Schottky cover should be well-defined in the sense that Schottky circles do not intersect and the Poincar\'e series are convergent. If this is not the case, one can always move the shrinking point $z\texteq0$ somewhere else. In what follows, we assume that shrinking a generator to $z\texteq0$ does not cause any of those problems.}
\begin{equation} \label{eqn:sep_deg_rescaling}
P_j \to \epsilon P_j, 
\quad
P_j^\prime \to \epsilon P_j^\prime, 
\quad
\lambda_j \to \lambda_j,
\end{equation}
with some rescaling parameter $0 < \epsilon \le 1$. The separating degeneration in Schottky uniformization is then achieved as the limit $\epsilon \,{\to}\, 0$. 

\paragraph{Degeneration of the period matrix.}

One can observe that the separating degeneration prescription given above indeed produces the period matrix of the block diagonal form as follows. For an arbitrary point $y\,{\in}\,\ComplexComplete$, let us first list some useful relations for the degeneration limit:
\begin{equation} \label{eqn:sep_deg_rels}
\sigma_j (y) =  \epsilon \,\tilde{\sigma}_j(\infty)+ \Order(\epsilon^2),
\qquad
\sigma_j(\epsilon y) = \epsilon\, \tilde{\sigma}_j(y) + \Order(\epsilon^2),
\qquad
\sigma_i (\epsilon y) = \sigma_i(0) + \Order(\epsilon),
\end{equation}
where $ \tilde{ \sigma}_j$ is the generator with non-rescaled fixed points $P_j$, $P_j^\prime$ and the multiplier $\lambda_j$. Now we are in the position to analyze the entries of the period matrix~\eqref{eq:period-matrix-schottky}. Let us separately consider $\periodmatrix_{jj}$, $\periodmatrix_{jl}$ and $\periodmatrix_{lm}$ for all $l, m \,{\ne}\, j$.
\begin{enumerate}
    \item $\periodmatrix_{jj}$: we can always decompose an element $\gamma\,{\in}\,(\SCosetLR{j}{j})\setminus\{\id\}$ as $\gamma \texteq \alpha \beta$, where $\alpha, \beta \,{\in}\, \SGroup$ such that $\alpha$ does not contain $\sigma_j$ or $\sigma_j^{-1}$ and $\beta\texteq\sigma_j^{\pm 1}\cdots$, i.e.~$\beta$ ends in a generator $\sigma_j^{\pm1}$ (or $\beta \texteq \id$). Due to \eqn{eqn:sep_deg_rels}, we deduce that $\gamma(\epsilon y) \texteq \alpha(0) \,{+}\, \Order(\epsilon)$ for all $y \,{\in}\, \ComplexComplete$. Therefore, the cross-ratio in the logarithm reads
    \begin{equation}
    \{\epsilon P_j', \alpha(0) + \Order(\epsilon), \epsilon P_j, \alpha(0) + \Order(\epsilon)\} = 1 + \Order(\epsilon),
    \end{equation}
    and thus each summand in the series vanishes in the degeneration limit. Accordingly, $\periodmatrix_{jj}$ coincides with $\periodmatrix_j$, the modular parameter of the $j$-th subcover. 
    \item $\periodmatrix_{jl}$: again, we can always decompose $\gamma\,{\in}\,\SCosetLR{l}{j}$ as $\gamma \texteq \alpha \beta$ like before and the cross-ratio reads
    \begin{equation}
    \{P_l', \alpha(0) + \Order(\epsilon), P_l, \alpha(0) + \Order(\epsilon)\} = 1 + \Order(\epsilon).
    \end{equation}
    Therefore, all off-diagonal terms for $j$-th entry are vanishing, $\periodmatrix_{ij} \texteq \periodmatrix_{jl} \texteq 0$ for all $i \,{\ne}\, j$.
    \item $\periodmatrix_{lm}$: now, if we decompose $\gamma \texteq \alpha \beta$ as before, there are two distinct cases: $\beta \texteq \sigma_j^n \cdots$ or $\beta \texteq \id$. In the former case, due to \eqn{eqn:sep_deg_rels}, $\beta (y) \texteq \Order(\epsilon)$ and thus $\gamma (y) \texteq \alpha(0) \,{+}\, \Order(\epsilon)$ and the cross-ratio 
    \begin{equation}
    \{\epsilon P_m', \alpha(0) + \Order(\epsilon), \epsilon P_m, \alpha(0) + \Order(\epsilon)\} = 1 + \Order(\epsilon),
    \end{equation}
    again vanishes in the degeneration limit. If $\beta \texteq \id$, then the cross-ratio does not depend on $\epsilon$ and it gives a non-trivial contribution to the period matrix. Thus, the non-degenerate entries of the period matrix reduce to \eqn{eq:period-matrix-schottky}, with $\SGroup$ being replaced by $\SGroup^{(\genus-1)}$ --- the Schottky group generated by $\{\sigma_{i}\}_{i \ne j}$ for the surface of genus $\genus{-}1$. 
\end{enumerate}

Now, let us turn our attention to what happens to the Schottky--Kronecker forms~\eqref{eqn:schottkykronecker} in the separating degeneration limit $\epsilon \,{\to}\, 0$. 

\begin{figure}[t]
	\centering
	\begin{subfigure}[b]{0.5\textwidth}
		\centering
		\mpostuse{schottky-degeneration-sep}
		\caption{}
        \label{fig:degeneration-sep-schottky}
	\end{subfigure}%
	\begin{subfigure}[b]{0.5\textwidth}
		\centering
		\includegraphics[width=\textwidth]{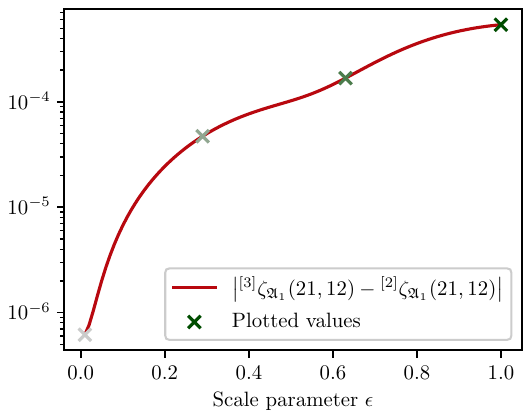}
		\caption{}
        \label{fig:degeneration-sep-numerics}
	\end{subfigure}
	\caption{Figure (a) depicts a degeneration of a genus-three surface by the rescaling 
    \protect\eqref{eqn:sep_deg_rescaling}, which corresponds to the separating degeneration from \figref{fig:degen_sep}. The absolute value of the difference between $\zetaAg{1}{3}(21,12)$, which depends on the scaling parameter $\epsilon$, and $\zetaAg{1}{2}(21,12)$, calculated for the surface formed by red and blue pairs of Schottky circles, is plotted in Figure (b). The fading green circles in (a) correspond to different values of $\epsilon$, which are also highlighted in (b) by the fading green crosses.}
    \label{fig:degeneration-sep}
\end{figure}
%

\paragraph{Degeneration of an unaffected direction.}

Let us consider the $k$-th component of the Schottky--Kronecker form $K_k$,
\begin{equation}
K_k(z, x) = \frac{1}{(-2 \pi \iunit)} \sum_{\gamma \in \SGroup}
\left[\frac{\der z}{z - \gamma x} - \frac{\der z}{z - \gamma P_k}\right] W(\gamma).
\end{equation}
We can split the summands in the Poincar\'e series above into two types.
\begin{enumerate}
    \item $\sigma_{j}^{\pm1} \,{\in}\, \gamma$, i.e.~$\gamma$ contains the $j$-th generator as a letter in the Schottky word. Similar considerations as for the period matrix allow us to conclude that 
    \begin{equation}
    \gamma (y) = c(\gamma) + \Order(\epsilon),
    \end{equation}
    for all $y \,{\in}\, \ComplexComplete$, where $c(\gamma) \,{\in}\, \ComplexComplete$ is some number that does not depend on $y$. Then for each summand we have 
    \begin{equation}
    \frac{\der z}{z - \gamma x} - \frac{\der z}{z - \gamma P_k} 
    = 
    \frac{\der z}{z - c(\gamma) + \Order(\epsilon)} - \frac{\der z}{z - c(\gamma) + \Order(\epsilon)}
    = \Order(\epsilon),
    \end{equation}
    which vanishes for $\epsilon\, {\to}\, 0$.
    \item $\sigma_{j}^{\pm1} \,{\notin}\, \gamma$, i.e.~$\gamma$ does not contain the $j$-th generator as a letter in the Schottky word. Then each summand does not depend on $\epsilon$.
\end{enumerate}
Therefore, the $k\,{\ne}\, j$-th component reduces to 
\begin{equation}
K_k(z, x) \to \frac{1}{(-2 \pi \iunit)} \sum_{\gamma \in \SGroup^{(\genus-1)}}
\left[\frac{\der z}{z - \gamma x} - \frac{\der z}{z - \gamma P_k}\right] W(\gamma).
\end{equation}
Thus, we recover the Schottky--Kronecker form components of genus $\genus{-}1$ for $k \,{\ne}\, j$.

This result for the Schottky--Kronecker form implies that the Enriquez kernels $\omega_{\mindx{i}k}(z, x)$ for $k \,{\ne}\, j$, that arise from $K_k(z, x)$, degenerate to the kernels of the lower-genus surface. Moreover, if the multi-index $\mindx{i}$ contains the index $j$, the corresponding kernel vanishes:
\begin{equation}
\omega_{\mindx{i}k}(z, x \mid \SGroup)  \overset{\epsilon\to 0}{\longrightarrow}
\begin{cases}
    \omega_{\mindx{i}k}(z, x \mid \SGroup^{(\genus-1)}), &  j \notin \mindx{i},
    \\
    0, & \text{otherwise}.
\end{cases}
\end{equation}
This also implies that \hgMZV{}s directly degenerate into lower-genus \hgMZV{}s as long as the weights do not contain the degeneration index $j$,
\begin{equation}
\zetaAg{k}{\genus}(\mindx{i}_1 k_1,\mindx{i}_2 k_2, \ldots, \mindx{i}_n k_n) \overset{\epsilon\to 0}{\longrightarrow}
\begin{cases}
    \zetaAg{k}{\genus-1}(\mindx{i}_1 k_1,\mindx{i}_2 k_2, \ldots, \mindx{i}_n k_n), & j \notin \mindx{i}_1, \mindx{i}_2, \ldots, \mindx{i}_n,
    \\
    0, & \text{otherwise},
\end{cases}
\end{equation}
where $k, k_s \,{\ne}\, j$ for $s \,{\in}\, \{1, \ldots, n\}$. \figref{fig:degeneration-sep-numerics} gives a numerical illustration of the separating degeneration of a genus-three \hgMZV{} to a genus-two \hgMZV{}.

\paragraph{Degeneration of the affected direction.}

In order to show that the $j$-th direction degenerates to the genus-one case, we introduce the rescaled one-form
\begin{equation}
\tilde{K}_j(z, x) := K_j(\epsilon z, \epsilon x) = \frac{1}{(-2 \pi \iunit)} \sum_{\gamma \in \SGroup}
\left[\frac{\epsilon\, \der z}{\epsilon z - \gamma (\epsilon x)} - \frac{\epsilon \,\der z}{\epsilon z - \gamma (\epsilon P_j)}\right]W(\gamma).
\end{equation}
We again consider several cases separately:
\begin{enumerate}
    \item $\gamma\texteq\sigma_i^{\pm 1}\cdots$ for $i\,{\neq}\, j$, i.e.~$\gamma$ does not end with the generator $\sigma_j$. Due to \eqn{eqn:sep_deg_rels}, we deduce that 
    \begin{equation}
    \gamma (\epsilon y) = c(\gamma) + \Order(\epsilon), 
    \end{equation}
    where $c(\gamma) \,{\in}\, \ComplexComplete$ does not depend on $y$. Therefore, each summand in the Poincar\'e series
    \begin{equation}
    \frac{\epsilon \,\der z}{\epsilon z - \gamma (\epsilon x)} - \frac{\epsilon \,\der z}{\epsilon z - \gamma (\epsilon P_j)} = 
    \frac{\epsilon \,\der z}{-c(\gamma) + \Order(\epsilon)} - \frac{\epsilon \,\der z}{-c(\gamma) + \Order(\epsilon)} = \Order(\epsilon)
    \end{equation}
    vanishes in the degeneration limit $\epsilon \,{\to}\, 0$.
    \item $\gamma\texteq\sigma_j^n\sigma_i^{\pm 1}\cdots$ for $i\,{\neq}\, j$, i.e.~$\gamma$ ends with the generator of the degenerate direction. Then one can show that 
    \begin{equation}
    \gamma (\epsilon y) = c(\gamma)\, \epsilon + \Order(\epsilon^2),
    \end{equation}
    for some $c(\gamma)\,{\in}\,\ComplexComplete$ that does not depend on $y$. Thus 
    \begin{equation}
    \frac{\epsilon \,\der z}{\epsilon z - \gamma (\epsilon x)} - \frac{\epsilon \,\der z}{\epsilon z - \gamma (\epsilon P_j)} = 
    \frac{ \der z}{ z - c(\gamma) + \Order(\epsilon)} - \frac{\epsilon \,\der z}{ z - c(\gamma) + \Order(\epsilon)} = \Order(\epsilon),
    \end{equation}
    which again vanishes in the degeneration limit.
    \item $\gamma \texteq \sigma_j^n, \, n\,{\in}\,\Integers$. Then, due to \eqn{eqn:sep_deg_rels} we have 
    \begin{equation}
    \frac{\epsilon \,\der z}{\epsilon z - \gamma (\epsilon x)} - \frac{\epsilon \,\der z}{\epsilon z - \gamma (\epsilon P_j)} = \frac{\der z}{z -  \tilde{ \sigma}_j^n (x)} - \frac{ \der z}{ z -  \tilde{ \sigma}_j^n (P_j)},
    \end{equation}
    which is finite in the degeneration limit.
\end{enumerate}
Altogether, we find that the $j$-th component of the Schottky--Kronecker form degenerates to the genus-one Kronecker form (upon rescaling):
\begin{equation}
\tilde{K}_j(z,x) \to \frac{1}{(-2 \pi \iunit)} \sum_{\gamma \in \tilde{\SGroup}_j}
\left[\frac{ \der z}{ z -  \gamma x} - \frac{ \der z}{ z - \gamma P_j}\right] W(\gamma),
\end{equation}
where $\tilde{\SGroup}_j$ is the genus-one Schottky group generated by $\tilde{\sigma}_j$.

By expanding the rescaled Schottky--Kronecker form, one can recover the genus-one meromorphic kernels upon rescaling of the genus-$\genus$ kernels in the $j$-th direction:
\begin{equation}
\tilde{\omega}_{\mindx{i} j}(z, x \mid \SGroup) 
=\omega_{\mindx{i} j}(\epsilon z, \epsilon x \mid \SGroup) 
\overset{\epsilon\to 0}{\longrightarrow}
\delta_{j \mindx{i}}\frac{1}{(-2\pi\iunit)} \skern{n}(z, x \mid \tilde{\SGroup}_j).
\end{equation}
This result directly translates to the \hgMZV{}s. We define the ``rescaled'' \hgMZV{} as
\begin{equation}
\zetaAgt{j}{\genus}(\mindx{i}_1j,\mindx{i}_2j,\ldots,\mindx{i}_kj) \coloneqq 
\int_{\Acycle_j}\tilde{\omega}_{\mindx{i}_1j}\circ \tilde{\omega}_{\mindx{i}_2j}\circ\cdots\circ \tilde{\omega}_{\mindx{i}_kj},
\end{equation}
where the contour $\acyc_j$ is the non-degenerate $\acyc$-cycle on the Schottky cover for the $j$-th direction. We therefore conclude that we find in the degeneration limit
\begin{equation} \label{eqn:sep_deg_ell_zeta}
\zetaAgt{j}{\genus}(\mindx{i}_1j,\mindx{i}_2j,\ldots,\mindx{i}_kj) 
\overset{\epsilon\to 0}{\longrightarrow}
\delta_{j \mindx{i}_1 \mindx{i}_2 \ldots \mindx{i}_k}\,(-2\pi\iunit)^{-\sum_l\!n_l}\,\omell_j(n_1, n_2, \ldots, n_k),
\end{equation}
where $n_l \texteq |\mindx{i}_l| \,{+}\, 1$.

Notice that we have not considered the kernels of the degenerate direction $\omega_{\mindx{i} j}(z, x \mid \SGroup)$ and the rescaled kernels of the non-degenerate directions $\tilde{\omega}_{\mindx{i} k}(z, x \mid \SGroup)$ for $k \,{\ne}\, j$. Naively, one could assume that those kernels must vanish in the degeneration limit. In fact, these functions are non-trivial, which is explained by the non-vanishing of the non-rescaled $j$-th component $K_j(z,x)$ as well as rescaled $l\,{\ne}\, j$-th components $\tilde{K}_l(z,x) \texteq K_l(\epsilon z, \epsilon x)$ in the degeneration limit $\epsilon \,{\to}\, 0$:
\begin{align}
K_j(z, x) &\to \frac{1}{(- 2 \pi \iunit)}\sum_{\gamma \in \SGroup^{(\genus-1)}}
\left[\frac{\der z}{z - \gamma x} - \frac{\der z}{z - \gamma(0)}\right] W(\gamma),
\\
\tilde{K}_{l\ne j}(z,x) &\to \frac{1}{(-2 \pi \iunit)} \sum_{\gamma \in \tilde{\SGroup}_j}
\left[\frac{ \der z}{ z -  \gamma x} - \frac{ \der z}{ z - \gamma(\infty)}\right] W(\gamma).
\end{align}
This is similar to the situation of the non-separating degeneration: the residual contributions to the components carry the information about the punctures due to the shrinking of the separating cycle, which are located at $z \texteq 0$ for the residual surface of genus $\genus {-}1$ (see the pole in $K_j(z,x)$) and at $z \texteq \infty$ in the rescaled coordinates for the genus-one surface (see the pole in $\tilde{K}_l(z,x)$). Therefore, the associated \hgMZV{}s are not expected to degenerate to the lower-genus values.

As a final comment, we could have rescaled several pairs of the fixed points simultaneously:
\begin{equation}
P_j \to \epsilon P_j,
\qquad
P_j' \to \epsilon P_j',
\qquad
j \in \{j_1, j_2, \ldots, j_k\} \subset \{1, \ldots, \genus\},
\end{equation}
which would allow one to degenerate a genus-\genus{} surface into two components of genera $k$ and $\genus {-} k$. Nevertheless, the considerations above would remain the same qualitatively.

\section{Relations for higher-genus MZVs}\label{sec:hgmzvrel}
In this section, we will explore identities among \hgMZV{}s for a fixed Riemann surface $\RSurf$ of genus~$\genus$. These relations can either rely on properties of Enriquez' kernels or be implied by connecting \hgMZV{}s defined as integrals along different cycles.

{
	\def\arraystretch{1.0}
	\begin{table}[h]
		\medskip\begin{center}
			\begin{tabular}{ |m{5.5cm}|m{9.6cm}|  }
				\hline
				\textbf{Identity} & \textbf{Formula/Example}\\
				\hline
				\hline
				\hyperlink{para:dep1}{Depth-one formula}~\cite{EnriquezHigher,EZ1} & {
					\setlength{\abovedisplayskip}{1pt}
					\setlength{\belowdisplayskip}{1pt}
					\begin{equation*}
						\hgzeta{k}(i_1\cdots i_r j)=
						\begin{cases}
							-\delta_{jki_1\cdots i_r}\frac{2\zeta_r}{(-2\pi \iunit)^r}, & \text{for }r \text{ even} \\
							\delta_{r1}\delta_{j i_1}(1- \delta_{jk})\frac{1}{2}, & \text{for } r \text{ odd}
						\end{cases}
					\end{equation*}
				}\\
				\hline
				\hyperlink{para:Fay}{Fay identities} & {
					\setlength{\abovedisplayskip}{1pt}
					\setlength{\belowdisplayskip}{1pt}
					\begin{equation*}0\texteq\zetaAg{1}{2}(21^2,1){+}\zetaAg{1}{2}(2^21,2){+}\zetaAg{1}{2}(2^2,21){+}\zetaAg{1}{2}(21,1^2)\end{equation*}
				} \\
				\hline
				\hyperlink{para:diffcyc}{\dct{}} (depth two)& {
					\setlength{\abovedisplayskip}{1pt}
					\setlength{\belowdisplayskip}{1pt}
					\begin{equation*}\zetaA{1}(12,23)\texteq\zetaA{2}(3,112)\,{+}\,\bn{1}\zetaA{2}(3,12)\phantom{\bigg|}\end{equation*}
				}\\
				\hline
				{\hyperlink{para:altid}{Alternating identity for hyperelliptic curves} (depth two)} & {
					\setlength{\abovedisplayskip}{1pt}
					\setlength{\belowdisplayskip}{1pt}
					\begin{equation*}\zetaAg{j}{2}(ji,ij)=\zetaAg{j}{2}(i,jij)\phantom{\bigg|}\end{equation*} 
				}\\
				\hline
				\hyperlink{para:weight}{Weight exchange}$^\ast$ (depth two)& {
					\setlength{\abovedisplayskip}{1pt}
					\setlength{\belowdisplayskip}{1pt}
					\begin{equation*}\hgzeta{j}(\mindx{i}_1, \mindx{i}_2) = -\hgzeta{j}(\mindx{i}_2, \mindx{i}_1)\, , \quad \text{for} \len{\mindx{i}_1} {+} \len{\mindx{i}_2} \in 2 \Integers {+} 1\phantom{\bigg|}\end{equation*}
				}\\
				\hline
				\hyperlink{para:emzvhemzv}{eMZV$\!\implies\!$\hgMZV{} identities}$^\ast$ & {
					\setlength{\abovedisplayskip}{1pt}
					\setlength{\belowdisplayskip}{1pt}\begin{equation*}\omell(0,5)=\omell(2,3)\ \implies\ \zetaA{j}(j,j^6)=\zetaA{j}(j^3,j^4)\end{equation*} 
				}\\
				\hline
			\end{tabular}
		\end{center}
		\caption{Known identities for higher-genus multiple zeta values, including a basic description of their range of validity as well as an example.}
		\label{tab:identities}
	\end{table}
}

\tabref{tab:identities} provides an overview of the relations that we found. The list is neither exhaustive nor claimed to be complete. The first four types of identities listed will be proven analytically. The last two relations, marked with an asterisk ($^\ast$), were tested numerically for a large number of examples, but will not be proven analytically in this article.
While the first formula was already addressed in \secref{sec:relations_depth_one}, we will show below the details of each of the remaining \hgMZV{} identities from the table, which can be accessed individually through the hyperlinks.

\subsection{Relations from Fay identities}\hypertarget{para:Fay}{}\label{sec:Fayid}

Generalizations of the genus-one Fay identities~\eqref{eqn:g1Fay} have been considered in \rcites{DHoker:2024ozn,Baune:2024ber}. One way to write these quadratic identities among Enriquez' kernels is given in \eqn{eqn:FaylikeId}.

Employing these Fay-like relations, in ref.~\cite[Sec.~7]{Baune:2024ber} an identity for \hgMPL{}s was derived to remove an occurrence of the argument $z$ in one of the labels. This is a higher-genus analogue to relations found in \rcite{Broedel:2014vla} for eMPLs.

Taking the integration contours of \hgMPL{}s in this type of identities to be the cycle $\acyc_j$, we find relations for \hgMZV{}s. One needs to be careful about possible poles when closing the contour: just as in the elliptic case, we will avoid endpoint divergences by not considering polylogs beginning with a kernel $\omega_{jj}$. For all other cases, we find non-trivial \hgMZV{} identities, which will for brevity be treated on a case-by-case basis below.
\begin{example}
	For genus two the above formalism yields, for example, the depth-two polylogarithm identity
	\begin{align}
		\Gargbare{(1),(211)}{-,z}{z}{z_0}&=\Gargbare{(211),(1)}{z_0,-}{z}{z_0}-\Gargbare{(2),(221)}{-,-}{z}{z_0}+\Gargbare{(22),(21)}{z_0,-}{z}{z_0}\notag\\
		&\quad-\Gargbare{(11),(21)}{z_0,-}{z}{z_0}+\Gargbare{(221),(2)}{-,-}{z}{z_0}+\Gargbare{(21),(11)}{-,z_0}{z}{z_0}\notag\\
		&\quad-\Gargbare{(21),(22)}{-,z_0}{z}{z_0}.
	\end{align}
	This implies --- after applying shuffle relations --- the \hgMZV{} relation (this equation also holds when integrating over the second \Atxt-cycle instead of the first, where one simply substitutes $\hgzeta{1}\,{\mapsto}\,\hgzeta{2}$ in this equation)
	\begin{align}
		0=\zetaAg{1}{2}(21^2,1)+\zetaAg{1}{2}(2^21,2)+\zetaAg{1}{2}(2^2,21)+\zetaAg{1}{2}(21,1^2).
	\end{align}
\end{example}
\begin{example}
	Another example of a hgMPL relation from the Fay-like identities reads (for any genus $\genus\,{\geq}\,3$)
	\begin{align}
		\Gargbare{(12),(33)}{-,z}{z}{z_0}&=\Gargbare{(12),(33)}{-,z_0}{z}{z_0}-\sum_{j=1}^\genus\Big[\Gargbare{(j)}{-}{z}{z_0}\Gargbare{(j12)}{-}{z}{z_0}\\
		&\hspace{30ex}+\Gargbare{(j),(1j2)}{-,z_0}{z}{z_0}+\Gargbare{(j2),(1j)}{z_0,z_0}{z}{z_0}\Big].\notag
	\end{align}
	This implies the \hgMZV{} relation (for any cycle $\acyc_i$ and any genus $\genus\,{\geq}\,2$)
	\begin{align}
		0&=\sum_{j=1}^\genus\left(\zetaAg{i}{\genus}(j,1j2)+\zetaAg{i}{\genus}(j2,1j)\right).
	\end{align}
	E.g.~for genus $\genus\texteq3$ this identity reduces to (for $i\texteq1$)
	\begin{align}
		0&=\zetaAg{1}{3}(1,112)+\zetaAg{1}{3}(2,122)+\zetaAg{1}{3}(3,132)
		+\zetaAg{1}{3}(12,11)+\zetaAg{1}{3}(22,12)+\zetaAg{1}{3}(32,13).
	\end{align}
\end{example}
%

\subsection{Cycle exchange}\hypertarget{para:diffcyc}{}\label{sec:relations_cyc_exchange}

A special type of identities relating \hgMZV{}s around different \Atxt-cycles can be proven using the kernel representation of \rcite{DHoker:2025dhv} (reviewed in \subsecref{sec:EnriquezDHS}). The statement can be phrased as follows. 
\begin{theorem}\label{thm:diffcyc}
	Let $r\geq0$ and $i,j_1,\ldots,j_r,k,l,m,n\,{\in}\,\{1,\ldots,\genus\}$ such that $k\,{\neq}\, l$, $i\,{\neq}\, m$ and $m\,{\neq}\, n$. Then, for any genus $\genus$, we have the \hgMZV{} identity
	\begin{equation}
	\begin{aligned}
		\hgzeta{i}(j_1\cdots j_rkl,mn)&=\hgzeta{m}(n,ij_1\cdots j_rkl)+\delta_{ij_1\cdots j_rk}\frac{\bn{r+1}}{(r+1)!}\hgzeta{m}(n,kl) \\
        &\quad+\sum_{p=1}^r\frac{\bn{p}}{p!}\delta_{ij_1\cdots j_p}\,\zetaA{m}(n,ij_{p+1}\cdots j_rkl).
	\end{aligned}
	\end{equation}
\end{theorem}
\begin{proof} 
	Let us expand the definition of the \hgMZV{} $\hgzeta{i}(j_1\cdots j_rkl,mn)$,
    \begin{equation}
        \hgzeta{i}(j_1\cdots j_rkl,mn)=\int_{t_1\in\acyc_i}\omega_{j_1\cdots j_rkl}(t_1,z_0)\int_{t_2=z_0}^{t_1}\omega_{mn}(t_2,z_0) \, ,
    \end{equation}
    where $z_0$ is an arbitrary point on the contour $\acyc_i$. Next, we use the representation of Enriquez' kernels in terms of the prime form~\eqref{eqn:omegaij2} in order to expand $\omega_{mn}(t_2,z_0)$. This yields
	\begin{equation}\label{eqn:diffcycstart}
		\begin{aligned}
			\hgzeta{i}(j_1\cdots j_rkl,mn)=\int_{t_1\in\acyc_i}\omega_{j_1\cdots j_rkl}(t_1,z_0)\int_{t_2=z_0}^{t_1}\left(\frac{1}{(-2\pi\iunit)}\int_{t\in\acyc_m}\omega_n(t)\,\Omega^{(z_1-t)}(t_2)\right) \, ,
		\end{aligned}
	\end{equation}
	where we have simultaneously used that $\omega_{mn}(t_2,z_0)\texteq\omega_{mn}(t_2,z_1)$ for $z_1$ located on the contour $\acyc_m$ and dropped the last term in \eqn{eqn:omegaij2}, both justified by the assumption $m\,{\neq}\, n$. Moreover, notice that we do not have to deform the integration contour $\acyc_i$ (as described in \secref{sec:EnriquezDHS}) since we require $i\,{\neq}\, m$. Hence the integrations are performed over disjoint cycles and no additional poles arise in this way. To proceed, we use the antisymmetry of the prime form yielding the identity
	\begin{equation}
		\int_{t_2=z_0}^{t_1}\Omega^{(z_1-t)}(t_2)=\log\left(\frac{E(t_1,z_1)E(z_0,t)}{E(t_1,t)E(z_0,z_1)}\right)=\log\left(\frac{E(z_1,t_1)E(t,z_0)}{E(t,t_1)E(z_1,z_0)}\right)=\int_{t_2=t}^{z_1}\Omega^{(t_1-z_0)}(t_2)
	\end{equation}
	for the normalized differential of the third kind (cf.~\eqn{eqn:fundamentalDiff3}),
	which is equivalent to the change of fibration basis formula from \rcite{DHoker:2024ozn}.
	Plugging this identity into \eqn{eqn:diffcycstart}, we arrive at 
	\begin{align}
		\hgzeta{i}(j_1\cdots j_rkl,mn)&=-\frac{1}{(-2\pi\iunit)}\int_{t\in\acyc_m}\omega_n(t)\int_{t_1\in\acyc_i}\omega_{j_1\cdots j_rkl}(t_1,z_0)\int_{t_2=z_1}^{t}\Omega^{(t_1-z_0)}(t_2) \, ,
	\end{align}
	where we used the assumption that $i\,{\neq}\, m$ to freely exchange the integrals and also swapped the direction of the innermost integration over $t_2$, yielding a minus sign in return. Next, notice that by~\cite[Thm.~1]{DHoker:2025dhv}, we have for any $p,q\,{\in}\,\{1,\ldots,\genus\}$
	\begin{align}
		(-2\pi\iunit)[\,\omega_{pq}(t_2,t_1)-\omega_{pq}(t_2,z_0)]&=\int_{t'\in\acyc_p}\omega_q(t')\left[\,\Omega^{(t_1-t')}(t_2)-\Omega^{(z_0-t')}(t_2)\right]\notag\\
		&=\int_{t'\in\acyc_p}\omega_q(t')\,\Omega^{(t_1-z_0)}(t_2)=\delta_{pq}\,\Omega^{(t_1-z_0)}(t_2) \, ,
	\end{align}
	where we have also used the definition~\eqref{eqn:fundamentalDiff3} of $\ndiff(x)$ in terms of prime forms. Using the above relation for $p\texteq q\texteq l$, we arrive at
	\begin{equation}\label{eqn:CycEx}
		\begin{aligned}
			\hgzeta{i}(j_1\cdots j_rkl,mn)=&-\int_{t\in\acyc_m}\omega_n(t)\int_{t_1\in\acyc_i}\omega_{j_1\cdots j_rkl}(t_1,z_0)\int_{t_2=z_1}^{t}[\,\omega_{ll}(t_2,t_1)-\omega_{ll}(t_2,z_0)] \\
			=&-\int_{t\in\acyc_m}\omega_n(t)\int_{t_2=z_1}^{t}\int_{t_1\in\acyc_i}\omega_{j_1\cdots j_rkl}(t_1,z_0)\,\omega_{ll}(t_2,t_1)+\hgzeta{i}(j_1\cdots j_rkl)\,\hgzeta{m}(n,ll) \\
            =&-\int_{t\in\acyc_m}\omega_n(t)\int_{t_2=z_1}^{t}\int_{t_1\in\acyc_i}\omega_{j_1\cdots j_rkl}(t_1,z_1)\,\omega_{ll}(t_2,t_1)\, ,
		\end{aligned}
	\end{equation}
	where we have again used $i\,{\neq}\, m$ to exchange the integrals. Furthermore, we dropped the last term due to \eqn{eqn:hgmzvdepth1} and the assumption that $k\,{\neq}\, l$ as well as swapped $z_0$ with $z_1$ in $\omega_{j_1\cdots j_rkl}(t_1,z_0)$ once more making use of $k\,{\neq}\, l$. To finish the argument, notice again that by~\cite[Thm.~2]{DHoker:2025dhv} (cf.~\eqn{eqn:DSthm2}), we have
	\begin{align}
		-\int_{t_1\in\acyc_i}\omega_{j_1\cdots j_rkl}(t_1,z_1)\,\omega_{ll}(t_2,t_1)&=\omega_{ij_1\cdots j_rkl}(t_2,z_1)+\sum_{\substack{p=1\\p\neq l}}^\genus\omega_{pl}(t_2)\underbrace{\int_{t\in\acyc_i}\omega_{j_1\cdots j_rkp}(t,z_1)}_{=\frac{\bn{r+1}}{(r+1)!}\delta_{ij_1\cdots j_rkp}}\\
		&\quad+\sum_{p=1}^r\frac{\bn{p}}{p!}\delta_{ij_1\cdots j_p}\,\omega_{ij_{p+1}\cdots j_rkl}(t_2,z_1)+\frac{\bn{r+2}}{(r+1)!}\delta_{ij_1\cdots j_rkl}\,\omega_l(t_2) \, ,\notag
	\end{align}
	where the last term vanishes due to the assumption $k\,{\neq}\,l$, and we used \eqn{{eqn:d1}} to evaluate the \Atxt-cycle integral on the right-hand side. Plugging this into \eqn{eqn:CycEx} finally yields
	\begin{equation}
		\begin{aligned}
			\hgzeta{i}(j_1\cdots j_rkl,mn)=\int_{t\in\acyc_m}\omega_n(t)\int_{t_2=z_1}^{t}\Big(&\omega_{ij_1\cdots j_rkl}(t_2,z_1)+\sum_{\substack{p=1\\p\neq l}}^\genus\omega_{pl}(t_2)\frac{\bn{r+1}}{(r+1)!}\delta_{ij_1\cdots j_rkp} \\
            &+\sum_{p=1}^r\frac{\bn{p}}{p!}\delta_{ij_1\cdots j_p}\,\omega_{ij_{p+1}\cdots j_rkl}(t_2,z_1)\Big)\\
			&\hspace*{-24ex}=\hgzeta{m}(n,ij_1\cdots j_rkl)+\delta_{ij_1\cdots j_rk}\frac{\bn{r+1}}{(r+1)!}\hgzeta{m}(n,kl) \\
			&+\sum_{p=1}^r\frac{\bn{p}}{p!}\delta_{ij_1\cdots j_p}\,\zetaA{m}(n,ij_{p+1}\cdots j_rkl)
		\end{aligned}
	\end{equation}
	by converting all (iterated) integrals back to $\acyc$-cycle \hgMZV{}s in the last line, thereby concluding the proof.
\end{proof}
Let us consider an instance of this identity as an example. We find (at arbitrary genus)
\begin{equation}
	\hgzeta{1}(12,21)=\hgzeta{2}(1,112)+\bn{1}\,\hgzeta{2}(1,12) \, .
\end{equation}
This is our first instance of an identity that mixes \hgMZV{}s arising from integrations around two \emph{different} \Atxt-cycles. 

Let us conclude with a few remarks on this identity. First, we assume that the current restrictions imposed on the indices can be (partially) lifted. However, this comes along several subtleties, for example the appearance of residue contributions from additional poles. Moreover, some assumptions have to be made in order to obtain an identity purely relating \hgMZV{}s. In general, it is expected that identities of this type contain additional contributions involving ordinary hgMPLs along paths connecting different $\acyc$-cycles. Secondly, we also expect that \thmref{thm:diffcyc} can be generalized to \hgMZV{}s of higher depth and, finally, one might also generalize the kernel in the second slot of $\hgzeta{i}(j_1\cdots j_rkl,mn)$ to higher weight by application of~\cite[Thm.~2]{DHoker:2025dhv} (c.f.~\eqn{eqn:DSthm2}). Once these subtleties are understood, it might be possible to find a more general statement. A detailed analysis will be devoted to future work.

\subsection{Alternating identity on hyperelliptic Riemann surfaces}\hypertarget{para:altid}{}\label{sec:altid}

Numerical experiments for genus-two and hyperelliptic surfaces of genus three suggested a special identity satisfied by \hgMZV{}s with alternating labels, the simplest instance reading
\begin{equation}
    \hgzeta{j}(ji,ij)=\hgzeta{j}(i,jij),
\end{equation}
where $i,j\,{\in}\,\{1,\ldots,\genus\}$, $i\,{\neq}\,j$ and $\genus\texteq2$ or the surface being hyperelliptic if $\genus\,{>}\,2$. As these \hgMZV{}s crucially rely on all moduli of the genus-two Riemann surface, it can be regarded as a genuine identity of \hgMZV{}s without an elliptic counterpart. In the following, we want to prove the validity of identities of this type on the Schottky uniformization  of hyperelliptic surfaces. To do this, we are going to rely on the techniques developed in~\secref{sec:TechniquesSchottky}.

Throughout this section, we restrict our considerations to the class of hyperelliptic Riemann surfaces described in \secref{sec:hyperelliptic} and \secref{sec:hyperellipticSchottky}. Let us start by introducing some useful notation. Let $\alt_{(i,j)}(n)$ be an alternating multi-index of length $n\,{+}\,1$ for some $i\,{\ne}\, j$ and $n \,{\ge}\, 0$, i.e.~$\alt_{(i,j)}(n) \texteq ijiji\cdots$ with $\alt_{(i,j)}(0) \,{:=}\, i$. Now, consider the two multi-indices  $\alt_{(i,j)}(n_1)$, $\alt_{(j,i)}(n_2)$. For simplicity, let us also assume that $n_1,n_2\,{>}\,0$ in order to avoid complications arising from additional singularities.
\begin{theorem}\label{thrm:altid}
    Let $i,j\,{\in}\,\{1,\ldots,\genus\}$ such that $i\,{\neq}\, j$ and let $n_1,n_2\,{>}\,0$ such that $n_1\,{+}\,n_2\,{\in}\,2\zN$. Then the identity
    \begin{equation}
        \label{eqn:altid}
        \normalfont \hgzeta{i}(\alt_{(i,j)}(n_1),\alt_{(j,i)}(n_2)) = \hgzeta{i}(\alt_{(j,i)}(n_1-1),\alt_{(i,j)}(n_2+1))
    \end{equation}
    holds on an arbitrary hyperelliptic surface of genus $\genus\,{\geq}\,2$.
\end{theorem}
\begin{proof}
    Let us denote $\mindx{i}_1\texteq\alt_{(i,j)}(n_1)$, $\mindx{i}_2\texteq\alt_{(j,i)}(n_2)$, $\mindx{i}_3\texteq\alt_{(j,i)}(n_1-1)$ and $\mindx{i}_4\texteq\alt_{(i,j)}(n_2+1)$. We start by writing out the definitions~\eqref{eqn:hgMZV} of the \hgMZV{}s in \eqn{eqn:altid}. This yields
    \begin{equation}
        \label{eqn:altidint}
        \int_{\acyc_i}\omega_{\mindx{i}_1}(t)\int_{t'=t_0}^t\omega_{\mindx{i}_2}(t')-\int_{\acyc_i}\omega_{\mindx{i}_3}(t)\int_{t'=t_0}^t\omega_{\mindx{i}_4}(t')=0 \, ,
    \end{equation}
    where we have deliberately ignored the second argument of the kernels as the last two indices of all multi-indices $\mindx{i}_k$ do not match due to their alternating nature.
    
    Before moving on, let us introduce some notation. Denote by $\mindx{i}_l^*$, $l\,{\in}\,\{1,\ldots,4\}$, the multi-index obtained from $\mindx{i}_l$ by removing the last index, i.e.~$\mindx{i}_1^* \texteq \alt_{(i,j)}(n_1\,{-}\, 1)$, etc. Also, define $e_l\,{\in}\,\{1,\ldots,\genus\}$ for $l\,{\in}\,\{1,2\}$ to be the last index\footnote{Notice that it is enough to define this for $l\texteq1,2$ as the last indices of $\mindx{i}_1$ and $\mindx{i}_3$ as well as $\mindx{i}_2$ and $\mindx{i}_4$ are equal by definition.} in $\mindx{i}_l$. With this notation, we can rewrite \eqn{eqn:altidint} by expressing Enriquez' kernels on the Schottky uniformization using \eqn{eqn:highergenuskernels}. Finally, we arrive at 
    \begin{equation}
        \label{eqn:altidschottky}
        \cL^{(i)}_{e_1,e_2}\equiv\sum_{\substack{\gamma\in \SCosetR{e_1}\\\delta\in \SCosetR{e_2}}}\frac{\Coeff(\gamma,\delta)}{(-2\pi\iunit)^2}\int_{\acyc_i}s_{e_1}^{(0)}(\gamma^{-1}t)\int_{t'=t_0}^ts_{e_2}^{(0)}(\delta^{-1}t')
    \end{equation}
    for the left-hand side of \eqn{eqn:altid}, where we have defined $\Coeff(\gamma,\delta)\texteq C(b_{\mindx{i}_1^*},\gamma)\,C(b_{\mindx{i}_2^*},\delta)-C(b_{\mindx{i}_3^*},\gamma)\,C(b_{\mindx{i}_4^*},\delta)$. Here we denote $b_\mindx{i} \texteq b_{\mindx{i}^{(1)}} b_{\mindx{i}^{(2)}} \cdots$, where $\mindx{i}^{(k)}$ is the $k$-th index of the multi-index $\mindx{i}$. The proof now proceeds in two parts. The first part is concerned with evaluating the iterated integral in \eqn{eqn:altidschottky}. To do this, we make use of the explicit representation~\eqref{eqn:schottky-holomorphic-basis} of the forms $s_{e_l}^{(0)}(t)$, $l\,{\in}\,\{1,2\}$. We obtain
    \begin{equation}
        \label{eqn:intschottky}
        \begin{aligned}
            \mathcal{I}_{e_1,e_2}(\gamma,\delta)\equiv&\int_{\acyc_i}s_{e_1}^{(0)}(\gamma^{-1}t)\int_{t'=t_0}^ts_{e_2}^{(0)}(\delta^{-1}t') \\
            =& \int_{\acyc_i}\dd t\left(\frac{1}{t-\gamma P_{e_1}'}-\frac{1}{t-\gamma P_{e_1}}\right)\log\left(\frac{(t-\delta P_{e_2}')(t_0-\delta P_{e_2})}{(t_0-\delta P_{e_2}')(t-\delta P_{e_2})}\right) \, ,
        \end{aligned}
    \end{equation}
    where we have used the invariance of the cross-ratio to move the action of $\delta$ from $z$ and $z_0$ to the fixed points $P_{e_2}'$ and $P_{e_2}$. In order to calculate the integral along $\acyc_i$, we employ the techniques demonstrated in \secref{sec:TechniquesSchottky} by applying of the residue theorem to this integral. We therefore have to analyze the location of the poles in the integrand and determine the different configurations contributing to the residue theorem. The treatment of the contributions entering the residue theorem requires to perform the same steps as described in \secref{sec:IntegrationSchottky} and essentially amounts to collecting all the positive intersections of the circles $\sigma_i^nC_i$, $n\,{\leq}\,0$ and $\sigma_i^nC_i'$, $n\,{\geq}\,0$ with the path associated to an element $\Upsilon\,{\in}\,\SCosetLR{e_2}{e_1}$, weighted by the associated residue contribution. While the details are spelled out in \appref{app:Technicalitiesaltid}, we merely state the result here:
    \begin{equation}
        \label{eqn:altidintermediatemain}
        \cL^{(i)}_{e_1,e_2}=\sum_{\Upsilon\in\SCosetLR{e_2}{e_1}}\left(\sum_{(\gamma,\delta)\in \nsplt{i}{\Upsilon}} \frac{\Coeff(\gamma,\delta)}{(-2\pi\iunit)^2}-\sum_{(\gamma,\delta)\in \psplt{i}{\Upsilon}}\frac{\Coeff(\gamma,\delta)}{(-2\pi\iunit)^2}\right)R^{-}_{e_1,e_2}(\Upsilon) \, ,
    \end{equation}
    where 
    \begin{equation}
        R^-_{e_1,e_2}(\Upsilon)=(-2\pi\iunit)\log\frac{(\Upsilon P_{e_1}'- P_{e_2}')(\Upsilon P_{e_1}- P_{e_2})}{(\Upsilon P_{e_1}- P_{e_2}')(\Upsilon P_{e_1}'- P_{e_2})} \, .
    \end{equation}
	Moreover, $\nsplt{i}{\Upsilon}\,,\,\psplt{i}{\Upsilon}\,\,{\subset}\, \SGroup\,{\times}\, \SGroup$ are defined in \secref{sec:splitting_of_Schottky_words}.
    
    The second part consists of investigating the coefficients of $R^-_{e_1,e_2}(\Upsilon)$ in \eqn{eqn:altidintermediatemain}. This amounts to carefully dissecting the coefficients $\Coeff(\gamma,\delta)$ for the splittings $\nsplt{i}{\Upsilon}$ and $\psplt{i}{\Upsilon}$ and applying the recursive definition~\eqref{eqn:Ccoeff} of the $C$-coefficients. The technical details are also spelled out in \appref{app:Technicalitiesaltid}. The derivation concludes with the formula
    \begin{equation}
        \label{eqn:coeffmain}
        \COeff(\Upsilon)\equiv\left(\sum_{(\gamma,\delta)\in \nsplt{i}{\Upsilon}} \Coeff(\gamma,\delta)-\sum_{(\gamma,\delta)\in \psplt{i}{\Upsilon}}\Coeff(\gamma,\delta)\right)=-\sum_{l\in L_i(\Upsilon)}m_l\, \Coeff(\gamma_l,\delta_l) \, ,
    \end{equation}
    where $L_i(\Upsilon)\texteq\{l\with j_l=i\}$ for $\Upsilon\texteq\sigma_{j_1}^{m_1}\cdots\sigma_{j_k}^{m_k}$, $k\,{\geq}\,0$ are the \textit{splitting locations}. Moreover, $\gamma_l,\delta_l$ represent a fixed splitting\footnote{This is well-defined as we have shown that $\Coeff(\gamma,\delta)$ only depends on the splitting location in \appref{app:Technicalitiesaltid}.} at the location $l\,{\in}\, L_i(\Upsilon)$. From \eqn{eqn:coeffmain}, it clearly follows that $\COeff(\Upsilon)$ admits a \textit{reflection symmetry}
    \begin{equation}
        \label{eqn:coeffsymm}
        \COeff(\tilde\Upsilon)=(-1)^{n+1}\COeff(\Upsilon) \, ,
    \end{equation}
    where $\tilde\Upsilon\texteq\sgen_{j_1}^{-m_1}\cdots\sgen_{j_k}^{-m_k}$ for $\Upsilon\texteq\sgen_{j_1}^{m_1}\cdots\sgen_{j_k}^{m_k}$ and $n$ is defined as the sum of the weights of the multi-indices $\mindx{i}_1$ and $\mindx{i}_2$ (or equivalently $\mindx{i}_3$ and $\mindx{i}_4$), i.e.~$n\texteq n_1\,{+}\,n_2\texteq n_3\,{+}\,n_4$. To see this, note first that the $C$-coefficients $C(b_\mindx{i},\Upsilon)$ for a fixed multi-index $\mindx{i}$ adhere to this symmetry, which can be proven inductively using the recursive definition~\eqref{eqn:Ccoeff}. Then \eqn{eqn:coeffsymm} directly follows from \eqn{eqn:coeffmain} and the definition of $\Coeff(\gamma,\delta)$ as a quadratic combination of $C$-coefficients. Since we assume the total weight of the multi-indices to be even, we can immediately see that $\COeff(\tilde\Upsilon)\texteq{-}\,\COeff(\Upsilon)$. Since we can partition the set $\SCosetLR{e_2}{e_1}$ into pairs $(\Upsilon,\tilde\Upsilon)$ and $\COeff(\id)\texteq0$, the proof of the identity finally boils down to showing that the quantity $R^-_{e_1,e_2}(\Upsilon)$ is invariant under $\Upsilon\,{\rightarrow}\,\tilde\Upsilon$. This however follows immediately from \lemref{lemma:altid}.
\end{proof}
%

\paragraph{Variations of the alternating identity.}

Notice that this identity can be combined with other identities to yield statements beyond what is covered by \thmref{thrm:altid}. As an example, for $n_1,n_2\,{>}\,0$ with $n_1\,{+}\,n_2\,{\in}\,2\zN$, consider the expression
\begin{equation}
    \label{eqn:altidshuffle}
    \hgzeta{j}(\alt_{(i,j)}(n_1),\alt_{(j,i)}(n_2)) \, ,
\end{equation}
which is not covered by the theorem above as we integrate over the cycle $\acyc_j$ and $i\,{\neq}\, j$ by assumption. However, using the shuffle relations~\eqref{eqn:shuffle} as well as the results~\eqref{eqn:hgmzvdepth1} at depth one, we can trace this back to \thmref{thrm:altid}. In particular, we get
\begin{equation}
    \begin{aligned}
        \hgzeta{j}(\alt_{(i,j)}(n_1),\alt_{(j,i)}(n_2))&=-\hgzeta{j}(\alt_{(j,i)}(n_2),\alt_{(i,j)}(n_1)) \\
        &=-\hgzeta{j}(\alt_{(i,j)}(n_2-1),\alt_{(j,i)}(n_1+1)) \\
        &=\hgzeta{j}(\alt_{(j,i)}(n_1+1),\alt_{(i,j)}(n_2-1)) \, ,
    \end{aligned}
\end{equation}
where we have used the shuffle identity in the first line, \thmref{thrm:altid} in the second line and again the shuffle relation in the last line. In other words, the shuffle relations allow to move indices in the opposite direction as well, given that the assumptions of the theorem still hold after the application of the shuffle relations.

\paragraph{Possible generalizations of the alternating identity.}

Numerical experiments suggest that the assumptions on the theorem can be weakened in certain situations. For example, it can be shown using the same method as for the general case above that
\begin{equation}
    \hgzeta{i}(\alt_{(i,j)}(n_1),\alt_{(j,i)}(0))=\hgzeta{i}(\alt_{(j,i)}(n_1-1),\alt_{(i,j)}(1)) \, , \quad i\neq j\, , \,\, n_1\in2\zN\setminus\{0\} \, ,
\end{equation}
still holds for $i,j\,{\in}\,\{1,\ldots,\genus\}$, even though the assumption $n_2\,{>}\,0$ is not satisfied anymore. However, it does not continue to hold if we exchange $i$ and $j$ in the second slot of the \hgMZV{}. A detailed analysis of this case reveals that it involves further contributions arising from additional pole configurations entering the residue theorem. There is also a similar phenomenon occurring for moving indices in the opposite direction, which is of course related to the former by the shuffle relations. Despite these subtleties, we suspect that stronger versions of \thmref{thrm:altid} can be proven along the same lines, employing the techniques provided by Schottky uniformization. However, since the involved calculations swiftly grow in complexity, we omit further investigations of these cases in this article.

While an analytical generalization remains to be found, we numerically discovered possible extensions, without finding a general statement for those generalizations yet. Such extensions boil down to \hgMZV{}s of the form $\zetaA{j}(\cdots ij,\cdots,ji)$, so that the considerations about the depth-two integral in \eqn{eqn:intschottky} of the proof of \thmref{thrm:altid} still hold, and where then the other indices of the \hgMZV{} need to fulfill a matching condition for their $C$-coefficients~\eqref{eqn:Ccoeff} in the Schottky expansion. This can then also include identities of odd total weight.

Numerically, we identified the following identities for genus two, which might constitute examples of generalizations of the alternating identity:
\begin{equation}	
\begin{aligned}
	\zetaA{1}(1^221,2)=\zetaA{1}(2&1,1^22),\qquad \zetaA{1}(12^21,2)=\zetaA{1}(12,2^21),\\
	\zetaA{1}(1212,2)&=\zetaA{1}(21^22,2)=\zetaA{1}(12,212)
	.
\end{aligned}
\end{equation}
%

\subsection{Weight exchange}\hypertarget{para:weight}{}\label{sec:weightex}

As will prove useful below, let us split a general \hgMZV{} into an ``elliptic part'' $\omell$ and 
a $\hgcorr$-term according to
\begin{equation}\label{eqn:hgmzvdecomp}
	\hgzeta{j}(\mindx{i}_1, \mindx{i}_2, \ldots, \mindx{i}_k)
	=
	\delta_{j \mindx{i}_1 \mindx{i}_2 \ldots \mindx{i}_k}\,(-2\pi\iunit)^{-\sum_j\!n_j}\,\omell_j(n_1, n_2, \ldots, n_k) 
	+ \hgcorr_j(\mindx{i}_1, \mindx{i}_2, \ldots, \mindx{i}_k),
\end{equation}
where $n_l \texteq |\mindx{i}_l|$ as defined below \eqn{eqn:hgMPL} and $\omell_j(\ldots)$ is an eMZV corresponding to the $j$-th subcover of the Schottky uniformization (i.e.~depending on $\tau_j$). The eMZV appears in the $k$-fold Poincar\'e series for \hgMZV{}s (cf.~\eqn{eqn:highergenuskernels}) as the term where all Schottky group elements are the identity. Such terms have a non-zero $C$-coefficient only when $\mindx{i}_l \texteq i_l^{n_l}$ for some $i_l \,{\in}\,\{1, \ldots, \genus\}$ for all $l$. However, in the case when $i_l\,{\ne}\, j$ for some $l$, we attribute the term to the $\hgcorr$-part. This decomposition not only explicitly identifies the genus-one contribution to the \hgMZV{}, but also makes the properties of the zeta values under permutations of the labels more transparent, as we will see below. 

Notice that this decomposition is also consistent with the separating degeneration~\eqref{eqn:sep_deg_ell_zeta}. If all multi-indices end with $j$-th index ($\mindx{i}_n \texteq \ldots j$ for all $n \,{\in}\, \{1,\ldots,k\}$), then the $\sigma$-terms vanish in the degeneration limit of the $j$-th direction.

The decomposition of \hgMZV{}s in \eqn{eqn:hgmzvdecomp} makes it possible to derive in many cases a reflection identity at depth two as we will show below. As we are not able to prove all cases, but can see the result for the remaining cases numerically, we state the following conjecture.
\begin{conjecture}\label{con:weightex}
	At depth two, \hgMZV{}s fulfill the reflection identity
	\begin{equation}
		\hgzeta{j}(\mindx{i}_1, \mindx{i}_2) =\begin{cases}\hgzeta{j}(\mindx{i}_2, \mindx{i}_1)=(-2\pi\iunit)^{-\len{\mindx{i}_1}-\len{\mindx{i}_2}}\,\omell_j(\len{\mindx{i}_1},\len{\mindx{i}_2}\mid\tau_j), & \text{if } \delta_{j\mindx{i}_1\mindx{i}_2}\texteq1\text{ and }\len{\mindx{i}_1}+\len{\mindx{i}_2}\in2\Integers,\\
		-\hgzeta{j}(\mindx{i}_2,\mindx{i}_1),& \text{otherwise.}\end{cases}
	\end{equation}
\end{conjecture}
\noindent This statement would be a first higher-genus generalization of the elliptic reflection identity~\cite[eq.~(2.13)]{Broedel:2015hia}
\begin{equation}\label{eqn:g1reflection}
	\omell_j(n_1, n_2, \ldots, n_k) = (-1)^{\sum_{l=1}^k n_l} \,\omell_j(n_k, n_{k-1}, \ldots, n_1).
\end{equation}

We argue for the above conjecture by inspecting the $\hgcorr$-term in \eqn{eqn:hgmzvdecomp} as follows:
The $\hgcorr$-term is a $k$-fold series composed of terms of the form (with possible prefactors)
\begin{equation}
	\int_{\acyc_j} \kappa_{\mindx{i}_1} \circ\cdots\circ \kappa_{\mindx{i}_k},
\end{equation}
where $\kappa_{\mindx{i}_l}$ is an Enriquez kernel $\omega_{\mindx{i}_l}$ with argument acted on by a non-trivial Schottky group element and/or with the last index $(\mindx{i}_l)_{n_l+1}\,{\ne}\, j$ for at least one value of $l \texteq 1, \ldots, k$. This ensures that for the $l$-th integrand
\begin{equation}
	\int_{\acyc_j} \kappa_{\mindx{i}_l} = 0
\end{equation}
holds. Therefore, due to the shuffle relations
\begin{equation}
	0=\int_{\acyc_j} \kappa_{\mindx{i}_1} \int_{\acyc_j} \kappa_{\mindx{i}_2} \cdots \int_{\acyc_j} \kappa_{\mindx{i}_k}
	=
	\sum_{\rho \in S_k} \int_{\acyc_j} \kappa_{\mindx{i}_{\rho(1)}} \circ\cdots\circ \kappa_{\mindx{i}_{\rho(k)}},
\end{equation}
and we can deduce the permutation property of the $\sigma$-terms
\begin{equation}\label{eqn:shufflecorrection}
	\sum_{\rho \in S_k} \hgcorr_j(\mindx{i}_{\rho(1)}, \mindx{i}_{\rho(2)},\ldots, \mindx{i}_{\rho(k)} ) = 0,
\end{equation}
for the permutation group $S_k$. 

At depth two, we encounter non-trivial $\hgcorr$-terms on top of the eMZVs $\omell$. The shuffle relation~\eqref{eqn:shufflecorrection} in this case gives us the following identity for the $\hgcorr$-terms from \eqn{eqn:hgmzvdecomp}:
\begin{equation}
	\hgcorr_j(\mindx{i}_1, \mindx{i}_2) = -\hgcorr_j(\mindx{i}_2, \mindx{i}_1),
\end{equation}
while for eMZVs the reflection identity~\eqref{eqn:g1reflection} holds.
From this we conclude that depth-two \hgMZV{}s $\zetaA{j}(\mindx{i}_1,\mindx{i}_2)$ with $\delta_{j\mindx{i}_1\mindx{i}_2}\,{\neq}\,1$ or odd total weight $|\mindx{i}_1| \,{+}\, |\mindx{i}_2| \,{\in}\, 2 \Integers \,{+}\, 1$ are antisymmetric under the exchange of kernels:
\begin{equation}
	\hgzeta{j}(\mindx{i}_1, \mindx{i}_2) = -\hgzeta{j}(\mindx{i}_2, \mindx{i}_1).
\end{equation}
This proves the first case stated in \conref{con:weightex}.
Note that this result can also be deduced from the shuffle relation for \hgMZV{}s and the depth-one result of \propref{prop:depth1}.

The remaining case, when all weights consist of $j$ labels (i.e.~$\delta_{j\mindx{i}_1\mindx{i}_2}\texteq1$) and the total weight is even, is more involved and could not be proven yet. Numerical investigations suggest that these \hgMZV{}s are symmetric under the exchange of the kernels as stated in the second case of \conref{con:weightex}. Therefore, the $\sigma$-part should vanish in this case
\begin{equation}
\hgcorr_j(j^{n_1}, j^{n_2}) = 0 \quad \text{for} \ n_1 + n_2 \in 2 \Integers,
\end{equation}
which would be a non-trivial statement. 
This would imply that the depth-two \hgMZV{}s of the form $\zetaA{j}(j^{n_1}, j^{n_2})$ for $n_1\,{+}\, n_2\,{\in}\, 2 \Integers$ reduce to eMZVs. 

\paragraph{Possible generalization.} 

Numerical tests actually indicate that \conref{con:weightex} might be generalizable to a \hgMZV{} reflection identity for arbitrary depth.
As we do not yet have an analytical handle on these higher-depth cases, showing such a statement is beyond the scope of this article.

\subsection{Higher-genus identities from elliptic identities}\hypertarget{para:emzvhemzv}{}\label{sec:emzvhemzv}

Curiously, despite the $\sigma$-term in \eqn{eqn:hgmzvdecomp}, \hgMZV{}s still seem to satisfy the known relations for the eMZVs: we lifted several known eMZV relations to \hgMZV{} relations in the sense of the conjecture below and tested them numerically. As all of those \hgMZV{} relations turned out to hold numerically, we formulate the following conjecture.
\begin{conjecture} Let $l,k_1,\ldots,k_l\,{\geq}\,1$ and $n_{1,1},\ldots,n_{k_1,1},\ldots,n_{1,l},\ldots,n_{k_l,l}\,{\geq}\,1$. Moreover, let \begin{equation}I(\omell(n_{1,1}-1,\ldots,n_{k_1,1}-1),\ldots,\omell(n_{1,\ell}-1,\ldots,n_{k_\ell,\ell}-1))\texteq0\end{equation} be an identity for eMZVs. Then 
\begin{equation}
	I\!\left(\zetaA{j}\left(j^{n_{1,1}},\ldots,j^{n_{k_1,1}}\right),\ldots,\zetaA{j}\left(j^{n_{1,\ell}},\ldots,j^{n_{k_\ell,\ell}}\right)\right)=0
\end{equation}
holds (with appropriate normalization of appearing genus-zero MZVs\footnote{In our conventions used here for the kernels, each appearance of $\zeta_k$ must include a factor of $(-2\pi\iunit)^{-k}$.}) for any $j\,{\in}\,\{1,\ldots,\genus\}$ and any genus $\genus\,{\geq}\,1$.
\end{conjecture}

\begin{example}
For several examples of eMZV relations from \rcites{Broedel:2014vla,Broedel:2015hia}, the above conjecture implies the following \hgMZV{} relations:
\begin{center}
\begin{tabular}{ c c c }
	\textbf{eMZV identity} & &\textbf{implied hgMZV identity} ($\forall j\,{\in}\,\{1,\ldots,\genus\}$)\\
	\hline
	\hline
	$\omell(0,5)=\omell(2,3)$&
	$\!\!\!\!\!\!\leadsto\!\!\!\!\!\!\phantom{\Bigg|}$ 
	& $\zetaA{j}(j,j^6)=\zetaA{j}(j^3,j^4)$\\
	
	\hline
	
	$\omell(1,2)=-2\,\zeta_2\,\omell(1,0)-\omell(0,3)$&
	$\!\!\!\!\!\!\leadsto\!\!\!\!\!\!\phantom{\Bigg|}$
	& $\zetaA{j}(j^2,j^3)=-\tfrac{2\,\zeta_2}{(-2\pi\iunit)^2}\zetaA{j}(j^2,j)-\zetaA{j}(j,j^4)$\\
	
	\hline
	\vspace{-1.5ex}
	\small	\hspace{-5ex}$\omell(0,6,2)=-\frac{21}{2}\zeta_8-14\,\zeta_6\,\omell(0,0,2)\phantom{\Bigg|}$& & \small\hspace{-6ex}$\zetaA{j}(j,j^7,j^3)=-\tfrac{21\,\zeta_8}{2(-2\pi\iunit)^8}-\tfrac{14\,\zeta_6}{(-2\pi\iunit)^6}\zetaA{j}(j,j,j^3)$\\
	\small\hspace{3.5ex}$-\frac25\omell(0,3,5)-6\,\zeta_4\,\omell(0,0,4)$ &$\!\!\!\!\!\!\leadsto\!\!\!\!\!\!\phantom{\bigg|}$ 
	&\small\hspace{10ex}$-\frac25\zetaA{j}(j,j^4,j^6)-\tfrac{6\,\zeta_4}{(-2\pi\iunit)^4}\zetaA{j}(j,j,j^5)$ \\
	\small\hspace{8ex}$-\frac92\omell(0,0,8)+2\,\omell(0,3)\,\omell(0,5)$& & \small\hspace{10ex}$-\frac92\zetaA{j}(j,j,j^9)+2\,\zetaA{j}(j,j^4)\,\zetaA{j}(j,j^6)$
\end{tabular}
\end{center}
We could verify these \hgMZV{} identities numerically. In fact, any known eMZV identity that we lifted to a \hgMZV{} identity in the sense of the above conjecture could be verified numerically at genus two and three.
\end{example}
While we cannot prove the above statement, we want to give some insights on why we believe this to be true. Let us recall that the known eMZV identities stem from three concepts:
\begin{enumerate}
	\setlength\itemsep{-0.2em}
	\item[(a)] Reflection identity~\eqref{eqn:g1reflection},
	\item[(b)] Shuffle relation~\eqref{eqn:shuffle},
	\item[(c)] Fay identities~\eqref{eqn:g1Fay}.
\end{enumerate}
Thus, the goal would be to show that these identities also hold for \hgMZV{}s with all indices being the same. We first note that the shuffle relation is automatically fulfilled by iterated integrals and is thus also true for \hgMZV{}s. For the reflection identity for eMZVs, a crucial property for the proof is the parity property $g^{(n)}(-z)\texteq(-1)^ng^{(n)}(z)$ for the elliptic kernels. While this symmetry does not seem to hold generally for the higher-genus kernels, we hope that the Schottky representation of Enriquez' kernels can help to prove a reflection identity, since it allows to write all higher-genus expressions in terms of elliptic kernels, for which we know the parity properties. A first step towards showing a reflection identity for \hgMZV{}s is shown in \secref{sec:weightex} above. For the Fay identity relations for \hgMZV{}s with all indices identical as in the conjecture above, we can use the procedures described in \secref{sec:Fayid}, where if all indices are equal to each other, the general higher-genus Fay-like identity~\eqref{eqn:FaylikeId} simplifies to read (implicit summation over $j\,{\in}\,\{1,\ldots,\genus\}$)
\begin{subequations}
	\begin{align}\label{eqn:Faylike}
		\omega_{i^{r+1}}(z,x)\,\omega_{i^{s+2}}(y,z)&=\omega_{i^{r+1}}(z,x)\,\omega_{i^{s+2}}(y,x)\notag\\
		&\quad-\sum_{\ell=0}^{r-1}\sum_{m=0}^s(-1)^{m-s}\binom{\ell+s-m}{\ell}\omega_{i^{\ell+s-m}ji^{r-\ell+1}}(z,x)\,\omega_{i^mj}(y,x)\notag\\
		&\quad-\sum_{m=0}^s(-1)^{m-s}\binom{r+s-m}{r}\omega_{i^{r+s-m}ji}(z,y)\,\omega_{i^mj}(y,x)\notag\\
		&\quad-\sum_{\ell=0}^r\sum_{m=0}^s(-1)^{m-s}\binom{\ell+s-m}{\ell}\omega_{i^{\ell+s-m}j}(z,y)\,\omega_{i^mji^{r-\ell+1}}(y,x).
	\end{align}
\end{subequations}
Lifting this identity to polylogarithms, one can obtain \hgMZV{} relations in the same way as it is done in \secref{sec:Fayid}.
We notice however, that the identity above introduces indices $j\,{\neq}\, i$. In order to explain that we find identities for \hgMZV{}s with all indices equal, it remains to be shown that the terms containing $j\,{\neq}\, i$ vanish or cancel each other. Unfortunately, showing this using the methods of the term-by-term residue analysis introduced in \secref{sec:TechniquesSchottky} and employed in \secref{sec:altid} for the alternating identities seems to be complicated: we were unable to identify any pattern in which the unwanted terms cancel each other. Numerically, the unwanted terms only vanish as the (truncated) series. Therefore, further analytical investigation is required.

\section{Open questions}\label{sec:openqs}

In this article, we have been defining general \hgMZV{}s, focussed on their \Atxt-cycle version and started the investigation of their properties. In particular, we explored the implications of known functional relations for hgMPLs on \hgMZV{}s. When evaluating \hgMZV{}s as iterated integrals on a \textit{closed path/cycle}, we observe a symmetry enhancement leading --- in conjunction with endpoint regularization --- to classes of additional relations. While we have been starting to investigate several classes of those relations, we have however not yet been able to draw the complete all-depth picture for some of them. 

The reason is of constructional nature: while functional relations for polylogarithms are available in a convenient Hopf-algebraic language governing the inner mechanics of those relations completely, the additional \hgMZV{} relations like the alternating identity or the weight exchange are not yet accessible in a nice algebraic formulation.

Accordingly, let us collect several open questions.
\begin{itemize}
	\item\textbf{Algebraic structures for \hgMZV{}s.} For genus-zero and genus-one MZVs algebraic structures have been identified, which could then be exploited leading to symbols and coactions: those tools allow to iteratively explore and access relations for polylogarithms and zeta values (see e.g.~\cite{Goncharov:2005sla,Brown:2011ik,Duhr:2012fh,Broedel:2018iwv}). 
	\\ 
	For \hgMZV{}s, such an algebraic structure is yet to be determined. Several of its properties, however, can be inferred already from the investigations in the current article:
\begin{itemize}[label=$\circ$]
	\item The algebra should account for relations between \hgMZV{}s defined along different cycles. A first example is the relation from \thmref{thm:diffcyc}. This could be extended to connect the algebras of \hgMZV{}s along different directions.\\ 
	Some intermediate results and mileage could be obtained from considering \hgMZV{}s, where path segments are allowed to be located on different cycles. In this situation, one would encounter a refinement of the shuffle algebra, where the common shuffling of path segments along the same cycle is extended by relations allowing shuffling of path segments along different cycles.
	\item The algebra is expected to be compatible with the higher modular structure inherent in the description of higher-genus Riemann surfaces.
\end{itemize}
Existence of and knowledge about an algebraic structure would allow derivation of a higher-genus version of the symbol/coproduct formalism right away.  
	
\item\textbf{Dimensions and data mines.} Once an algebraic structure has been identified, one could ask for an extension of the Broadhurst--Kreimer conjecture~\cite{BroadKrei}: what is the number of \hgMZV{}s of given weight and depth independent over $\zQ$? An extension of the two available data mines at genus zero \cite{Blumlein:2009cf} and genus one \cite{eMZVWebsite} would then be an obvious next step. Whether such a tool would find wide application in particle physics is yet to be determined. 

\item\textbf{Eisenstein formulation.} At genus one, many insights about eMZVs have been found by reformulating them in terms of iterated Eisenstein integrals, which is traced back to employing the mixed-heat equation to iteratively rewrite the former into the latter. At this point the derivation algebra considerations set in: it is not unreasonable to expect higher versions of derivation algebras to trigger relations between \hgMZV{}s similar to those at genus one.

\item\textbf{Complete \hgMZV{} relations.} The \hgMZV{} identities presented in this article are not a complete picture of the structure of \hgMZV{}s yet: there are for sure more relations to be explored. During our study of \hgMZV{}s we encountered further relations, which hold true numerically, e.g.~$\zetaAg{1}{2}(21^2,2)\texteq\zetaAg{1}{2}(2^3,2)$. Those relations would need to be explained and proven for the final goal of finding all relations among \hgMZV{}s.

\item\textbf{Physics applications.} It will be interesting to see how \hgMZV{}s appear in the context of physics and how the identities shown here provide new insights: for Feynman integral calculations, higher-genus curves have appeared in some examples~\cite{Marzucca:2023gto, Duhr:2024uid}, where one could hope to find \hgMZV{}s as coefficients in special limits. In string theory, integrations on higher-genus surfaces naturally arise for loop-amplitudes~\cite{DHoker:2002hof} and first applications of the hgMPLs have been discussed for example in \rcites{SnowmassString,DHoker:2023khh,DHoker:2025jgb}.
\end{itemize}


\vspace*{2cm}
\noindent{\LARGE\textbf{Appendix}}
\addcontentsline{toc}{section}{Appendix}
\appendix
\addtocontents{toc}{\protect\setcounter{tocdepth}{1}}

\appendixsection{The normalized differential of the third kind}\label{app:Ndiff}

The normalized differential of the third kind is a crucial object in Riemann surface theory as it can be used to generate all differential forms with at most simple poles on a given Riemann surface\footnote{In the sense that each such differential can be represented as a linear combination of the normalized differential of the third kind and the basis of normalized holomorphic differentials.}~\cite{mumford1984tata}. In this section, we want to derive its representation~\eqref{eqn:fundDiffschottky} on the Schottky uniformization. To do so, first recall the analytic properties of the normalized differential of the third kind defined in \eqn{eqn:fundamentalDiff3}. It is a single-valued differential one-form in $x$ and satisfies~\cite{fay}
\begin{subequations}
    \begin{align}
        \Res_{x=y}\left(\ndiff(x)\right)&=-\Res_{x=t}\left(\ndiff(x)\right)=1 \, , \\
        \int_{\acyc_i}\ndiff(x)&=0 \, , \quad \forall i\in\{1,\ldots,\genus\} \, .
    \end{align}
\end{subequations}
These conditions uniquely determine $\ndiff(x)$ for fixed $y,t\,{\in}\,\RSurf$. On the other hand, the expression
\begin{equation}
    \label{eqn:fundDiffschottkyApp}
    s^{(y-t)}(x)\coloneqq\sum_{\gamma\in\SGroup}\left(\frac{1}{\gamma x-y}-\frac{1}{\gamma x - t}\right)\dd(\gamma x)=\sum_{\gamma\in\SGroup}\left(\frac{1}{x - \gamma y}-\frac{1}{x -\gamma t}\right)\dd x
\end{equation}
for $y,t\,{\in}\,\funddom$ is a well-defined differential one-form on the Schottky cover $\Omega(\SGroup)$. It is invariant under the action of moving along the lift of a $\bcyc$-cycle to the covering space\footnote{Periodicity along $\acyc_k$, $k\,{\in}\,\{1,\ldots,\genus\}$ follows immediately from the fact that these cycles remain closed contours on the covering space.}. To see this, we can use that a cycle $\bcyc_k$ can be realized as $x\,{\rightarrow}\,\sigma_k x$ for all $k\,{\in}\,\{1,\ldots,\genus\}$ and $x$ in the fundamental domain\footnote{Recall that a point on the Riemann surface $\RSurf$ can be (uniquely) identified with a point on the fundamental domain of the associated Schottky group $\SGroup$.}. However, this action can be absorbed by redefining the summation in \eqn{eqn:fundDiffschottkyApp} by $\gamma\,{\rightarrow}\,\gamma\sigma_k$. Since this defines a bijection on $\SGroup$, the desired invariance follows. Hence $s^{(y-t)}(x)$ descends to a well-defined differential one-form on the Riemann surface $\RSurf\texteq\Omega(\SGroup)\slash\SGroup$. From \eqn{eqn:fundDiffschottkyApp}, it is also immediately clear that $s^{(y-t)}(x)$ has simple poles at $y$ and $t$ with residues $+1$ and $-1$, respectively, and is holomorphic elsewhere (regarded as a differential form on the Riemann surface and not on the whole covering space). Finally, to calculate its integral around $\acyc_k$ for $k\,{\in}\,\{1,\ldots,\genus\}$, we employ the residue theorem. Notice that, as an expression on the Riemann sphere $\ComplexComplete$, \eqn{eqn:fundDiffschottkyApp} has poles at $p_y(\gamma)\texteq\gamma y$ and $p_t(\gamma)\texteq\gamma t$ with residues $+1$ and $-1$, respectively. These poles are always both located on the same side of the contour $\acyc_k$. Therefore, by inverting the orientation of the contour if necessary, we can assume that the contour does not enclose the poles. Hence, by the residue theorem, the integral around it vanishes in every case. This implies that
\begin{equation}
    \int_{\acyc_k}s^{(y-t)}(x)=0\, , \quad \forall k\in\{1,\ldots,\genus\} \, .
\end{equation}
In conclusion, we have shown that $\ndiff(x)$ and $s^{(y-t)}(x)$ have the same analytic properties regarded as differential one-forms on the Riemann surface, hence they are equal on $\RSurf$ (or equivalently the fundamental domain $\funddom\,{\subset}\,\ComplexComplete$). Furthermore, comparison of analytic properties for e.g.~$t'\texteq t\,{-}\,\bcyc_i$ shifted outside of the fundamental domain $\funddom$ by the cycle $\bcyc_i$ shows that
\begin{equation}
    \Omega^{(y-t')}(x)=\ndiff(x)+2\pi\iunit\,\omega_i(x)=s^{(y-t')}(x) \, .
\end{equation}
Hence, we can lift the equality to the covering space, regarding $\ndiff(x)$ as a differential in $x$ for $x,y,t$ being variables on the universal covering of the Riemann surface $\RSurf$.

\appendixsection{Additional calculations}

\appendixsubsection{\texorpdfstring{Additional calculations for \subsecref{sec:IntegrationSchottky}}{Additional calculations for §5.1}}\label{app:TechnicalitiesIntegrationSchottky}

Collecting all the different contributions from the residue theorem given in~\sitref{item:1} -- \sitref{item:4} allows to rewrite \eqn{eqn:omegaij-schottky} as
\begin{equation}
    \begin{aligned}
    \label{eqn:omegaij-schottkyApp}
    \omega_{ij}(x,y)=&\frac{1}{(-2\pi\iunit)^2}\sum_{\gamma_1\in \SGroup^{[-]}_{i\rightarrow}}\sum_{\substack{\gamma_2\in \mathrm{C}(\SCosetR{j})^{[-]}_{i\rightarrow} \\ i=j:\gamma_2\neq\id}}\dd(\gamma_1x)(2\pi\iunit)\left(\frac{1}{\gamma_1x-\gamma_2 P_j'}-\frac{1}{\gamma_1x-\gamma_2 P_j} \right) \\
    &+\frac{1}{(-2\pi\iunit)^2}\sum_{\gamma_1\in \mathrm{C}\SGroup^{[-]}_{i\rightarrow}}\sum_{\gamma_2\in (\SCosetR{j})^{[-]}_{i\rightarrow}}\dd(\gamma_1x)(-2\pi\iunit)\left(\frac{1}{\gamma_1x-\gamma_2 P_j'}-\frac{1}{\gamma_1x-\gamma_2 P_j} \right) \\
    &+\delta_{ij}\left\{\frac{1}{2\pi\iunit}\left(\sum_{\gamma\in \SGroup^{[-]}_{i\rightarrow}}\frac{\dd(\gamma x)}{\gamma x-P_j'}+\sum_{\gamma\in \mathrm{C}\SGroup^{[-]}_{i\rightarrow}}\frac{\dd(\gamma x)}{\gamma x-P_j}-\sum_{\gamma\in \SGroup}\frac{\dd(\gamma x)}{\gamma x-y}\right)+\frac12\omega_j(x)\right\} \, ,
\end{aligned}
\end{equation}
where we have employed the notation~\eqref{eqn:set1}--\eqref{eqn:cset2} defined in \secref{sec:splitting_of_Schottky_words}. To see that this is in fact the case, consider for example the conditions on $\gamma_1$ and $\gamma_2$ stated in~\sitref{item:2}. These exactly translate to the second line in \eqn{eqn:omegaij-schottkyApp}. The other terms are obtained analogously. Next, we can shift the action of $\gamma_1$ from $x$ to $\gamma_2P_j'$ and $\gamma_2P_j$, respectively, by using M\"obius invariance and relabeling $\gamma_1\,{\rightarrow}\,\gamma_1^{-1}$. This results in 
\begin{equation}
\begin{aligned}
    \label{eqn:schottky1}
    \omega_{ij}(x,y)=&-\frac{1}{(-2\pi\iunit)}\sum_{\gamma_1\in \SGroup^{[+]}_{\leftarrow i}}\sum_{\substack{\gamma_2\in \mathrm{C}(\SCosetR{j})^{[-]}_{i\rightarrow} \\ i=j:\gamma_2\neq\id}}\dd x\left(\frac{1}{x-\gamma_1\gamma_2 P_j'}-\frac{1}{x-\gamma_1\gamma_2 P_j} \right) \\
    &+\frac{1}{(-2\pi\iunit)}\sum_{\gamma_1\in \mathrm{C}\SGroup^{[+]}_{\leftarrow i}}\sum_{\gamma_2\in (\SCosetR{j})^{[-]}_{i\rightarrow}}\dd x\left(\frac{1}{x-\gamma_1\gamma_2 P_j'}-\frac{1}{x-\gamma_1\gamma_2 P_j} \right) \\
    &+\delta_{ij}\left\{\frac{1}{2\pi\iunit}\left(\sum_{\gamma\in \SGroup^{[-]}_{i\rightarrow}}\frac{\dd(\gamma x)}{\gamma x-P_j'}+\sum_{\gamma\in \mathrm{C}\SGroup^{[-]}_{i\rightarrow}}\frac{\dd(\gamma x)}{\gamma x-P_j}-\sum_{\gamma\in \SGroup}\frac{\dd(\gamma x)}{\gamma x-y}\right)+\frac12\omega_j(x)\right\} \, ,
\end{aligned}
\end{equation}
where we have cancelled some factors of $2\pi\iunit$ and used that, by definition, $\SGroup^{[-]}_{i \rightarrow}\,{\rightarrow}\, \SGroup^{[+]}_{\leftarrow i}$ and $\mathrm{C}\SGroup^{[-]}_{i \rightarrow}\,{\rightarrow}\, \mathrm{C}\SGroup^{[+]}_{\leftarrow i}$ under inversion of elements. Next, we want to combine the sums over $\gamma_1$ and $\gamma_2$ into a single sum over $\Upsilon\,{\in}\, \SCosetR{j}$. However, in doing this, we overcount as a fixed $\Upsilon\,{\in}\, \SCosetR{j}$ can be realized as a product $\gamma_1\gamma_2$ for $\gamma_1$, $\gamma_2$ in the respective subsets in different ways. To illustrate this, let us consider the sets in the second line of \eqn{eqn:schottky1}, i.e.~$\gamma_1\,{\in}\, \mathrm{C}\SGroup^{[+]}_{\leftarrow i}$ and $\gamma_2\,{\in}\,(\SCosetR{j})^{[-]}_{i\rightarrow}$. Now, as an example, fix $\Upsilon\texteq\sigma_k\sigma_l\sigma_i^{-2}\sigma_m\,{\in}\,\SCosetR{j}$ for $k,l\,{\neq}\, i$ and $m\,{\neq}\, i,j$. Then we can split $\Upsilon$ into a product by either defining 
\begin{equation} \gamma_1=\sigma_k\sigma_l\sigma_i^{-1}\in\mathrm{C}\SGroup^{[+]}_{\leftarrow i}\, ,\quad \gamma_2=\sigma_i^{-1}\sigma_m\in(\SCosetR{j})^{[-]}_{i\rightarrow},
\end{equation} 
or 
\begin{equation}
\gamma_1=\sigma_k\sigma_l\in\mathrm{C}\SGroup^{[+]}_{\leftarrow i} \, , \quad \gamma_2=\sigma_i^{-2}\sigma_m\in(\SCosetR{j})^{[-]}_{i\rightarrow},
\end{equation}
respectively. These are both valid splittings of $\Upsilon$ in the sense of \defref{def:splitting}. The number of possible splittings can easily be related to the number of occurrences of the generator $\sigma_i^{-1}$ as each such $\sigma_i^{-1}$ defines a splitting in $\ssplit^-_i(\Upsilon)$ in the sense of \defref{def:splitting}. This can be easily shown to hold generally, and so we can write
\begin{equation}
    \sum_{\substack{\gamma_1\in \mathrm{C}\SGroup^{[+]}_{\leftarrow i}\\\gamma_2\in (\SCosetR{j})^{[-]}_{i\rightarrow}}}\dd x\left(\frac{1}{x-\gamma_1\gamma_2 P_j'}-\frac{1}{x-\gamma_1\gamma_2 P_j} \right)= \sum_{\Upsilon\in \SCosetR{j}}\dd x \,N_i^-(\Upsilon)\left(\frac{1}{x-\Upsilon P_j'}-\frac{1}{x- \Upsilon P_j} \right) \, ,
\end{equation}
where $N_i^-(\Upsilon)$ denotes the splitting number from \defref{def:splitting}. The same argument works for the first line in \eqn{eqn:schottky1} with the difference that we split after the generator $\sigma_i$ and thus overcount by $N_i^+(\Upsilon)$. In total, we obtain
\begin{equation}
\begin{aligned}
    \label{eqn:schottky2}
    \omega_{ij}(x,y)=&\frac{1}{(-2\pi\iunit)}\sum_{\Upsilon\in \SCosetR{j}}\dd x\,(N_i^-(\Upsilon)-N_i^+(\Upsilon))\left(\frac{1}{x-\Upsilon P_j'}-\frac{1}{x-\Upsilon P_j} \right) \\
    &+\delta_{ij}\left\{\frac{1}{2\pi\iunit}\left(\sum_{\gamma\in \SGroup^{[-]}_{i\rightarrow}}\frac{\dd(\gamma x)}{\gamma x-P_j'}+\sum_{\gamma\in \mathrm{C}\SGroup^{[-]}_{i\rightarrow}}\frac{\dd(\gamma x)}{\gamma x-P_j}-\sum_{\gamma\in \SGroup}\frac{\dd(\gamma x)}{\gamma x-y}\right)+\frac12\omega_j(x)\right\} \, .
\end{aligned}
\end{equation}
Since the coefficient in~\eqref{eqn:schottky2} quantifies the difference in occurrences of $\sigma_i$ and $\sigma_i^{-1}$, we can see that it exactly matches the definition of $-C(b_i,\Upsilon)$. So we finally get
\begin{equation}
    \begin{aligned}
    \label{eqn:schottky3}
    \omega_{ij}(x,y)=&\frac{1}{(-2\pi\iunit)}\sum_{\Upsilon\in \SCosetR{j}}C(b_i,\Upsilon)\,s_j^{(0)}(\Upsilon^{-1}x) \\
    &+\delta_{ij}\left\{\frac{1}{2\pi\iunit}\left(\sum_{\gamma\in \SGroup^{[-]}_{i \rightarrow}}\frac{\dd(\gamma x)}{\gamma x-P_j'}+\sum_{\gamma\in \mathrm{C}\SGroup^{[-]}_{i \rightarrow}}\frac{\dd(\gamma x)}{\gamma x-P_j}-\sum_{\gamma\in \SGroup}\frac{\dd(\gamma x)}{\gamma x-y}\right)+\frac12\omega_j(x)\right\} \, ,
    \end{aligned}
\end{equation}
where we have used the definition~\eqref{eqn:sn} of $s_j^{(0)}(x)$ as well. Let us now finally focus on the second term. We make use of the fact that $\SGroup\texteq\{\sigma_i^n\Upsilon\mid n\,{\in}\,\zZ, \Upsilon\in \SCosetL{i}\}$. This allows us to write
\begin{align}
	\omega_{ij}(x,y)=&\frac{1}{(-2\pi\iunit)}\sum_{\Upsilon\in \SCosetR{j}}C(b_i,\Upsilon)\,s_j^{(0)}(\Upsilon^{-1}x) \notag\\
    &+\delta_{ij}\frac{1}{2\pi\iunit}\sum_{\Upsilon\in \SCosetL{i}}\Bigg\{\left(\sum_{n<0}\frac{\dd(\sigma_i^n\Upsilon x)}{\sigma_i^n\Upsilon x-P_j'}+\sum_{n\geq0}\frac{\dd(\sigma_i^n\Upsilon x)}{\sigma_i^n\Upsilon x-P_j} -\sum_{n\in\zZ}\frac{\dd(\sigma_i^n\Upsilon x)}{\sigma_i^n\Upsilon x-y}\right)\notag \\
    &+\frac{\dd(\Upsilon x)}{2}\left(\frac{1}{\Upsilon x-P_j'}-\frac{1}{\Upsilon x-P_j}\right)\Bigg\} \, ,
\end{align}
where we have also used the definitions~\eqref{eqn:set1}-\eqref{eqn:cset2}. We can now split the sum over $\zZ$ into positive and negative parts to obtain
\begin{align}
    \label{eqn:omegaij-fullyexp}
    \omega_{ij}(x,y)=&\frac{1}{(-2\pi\iunit)}\sum_{\Upsilon\in \SCosetR{j}}C(b_i,\Upsilon)\,s_j^{(0)}(\Upsilon^{-1}x)\notag \\
    &+\delta_{ij}\frac{1}{2\pi\iunit}\sum_{\Upsilon\in \SCosetL{i}}\Bigg\{\sum_{n<0}\frac{(P_j'-y)\,\dd(\sigma_i^n\Upsilon x)}{(\sigma_i^n\Upsilon x-P_j')(\sigma_i^n\Upsilon x-y)}+\sum_{n>0}\frac{(P_j-y)\,\dd(\sigma_i^n\Upsilon x)}{(\sigma_i^n\Upsilon x-P_j)(\sigma_i^n\Upsilon x-y)}\notag \\ 
    &+\frac{\dd(\Upsilon x)}{2}\left(\frac{1}{\Upsilon x-P_j'}+\frac{1}{\Upsilon x-P_j}\right)-\frac{\dd(\Upsilon x)}{\Upsilon x-y}\Bigg\} \, .
\end{align}
On the other hand, we can explicitly compute $s_j^{(1)}(x,y)$ using eqs.~\eqref{eqn:sn}, \eqref{eqn:g1q}, \eqref{eq:abel-map-schottky} and \eqref{eqn:schottky-holomorphic-basis}. Let us compute this in several steps. For the first term in \eqn{eqn:g1q}, we can use the explicit form of Abel's map~\eqref{eq:abel-map-schottky} and the holomorphic differentials~\eqref{eqn:schottky-holomorphic-basis} to expand the cotangent. We obtain
\begin{equation}
    \label{eqn:qExpDiv}
    \pi\cot(\pi\exp(2\pi\iunit\,\abel_j(x,y)))\,\omega(x\mid\SGroup_j) = -\frac{\dd x}{2}\left(\frac{1}{x-P_j'}+\frac{1}{x-P_j}\right)+\frac{\dd x}{x-y}
\end{equation}
after various simplifications. For the second term
\begin{equation}\label{eqn:qExpcon}
    (-2\pi\iunit)\sum_{n,m>0}q_j^{nm}\big(\!\exp(2\pi\iunit\,\abel_j(x,y))^m-\exp(2\pi\iunit\,\abel_j(x,y))^{-m}\big)\,\omega(x\mid\SGroup_j) \, ,
\end{equation}
we can use the geometric series (convergence is ensured by $|q_j|\,{<}\,1$). Concretely, we can derive
\begin{align}
    (-2\pi\iunit)\!\sum_{n,m>0}\!q_j^{nm}\exp(2\pi\iunit\,\abel_j(x,y))^m\,\omega(x|\SGroup_j)&=(-2\pi\iunit)\!\sum_{n>0}\!\left(\frac{1}{1{-}q_j^n\exp(2\pi\iunit\,\abel_j(x,y))}{-}1\right)\omega(x|\SGroup_j)\notag \\
    &=(-2\pi\iunit)\!\sum_{n>0}\frac{(\sigma_j^nx-P_j')(y-P_j)}{(\sigma_j^nx-y)(P_j-P_j')}\omega(x\mid\SGroup_j) \notag\\
    &=\sum_{n>0}\frac{(y-P_j)\,\dd(\sigma_j^nx)}{(\sigma_j^nx-y)(\sigma_j^nx-P_j)} \, ,\label{eqn:qExpNondiv}
\end{align}
where we have used the definition of the holomorphic differential~\eqref{eqn:schottky-holomorphic-basis}, the invariance of the cross-ratio as well as the fact that
\begin{equation}
    q_j^n\exp(2\pi\iunit\,\abel_j(x,y))=q_j^n\frac{(x-P_j')(y-P_j)}{(x-P_j)(y-P_j')}=\frac{(\sigma_j^nx-P_j')(y-P_j)}{(\sigma_j^nx-P_j)(y-P_j')} 
\end{equation}
by definition of Abel's map~\eqref{eq:abel-map-schottky} and the fixed point equation~\eqref{eqn:fixedpteq}. Similarly, the treatment of the remaining term in \eqn{eqn:qExpcon} yields
\begin{equation}\label{eqn:qExpNondiv2}
    (-2\pi\iunit)\sum_{n,m>0}q_j^{nm}\exp(2\pi\iunit\,\abel_j(x,y))^m\,\omega(x\mid\SGroup_j)=\sum_{n<0}\frac{(y-P_j')\,\dd(\sigma_j^nx)}{(\sigma_j^nx-y)(\sigma_j^nx-P_j')} \, .
\end{equation}
Combining eqs.~\eqref{eqn:qExpDiv}, \eqref{eqn:qExpNondiv} and \eqref{eqn:qExpNondiv2}, we can see that \eqn{eqn:omegaij-fullyexp} can be rewritten as 
\begin{equation}
    \begin{aligned}
        \label{eqn:omegaij-finalexp}
        \omega_{ij}(x,y)=&\frac{1}{(-2\pi\iunit)}\sum_{\Upsilon\in \SCosetR{j}}C(b_i,\Upsilon)\,s_j^{(0)}(\Upsilon^{-1}x) +\delta_{ij}\frac{1}{(-2\pi\iunit)}\sum_{\Upsilon\in \SCosetL{i}}s_j^{(1)}(\Upsilon x,y) \\
        =&\frac{1}{(-2\pi\iunit)}\sum_{\Upsilon\in \SCosetR{j}}\left(C(b_i,\Upsilon)\,s_j^{(0)}(\Upsilon^{-1}x) +\delta_{ij}\,s_j^{(1)}(\Upsilon^{-1} x,y)\right) \, ,
    \end{aligned}
\end{equation}
where we have relabelled $\Upsilon\,{\rightarrow}\,\Upsilon^{-1}$ in the second term to arrive at the final expression.

\appendixsubsection{\texorpdfstring{Additional calculations for \subsecref{sec:thm2}}{Additional calculations for §5.2}}\label{app:33}

In \rcite[Thm.~3]{DHoker:2025dhv}, a convolution identity among Enriquez' kernels of different weights was proven:
\begin{align}\label{eqn:DSthm3}
    \int_{t \in \acyc_j} \omega_{p_1 \cdots p_s ab}(y, t)\, \omega_{i_1 \cdots i_rk}(t, z)&=\frac{\bn{r}}{r!} \delta_{i_1 \cdots i_r jk} \,\omega_{p_1 \cdots p_s ab}(y, z)\notag \\
    &\quad+\delta_{ab} \sum_{\ell=0}^s\bigg\{\rn{r, s+1-\ell} \,\delta_{p_s \cdots p_{\ell+1} i_1 \cdots i_r jk}\, \omega_{p_1 \cdots p_{\ell}j}(y, z)\\
    &\hspace{10ex}-\sum_{k=0}^r \rn{k, s-\ell} \,\delta_{p_s \cdots p_{\ell+1} i_1 \cdots i_kj}\, \omega_{p_1 \cdots p_{\ell} j i_{k+1} \cdots i_rk}(y, z)\bigg\},\notag
\end{align}
with the coefficient
\begin{equation}
	\rn{m,n}=(-1)^n\frac{\bn{m+n}}{m!\,n!},
\end{equation}
where $\bn{r}$ are the Bernoulli numbers from \eqn{eqn:Bernoulli}.

In this appendix, we focus on the special case where $s \texteq 0$ and $a \texteq b$, and show how the corresponding identity can be derived using the Schottky formalism. To facilitate the derivation, we divide the discussion into two parts: the cases $j \,{\neq}\, k$ and $j \texteq k$. For each, we provide illustrative examples that introduce the necessary concepts and techniques required to obtain the final result.

\subsubsection{\texorpdfstring{$j\,{\neq}\, k$: the combinatorial problem}{j ≠ k: the combinatorial problem}}\label{app:combinatoricProblem}

In this subsection we derive \eqn{eqn:DSthm3} in the case when $j \,{\neq}\, k$, i.e.
\begin{align}\label{eqn:thm3-general-result-j-neq-k}
	\int_{t \in \acyc_j} \omega_{aa}(y, t)\, \omega_{i_1 \cdots i_rk}(t, z) = -\sum_{l=0}^r \frac{\bn{l}}{l!} \delta_{i_1 \cdots i_lj}\, \omega_{j i_{l+1} \cdots i_rk}(y, z),
\end{align}
by reducing it to a combinatorial problem for the $C$-coefficients \eqref{eqn:Ccoeff}. We only consider the most non-trivial case of $i_1 \texteq i_2 \texteq \ldots \texteq i_r \texteq j$, since the other cases can be seen as trivial extensions of the former one. The derivation depends on the following equality for a convolution of the Enriquez' kernels
\begin{align}\label{eqn:recursiveResult}
    \sum_{l = 1}^{r} \frac{1}{l!} \int_{t \in \acyc_j} \omega_{aa}(y, t)\, \omega_{j^{r-l}k}(t, z) 
    = -\omega_{j^{r}k}(y, z),
\end{align}
where we used the multi-index notation introduced in \secref{sec:hgMZV}. We prove it by reducing it to an identity among $C$-coefficients. 

First of all, notice that each integral in the left-hand side has the same pole structure as the example considered in \secref{sec:thm2}. Therefore, we can apply the same residue analysis at the level of the Schottky expansions to \eqn{eqn:recursiveResult}, obtaining 
\begin{align}
\frac{1}{(-2\pi\iunit)}\sum_{\Upsilon \in \SCosetR{k}}\sum_{l=1}^r \frac{1}{l!}
&\left[
\sum_{(\ldots, \gamma^+) \in \psplt{j}{\Upsilon}} C(b_j^{r-l}, \gamma^+) 
-
\sum_{(\ldots, \gamma^-) \in \nsplt{j}{\Upsilon}} C(b_j^{r-l}, \gamma^-) 
\right] \skern{0}_k(\Upsilon^{-1} y) \nonumber
\\
&=
\frac{1}{(-2\pi\iunit)}\sum_{\Upsilon \in \SCosetR{k}}C(b_j^{r}, \Upsilon)\,  \skern{0}_k(\Upsilon^{-1} y).
\label{eqn:recursiveResult_coefs}
\end{align}
In the following we show that this equality holds term-by-term for each $\Upsilon \,{\in}\, \SCosetR{k}$. Using \lemref{lem:Ccoeff} we can rewrite the sums over $l$ in the left-hand side as 
\begin{equation}
\sum_{l=1}^r \frac{1}{l!} C(b_j^{r-l}, \gamma^\pm) = 
C(b_j^{r}, \sigma_j \gamma^\pm) - C(b_j^{r}, \gamma^\pm).
\end{equation}
Thus, the left-hand side of \eqn{eqn:recursiveResult_coefs} for a fixed $\Upsilon$ reduces to 
\begin{equation} \label{eqn:ccoef_lhs}
\sum_{(\ldots, \gamma^+) \in \psplt{j}{\Upsilon}} [C(b_j^{r}, \sigma_j\gamma^+) - 
C(b_j^{r}, \gamma^+) ]
-
\sum_{(\ldots, \gamma^-) \in \nsplt{j}{\Upsilon}} [C(b_j^{r}, \sigma_j\gamma^-) - 
C(b_j^{r}, \gamma^-) ].
\end{equation}
We need to show that \eqn{eqn:ccoef_lhs} equals $C(b_j^r, \Upsilon)$ for all $\Upsilon\,{\in}\, \SCosetR{k}$.
Generally, $\Upsilon$ can be put to the form 
\begin{equation}
\Upsilon = \alpha_0 \sigma_j^{n_1} \alpha_1 \sigma_j^{n_2} \alpha_2 \ldots \sigma_j^{n_m} \alpha_m,
\end{equation}
where $n_i\,{\in}\, \Integers \setminus \{0\}$ for $i\,{\in}\, \{1, \ldots, m\}$, $\alpha_0$, $\alpha_i$ and $\alpha_m$ are the Schottky words that do not contain $\sigma_j$ with $\alpha_i\,{\ne}\,\id$ for $i\,{\in}\, \{1, \ldots, m-1\}$ for some $m\,{\ge}\, 0$. Notice that in the summation over $\gamma^+$ ($\gamma^-$) we essentially sum over all suffixes of $\Upsilon$ taken after (before) each $\sigma_j$ ($\sigma_j^{-1}$). Therefore, we can rewrite the sums in \eqn{eqn:ccoef_lhs} as
\begin{equation}
\sum_{s=1}^m \sum_{l=1}^{n_s} \sgn(n_s) [C(b_j^{r}, \gamma_{s,l-1}) - 
C(b_j^{r}, \gamma_{s,l})],
\end{equation}
where
\begin{equation}
\gamma_{s,l} = \sigma_j^{n_s -l} \alpha_s\, \sigma_j^{n_{s+1}} \ldots \alpha_m,
\end{equation}
and we defined the sum with a negative upper limit $n\,{<}\,0$ according to 
$\sum_{l=1}^{n} \texteq \sum_{l = n+1}^{0}$. In this form it is obvious that for a fixed $s$ the inner sum is telescoping and we obtain 
\begin{equation}
\sum_{s=1}^m [C(b_j^{r}, \gamma_{s,0}) - C(b_j^{r}, \gamma_{s,n_s})].
\end{equation}
We argue that the residual sum over $s$ is also telescoping. In order to see this, we notice that $\gamma_{s, n_s} \texteq \alpha_s \gamma_{s+1, 0}$, where $\alpha_s$ does not contain $\sigma_j$ by definition. Thus, according to the recursive property of the $C$-coefficients \eqref{eqn:Ccoeff}, $C(b_j^r, \gamma_{s, n_s}) \texteq C(b_j^r, \gamma_{s+1, 0})$. Therefore, the whole sum reduces to 
\begin{equation}
    C(b_j^r, \gamma_{1,0}) -  C(b_j^r, \gamma_{m,n_m}) = C(b_j^r, \alpha_0^{-1} \Upsilon ) -  C(b_j^r, \alpha_m) = C(b_j^r,  \Upsilon ),
\end{equation}
where in the last equality we again used the recursive property of the $C$-coefficients. This concludes the proof of \eqn{eqn:recursiveResult}.

Now, to derive \eqn{eqn:thm3-general-result-j-neq-k} for $i_1 \texteq \ldots \texteq i_r \texteq j$ we use induction. 
We have already explicitly proven \eqn{eqn:thm3-general-result-j-neq-k} for $r\texteq 0$ in \secref{sec:thm2}, and now assume it holds for $r\texteq n$.
Using \eqn{eqn:recursiveResult} with $r\texteq n\,{+}\,1$, we find 
\begin{align}
   - \omega_{j^{n+1}k}(y, z)
   &=\sum_{l=1}^{n+1}\frac{1}{l!}\int_{t \in \Acycle_j}\omega_{aa}(y,t)\,\omega_{j^{n-l+1}k}(t, z)\nonumber\\
    &=\int_{t \in \Acycle_j}\omega_{aa}(y,t)\,\omega_{j^{n+1}k}(t, z)
    +\sum_{l=2}^{n+1}\frac{1}{l!}\left(-\sum_{p=0}^{n-l+1}\frac{\bn{p}}{p!}\,\omega_{j^{n-l-p+2}k}(y, z)\right).
\end{align}
Rearranging and changing the summation variable to $m\texteq p\,{+}\,l\,{-}\,1$ yields
\begin{align}
    \int_{t \in \Acycle_j}\omega_{aa}(y,t)\,\omega_{j^{n+1}k}(t, z)
    &= -\omega_{j^{n+1}k}(y, z)
    + \sum_{l=2}^{n+1}\frac{1}{l!}\sum_{m=l-1}^{n}\frac{\bn{m-l+1}}{(m-l+1)!}\, 
    \omega_{j^{n+1-m}k}(y, z)\nonumber\\
    &=-\omega_{j^{n+1}k}(y, z)
    + \sum_{m=1}^{n}\sum_{l=2}^{m+1}\frac{1}{l!}\frac{\bn{m-l+1}}{(m-l+1)!}\,
    \omega_{j^{n+1-m}k}(y, z)
\end{align}
where in the second line we swapped the order of summation.  Next, we relabel $q\texteq m\,{-}\,l\,{+}\,1$ and use the property of Bernoulli numbers (for $m\,{>}\,0$)
\begin{align}
    \frac{\bn{m}}{m!}=-\sum_{q=0}^{m-1}\frac{1}{(m-q+1)!}\frac{\bn{q}}{q!},
\end{align}
to find
\begin{align}
     \int_{t \in \Acycle_j}\omega_{aa}(y,t)\,\omega_{j^{n+1}k}(t, z)
     &= -\omega_{j^{n+1}k}(y, z)
     + \sum_{m=1}^{n}\sum_{q=0}^{m-1}\frac{1}{(m-q+1)!}\frac{\bn{q}}{q!}
     \,\omega_{j^{n+1-m}k}(y, z)\nonumber\\
     &=-\omega_{j^{n+1}k}(y, z)
     -\sum_{m=1}^{n}\frac{\bn{m}}{m!}
     \,\omega_{j^{n+1-m}k}(y, z)\nonumber\\
     &=-\sum_{m=0}^{n}\frac{\bn{m}}{m!}
     \,\omega_{j^{n+1-m}k}(y, z),
\end{align}
which completes the proof of the general case when $j \,{\neq}\, k$.

\subsubsection{\texorpdfstring{$j\texteq k$: genus-one recursion and resummation}{j = k: genus-one recursion and resummation}}\label{app:L=Kproof}

In this subsection, we conclude the general case of the theorem from \rcite{DHoker:2025dhv}, as stated in \eqn{eqn:DSthm3}, by addressing the final case where $j \texteq k$ and $a \texteq b$. In this case, all terms in \eqn{eqn:DSthm3} contribute, and the full expression takes the form:
\begin{align}
    \int_{t \in \acyc_j} \omega_{jj}(y, t)\, \omega_{i_1 \cdots i_r j}(t, z) 
    =&\, \frac{\bn{r}}{r!} \delta_{i_1 \cdots i_rj} \omega_{aa }(y, z)
    -  \frac{\bn{r+1}}{r!} \delta_{i_1 \cdots i_rj}\, \omega_j(y) \nonumber\\
    &- \sum_{l = 0}^r \frac{\bn{l}}{l!} \delta_{i_1 \cdots i_lj} \,\omega_{j i_{l+1} \cdots i_r j}(y, z).
\end{align}
For more general cases, the analysis is similar. We again use the freedom to choose $a$ and fix it to be $a \texteq j$ for convenience. Eq.~\eqref{eqn:DSthm3} becomes
\begin{align}
    \int_{t \in \acyc_j} \omega_{jj}(y, t)\, \omega_{i_1 \cdots i_r j}(t, z) 
    =&\, \frac{\bn{r}}{r!} \delta_{i_1 \cdots i_rj} \omega_{jj}(y, z)
    -  \frac{\bn{r+1}}{r!} \delta_{i_1 \cdots i_rj}\, \omega_j(y) \nonumber\\
    &- \sum_{l = 0}^r \frac{\bn{l}}{l!} \delta_{i_1 \cdots i_lj} \,\omega_{j i_{l+1} \cdots i_r j}(y, z).
\end{align}
Note that in the final summation term, when $l \texteq r$, the expression 
\begin{equation}
\frac{\bn{r}}{r!} \,\delta_{i_1 \cdots i_rj} \,\omega_{jj}(y, z)
\end{equation}
matches the first term, so these two terms always cancel.

Additionally, we now consider the case where $i_1 \texteq \cdots \texteq i_r \texteq j$, which represents the situation with maximal complexity: every Kronekcer delta in \eqn{eqn:DSthm3} is non-zero.
Hence, we aim to prove
\begin{align}
     \int_{t \in \acyc_j} \omega_{jj}(y, t)\, \omega_{j^{n+1}}(t, z) 
    =& -  \frac{\bn{n+1}}{n!} \, \omega_j(y) - \sum_{l = 0}^{n-1} \frac{\bn{l}}{l!}  \,\omega_{j^{n-l+2}}(y, z).
\end{align}
To begin, we compute the left-hand side explicitly:
\begin{align}
    \int_{t \in \acyc_j} \omega_{jj}(y, t)\, \omega_{j^{n+1}}(t, z)&\sum_{\bar{\gamma}_j\in\SCosetR{j}}\frac{1}{(-2\pi \iunit)} \int_{t \in \Acycle_j}C(b_j,\bar{\gamma}_j)\,\skern{0}_j(\bar{\gamma}_j^{-1}y,t)\,\omega_{j^{n+1}}(t, z)\nonumber\\
    &+\sum_{\bar{\gamma}_j\in\SCosetR{j}}\frac{1}{(-2\pi \iunit)} \int_{t \in \Acycle_j}\skern{1}_j(\bar{\gamma}_j^{-1}y,t)\,\omega_{j^{n+1}}(t, z)\nonumber\\
    =&\sum_{\bar{\gamma}_j\in\SCosetR{j}}\frac{1}{(-2\pi \iunit)} C(b_j,\bar{\gamma}_j)\,\skern{0}_j(\bar{\gamma}_j^{-1}y,z)\frac{\bn{n}}{n!}\\
    &+\frac{1}{(-2\pi \iunit)^2}\sum_{\substack{\bar{\gamma}_j\in\SCosetR{j}\\\gamma_j\in\SCosetR{j}}}\sum_{k=0}^n C(b_j^{n-k},\gamma_j)\int_{t \in \Acycle_j}\skern{1}_j(\bar{\gamma}_j^{-1}y,t)\,\skern{k}_j(\gamma_j^{-1}t,z)\notag
\end{align}
In the last step, we used the depth-one result to evaluate the first integral. For the second term, the contribution from $\gamma_j \texteq \id$ can be evaluated separately using the genus-one recursion relation for the $s$-kernels as follows.

According to~\rcite{DHoker:2025dhv} we have the genus-one kernel convolution identity for $r\,{\ge}\, 1$
\begin{align}\label{eqn:genus-onekernelrecursiondefinition}
    \gkern{r+1}&(\xi - \chi) \\
    &= -\int_0^1 \dd \nu \, \gkern{1}(\xi - \nu)\, \gkern{r}(\nu - \chi) 
    - (-2\pi\iunit)^{r+1} \frac{\bn{r+1}}{r!} 
    - \sum_{k = 1}^{r - 1} (-2\pi\iunit)^k \frac{\bn{k}}{k!} \gkern{r+1-k}(\xi - \chi).\notag
\end{align}
The case $r \texteq 0$ is treated as a special case, since we have
\begin{align}
    \int_0^1 \dd \nu \, \gkern{1}(\xi - \nu)\,\gkern{0}(\nu - \chi)
    = \int_0^1 \dd \nu \, \gkern{1}(\xi - \nu)
    = -\pi \iunit
    = \frac{1}{2} (-2\pi \iunit)\,\gkern{0}(\xi - \chi),
\end{align}
where we used the property $\gkern{0}(\xi) \texteq 1$.
Using the relation in \eqn{eqn:sn} the $s$-kernels can be written as
\begin{align}
    \skern{n}_j(z, x) 
    = (-2\pi\iunit)^{1 - n}\, \gkern{n}_j(\abel_j(z,x_0) - \abel_j(x,x_0))\, \dd\abel_j(z,x_0),
\end{align}
where $x_0$ is a basepoint on the fundamental domain of the Schottky cover.

Collecting these results, the genus-one recursion relation for the $s$-kernels in the case $n \,{\ge}\, 0$ takes the form
\begin{align}\label{eqn:genusOneRescursionSKernel}
    \int_{t \in \acyc_j} \skern{1}_j(y, t)\, \skern{n}_j(t, z) 
    = (-2\pi\iunit) \left( -\sum_{k=0}^{n-1}\frac{\bn{k}}{k!} \skern{n+1-k}_j(y,z) -\frac{\bn{n+1}}{n!}\skern{0}_j(y,z)
    \right).
\end{align}
Hence, we can isolate this case for subsequent application:
\begin{align}
    \int_{t \in \acyc_j} \omega_{jj}(y, t)\, \omega_{j^{n+1}}&(t, z) \nonumber\\
    =&\sum_{\bar{\gamma}_j\in\SCosetR{j}}\frac{1}{(-2\pi \iunit)} C(b_j,\bar{\gamma}_j)\,\skern{0}_j(\bar{\gamma}_j^{-1}y,z)\frac{\bn{n}}{n!}\nonumber\\
    &+\frac{1}{(-2\pi \iunit)^2}\sum_{\bar{\gamma}_j\in\SCosetR{j}}\sum_{\substack{\gamma_j\in\SCosetR{j}\\\gamma_j\neq \id}}\sum_{k=0}^n C(b_j^{n-k},\gamma_j)\int_{t \in \Acycle_j}\skern{1}_j(\bar{\gamma}_j^{-1}y,t)\,\skern{k}_j(\gamma_j^{-1}t,z)\nonumber\\
    &+\frac{1}{(-2\pi \iunit)^2}\sum_{\bar{\gamma}_j\in\SCosetR{j}}\int_{t \in \Acycle_j}\skern{1}_j(\bar{\gamma}_j^{-1}y,t)\,\skern{n}_j(\gamma_j^{-1}t,z).
\end{align}
To evaluate the double Schottky sum in the second line, we apply the residue analysis established in \secref{sec:TechniquesSchottky}, noting that only the pole structure differs. As an illustrative example, we consider the case $k\texteq0$.

\paragraph{Residue analysis.}

Now we focus on the integrand
\begin{align}
	\mathcal{I}_j 
	=\skern{1}_j(\bar{\gamma}_j^{-1} y,  t)\, \skern{0}_j(\gamma_j^{-1}t),
\end{align}
where we again apply the residue analysis, with poles at
\begin{align}
    &p_{P}(\gamma_j) = \gamma_j P_j,\qquad 
    p_{P'}(\gamma_j) = \gamma_j P'_j,
    &&\quad \text{from } \skern{0}_j(\gamma_j^{-1} t) ,
    \nonumber\\
    &p_y^n(\bar{\gamma}_j) = \sigma_j^n \bar{\gamma}_j^{-1} y,
    &&\quad \text{from } \skern{1}_j(\bar{\gamma}_j^{-1}y,t) .
\end{align}
\begin{situations}
    \item\label{appitem:1} All simple poles lie on the same side of the contour. This occurs when $\gamma_j \,{\in}\, \mathrm{C}(\SCosetR{j})^{[-]}_{j \rightarrow}$ with $n\,{\ge}\, 0$, or $\gamma_j \,{\in}\, (\SCosetR{j})^{[-]}_{j \rightarrow}$ with $n\,{<}\,0$. In these cases, the residue contributions are zero.
	\item\label{appitem:2} The point $p_y^n(\bar{\gamma}_j)$ lies outside the integration contour, while $p_{P}(\gamma_j)$ and $p_{P'}(\gamma_j)$ are inside. This happens when $\gamma_j \,{\in}\, (\SCosetR{j})^{[-]}_{j \rightarrow}$ with $n\,{\ge}\, 0$. We take the contour to enclose the poles outside, yielding the residue contribution
	\begin{align}
	    \operatorname{Res}_{p_y^n}\mathcal{I}_{j}(t) 
	    = -{2\pi \iunit}\, \skern{0}_k(\gamma_j^{-1} \sigma_j^n \bar{\gamma}_j^{-1}y)\quad \text{for}\quad  n\ge 0.
	\end{align}
	\item\label{appitem:3} The point $p_y^n(\bar{\gamma}_j)$ lies inside the integration contour, while $p_{P}(\gamma_j)$ and $p_{P'}(\gamma_j)$ are outside. This happens when $\gamma_j \,{\in}\, \mathrm{C}(\SCosetR{j})^{[-]}_{j \rightarrow}$ with $n\,{<}\,0$. We take the contour to enclose the poles outside, yielding the residue contribution
	\begin{align}
	    \operatorname{Res}_{p_y^n}\mathcal{I}_{j}(t) 
	    = {2\pi \iunit}\, \skern{0}_k(\gamma_j^{-1}\sigma_j^n \bar{\gamma}_j^{-1}y) \quad \text{for} \quad n<0.
	\end{align}
\end{situations}

\paragraph{Manipulating Schottky sums.}

Considering the splitting of words, we get
\begin{align}
    \sum_{\bar{\gamma}_j\in \SCosetR{j}}
    \sum_{\substack{\gamma_j\in\SCosetR{j}\\\gamma_j\neq \id}}
    \frac{1}{(-2\pi\iunit)^2}
    \int_{ \acyc_j} &\skern{1}_j(\bar{\gamma}_j^{-1} y,  t)\, \skern{0}_j(\gamma_j^{-1}t)\nonumber\\
    &= \sum_{\Upsilon\in \SCosetR{j}} 
    \frac{1}{(-2\pi \iunit)} \left(N_j^-(\Upsilon)-N_j^+(\Upsilon)\right) 
    \skern{0}_j(\Upsilon^{-1}y)\nonumber\\
    &= -\sum_{\Upsilon\in \SCosetR{j}} 
    \frac{1}{-2\pi\iunit} C(b_j, \Upsilon)\, \skern{0}_j(\Upsilon^{-1} y).
\end{align}
This result also holds for general $\skern{r}_k$ due to similar residue analysis. 
Thus, we obtain
\begin{align}
    \sum_{\bar{\gamma}_j\in \SCosetR{j}}\sum_{\substack{\gamma_j\in\SCosetR{j}\\\gamma_j\neq \id}}C(b_j^{r-k},\gamma_j)&\int_{t \in \Acycle_j}\skern{1}_j(\bar{\gamma}_j^{-1}y,t)\,\skern{k}_j(\gamma_j^{-1}t,z)\nonumber\\
    &=2\pi \iunit \sum_{l=0}^{r-k}\frac{\bn{l}}{l!}\sum_{\Upsilon\in \SCosetR{j}}C(b_j^{r-k-l+1},\Upsilon)\,\skern{k}_j(\Upsilon^{-1}y,z).
\end{align}
We emphasize that $\gamma_j \texteq \id$ must be excluded to ensure the composed word $\Upsilon \texteq \bar{\gamma}_j \sigma_j^n \gamma_j$ remains in the coset $\SCosetR{j}$. Using this, we write
\begin{align}
    \int_{t \in \acyc_j}& \omega_{jj}(y, t)\, \omega_{j^{n+1}}(t, z) \nonumber\\
    &=\sum_{\bar{\gamma}_j\in\SCosetR{j}}\frac{1}{(-2\pi \iunit)} C(b_j,\bar{\gamma}_j)\,\skern{0}_j(\bar{\gamma}_j^{-1}y,z)\frac{\bn{n}}{n!}\nonumber\\
    &\quad-\frac{1}{(-2\pi \iunit)}\sum_{\Upsilon\in \SCosetR{j}}\sum_{k=0}^n \sum_{l=0}^{n-k}\frac{\bn{l}}{l!}C(b_j^{n-k-l+1},\Upsilon)\,\skern{k}_j(\Upsilon^{-1}y,z)\nonumber\\
    &\quad-\frac{1}{(-2\pi \iunit)}\sum_{\bar{\gamma}_j\in\SCosetR{j}} \left(\sum_{l=0}^{n-1}\frac{\bn{l}}{l!}\skern{n+1-l}_j(\bar{\gamma}_j^{-1}y,z)+\frac{\bn{n+1}}{n!}\skern{0}_j(\bar{\gamma}_j^{-1}y,z)
    \right)\nonumber\\
    &=\sum_{\bar{\gamma}_j\in\SCosetR{j}}\frac{1}{(-2\pi \iunit)} C(b_j,\bar{\gamma}_j)\,\skern{0}_j(\bar{\gamma}_j^{-1}y,z)\frac{\bn{n}}{n!}\nonumber\\
    &\quad-\frac{1}{(-2\pi \iunit)}\sum_{\bar{\gamma}_j\in \SCosetR{j}}\!\!\left(\sum_{l=0}^{n-1}\frac{\bn{l}}{l!}\sum_{k=0}^{n-l}C(b_j^{n-k-l+1},\bar{\gamma}_j)\,\skern{k}_j(\bar{\gamma}_j^{-1}y,z)+\frac{\bn{n}}{n!}C(b_j,\bar{\gamma}_j)\,\skern{0}_j(\bar{\gamma}_j^{-1}y,z)\!\right)
    \nonumber\\
    &\quad-\frac{1}{(-2\pi \iunit)}\sum_{\bar{\gamma}_j\in\SCosetR{j}} \left(\sum_{l=0}^{n-1}\frac{\bn{l}}{l!}\skern{n+1-l}_j(\bar{\gamma}_j^{-1}y,z)+\frac{\bn{n+1}}{n!}\skern{0}_j(\bar{\gamma}_j^{-1}y,z)
    \right).
\end{align}
In the second equality, we reorganized and renamed summation variables for clarity. We also notice that the first term in the last line can be incorporated as the $k \texteq n \,{-}\, l \,{+}\, 1$ term of the first summation in the line before the last line. 

We now perform the resummation and simplify the terms,
\begin{align}
    \int_{t \in \acyc_j} \omega_{jj}(y, t)\, \omega_{j^{n+1}}(t, z)
    &= -\frac{1}{(-2\pi \iunit)}\sum_{\bar{\gamma}_j\in \SCosetR{j}}\sum_{l=0}^{n-1}\frac{\bn{l}}{l!}\sum_{k=0}^{n-l+1}C(b_j^{n-k-l+1},\bar{\gamma}_j)\,\skern{k}_j(\bar{\gamma}_j^{-1}y,z)\notag\\
    &\quad-\frac{1}{(-2\pi \iunit)}\sum_{\bar{\gamma}_j\in\SCosetR{j}} \frac{\bn{n+1}}{n!}\skern{0}_j(\bar{\gamma}_j^{-1}y,z)\nonumber\\
    &=- \sum_{l = 0}^{n-1} \frac{\bn{l}}{l!}  \,\omega_{j^{n-l+2}}(y, z)-\frac{\bn{n+1}}{n!}\omega_j(y),
\end{align}
confirming the claimed identity.

\appendixsubsection{\texorpdfstring{Additional calculations for \secref{sec:altid}}{Additional calculations for §8.3}}\label{app:Technicalitiesaltid}

In this appendix, we want to supplement the details of the proof of the alternating identity stated in \secref{sec:altid}. 

\paragraph{Residue analysis.} Analyzing the locations of the poles in \eqn{eqn:intschottky}, we arrive at the following contributions entering the residue theorem. From \eqn{eqn:intschottky}, we can read off that the poles are located at $t_{P_{e_1}'}(\gamma)\texteq\gamma P_{e_1}'$, $t_{P_{e_1}}(\gamma)\texteq\gamma P_{e_1}$ (simple poles) and $t_{P_{e_2}'}(\delta)\texteq\delta P_{e_2}'$, $t_{P_{e_2}}(\delta)\texteq\delta P_{e_2}$ (branch points of the logarithm), respectively. Moreover, notice that $t_{P'_{e_l}}(\gamma)$ and $t_{P_{e_l}}(\gamma)$ (resp.~$t_{P'_{e_l}}(\delta)$ and $t_{P_{e_l}}(\delta)$), $l\texteq1,2$, are always located on the same side of the integration contour for $\gamma\,{\neq}\,\id$ (resp.~$\delta\,{\neq}\,\id$). However, the situations when either $\gamma\texteq\id$ or $\delta\texteq\id$ (or $\gamma\texteq\delta\texteq\id$) can be neglected as $\Coeff(\id,\delta)\texteq\Coeff(\gamma,\id)\texteq\Coeff(\id,\id)\texteq0$ for the coefficient\footnote{This furthermore ensures that the branch cut connecting $z_1(\delta)$ and $z_2(\delta)$ is always located entirely on one side of the contour and does not intersect it in any situation. Therefore, the residue theorem is indeed applicable for all relevant configurations. Moreover, $\gamma\,{\neq}\,\id$ makes sure that no additional pole configurations enter the residue theorem.} $\Coeff(\gamma,\delta)$ in \eqn{eqn:altidschottky}. This directly follows from the assumption $n_1,n_2\,{>}\,0$ and the definition~\eqref{eqn:Ccoeff} of the $C$-coefficients. The relevant pole distributions can therefore be summarized in the following list.
\begin{situations}
    \item\label{sit:1} The first configuration occurs when all the singularities are located on the same side of the contour. Then, reversing the orientation if necessary, we can assume that the integrand is holomorphic on the side enclosed by the contour. It follows immediately from the residue theorem that the integral vanishes in this situation.
    \item\label{sit:2} The second configuration is that the simple poles at $t_{P_{e_1}'}(\gamma)$ and $t_{P_{e_1}}(\gamma)$ are enclosed by the contour, whereas the branch points $t_{P_{e_2}'}(\delta)$ and $t_{P_{e_2}}(\delta)$ are not. In order for this to happen, we require $\gamma\,{\in}\, \SCosetR{e_1}$ and $\delta\,{\in}\, \SCosetR{e_2}$ to satisfy certain conditions. The requirement that $t_{P_{e_1}'}(\gamma)$ and $t_{P_{e_1}}(\gamma)$ are enclosed by the contour $\acyc_i$ translates to $\gamma\,{\in}\,(\SCosetR{e_1})^{[-]}_{i\rightarrow}$ (see \secref{sec:splitting_of_Schottky_words} for the definition of the notation) as $\sigma_i^n$ for $n\,{<}\,0$ maps points to the side enclosed by the contour $\acyc_i$ by definition of the generator $\sigma_i^{-1}$. Similarly, the condition on $t_{P_{e_2}'}(\delta)$ and $t_{P_{e_2}}(\delta)$ translates to $\delta\,{\in}\, \mathrm{C}(\SCosetR{e_2})^{[-]}_{i\rightarrow}$ as this is the condition for $\delta$ not mapping points to the side enclosed by $\acyc_i$. The residue attributed to such a configuration is given by 
    \begin{equation}
    \label{eqn:R+}
    R^{-}_{e_1,e_2}(\delta^{-1}\gamma)\equiv(-2\pi\iunit)\log\left(\frac{(\gamma P_{e_1}'-\delta P_{e_2}')(\gamma P_{e_1}-\delta P_{e_2})}{(\gamma P_{e_1}-\delta P_{e_2}')(\gamma P_{e_1}'-\delta P_{e_2})}\right) \, ,
    \end{equation}
    where we have included a minus sign to account for the negative orientation of the contour $\acyc_i$. Notice that $R^-_{e_1,e_2}(\delta^{-1}\gamma)$ in fact only depends on the product $\delta^{-1}\gamma$, which justifies the notation. This can be seen from the definition~\eqref{eqn:R+} by applying the invariance of the cross-ratio.
    \item\label{sit:3} The third configuration is obtained by exchanging the locations of the singularities in the previous case, hence we can trace this configuration back to the previous case by inverting the orientation of the contour. The conditions on $\gamma$ and $\delta$ thus become $\gamma\,{\in}\, \mathrm{C}(\SCosetR{e_1})^{[-]}_{i\rightarrow}$ and $\delta\,{\in}\, (\SCosetR{e_2})^{[-]}_{i\rightarrow}$ for this configuration. Additionally, we are taking up a minus sign accounting for the inversion of the contour. In total, the residue attributed to such a configuration is
    \begin{equation}
    \label{eqn:R-}
    R^{+}_{e_1,e_2}(\delta^{-1}\gamma)\equiv 2\pi\iunit\log\left(\frac{(\gamma P_{e_1}'-\delta P_{e_2}')(\gamma P_{e_1}-\delta P_{e_2})}{(\gamma P_{e_1}-\delta P_{e_2}')(\gamma P_{e_1}'-\delta P_{e_2})}\right)=-R^-_{e_1,e_2}(\delta^{-1}\gamma) \, .
    \end{equation}
\end{situations}
Putting all of this together, \eqn{eqn:altidschottky} becomes
\begin{equation}
    \label{eqn:rescontrib}
    \begin{aligned}
        \cL^{(i)}_{e_1,e_2}=&\sum_{\substack{\gamma\in (\SCosetR{e_1})^{[-]}_{i\rightarrow} \\ \delta\in \mathrm{C}(\SCosetR{e_2})^{[-]}_{i\rightarrow}}}\frac{\Coeff(\gamma,\delta)}{(-2\pi\iunit)^2}R^{-}_{e_1,e_2}(\delta^{-1}\gamma)+\sum_{\substack{\gamma\in \mathrm{C}(\SCosetR{e_1})^{[-]}_{i\rightarrow} \\ \delta\in (\SCosetR{e_2})^{[-]}_{i\rightarrow}}}\frac{\Coeff(\gamma,\delta)}{(-2\pi\iunit)^2}R^{+}_{e_1,e_2}(\delta^{-1}\gamma) \, .
    \end{aligned}
\end{equation}
The strategy is now to define $\Upsilon\,{\equiv}\,\delta^{-1}\gamma$ and to rewrite \eqn{eqn:rescontrib} in terms of a single sum over $\Upsilon$. First of all, it can be seen from the definitions~\eqref{eqn:set1}-\eqref{eqn:cset2} that $\Upsilon\texteq\delta^{-1}\gamma\,{\in}\,\SCosetLR{e_2}{e_1}$ for all tuples $(\gamma,\delta)$ appearing in \eqn{eqn:rescontrib}. However, there is not a one-to-one correspondence between those tuples and $\Upsilon\,{\in}\,\SCosetLR{e_2}{e_1}$. More precisely, given $\Upsilon\,{\in}\,\SCosetLR{e_2}{e_1}$, there are several different ways to divide it into a product $\delta^{-1}\gamma$ for $\gamma$ and $\delta$ contained in the respective sets appearing in \eqn{eqn:rescontrib}. Let us illustrate this procedure in more detail for $(\gamma,\delta)\,{\in}\,(\SCosetR{e_1})^{[-]}_{i\rightarrow}\times \mathrm{C}(\SCosetR{e_2})^{[-]}_{i\rightarrow}$, which corresponds to rewriting the first term in \eqn{eqn:rescontrib}. Unpacking the definitions, $\gamma\,{\in}\,(\SCosetR{e_1})^{[-]}_{i\rightarrow}$ in principle just means that $\gamma\,{\in}\,\SCosetR{e_1}$ must start with a (strictly) negative power of the generator $\sigma_i$. Analogously, $\delta\,{\in}\, \mathrm{C}(\SCosetR{e_2})^{[-]}_{i\rightarrow}$ translates to the statement that $\delta$ must not start with a (strictly) negative power of $\sigma_i$. By inversion, this implies that $\delta^{-1}$ must not end in a (strictly) positive power of $\sigma_i$. Let now $\Upsilon\,{\in}\,\SCosetLR{e_2}{e_1}$. Also recall that we neglect $\delta\texteq\id$ as its coefficient in \eqn{eqn:altidschottky} vanishes. To divide $\Upsilon$ into a product $\delta^{-1}\gamma$ for valid $\gamma$ and $\delta$ now becomes equivalent to splitting $\Upsilon$ before a generator $\sigma_i^{-1}$. In this way, we guarantee that the condition on $\gamma$ is satisfied. The condition on $\delta^{-1}$ on the other hand prohibits trivial splittings obtained from the insertion of an identity $\id\texteq\sigma_i^{-1}\sigma_i$ into $\Upsilon$. To make this a bit more concrete, let us consider the example of $\Upsilon\texteq\sigma_2^2\sigma_1^{-2}\sigma_2^2\sigma_1$ (at genus two). Then
\begin{equation}
    \begin{aligned}
        (\gamma,\delta)=&(\sigma_1^{-2}\sigma_2^2\sigma_1,\sigma_2^{-2}) \, , \\
        (\gamma,\delta)=&(\sigma_1^{-1}\sigma_2^2\sigma_1,\sigma_1\sigma_2^{-2})
    \end{aligned}
\end{equation}
are both valid splittings along $\sigma_1$, while $(\gamma,\delta)\texteq(\sigma_1^{-1}\sigma_2\sigma_1^{-2}\sigma_2^2\sigma_1,\sigma_1^{-1}\sigma_2^{-1})$ is not as it violates the condition on $\delta$. We therefore obtain the set of all possible splittings $\nsplt{i}{\Upsilon}\,{\subset}\, \SGroup\times \SGroup$ of $\Upsilon\,{\in}\,\SCosetLR{e_2}{e_1}$ before $\sigma_i^{-1}$ as defined in \secref{sec:splitting_of_Schottky_words}. We can therefore rewrite the first line of \eqn{eqn:rescontrib} as
\begin{equation}
    \sum_{\Upsilon\in\SCosetLR{e_2}{e_1}}\sum_{(\gamma,\delta)\in \nsplt{i}{\Upsilon}} \frac{\Coeff(\gamma,\delta)}{(-2\pi\iunit)^2}R^{-}_{e_1,e_2}(\Upsilon) \, .
\end{equation}
A completely analogous analysis can be performed for the second term in \eqn{eqn:rescontrib}. It turns out that in this case, $\Upsilon$ has to be split after $\sigma_i$ in order to match the conditions on $\gamma$ and $\delta$. We thus get the set $\psplt{i}{\Upsilon}$. Then we can put everything together and finally write
\begin{equation}
    \label{eqn:altidintermediate}
    \cL^{(i)}_{e_1,e_2}=\sum_{\Upsilon\in\SCosetLR{e_2}{e_1}}\left(\sum_{(\gamma,\delta)\in \nsplt{i}{\Upsilon}} \frac{\Coeff(\gamma,\delta)}{(-2\pi\iunit)^2}-\sum_{(\gamma,\delta)\in \psplt{i}{\Upsilon}}\frac{\Coeff(\gamma,\delta)}{(-2\pi\iunit)^2}\right)R^{-}_{e_1,e_2}(\Upsilon) \, ,
\end{equation}
where we have also made use of the fact that $R^+_{e_1,e_2}(\Upsilon)\texteq{-}\,R^-_{e_1,e_2}(\Upsilon)$. This completes the evaluation of the integral.

\paragraph{Coefficient analysis.} Let us take a look at the coefficients in \eqn{eqn:altidintermediate}. The splittings in $\nsplt{i}{\Upsilon}\,,\,\psplt{i}{\Upsilon}$ can be parameterized by the location of a generator $\sgen_i^{\mp1}$ inside $\Upsilon$. More detailed, consider a general element $\Upsilon\texteq\sgen_{j_1}^{m_1}\cdots\sgen_{j_s}^{m_s}\,{\in}\,\SCosetLR{e_2}{e_1}$. Define the subsets 
\begin{equation} 
	L_i^\pm(\Upsilon)\texteq\{l\mid j_l\texteq i,\,\pm m_l\,{>}\,0\}
\end{equation} 
of $\{1,\ldots,s\}$, which encapsulate the locations of the generators $\sgen_i^{\pm 1}$ in $\Upsilon$. We call the elements of those sets the \textit{positive\slash negative splitting locations}. We can then write
\begin{subequations}
    \begin{align}
        \label{eqn:splitting+}
        \sum_{(\gamma,\delta)\in \nsplt{i}{\Upsilon}} \Coeff(\gamma,\delta)=&\sum_{l\in L_i^-(\Upsilon)}\sum_{p=0}^{-m_l-1}\Coeff(\sigma_i^{m_l+p}\sgen_{j_{l+1}}^{m_{l+1}}\cdots\sgen_{j_s}^{m_s}, \sgen_i^p\sgen_{j_{l-1}}^{-m_{l-1}}\cdots\sgen_{j_1}^{-m_1}) \, , \\
        \label{eqn:splitting-}
        \sum_{(\gamma,\delta)\in \psplt{i}{\Upsilon}} \Coeff(\gamma,\delta)=&\sum_{l\in L_i^+(\Upsilon)}\sum_{p=1}^{m_l}\Coeff(\sgen_i^{m_l-p}\sgen_{j_{l+1}}^{m_{l+1}}\cdots\sgen_{j_s}^{m_s}, \sgen_i^{-p}\sgen_{j_{l-1}}^{-m_{l-1}}\cdots\sgen_{j_1}^{-m_1}) \, ,
    \end{align}
\end{subequations}
where we have taken into account that we must split before $\sgen_i^{-1}$ and after $\sgen_i$. Let us now investigate the exact dependence of the summands on the splitting locations. We can use the recursive definition~\eqref{eqn:Ccoeff} of the $C$-coefficients to write (for $p\,{\neq}\, 0$)
\begin{equation}\label{eqn:splittingsC}
\begin{aligned}
    C(b_{\mindx{i}_1}^*,\gamma)&=C(b_{\mindx{i}_1}^*,\gamma_{\mathrm{red}})+(m_l+p)\,C(b_{\mindx{i}_3}^*,\gamma_{\mathrm{red}}) \, , \\
    C(b_{\mindx{i}_2}^*,\delta)&=C(b_{\mindx{i}_2}^*,\delta_{\mathrm{red}}) \, , \\
    C(b_{\mindx{i}_3}^*,\gamma)&=C(b_{\mindx{i}_3}^*,\gamma_{\mathrm{red}}) \, , \\
    C(b_{\mindx{i}_4}^*,\delta)&=C(b_{\mindx{i}_4}^*,\delta_{\mathrm{red}})+p\,C(b_{\mindx{i}_2}^*,\delta_{\mathrm{red}}),
\end{aligned}
\end{equation}
for $\gamma,\delta$ as in \eqn{eqn:splitting+}. In the above, we made use of the definitions of the multi-indices $\mindx{i}_l^*$, $l\,{\in}\,\{1,\ldots,4\}$, as well as defined $\rho_{\mathrm{red}}\texteq\sgen_{j_2}^{m_2}\cdots\sgen_{j_s}^{m_s}$ as the reduced version of $\rho\texteq\sgen_{j_1}^{m_1}\cdots\sgen_{j_s}^{m_s}$. In total, this yields for negative splitting locations
\begin{equation}\label{eqn:splittingC2}
    \begin{aligned}
        \Coeff(\gamma,\delta)=&p\,C(b_{\mindx{i}_3}^*,\gamma_{\mathrm{red}})\,C(b_{\mindx{i}_2}^*,\delta_{\mathrm{red}})-p\,C(b_{\mindx{i}_3}^*,\gamma_{\mathrm{red}})\,C(b_{\mindx{i}_2}^*,\delta_{\mathrm{red}})+(\text{$p$-independent terms}) \\
        =&~(\text{$p$-independent terms}).
    \end{aligned}
\end{equation}
A similar analysis shows \eqn{eqn:splittingsC} to hold for $(\gamma,\delta)$ as in \eqn{eqn:splitting-} under $p\,{\rightarrow}\,{-}\,p$. Accordingly, \eqn{eqn:splittingC2} analogously holds for positive splitting locations. This shows that $C(\gamma,\delta)$ does only depend on the splitting location, but not on the exact splitting at this location itself. We can therefore simplify the coefficient in \eqn{eqn:altidintermediate} as
\begin{equation}
    \label{eqn:coeff}
    \COeff(\Upsilon)\equiv\left(\sum_{(\gamma,\delta)\in \nsplt{i}{\Upsilon}} \Coeff(\gamma,\delta)-\sum_{(\gamma,\delta)\in \psplt{i}{\Upsilon}}\Coeff(\gamma,\delta)\right)=-\sum_{l\in L_i(\Upsilon)}m_l\, \Coeff(\gamma_l,\delta_l) \, ,
\end{equation}
where we defined $L_i(\Upsilon)\texteq L_i^+(\Upsilon)\,{\cup}\, L_i^-(\Upsilon)$. Moreover, $\gamma_l,\delta_l$ represent an arbitrary, but fixed, splitting at the location $l\,{\in}\, L_i(\Upsilon)$, which is well-defined by the above.

\newpage


\bibliographystyle{heMZV}
\bibliography{heMZV}

\begin{thebibliography}{10}
\ifx\href\asklfhas\newcommand{\href}[2]{#2}\fi
\ifx\arxivref\asklfhas\newcommand{\arxivref}[2]{\href{http://arxiv.org/abs/#1}{#2}}\fi
\ifx\doiref\asklfhas\newcommand{\doiref}[2]{\href{http://dx.doi.org/#1}{#2}}\fi
\raggedright
\small
\parskip 0pt

\bibitem{Basel1}
L.~Euler,
\textit{``{De summis serierum reciprocarum}''},
\textsf{Commentarii~academiae~scientiarum~Petropolitanae~7,~123~(1740)}.

\bibitem{Basel2}
L.~Euler,
\textit{``On the sums of series of reciprocals (English translation of {`De
  summis serierum reciprocarum (1740)'})''},
\texttt{\arxivref{math/0506415}{math/0506415}}.

\bibitem{GONCHAROV1995197}
A.~Goncharov,
\textit{``Geometry of Configurations, Polylogarithms, and Motivic
  Cohomology''},
\textsf{\doiref{10.1006/aima.1995.1045}{Advances~in~Mathematics~114,~197
  ~(1995)}}.

\bibitem{Goncharov:1998kja}
A.~B.~Goncharov,
\textit{``{Multiple polylogarithms, cyclotomy and modular complexes}''},
\textsf{\doiref{10.4310/MRL.1998.v5.n4.a7}{Math.~Res.~Lett.~5,~497~(1998)}},
\texttt{\arxivref{1105.2076}{arxiv:1105.2076}}.

\bibitem{Remiddi:2000}
E.~Remiddi and J.~A.~M.~Vermaseren,
\textit{``Harmonic polylogarithms''},
\textsf{\doiref{10.1142/s0217751x00000367}{International~Journal~of~Modern~Physics~A~15,~725–754~(2000)}}.

\bibitem{Goncharov:2001iea}
A.~Goncharov,
\textit{``{Multiple polylogarithms and mixed Tate motives}''},
\texttt{\arxivref{math/0103059}{math/0103059}}.

\bibitem{GiNaC:2005}
J.~Vollinga and S.~Weinzierl,
\textit{``Numerical evaluation of multiple polylogarithms''},
\textsf{\doiref{10.1016/j.cpc.2004.12.009}{Computer~Physics~Communications~167,~177–194~(2005)}}.

\bibitem{Chen}
K.~Chen,
\textit{``Iterated path integrals''},
\textsf{\doiref{10.1090/S0002-9904-1977-14320-6}{Bull.~Amer.~Math.~Soc.~83,~831~(1977)}}.

\bibitem{Borwein1997}
J.~Borwein, D.~Bradley and D.~Broadhurst,
\textit{``Evaluations of $k$-fold Euler/Zagier sums: a compendium of results
  for arbitrary $k$.''},
\textsf{\doiref{https://doi.org/10.37236/1320}{The~Electronic~Journal~of~Combinatorics~[electronic~only]~4,~Research
  paper R5~(1997)}}.

\bibitem{Hoffman:MHS}
M.~E.~Hoffman,
\textit{``Multiple harmonic series''},
\textsf{\doiref{10.2140/pjm.1992.152.275}{Pacific~J.~Math.~152,~275~(1992)}}.

\bibitem{Zagier23}
D.~Zagier,
\textit{``Values of zeta functions and their applications''},
in: \textit{``First European Congress of Mathematics Paris, July 6--10, 1992:
  Vol. II: Invited Lectures (Part 2)''},
Birkh{\"a}user Basel (1994),
Basel,
497--512p,
\href{https://doi.org/10.1007/978-3-0348-9112-7\_23}{\texttt{https://doi.org/10.1007/978-3-0348-9112-7\_23}}.

\bibitem{BrownThesis}
F.~Brown,
\textit{``Multiple zeta values and periods of moduli spaces
  $\overline{\mathfrak M}_{0,n}$''},
\textsf{\doiref{10.24033/asens.2099}{Ann.~Sci.~\'Ec.~Norm.~Sup\'er.~(4)~42,~371~(2009)}},
\texttt{\arxivref{math/0606419}{math/0606419}}.

\bibitem{FresanGil}
J.~Fresan and J.~Gil,
\textit{``Multiple zeta values: from numbers to motives''},
\textsf{Clay~Mathematics~Proceedings~\!\!,~to appear},
\href{http://javier.fresan.perso.math.cnrs.fr/mzv.pdf}{\texttt{http://javier.fresan.perso.math.cnrs.fr/mzv.pdf}}.

\bibitem{DG}
P.~Deligne and A.~B.~Goncharov,
\textit{``Groupes fondamentaux motiviques de {T}ate mixte''},
\textsf{\doiref{10.1016/j.ansens.2004.11.001}{Ann.~Sci.~\'Ecole~Norm.~Sup.~(4)~38,~1~(2005)}}.

\bibitem{Brown1102.1312}
F.~Brown,
\textit{``Mixed Tate motives over $\mathbb{Z}$''},
\textsf{\doiref{10.4007/annals.2012.175.2.10}{Annals~of~Mathematics~175,~949~(2012)}}.

\bibitem{browndepthgraded}
F.~Brown,
\textit{``Depth-graded motivic multiple zeta values''},
\textsf{\doiref{10.1112/S0010437X20007654}{Compositio~Mathematica~157,~529–572~(2021)}},
\texttt{\arxivref{1301.3053}{arxiv:1301.3053}}.

\bibitem{Glanois:Thesis}
C.~Glanois,
\textit{``{Periods of the motivic fundamental groupoid of $\mathbb{P}^1
  \setminus \{0,\mu_N,\infty\}$}''},
PhD thesis,
Universit\'e Pierre et Marie Curie,
2016.

\bibitem{Blumlein:2009cf}
J.~Blumlein, D.~Broadhurst and J.~Vermaseren,
\textit{``{The Multiple Zeta Value Data Mine}''},
\textsf{\doiref{10.1016/j.cpc.2009.11.007}{Comput.Phys.Commun.~181,~582~(2010)}},
\texttt{\arxivref{0907.2557}{arxiv:0907.2557}}.

\bibitem{SnowmassString}
N.~Berkovits, E.~D'Hoker, M.~B.~Green, H.~Johansson and O.~Schlotterer,
\textit{``{Snowmass White Paper: String Perturbation Theory}''},
\texttt{\arxivref{2203.09099}{arxiv:2203.09099}},
in: \textit{``{Snowmass 2021}''}.

\bibitem{LocalizationReview}
V.~Pestun, M.~Zabzine, F.~Benini, T.~Dimofte, T.~T.~Dumitrescu, K.~Hosomichi,
  S.~Kim, K.~Lee, B.~Le~Floch, M.~Mariño, J.~A.~Minahan, D.~R.~Morrison,
  S.~Pasquetti, J.~Qiu, L.~Rastelli, S.~S.~Razamat, S.~S.~Pufu, Y.~Tachikawa,
  B.~Willett and K.~Zarembo,
\textit{``Localization techniques in quantum field theories''},
\textsf{\doiref{10.1088/1751-8121/aa63c1}{Journal~of~Physics~A:~Mathematical~and~Theoretical~50,~440301~(2017)}},
\href{http://dx.doi.org/10.1088/1751-8121/aa63c1}{\texttt{http://dx.doi.org/10.1088/1751-8121/aa63c1}}.

\bibitem{Alday:2023mvu}
L.~F.~Alday and T.~Hansen,
\textit{``{The AdS Virasoro-Shapiro amplitude}''},
\textsf{\doiref{10.1007/JHEP10(2023)023}{JHEP~2310,~023~(2023)}},
\texttt{\arxivref{2306.12786}{arxiv:2306.12786}}.

\bibitem{Enriquez:EllAss}
B.~Enriquez,
\textit{``Elliptic associators''},
\textsf{\doiref{10.1007/s00029-013-0137-3}{Selecta~Math.~(N.S.)~20,~491~(2014)}}.

\bibitem{Enriquez:Emzv}
B.~Enriquez,
\textit{``Analogues elliptiques des nombres multiz\'etas''},
\textsf{\doiref{10.24033/bsmf.2718}{Bull.~Soc.~Math.~France~144,~395~(2016)}},
\texttt{\arxivref{1301.3042}{arxiv:1301.3042}}.

\bibitem{Levin}
A.~Levin,
\textit{``Elliptic polylogarithms: An analytic theory''},
\textsf{\doiref{10.1023/A:1000193320513}{Compositio~Mathematica~106,~267~(1997)}}.

\bibitem{BrownLevin}
F.~Brown and A.~Levin,
\textit{``{Multiple elliptic polylogarithms}''},
\texttt{\arxivref{1110.6917v2}{arxiv:1110.6917v2}}.

\bibitem{Broedel:2014vla}
J.~Broedel, C.~R.~Mafra, N.~Matthes and O.~Schlotterer,
\textit{``{Elliptic multiple zeta values and one-loop superstring
  amplitudes}''},
\textsf{\doiref{10.1088/1126-6708/2015/07/112}{JHEP~1507,~112~(2015)}},
\texttt{\arxivref{1412.5535}{arxiv:1412.5535}}.

\bibitem{Pollack}
A.~Pollack,
\textit{``{Relations between derivations arising from modular forms}''},
Undergraduate thesis, Duke University.

\bibitem{CEE}
D.~Calaque, B.~Enriquez and P.~Etingof,
\textit{``Universal {KZB} equations: the elliptic case''},
in: \textit{``Algebra, arithmetic, and geometry: in honor of {Y}u. {I}.
  {M}anin. {V}ol. {I}''},
Birkh\"auser Boston, Inc., Boston, MA (2009),
165--266p.

\bibitem{Broedel:2015hia}
J.~Broedel, N.~Matthes and O.~Schlotterer,
\textit{``{Relations between elliptic multiple zeta values and a special
  derivation algebra}''},
\textsf{\doiref{10.1088/1751-8113/49/15/155203}{J.~Phys.~A49,~155203~(2016)}},
\texttt{\arxivref{1507.02254}{arxiv:1507.02254}}.

\bibitem{Broedel:2017kkb}
J.~Broedel, C.~Duhr, F.~Dulat and L.~Tancredi,
\textit{``{Elliptic polylogarithms and iterated integrals on elliptic curves.
  Part I: general formalism}''},
\textsf{\doiref{10.1007/JHEP05(2018)093}{JHEP~1805,~093~(2018)}},
\texttt{\arxivref{1712.07089}{arxiv:1712.07089}}.

\bibitem{Bourjaily:2022bwx}
J.~L.~Bourjaily et~al.,
\textit{``{Functions Beyond Multiple Polylogarithms for Precision Collider
  Physics}''},
\texttt{\arxivref{2203.07088}{arxiv:2203.07088}},
in: \textit{``{Snowmass 2021}''}.

\bibitem{eMZVWebsite}
J.~Broedel, N.~Matthes and O.~Schlotterer,
\textit{``Elliptic multiple zeta values datamine''},
\href{https://tools.aei.mpg.de/emzv/datamine.html}{\texttt{https://tools.aei.mpg.de/emzv/datamine.html}}.

\bibitem{Matthes:Thesis}
N.~Matthes,
\textit{``{Elliptic multiple zeta values}''},
PhD thesis,
Universit\"at Hamburg,
2016.

\bibitem{ZagierGangl}
D.~Zagier and H.~Gangl,
\textit{``Classical and Elliptic Polylogarithms and Special Values of
  L-Series''},
\textsf{\doiref{10.1007/978-94-011-4098-0_21}{The~Arithmetic~and~Geometry~of~Algebraic~Cycles,~Nato~Science~Series~C~548,~61~(2000)}}.

\bibitem{Broedel_2020}
J.~Broedel and A.~Kaderli,
\textit{``Functional relations for elliptic polylogarithms''},
\textsf{\doiref{10.1088/1751-8121/ab81d7}{Journal~of~Physics~A:~Mathematical~and~Theoretical~53,~245201~(2020)}},
\href{http://dx.doi.org/10.1088/1751-8121/ab81d7}{\texttt{http://dx.doi.org/10.1088/1751-8121/ab81d7}}.

\bibitem{Broedel:2018qkq}
J.~Broedel, C.~Duhr, F.~Dulat, B.~Penante and L.~Tancredi,
\textit{``{Elliptic Feynman integrals and pure functions}''},
\textsf{\doiref{10.1007/JHEP01(2019)023}{JHEP~1901,~023~(2019)}},
\texttt{\arxivref{1809.10698}{arxiv:1809.10698}}.

\bibitem{DHoker:2019blr}
E.~D'Hoker and M.~B.~Green,
\textit{``{Exploring transcendentality in superstring amplitudes}''},
\textsf{\doiref{10.1007/JHEP07(2019)149}{JHEP~1907,~149~(2019)}},
\texttt{\arxivref{1906.01652}{arxiv:1906.01652}}.

\bibitem{DHoker:2025dhv}
E.~D'Hoker and O.~Schlotterer,
\textit{``{Meromorphic higher-genus integration kernels via convolution over
  homology cycles}''},
\texttt{\arxivref{2502.14769}{arxiv:2502.14769}}.

\bibitem{EnriquezHigher}
B.~Enriquez,
\textit{``Flat connections on configuration spaces and braid groups of
  surfaces''},
\textsf{\doiref{https://doi.org/10.1016/j.aim.2013.10.025}{Advances~in~Mathematics~252,~204~(2014)}},
\texttt{\arxivref{1112.0864}{arxiv:1112.0864}}.

\bibitem{EZ1}
B.~Enriquez and F.~Zerbini,
\textit{``{Construction of Maurer-Cartan elements over configuration spaces of
  curves}''},
\texttt{\arxivref{2110.09341}{arxiv:2110.09341}}.

\bibitem{SchottkyTools}
E.~Im,
\textit{``\texttt{SchottkyTools}''},
Software package for numerical evaluation of higher-genus polylogarithms. Work
  in progress.

\bibitem{Schottky1887}
F.~Schottky,
\textit{``Ueber eine specielle Function, welche bei einer bestimmten linearen
  Transformation ihres Arguments unverändert bleibt.''},
\textsf{\doiref{doi:10.1515/crll.1887.101.227}{Journal~für~die~reine~und~angewandte~Mathematik~1887,~227~(1887)}}.

\bibitem{Bobenko:2011}
A.~I.~Bobenko and C.~Klein,
\textit{``{Computational approach to Riemann surfaces}''},
Springer Berlin, Heidelberg (2011).

\bibitem{herrlich1schottky}
F.~Herrlich,
\textit{``Schottky space and Teichm{\"u}ller disks''},
\textsf{Handbook~of~group~actions~1,~289~(2015)}.

\bibitem{Baune:2024biq}
K.~Baune, J.~Broedel, E.~Im, A.~Lisitsyn and F.~Zerbini,
\textit{``{Schottky\textendash{}Kronecker forms and hyperelliptic
  polylogarithms}''},
\textsf{\doiref{10.1088/1751-8121/ad8197}{J.~Phys.~A~57,~445202~(2024)}},
\texttt{\arxivref{2406.10051}{arxiv:2406.10051}}.

\bibitem{Baker}
H.~F.~Baker,
\textit{``Chapter XII: A particular form of fundamental surface''},
in: \textit{``Abelian functions: Abel's theorem and the allied theory''},
Cornell University Library (1900),
343--374p.

\bibitem{Keen_1980}
L.~Keen,
\textit{``On hyperelliptic Schottky groups''},
\textsf{\doiref{10.5186/aasfm.1980.0512}{Annales~Fennici~Mathematici~5,~165–174~(1980)}}.

\bibitem{Brown:2013qva}
F.~Brown,
\textit{``{Iterated integrals in quantum field theory}''},
in: \textit{``{Proceedings, Geometric and Topological Methods for Quantum Field
  Theory : 6th Summer School: Villa de Leyva, Colombia, July 6-23, 2009}''},
188-240p.

\bibitem{Broedel:2018iwv}
J.~Broedel, C.~Duhr, F.~Dulat, B.~Penante and L.~Tancredi,
\textit{``{Elliptic symbol calculus: from elliptic polylogarithms to iterated
  integrals of Eisenstein series}''},
\textsf{\doiref{10.1007/JHEP08(2018)014}{JHEP~1808,~014~(2018)}},
\texttt{\arxivref{1803.10256}{arxiv:1803.10256}}.

\bibitem{EZ3}
B.~Enriquez and F.~Zerbini,
\textit{``{Elliptic hyperlogarithms}''},
\textsf{\doiref{10.4153/S0008414X24001068}{Canadian~Journal~of~Mathematics~online,~1–36~(2025)}},
\texttt{\arxivref{2307.01833}{arxiv:2307.01833}}.

\bibitem{Broedel:2019gba}
J.~Broedel and A.~Kaderli,
\textit{``{Amplitude recursions with an extra marked point}''},
\textsf{\doiref{10.4310/CNTP.2022.v16.n1.a3}{Commun.~Num.~Theor.~Phys.~16,~75~(2022)}},
\texttt{\arxivref{1912.09927}{arxiv:1912.09927}}.

\bibitem{DHoker:2023vax}
E.~D'Hoker, M.~Hidding and O.~Schlotterer,
\textit{``{Constructing polylogarithms on higher-genus Riemann surfaces}''},
\textsf{\doiref{10.4310/cntp.250531031558}{Commun.~Num.~Theor.~Phys.~19,~355~(2025)}},
\texttt{\arxivref{2306.08644}{arxiv:2306.08644}}.

\bibitem{DHoker:2025szl}
E.~D'Hoker, B.~Enriquez, O.~Schlotterer and F.~Zerbini,
\textit{``{Relating flat connections and polylogarithms on higher genus Riemann
  surfaces}''},
\texttt{\arxivref{2501.07640}{arxiv:2501.07640}}.

\bibitem{LevinRacinet}
A.~Levin and G.~Racinet,
\textit{``{Towards multiple elliptic polylogarithms}''},
\texttt{\arxivref{math/0703237}{math/0703237}}.

\bibitem{Baune:2024ber}
K.~Baune, J.~Broedel, E.~Im, A.~Lisitsyn and Y.~Moeckli,
\textit{``{Higher-genus Fay-like identities from meromorphic generating
  functions}''},
\textsf{\doiref{10.21468/SciPostPhys.18.3.093}{SciPost~Phys.~18,~093~(2025)}},
\texttt{\arxivref{2409.08208}{arxiv:2409.08208}}.

\bibitem{DHoker:2024ozn}
E.~D'Hoker and O.~Schlotterer,
\textit{``{Fay identities for polylogarithms on higher-genus Riemann
  surfaces}''},
\texttt{\arxivref{2407.11476}{arxiv:2407.11476}}.

\bibitem{fay}
J.~Fay,
\textit{``Theta functions on Riemann surfaces''},
Springer (1973).

\bibitem{EZ2}
B.~Enriquez and F.~Zerbini,
\textit{``{Analogues of hyperlogarithm functions on affine complex curves}''},
\texttt{\arxivref{2212.03119}{arxiv:2212.03119}}.

\bibitem{gonzalez2020surfacedrinfeldtorsorsi}
M.~Gonzalez,
\textit{``Surface Drinfeld Torsors I : Higher Genus Associators''},
\texttt{\arxivref{2004.07303}{arxiv:2004.07303}},
\href{https://arxiv.org/abs/2004.07303}{\texttt{https://arxiv.org/abs/2004.07303}}.

\bibitem{Goncharov:2005sla}
A.~B.~Goncharov,
\textit{``Galois symmetries of fundamental groupoids and noncommutative
  geometry''},
\textsf{\doiref{10.1215/S0012-7094-04-12822-2}{Duke~Math.~J.~128,~209~(2005)}},
\texttt{\arxivref{math/0208144}{math/0208144}}.

\bibitem{Brown:2011ik}
F.~C.~S.~Brown,
\textit{``On the decomposition of motivic multiple zeta values''},
\texttt{\arxivref{1102.1310}{arxiv:1102.1310}},
in: \textit{``Galois-{T}eichm\"uller theory and arithmetic geometry''},
Math. Soc. Japan, Tokyo (2012),
31--58p.

\bibitem{Duhr:2012fh}
C.~Duhr,
\textit{``{Hopf algebras, coproducts and symbols: an application to Higgs boson
  amplitudes}''},
\textsf{\doiref{10.1007/JHEP08(2012)043}{JHEP~1208,~043~(2012)}},
\texttt{\arxivref{1203.0454}{arxiv:1203.0454}}.

\bibitem{BroadKrei}
D.~J.~Broadhurst and D.~Kreimer,
\textit{``{Association of multiple zeta values with positive knots via Feynman
  diagrams up to 9 loops}''},
\textsf{\doiref{10.1016/S0370-2693(96)01623-1}{Phys.~Lett.~B~393,~403~(1997)}},
\texttt{\arxivref{hep-th/9609128}{hep-th/9609128}}.

\bibitem{Marzucca:2023gto}
R.~Marzucca, A.~J.~McLeod, B.~Page, S.~P\"ogel and S.~Weinzierl,
\textit{``{Genus drop in hyperelliptic Feynman integrals}''},
\textsf{\doiref{10.1103/PhysRevD.109.L031901}{Phys.~Rev.~D~109,~L031901~(2024)}},
\texttt{\arxivref{2307.11497}{arxiv:2307.11497}}.

\bibitem{Duhr:2024uid}
C.~Duhr, F.~Porkert and S.~F.~Stawinski,
\textit{``{Canonical differential equations beyond genus one}''},
\textsf{\doiref{10.1007/JHEP02(2025)014}{JHEP~2502,~014~(2025)}},
\texttt{\arxivref{2412.02300}{arxiv:2412.02300}}.

\bibitem{DHoker:2002hof}
E.~D'Hoker and D.~H.~Phong,
\textit{``{Lectures on two loop superstrings}''},
\textsf{Conf.~Proc.~C~0208124,~85~(2002)},
\texttt{\arxivref{hep-th/0211111}{hep-th/0211111}}.

\bibitem{DHoker:2023khh}
E.~D'Hoker, M.~Hidding and O.~Schlotterer,
\textit{``{Cyclic Products of Higher-Genus Szeg{\"o} Kernels, Modular Tensors,
  and Polylogarithms}''},
\textsf{\doiref{10.1103/PhysRevLett.133.021602}{Phys.~Rev.~Lett.~133,~021602~(2024)}},
\texttt{\arxivref{2308.05044}{arxiv:2308.05044}}.

\bibitem{DHoker:2025jgb}
E.~D'Hoker and O.~Schlotterer,
\textit{``{Worldsheet fermion correlators, modular tensors and higher genus
  integration kernels}''},
\texttt{\arxivref{2505.07947}{arxiv:2505.07947}}.

\bibitem{mumford1984tata}
D.~Mumford,
\textit{``Tata Lectures on Theta I, II''},
Birkh{\"a}user Boston, MA (1983, 1984).

\end{thebibliography}

\end{document}